\documentclass[12pt]{article}

\usepackage{graphicx}
\usepackage{multirow}%
\usepackage{amsmath,amssymb,amsfonts,amsthm}%
\usepackage{mathrsfs}%
\usepackage{xcolor}%
\usepackage{textcomp}%
\usepackage{manyfoot}%
\usepackage{booktabs}%
\usepackage{algorithm}%
\usepackage{algorithmicx}%
\usepackage{algpseudocode}%
\usepackage{listings}%
\usepackage{float}

\usepackage[letterpaper,margin=1in]{geometry}

\newtheorem{theorem}{Theorem}%
\newtheorem{proposition}[theorem]{Proposition}%

\newtheorem{remark}{Remark}%
\newtheorem{assumption}{Assumption}
\newtheorem{definition}{Definition}%
\newtheorem{lemma}{Lemma}

\newcommand{\brac}[1]{\left(#1\right)}
\newcommand{\abs}[1]{\left\vert#1\right\vert}
\newcommand{\norm}[1]{\left\Vert#1\right\Vert}
\newcommand{\vect}[1]{\boldsymbol{#1}}
\newcommand{\Mod}[1]{\ (\mathrm{mod}\ #1)}

\renewenvironment{abstract}
	{\quotation}
	{\endquotation}

\date{}

\makeatletter
\renewcommand{\fnum@figure}{\textbf{Figure \thefigure}}
\renewcommand{\fnum@table}{\textbf{Table \thetable}}
\makeatother

\usepackage{scicite}

\usepackage{url}

\def\scititle{
	Topological Mechanics of Entangled Networks
}

\title{\Large\bfseries \boldmath \scititle}

\author{
	Juntao~Huang$^{1\ast\dagger}$,
	Jiabin~Liu$^{2\dagger}$,
	Shaoting~Lin$^{2\ast}$\and
	\small$^{1}$Department of Mathematical Sciences, University of Delaware, Newark, DE, USA.\and
	\small$^{2}$Department of Mechanical Engineering, Michigan State University, East Lansing, MI, USA.\and
	\small$^\ast$Corresponding authors. Email: huangjt@udel.edu; linshaot@msu.edu.\and
	\small$^\dagger$These authors contributed equally to this work.
}

\begin{document} 

\maketitle

\begin{abstract} \bfseries \boldmath
  Entangled networks are ubiquitous in tissues, polymers, and fabrics. However, their mechanics remain insufficiently understood due to the complexity of the topological constraints at the network level. Here, we develop a mathematical framework that models entangled networks as graphs, capturing topological constraints of entanglements. We prove that entanglements reduce system energy by enabling uniform tension along chains crossing entanglements and by redistributing stress through sliding. Under this framework, we study elasticity and fracture, validated by experiments on entangled fabrics and hydrogels. 
  For elasticity, entanglements increase strength by enabling stress homogeneity in the network.
  For fracture, entanglements enhance toughness by mitigating stress concentration around crack tips. We discover counterintuitive physical laws governing crack-tip stretch during crack opening: stress deconcentration at small deformation, constitutive-law independence at intermediate deformation, and linear scaling at large deformation. This framework establishes fundamental principles of linking topology to mechanics of entangled networks and offers a foundational tool for designing reconfigurable materials.
\end{abstract}

Entangled networks, ubiquitous in biological \cite{schulz2025second} and physical \cite{wang2024trapped} systems, play an important role in governing fundamental processes in life sciences and engineering applications. In contrast to classical spring networks, entangled networks are made of chains interconnected by entanglements as topological constraints that can slide along chains while maintaining intersections between adjacent chains. Recent experimental findings highlight the crucial role of entanglements in enabling extreme mechanical properties of tissues \cite{schulz2025second}, gels \cite{kim2021fracture, fu2023cartilage, zheng2022fracture}, elastomers \cite{wang2025self}, and textiles \cite{zhou20253d, singal2024programming} across multiple length scales. 
At the molecular scale, introducing a high density of molecular entanglements or entanglement-like slide rings into hydrogels signficiantly enhances the toughness from $\sim$ 10 J/m$^2$ to $\sim$ 3000 J/m$^2$ with negligible hysteresis \cite{liu2021tough,kim2021fracture}. The high density of entanglements also resolves the stiffness-toughness conflict \cite{kim2021fracture}, an intrinsic limitation for conventional cross-linked polymer networks \cite{lake1967strength}. 
At the macroscopic scale, incorporating entanglement-like concatenated rings into textile architectures leads to a new class of metamaterials that exhibit non-Newtonian fluid behaviors \cite{zhou20253d}, unattainable in conventional architected materials. Despite the crucial role of entanglements in enabling extreme mechanical properties of polymer networks and textile architectures, the fundamental mechanisms that govern the large deformation of entangled networks remain poorly understood. Notably, there currently exists no comprehensive mathematical model linking the topological structure of entangled networks to their mechanical properties.

Over the past few decades, numerous models have been proposed to capture the mechanics of entangled networks in polymers.
At the microscopic scale, molecular dynamics provide detailed representations of chain motions in the molecular level, but remain computationally intractable for predicting bulk mechanical behavior \cite{ge2013molecular, gula2020computational, gusev2024molecular}.
At the macroscopic scale, constitutive relations have been developed to describe stress-strain behaviors arising from entanglements, yet these models are often phenomenological and  not able to explicitly capture the topological constraints \cite{rubinstein2002elasticity, davidson2013nonaffine, liu2024tuning}. 
Mesoscopic models, lying between microscopic detail and macroscopic phenomenology, aim to balance between computational cost and physical fidelity. 
Among them, the tube model \cite{edwards1967statistical,de1971reptation,doi1978dynamics,doi1988theory,edwards1988tube}, which represents entangled polymers as a single chain confined within a tube formed by neighboring constraints and undergoing reptation along the tube, has been successful in describing rheological behavior. However, this model neglects topological features at the network level. 
Building on similar spirit, the slip-link model \cite{hua1998segment,doi2003molecular} approximates entanglement by discrete sliding constraints along the chain, and has been extended to multiple chains \cite{masubuchi2001brownian,ramirez2013theoretically,schneider2023entanglements,lamont2023micromechanics} and networks \cite{riggleman2009entanglement,chappa2012translationally, assadi2025nonaffine}. However, they still cannot systematically capture how the network topology governs bulk mechanical behavior. 
Moreover, most existing models are limited to predicting small-strain behaviors such as rheology but fail to capture large-deformation behaviors such as fracture.
To address these limitations, we develop a discrete network model, equipped with rigorous mathematical analysis and an efficient and robust optimization algorithm, to study the topological mechanics of entangled networks, particularly focusing on how the fraction of entanglements and their geometric distribution govern elasticity and fracture at the network scale. This framework will serve as a foundational tool for tissue mechanics, polymer physics, to textile manufacturing.

\subsection*{Modeling entangled networks}

Inspired by the pervasive entanglements in biological and physical systems spanning from tissues and gels to cotton and fabrics across multiple length scales (Fig.~1A), we introduce a graph-based representation of entangled networks (Fig.~1B). This formulation captures the topological constraints of the entanglements, allowing entanglements to slide along chains while maintaining the crossings. The central idea is to introduce an additional degree of freedom to represent the entanglement, ensuring that the entangled chains always share a common intersection point under any deformation.
\begin{definition}\label{def:main-entangled-network}
  An entangled network is an ordered pair $G_e = (V, C)$ consisting of:
  \begin{itemize}
    \item A set of nodes $V=\{ v_i, \, i = 1, \dots, N \}$, where $N$ is the total number of nodes.
    
    \item A set of chains $C=\{ c_k = ( v_{i_{k,1}}, v_{i_{k,2}}, \dots, v_{i_{k,n_k}} ), \, k = 1,\dots, M \}$, where each chain $c_k$ is an ordered tuple of nodes, $M$ is the total number of chains, and $n_k\ge 2$ is the number of nodes in the $k$-th chain.
  \end{itemize}
\end{definition}
\noindent Based on the definition of the entangled network, we define the non-slidable nodes as the endpoints (e.g. $v_{i_{k,1}}$ and $v_{i_{k,n_k}}$) and the slidable nodes (or called entanglements) as the internal nodes (e.g. $v_{i_{k,2}}, \dots, v_{i_{k,n_k-1}}$) in the chain $c_k = (v_{i_{k,1}}, v_{i_{k,2}}, \dots, v_{i_{k,n_k}})$. To characterize how the chains interact with each other through entanglements as the slidable nodes, we define the orientation of the slidable node as the directions of the chains that pass through it, see Definition \ref{def-supp:orientation-slidable-node}. 

We highlight the difference between our entangled network and the classical spring network. A spring network can be viewed as a graph, in terms of its topological structure, consisting of nodes and chains (or called edges), where each chain connects exactly two nodes (fig.~\ref{fig-supp:graph-hypergraph-entangled-network}a). In contrast, an entangled network allows each chain to pass through multiple nodes, forming an ordered tuple of nodes (fig.~\ref{fig-supp:graph-hypergraph-entangled-network}c).
From a topological perspective, a spring network is fully characterized by its node set and the edge set connecting the nodes, whereas an entangled network is more complex and depends on by four components: (1) the set of slidable nodes; (2) the set of non-slidable nodes; (3) the adjacency relations among all nodes; (4) the orientation associated with each slidable node; see details in Section \ref{sec:graph-entangle-network}. These components provide a systematic basis to tune the topology of entangled networks and explore its impact on the mechanics.

Since our interest lies in the mechanical behaviors of the entangled network, we extend the definition to include the length and force of each chain. In the undeformed state, each node $v_i\in V$ is assigned an initial position $\vect{X}_{i}\in\mathbb{R}^d$. Then, for each adjacent pair of nodes $v_i$ and $v_j$, we define its initial distance $L_{i,j} = \norm{\vect{X}_{i} - \vect{X}_{j}}$. Moreover, we assume that all chains follow the same constitutive law, described by a force-stretch function $f = f(\Lambda)$, where $\Lambda$ denotes the stretch ratio of a single chain and $f$ represents the stretching force along the chain.

We determine the equilibrium state of the entangled network from the minimization of its elastic energy. To this end, we first define the elastic energy of a single chain in the network. For a chain $c_k = (v_{i_{k,1}}, v_{i_{k,2}}, \cdots, v_{i_{k,n_k}})\in C$, its elastic energy is given by
\begin{equation}
  U_{k} = L_{k} \int_{1}^{\Lambda_{k}} f(\Lambda) \, d\Lambda.
\end{equation}
Here $L_{k}$ is the initial length of chain $c_k$, computed as the sum of the length for each segement $L_{k} = \sum_{j=1}^{n_k-1} L_{i_{k,j},i_{k,j+1}}$. In the deformed state, the position of the node $v_i$ is denoted by $\vect{x}_{i}$. The stretch ratio of chain $c_k$ is defined as $\Lambda_k = {l_k}/{L_{k}}$ where $l_k$ is its deformed length, given by $l_{k} = \sum_{j=1}^{n_k-1} l_{i_{k,j},i_{k,j+1}}$ with $l_{i_{k,j},i_{k,j+1}} = \norm{\vect{x}_{i_{k,j}} - \vect{x}_{i_{k,j+1}}}$. Here we assume that there is no friction at the entanglements, so the tension is uniform across all segments of a chain crossing entanglements. The total elastic energy of the network is obtained as the sum of the elastic energies of all chains:
\begin{equation}\label{eq:main-total-energy-entangled-network}
  U_e = \sum_{k=1}^M U_{k}.
\end{equation}
This total energy $U_e$ is a function of the positions of all nodes $\vect{x} = (\vect{x}_1, \vect{x}_2, \cdots, \vect{x}_N)\in\mathbb{R}^{dN}$. The equilibrium state of the entangled network is the state $\vect{x}_e^*$ that minimizes the total elastic energy $U_e(\vect{x})$, with the minimum elastic energy denoted by $U_e^* = U_e(\vect{x}_e^*)$. 

We analyze the theoretical properties of the energy function $U_e(\vect{x})$ defined in Eq.~\eqref{eq:main-total-energy-entangled-network}. Assuming that the constitutive law $f=f(\Lambda)$ is an increasing and non-negative function, we prove that the energy function $U_e(\vect{x})$ is a convex function of $\vect{x}$ (see Theorems \ref{thm:convexity-energy-entangled-network} and \ref{thm:convexity-energy-entangled-network-displacement}). This convexity implies that the local minimum of the energy function is also the global minimum and provides guidance on the algorithm design.

For a given entangled network $G_e = (V, C)$, we introduce its reference spring network $G_s = (V, C_s)$ by replacing all the slidable nodes with non-slidable nodes (see Definition~\ref{def-supp:reference-non-slidable-network} and fig.~\ref{fig-supp:entangleVSreference}).
The elastic energy of the spring network $U_s = U_s(\vect{x})$ can be defined in the same way as that of the entangled network. The equilibrium state of the spring network is the configuration $\vect{x}_s^*$ that minimizes the total elastic energy $U_s(\vect{x})$, with the minimum energy $U_s^* = U_s(\vect{x}_s^*)$.

We prove that, the elastic energy of the entangled network is always no larger than that of the reference spring network with the same node positions (Fig.~1C), see the theorem below. Mathematically, it means that the energy surface of the entangled network always lies below that of its reference spring network. Moreover, this result holds for entangled networks of arbitrary size and topology, making it a fundamental principle governing their mechanics.
\begin{theorem}\label{thm:main-theorem-energy-order}
  Let $G_e = (V, C)$ be an entangled network  and $G_s = (V, C_s)$ its reference spring network. Assume that the nodes in both networks have the same initial position in the undeformed state and the constitutive law $f=f(\Lambda)$ is an increasing function. Then, the total energy of the entangled network is no larger than that of the spring network with the same node positions in the deformed state:
  \begin{equation}\label{eq:main-theorem-energy-order}
    U_e(\vect{x}) \le U_{s}(\vect{x}), \quad \forall \, \vect{x}\in \mathbb{R}^{dN}.
  \end{equation}
\end{theorem}
\noindent We further investigate the physical mechanisms underlying this energy reduction. By taking $\vect{x}$ in Eq.~\eqref{eq:main-theorem-energy-order} at the equilibrium state $\vect{x}_s^*$ of the reference spring network, we obtain 
\begin{equation}
  U_e(\vect{x}_s^*) \le U_s^*.
\end{equation}
Physically, this implies that the system energy can be reduced simply by replacing some non-slidable nodes with slidable ones, even without allowing the slidable nodes to slide (Fig.~1D). This reduction is due to uniform tension across all segments of a chain crossing entanglements in the entangled network.
Next, we consider the relaxation of the entangled network from $\vect{x}_s^*$ to its own equilibrium state $\vect{x}_e^*$ and obtain 
\begin{equation}
  U_e^* \le U_e(\vect{x}_s^*).
\end{equation}
This reduction is due to stress redistribution enabled by the sliding of entanglements (Fig.~1E). The above conclusions are visualized in Fig.~1C.

\section*{Two-chain model}

To gain further insight into the underlying physics, we study a toy model consisting of two entangled chains. In the setup, two chains intersect in a two-dimensional space with the intersection being either a non-slidable node, representing a spring network (Fig.~2A), or a slidable node, representing an entangled network (Fig.~2C). Our in-situ photoelastic experiments \cite{liu2024fatigue} (fig.~\ref{fig-supp:SI_Fig1_SingleFiber}) enable direct visualization of stress distribution within individual chains, as the two-chain models are subjected to a directional displacement load and reach their respective equilibrium states (Fig.~2, B and D, and figs.~\ref{fig-supp:SI_Fig2_Nonslidable_TwoFiber_Exp_Sim} and \ref{fig-supp:SI_Fig3_SlidableTwoFiber_Exp_Sim}). 
As shown in Fig.~2E and fig.~\ref{fig-supp:SI_Fig4_SlidableTwoFiber_Exp_Sim}, the entangled network stores less elastic energy than the spring network, quantitatively validating the theoretical analysis in  Fig.~1C.
The entangled network also displays distinct behaviors in the entanglement motion and energy distribution. The entanglement initially slides along the chains and eventually becomes locked when it reaches the terminal end of one chain (Fig.~2, C and D). The energy remains relatively balanced between two chains during sliding, but significantly diverges once the entanglement locks (Fig.~2F).
From a topological perspective, locking transforms the system from entangled network with two chains to spring network with three chains, thereby altering the network topology (fig.~\ref{fig-supp:topologychange}). This further suggests that, under large deformations, entangled networks may be analyzed through topologically equivalent networks with fewer entanglements.

We plot the energy of both systems as a function of the intersection node position under the same displacement. Both energy surfaces are convex with a unique minimum, and the energy surface of the entangled network lies below that of the spring network (Fig.~2G and fig.~\ref{fig-supp:SI_Fig7_EnergylandscapeOfExperiment}), quantitatively validating our theoretical analysis in Theorem~\ref{thm:main-theorem-energy-order} and Fig.~1C. Interestingly, we observe that the energy surface of the entangled network has a singularity at the minimum, which is continuous but not differentiable, highlighted more clearly in the 1D slices (Fig.~2H and fig.~\ref{fig-supp:SI_Fig7_EnergylandscapeOfExperiment}). We further analyze this singular behavior and prove the existence of a transition point where the energy function switches from smooth to non-smooth as the displacement increases, see Theorem \ref{thm-supp:non-smoothness-energy-two-chain}.

The non-smoothness of the energy function in the entangled network presents a significant challenge for optimization: Newton's method fails to converge for the slidable system (Fig.~2I). To address this, we adopt a gradient descent method which robustly converges to the equilibrium state (Fig.~2I and fig.~\ref{fig-supp:SI_Fig6_Slidable_linear_simulation}). 
Furthermore, to extend the algorithm from the two-chain toy model to entangled networks of arbitrary topology and arbitrary number of nodes, we further resolve several key issues to balance accuracy and computational cost. Specifically, we (i) develop a new data structure to efficiently represent the entangled network; (ii) incorporate gradient descent with momentum, which leverages accumulated gradient information to accelerate convergence and suppress oscillations; (iii) apply a step-size decay strategy, which allows rapid progress in the early stage and stable refinement near equilibrium; (iv) select the initial guess based on the underlying mechanical problem to improve robustness; (v) introduce a new stopping criterion that achieves a practical balance between accuracy and computational cost; and (vi) design an efficient method for computing the gradient of the network energy. Further details are provided in Section~\ref{sec:numerical-method}. 

For the network topology, we focus on two-dimensional entangled networks arranged in a square lattice. The square lattice is selected over other periodic geometries (e.g., triangular or honeycomb lattices) because each node connects to exactly four neighbors. This connectivity allows each node to potentially act as a slidable node shared by two chains (Proposition \ref{prop:slidable-node-four-adjacent}). Our framework can also be extended to other lattice geometries, random networks, and three dimensions. Within the square lattice, the topology of an entangled network is fully characterized by the sets of slidable and non-slidable nodes, together with the orientation associated with each slidable node. In this study, for simplicity, we assume a uniform orientation for all nodes (fig.~\ref{fig-supp:orientation}) and vary only the sets of slidable and non-slidable nodes. Specifically, we consider two types of networks: periodic lattice networks generated from a unit cell, and random lattice networks where each node is randomly designated as slidable or non-slidable according to a prescribed probability.

\section*{Elasticity of entangled networks}

We apply the entangled network model to investigate the elasticity behavior. We begin with the periodic lattice network generated from a unit cell of size $2\times2$ and consider four periodic configurations with slidable node fractions $\varphi_s = 0\%$ (spring network), 25\%, 50\%, and $75\%$ (Fig.~3A). 
To characterize the strength of the network at failure, we introduce a failure criterion for each single chain: each chain breaks at a force $f_f$ corresponding to a stretch ratio $\Lambda_f$, with the associated energy required to rupture a chain in unit length denoted by $U_{\textrm{chain}}$.
The displacement boundary conditions are applied to the top and bottom, which quasi-statically stretch the network from the undeformed state until chain failure occurs, see the implementation details in Section \ref{sec:setup-entangle-network-elasticity}.

Our simulations show that, under a network stretch ratio $\lambda$, vertically oriented chains in the spring network ($\varphi_s=0\%$) stretch by $\lambda$, while horizontally oriented chains remain unstretched (Fig.~3B and fig.~\ref{fig-supp:SI_Fig8_0_Period_Network}). In contrast, the entangled network with $\varphi_s = 50\%$ shows homogenized deformation, with all chains stretched by $\Lambda = (\lambda + 1)/2$ due to stress redistribution enabled by slidable nodes (Fig.~3B and fig.~\ref{fig-supp:SI_Fig10_50_Period_Network}).
The case of $\varphi_s = 75\%$ behaves similarly (fig.~\ref{fig-supp:SI_Fig11_75_Period_Network}), while $\varphi_s = 25\%$ shows a non-uniform deformation pattern combining features of the spring network and entangled network with $\varphi_s = 50\%$ (fig.~\ref{fig-supp:SI_Fig9_25_Period_Network}). 
We compute the total energy of the network $U$, normalized by the total energy of the spring network at failure, as a function of the network stretch ratio $\lambda$.
As shown in Fig.~3C, the normalized total network energy at a fixed stretch ratio decreases with increasing slidable node fraction up to 50\% and then saturates. The maximum energy at failure varies non-monotonically: at low densities ($\varphi_s=25\%$), non-uniform energy distribution reduces failure energy, whereas at $\varphi_s=50\%$, homogenized chain deformations lead to a two-fold increase compared to the spring network.
By differentiating the network total energy $U$ with respect to the network stretch ratio $\lambda$, we calculate the nominal stress applied on the network, normalized by the nominal strength of the spring network (Fig.~3D). The modulus decreases with increasing $\varphi_s$, while the nominal strength shows a similar non-monotonic dependence on $\varphi_s$ (fig.~\ref{fig-supp:SI_Fig12_SimulationOfPeriodNetwork}). 
Finally, we derive analytical expressions for the total energy, nominal force, and modulus in Section \ref{sec:analytical-solution-elasticity}, and validate the simulation results.

We further perform in-situ photoelastic characterizations on hydrogel fabrics to experimentally study the role of entanglements as slidable nodes. The distribution of chain stretch ratios in fabrics with and without entanglements are visualized in Fig. 3E, figs.~\ref{fig-supp:SI_Fig13_Elasticity_Experiment_Nonslidable} and \ref{fig-supp:SI_Fig15_Elasticity_Experiment_50_slidable}. In both experiment and simulation, fabrics without entanglements exhibit a bimodal chain stretch ratio distribution, with one peak corresponding to load-bearing chains aligned with the stretch direction and the other peak corresponding to load-free chains perpendicular to the stretch direction (Fig.~3F, fig.~\ref{fig-supp:SI_Fig14_DistributionOfNonslidableNetwork}). In contrast, fabrics with 50\% entanglements exhibit a single peak, resulting from entanglement sliding that promotes uniform distribution of chain stretch ratios across the network (Fig.~3F, fig.~\ref{fig-supp:SI_Fig16_DistributionOfSlidableNetwork}). 

We proceed to simulate the random entangled network, where the slidable and non-slidable nodes are randomly distributed in the square lattice according to a given probability $0<\varphi_s<1$ (Fig.~3G and Section \ref{sec:data-structure-entangle-network}). 
In contrast to periodic networks, random networks exhibit a more scattered distribution of chain stretch ratios for $\varphi_s=50\%$ (Fig.~3H and fig.~\ref{fig-supp:SI_Fig18_50_Random_Network}). 
This behavior arises because the disorder in the node arrangement disrupts uniform stress redistribution between vertical and horizontal chains and leads to heterogeneous stretching across the network. Moreover, under large deformation, some slidable nodes become locked, preventing further sliding and amplifying the heterogeneity (Fig.~3H and fig.~\ref{fig-supp:SI_Fig20_10_50_99_networkstretching}).
As the slidable node fraction increases to $\varphi_s = 90\%$ (highly entangled regime), the likelihood of locking is greatly reduced, leading to a uniform stress distribution and enhanced strength (Fig. 3H and fig.~\ref{fig-supp:SI_Fig19_99_Random_Network}). To quantify these observations, we plot the fraction of locked slidable nodes among all slidable nodes (Fig.~3I). The fraction of locked nodes decreases with increasing $\varphi_s$, and no locking occurs in the highly entangled regime ($\varphi_s=90\%$).
Finally, we examine the chain stretch ratios during the stretching process (Fig.~3J). In the periodic network with $\varphi_s = 50\%$, all chains show a uniform linear stretch, following $\Lambda = (\lambda + 1)/2$. In contrast, the random network with $\varphi_s = 50\%$ shows a scattered distribution of chain stretch ratios due to heterogeneity.

We also plot the nominal stress during stretching (Fig.~3K) and the corresponding strength (Fig.~3L). Despite the randomness, the stress-stretch curves of random networks remain remarkably consistent across different samples for each fixed slidable node fraction (fig.~\ref{fig-supp:SI_Fig24_Simulation_ForceEnergyCurve}). Similar to periodic networks, the modulus of random networks decreases with increasing $\varphi_s$ and the strength shows a non-monotonic trend: it first decreases with increasing $\varphi_s$ due to heterogeneity, and then increases as the system approaches the highly entangled regime (fig.~\ref{fig-supp:SI_Fig25_Energy_ModulusOfRandomNetwork}).
Experiments are also conducted to validate the simulation results (Fig.~3, M and N). 
We synthesize a series of hydrogels with systematically tuned entanglements to cross-links ratios to investigate the role of entanglements in real material systems. The ratio is adjusted by maintaining a constant monomer concentration while varying the cross-link concentration. As the monomer to cross-linker ratio $M$ increases from 500 to 50,000, the entangled hydrogels exhibit a pronounced reduction in modulus (Fig.~3M and fig.~\ref{fig-supp:SI_Fig26_HydrogelExperimentData}) and a significant enhancement in strength (Fig. 3N  and fig.~\ref{fig-supp:SI_Fig26_HydrogelExperimentData}), which qualitatively align with our simulation results.

\section*{Fracture of entangled networks}

One key interest is to apply our model to quantify the role of entanglements on fracture of entangled networks, an emerging topic that remains poorly understood. We first consider the periodic network as in the elasticity simulation, but with a crack in the middle of the network, called the notched sample (Fig.~4A and fig.~\ref{fig-supp:SI_Fig27_CrackNetworkConfiguration}). We consider both linear and nonlinear constitutive laws of a single chain, with the breaking force and stretch denoted by $f_f$ and $\Lambda_f$  (Fig.~4B).

To compute the intrinsic fracture energy $\Gamma_0$, we conduct the standard pure shear test \cite{long2016fracture, deng2023nonlocal} in two steps: (i) measure the nominal stress versus stretch curve $S(\lambda)$ of an unnotched sample; (ii) load a notched sample until the chain around crack tip breaks (called crack initiation) to determine the critical stretch of the sample $\lambda_c$. Then $\Gamma_0$ is calculated as $\Gamma_0 = h_0 \int_{1}^{\lambda_c} S \, d\lambda$ where $h_0$ is the initial height of the sample. 
Directly simulating $\lambda_c$ by incrementally loading an entangled network until crack initiation is computationally expensive due to the stress concentration in fracture. To improve efficiency, we instead compute $\lambda_c$ using a bisection method (see Section \ref{sec:bisection-method-critical-stretch} for details).

We plot the distributions of chain stretch ratios at the crack initiation, for both linear and nonlinear constitutive laws and for different slidable node fractions $\varphi_s = 0\%$, $25\%$, $50\%$, and $75\%$ (Fig.~4C, figs.~\ref{fig-supp:SI_Fig28_Undeformed_20_80} to \ref{fig-supp:SI_Fig32_Deformed_100_400}).
As $\varphi_s$ increases, more chains near the crack tip experience large stretch ratios, and the far-field stretch ratio also increases; both effects are further amplified in nonlinear networks compared to linear ones.
To quantify the energy distribution near the crack tip, we define $U_n$ as the average energy per unit chain length in the $n$-th layer away from the crack tip, and $U_\infty$ as the corresponding far-field value. The ratio $U_n / U_\infty$ shows that both nonlinearity and entanglements reduce energy concentration around crack tip (Fig.~4D).
We further compare $\lambda_c$ (Fig.~4E) and $\Gamma_0$ (Fig.~4F) for both linear and nonlinear periodic networks. Both quantities increase with $\varphi_s$, indicating that entanglement can enhance toughness.
Together, these simulations provide the first quantitative evidence that entanglements can mitigate stress concentration and enable extreme toughening, an effect observed experimentally in soft materials but lacking quantitative foundations \cite{kim2021fracture, steck2023multiscale}.

Since the crack initiation is governed by when the crack-tip chain reaches its breaking threshold, it is essential to study the crack-tip chain stretch ratio, denoted by $\Lambda_{\textrm{tip}}$, as the crack opens. Our simulations reveal three distinct regimes of counterintuitive physical laws governing the crack-tip stretch.
(1) Small deformation regime (Fig.~4G, figs.~\ref{fig-supp:SI_Fig35_Nonslidable_linearlaw} and~\ref{fig-supp:SI_Fig36_Slidable_linearlaw}). In spring networks, $\Lambda_{\textrm{tip}}$ is nearly identical to the far-field stretch $\Lambda_{\infty}$, showing no stress concentration.
In entangled networks, $\Lambda_{\textrm{tip}}$ is even lower than $\Lambda_{\infty}$, revealing a stress deconcentration around the crack tip due to the entanglements, a counterintuitive phenomenon not previously reported or explained. See also the energy distribution near the crack tip in fig.~\ref{fig-supp:SI_Fig37_EnergySmallDeform}.
(2) Intermediate deformation regime (Fig.~4H, figs.~\ref{fig-supp:SI_Fig35_Nonslidable_linearlaw} and~\ref{fig-supp:SI_Fig36_Slidable_linearlaw}). The values of $\Lambda_{\textrm{tip}}$ for both linear and nonlinear constitutive laws coincide, indicating that $\Lambda_{\textrm{tip}}$ is independent of the specific form of the force-stretch relationship of chains and determined solely by the topology of the network. This behavior arises from the fact that any smooth nonlinear constitutive law can be locally approximated by a linear function through Taylor expansion (see Section \ref{sec:three-regime-crack-opening}).
(3) Large deformation regime (Fig.~4I, figs.~\ref{fig-supp:SI_Fig35_Nonslidable_linearlaw} and~\ref{fig-supp:SI_Fig36_Slidable_linearlaw}). Surprisingly, we observe that $\Lambda_{\textrm{tip}}$ increases linearly with the network stretch ratio $\lambda$, regardless of the constitutive law (linear or nonlinear) or network type (spring network or entangled network). Even more surprisingly, this linear scaling turns out to be closely related to the scaling law for the spring network reported by \cite{hartquist2025scaling}, which states that at large deformation the intrinsic fracture energy satisfies $\Gamma_0 = \alpha l_0 f_f (\Lambda_f - 1)$, where $\alpha$ is a constant determined by the network topology, $l_0$ is the unit length of the network.
Assuming a monomial form of the constitutive law $f(\Lambda) = (\Lambda - 1)^p$ with $p>0$, we prove that $\Lambda_{\textrm{tip}}$ linearly increases with $\lambda$ with the slope
\begin{equation}\label{eq:main-slope-large-deformation}
  k_s = \brac{\frac{h_0}{\alpha l_0 (p+1)}}^{\frac{1}{p+1}}.
\end{equation}
This implies that stronger nonlinearity (i.e., larger $p$) leads to a smaller increasing slope, providing a quantitative description of how nonlinearity mitigates stress concentration at the crack tip. A similar derivation for periodic entangled networks with $\varphi_s = 50\%$ yields the slope:
\begin{equation}\label{eq:main-slope-large-deformation-entangle}
  k_e = \brac{\frac{h_0}{\alpha l_0 2^p(p+1)}}^{\frac{1}{p+1}}.
\end{equation}
The additional factor of $2^p$ in Eq.~\eqref{eq:main-slope-large-deformation-entangle} further reduces the slope, demonstrating that entangled networks are even more effective in mitigating stress concentration compared to spring networks. The detailed derivations of Eqs.~\eqref{eq:main-slope-large-deformation} and \eqref{eq:main-slope-large-deformation-entangle} are provided in Section \ref{sec:three-regime-crack-opening}.

To experimentally examine the stress distribution near the crack tip, we conduct in-situ photoelastic characterizations on notched hyrogel fabrics with and without entanglements. As shown in Fig.~4J, figs.~\ref{fig-supp:SI_Fig38_Nonslidable_linear_Fiber}-\ref{fig-supp:SI_Fig41_Slidable_nonlinear_Fiber}, the notched hydrogel fabric containing entanglements with $\varphi_s = 50\%$ exhibits a more uniform stress distribution. We also plot the crack-tip stretch as a function of the network stretch ratio (Fig.~4K and fig.~\ref{fig-supp:SI_Fig42_SmallDeform_Experiment}), which confirms stress deconcentration at small deformation and linear scaling at large deformation, consistent with our simulation findings.

To demonstrate the capability of our model in simulating more realistic systems, we investigate the fracture behavior of the random entangled networks, which are generated by randomly assigning slidable and non-slidable nodes in the square lattice with a given slidable node fraction $0 < \varphi_s < 1$ (Fig.~5A, see details in Section \ref{sec:data-structure-entangle-network}).
For each $\varphi_s$, we generate 100 independent random samples of networks. Fig.~5B shows the distributions of chain stretch ratios for the sample corresponding to the median value of $\Gamma_0$ for $\varphi_s = 10\%$, $50\%$, $90\%$, $99\%$, under both linear and nonlinear constitutive laws, at the crack initiation. We observe that the nonlinearity and entanglements mitigate stress concentration, consistent with our observations in periodic networks. The simulation results with other slidable node fraction ($30\%$ and $70\%$) and other lattice sizes are included in figs.~\ref{fig-supp:SI_Fig43_Random_10_20_undeformed} to~\ref{fig-supp:SI_Fig58_Random_99_100_deformed}.
Similar to our analysis on periodic networks, we compute the average energy per unit chain length in the $n$-th layer away from the crack tip, normalized by that in the far field (Fig.~5C). The plots confirm that both nonlinearity and entanglements help redistribute energy more uniformly around the crack tip and thus reduce stress concentration.

We further present the distributions of $\Gamma_0$ for all 100 samples under both linear and nonlinear constitutive laws in Fig.~5D and E, figs.~\ref{fig-supp:SI_Fig59_Critical_stretch_random_box_plot} and~\ref{fig-supp:SI_Fig60_Gamma0_random_box_plot}. We have several key observations: (1) in average, $\Gamma_0$ increases with the slidable node fraction $\varphi_s$; (2) random networks may result in $\Gamma_0$ values significantly higher than those of periodic networks with the same $\varphi_s$; (3) the distribution of $\Gamma_0$ becomes increasingly dispersed as $\varphi_s$ increases.
To investigate the origin of the dispersed behavior, we examine the dependence of $\Gamma_0$ on the crack-tip chain length, denoted by $l_{\textrm{tip}}$ (Fig.~5F, figs.~\ref{fig-supp:SI_Fig61_ChainlengthVSGamma0_60_linear} and \ref{fig-supp:SI_Fig62_ChainlengthVSGamma0_100}). We find that  $\Gamma_0$ generally increases with $l_{\textrm{tip}}$, but the growth slows down and eventually saturates for longer chains. Physically, this reflects two effects: (1) longer chains require more energy to break, resulting in larger $\Gamma_0$; (2) beyond a certain length, the crack tip chain alone is no longer dominant, and surrounding chains play an increasingly important role in stress redistribution and toughness.
We further analyze the probablity distribution of chain lengths at the crack tip. Assuming a sufficiently large network, we analytically derive the probability of the chain at the crack tip:
\begin{equation}
  P(l_{\textrm{tip}} = k) = 
  \begin{cases}
    (1-\varphi_s)^2, & k = 1, \\
    \varphi_s^{k-1} (1 - \varphi_s) (2 - \varphi_s), & k = 2, 3, \dots
  \end{cases}
\end{equation}
See the details in Section \ref{sec:prob-crack-tip-length-rand-network}. As shown in Fig.~5G, this theoretical prediction agrees closely with numerical results. As $\varphi_s$ increases, the distribution of $l_{\textrm{tip}}$ becomes more dispersed, which in turn explains the broader distribution of $\Gamma_0$ observed for larger $\varphi_s$ in Fig.~5E. We further perform fatigue characterizations to measure the intrinsic fracture energy $\Gamma_0$ of hydrogels with systematically tuned entanglement to cross-link ratios. As the monomer to cross-linker ratio increases from 500 to 50,000, the entangled hydrogels have a significant enhancement in $\Gamma_0$ (Fig.~5H, figs.~\ref{fig-supp:SI_Fig63_fatigue_20K} and \ref{fig-supp:SI_Fig64_fatigue_all}), qualitatively validating our simulation results (Figs. 5D and 5E).

\section*{Discussion}

In this work, we developed a mathematical framework for the mechanics of entangled networks, consisting of three key components: (1) a graph-based model that represents the topological constraints of entanglements; (2) a theoretical analysis that establishes the fundamental mathematical and physical properties of the model; (3) an optimization algorithm that enable efficient and robust solutions. Under this framework, we discovered new physical principles that govern the elasticity and fracture of entangled networks, validated by experiments on entangled fabrics and hydrogels.

For physics and mechanics, this framework goes beyond a new methodology. More importantly, it redefines how entangled networks can be systematically studied, by establishing topology as a fundamental determinant of their mechanics. This provides a unified basis for exploring the physical principles that govern diverse entangled systems across multiple scales.
For applied math, this work introduces a new class of graphs inspired by the topology of entangled networks, which could open avenues for exploring its fundamental properties within graph theory and related mathematical disciplines.
For material science, these physical insights lay the foundation for designing and optimizing reconfigurable materials, including biological tissues, active gels, and intelligent metamaterials.

While the present work focuses on equilibrium states and chains idealized as stretch-only strands, future studies will extend the framework to capture the dynamic behavior of entangled networks and to include bending effects for semi-flexible networks. Together, these efforts will broaden the applicability of the framework and deepen our understanding of the mechanics of entangled networks.

\newpage

\vspace*{0pt}
\begin{figure}[H]
  \centering
  \includegraphics[trim={0cm 0cm 0cm 0cm},clip, width=1.0\textwidth]{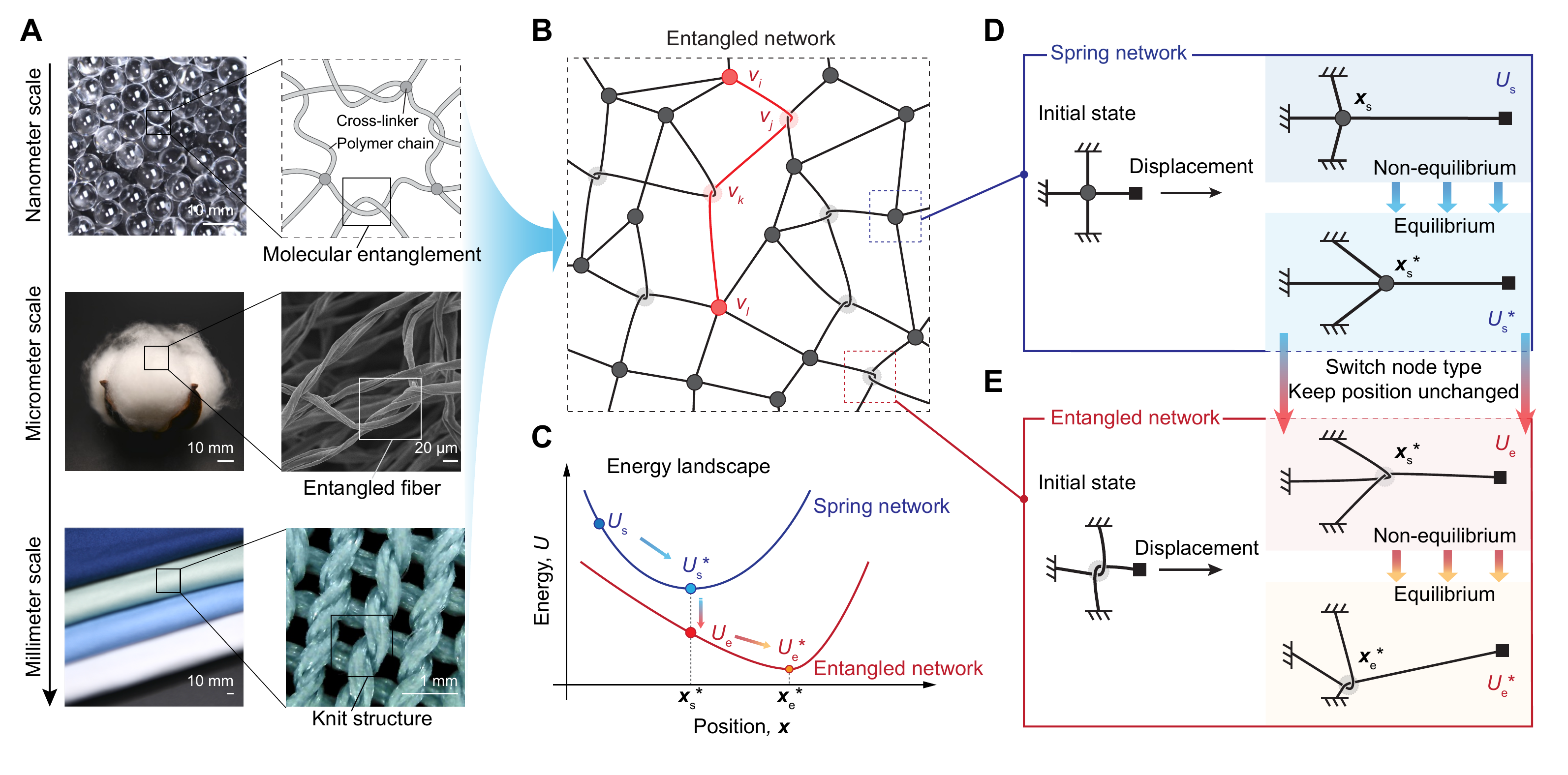}
\caption{
\textbf{Graph-based formulation of entangled networks.}
\textbf{(A)} Examples of entangled network across multiple length scales: molecular entanglements in hydrogels (nanometer scale), entangled fibers in cotton (micrometer scale), and knitted structures in fabrics (millimeter scale).
\textbf{(B)} Graph-based representation of an entangled network. Chain endpoints are non-slidable nodes in black color, while internal nodes are slidable nodes (entanglements) in grey color. 
\textbf{(C)} Comparison of the total elastic energy surfaces of an entangled network and its reference spring network, showing that the entangled network always has lower energy.
\textbf{(D-E)} Illustration of energy reduction mechanisms of entangled networks. In the spring network, the system reduces energy by relaxing from a non-equilibrium to an equilibrium state (shown in (C) from $U_s$ to $U_s^*$). For the equilibrium state of the spring network, replacing non-slidable nodes with slidable nodes further decreases the energy (shown in (C) from $U_s^*$ to $U_e$), due to uniform tension along chains crossing entanglements. Allowing slidable nodes to slide enables the relaxation, leading to lower equilibrium energy in the entangled network (shown in (C) from $U_e$ to $U_e^*$).
}
  \label{fig-main-1}
\end{figure}

\vspace*{0pt}
\begin{figure}[H]
  \centering
  \includegraphics[trim={0cm 0cm 0cm 0cm},clip, width=1.0\textwidth]{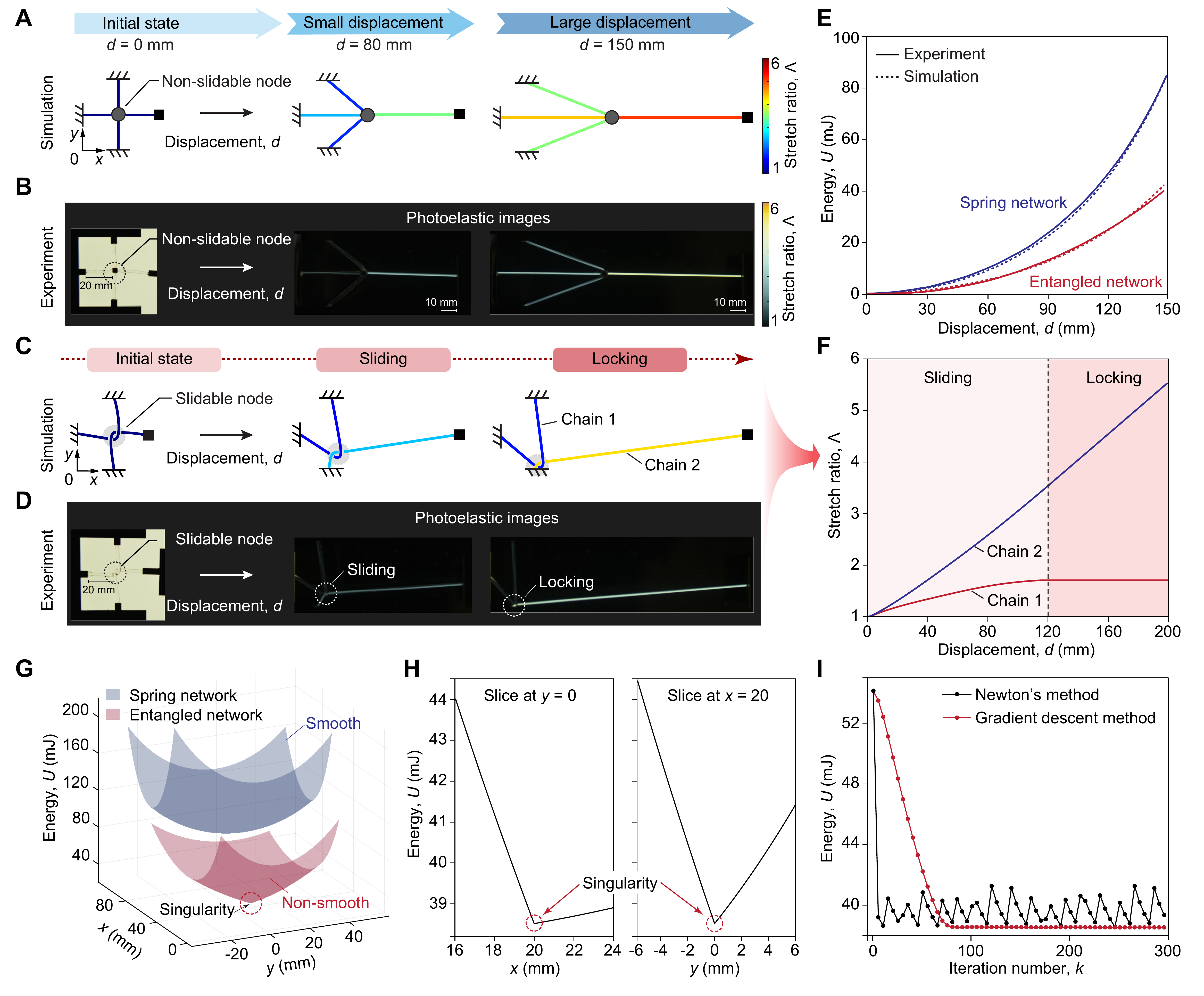}
\caption{See caption on the next page.}
  \label{fig-main-2}
\end{figure}

\newpage

\noindent\textbf{Figure 2: Two-chain model.}
\textbf{(A-D)} Simulation and photoelastic experiment visualization of two-chain model with either a non-slidable node, representing a spring network {(A, B)} or a slidable node, representing an entangled network {(C, D)}, subjected to directional displacement. In the entangled network, the entanglement slides along the chains and eventually locks once it reaches the terminal end of a chain.
\textbf{(E)} Comparison of elastic energy between simulation and experiment for entangled network and spring network, showing quantitative agreement. The entangled network stores less elastic energy than the spring network, validating the theoretical prediction.
\textbf{(F)} Energy distribution between the two chains in the entangled network. The energy remains relatively balanced between chains during sliding, but diverges once the entanglement locks.
\textbf{(G)} Energy surfaces of the two systems as functions of intersection node position, demonstrating convex surfaces with a unique minimum. The surface of the entangled network lies below that of the spring network and exhibits a singularity at the minimum.
\textbf{(H)} One-dimensional slices of the energy surface of entangled network at $y=0$ and $x=20$, highlighting the singularity at the minimum, where the energy function is continuous but non-differentiable.
\textbf{(I)} Optimization algorithms for solving the energy minimization problem of the entangled network. Newton's method exhibits oscillations and fails to converge due to non-smoothness of the energy function, while gradient descent converges to the equilibrium state.

\newpage
\vspace*{0pt}
\begin{figure}[H]
  \centering
  \includegraphics[trim={0cm 0cm 0cm 0cm},clip, width=1.0\textwidth]{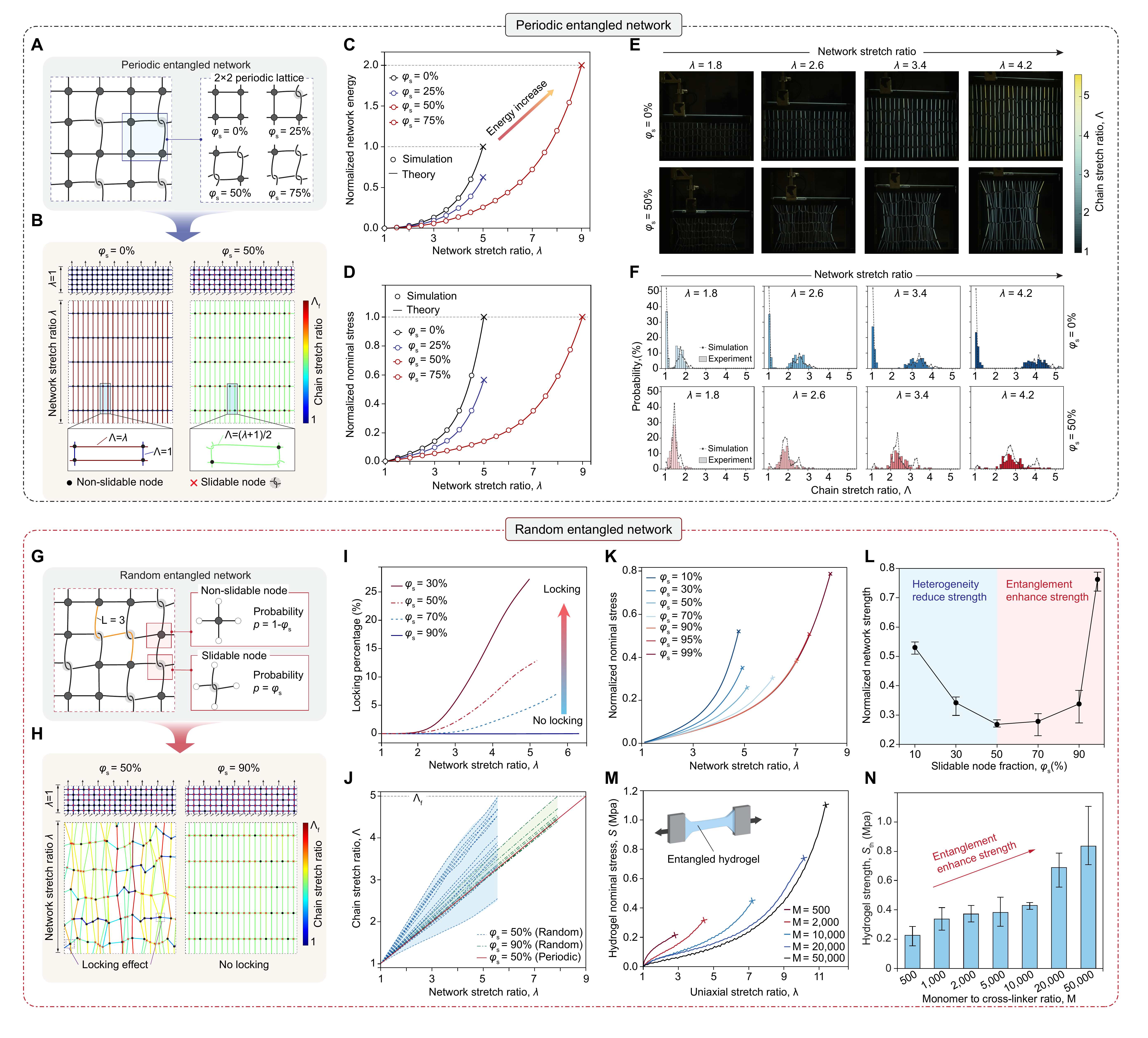}
\caption{See caption on the next page.}
  \label{fig-main-3}
\end{figure}

\newpage

\noindent\textbf{Figure 3. Elasticity of entangled networks.}
\textbf{(A)} Periodic entangle network configurations with slidable node fractions $\varphi_s=0\%$ (spring network), $25\%$, $50\%$, and $75\%$.
\textbf{(B)} Simulations of periodic entangled networks under uniaxial tension with network stretch ratio $\lambda$ in the vertical direction. For the spring network ($\varphi_s=0\%$), vertical chains stretched by $\lambda$ while horizontal chains remain unstretched. For the entangled network ($\varphi_s=50\%$), all chains stretch by $(\lambda+1)/2$ due to stress redistribution enabled by slidable nodes.
\textbf{(C)} Total network energy, normalized by the total energy of the spring network at failure, as a function of network stretch ratio $\lambda$, showing decreasing energy at a fix stretch ratio with increasing $\varphi_s$ and a non-monotonic energy at failure. 
\textbf{(D)} Nominal stress-stretch curves normalized by the strength of the spring network. The modulus decreases with $\varphi_s$, while strength shows non-monotonic dependence.
\textbf{(E-F)} In-situ photoelastic characterization of hydrogel fabrics. Fabrics without entanglements exhibit a bimodal chain stretch ratio distribution, whereas fabrics with $\varphi_s=50\%$ show a single peak due to homogenized stress redistribution. Simulations  confirm this transition.
\textbf{(G)} Random entangle network configurations with a given fraction of slidable node $0< \varphi_s < 1$.
\textbf{(H)} Simulations of random networks under uniaxial tension with $\varphi_s=50\%$ and $90\%$. Scattered stretch ratio distributions at $\varphi_s=50\%$, with some slidable nodes locked. At $\varphi_s=90\%$, locking is suppressed and stress distribution becomes more uniform.
\textbf{(I)} Fraction of locked slidable nodes  decreases with increasing entanglement and vanishes in the highly entangled regime ($\varphi_s=90\%$).
\textbf{(J)} Chain stretch ratios during stretching: periodic network ($\varphi_s=50\%$) shows a uniform linear stretch for all chains, following $\Lambda = (\lambda + 1)/2$, while random network ($\varphi_s=50\%$) exhibits scattered distribution of chain stretch ratios due to heterogeneity.
\textbf{(K)} Nominal force-stretch curves of random networks with different slidable node fractions.
\textbf{(L)} Nominal strength of random networks with different slidable node fractions, shows a non-monotonic trend, the strength first decrease with $\varphi_s$ due to heterogeneity and increase again in the highly entangled regime.
\textbf{(M)} Experimental modulus of hydrogels, showing modulus reduction with increasing entanglement density, consistent with simulations.
\textbf{(N)} Experimental strength of hydrogels, showing significant strengthening with increasing entanglement density, consistent with simulations.

\newpage
\vspace*{0pt}
\begin{figure}[H]
  \centering
  \includegraphics[trim={0cm 1.5cm 0cm 1cm},clip, width=1.0\textwidth]{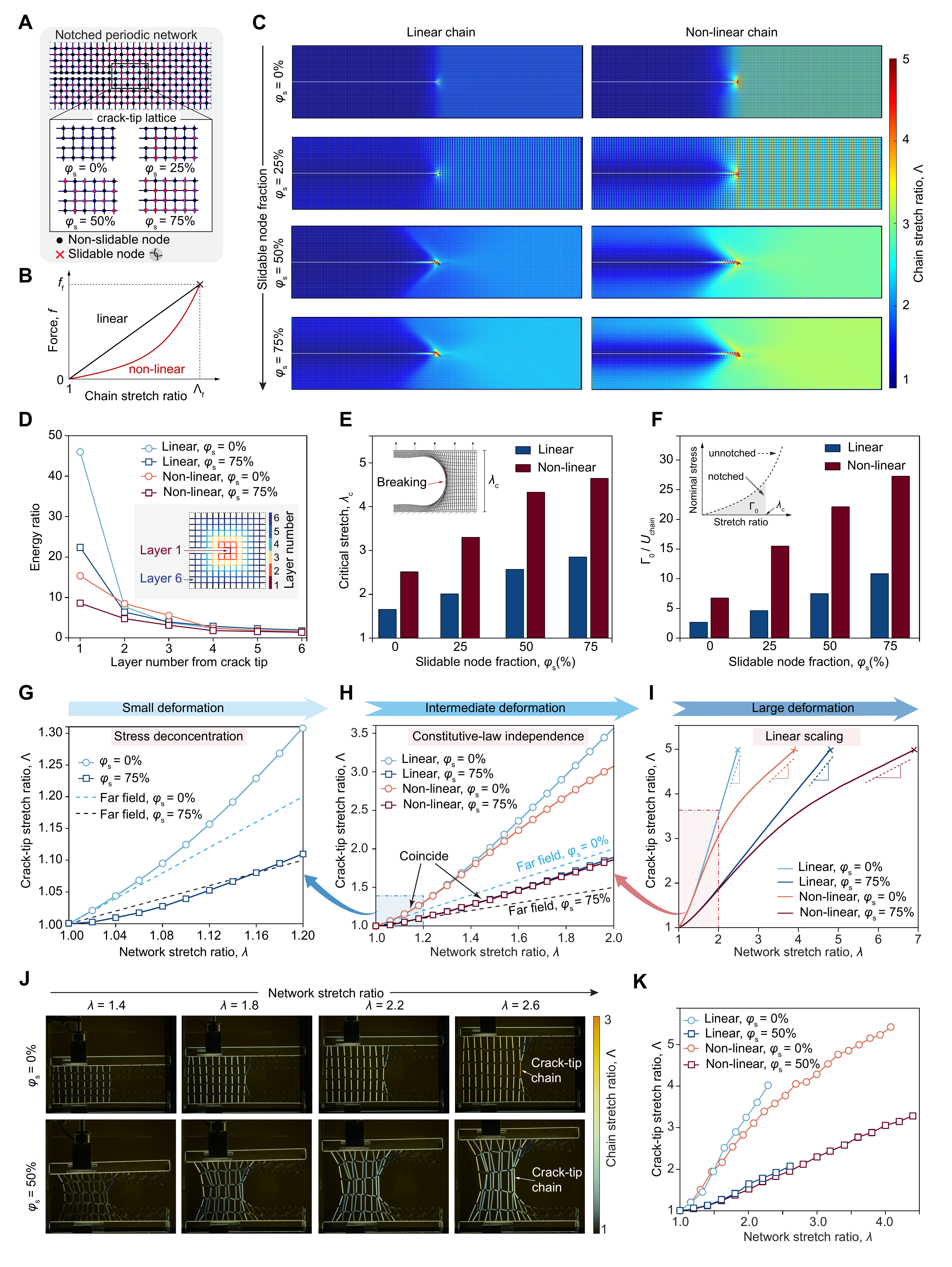}
\caption{See caption on the next page.}
  \label{fig-main-4}
\end{figure}

\newpage

\noindent\textbf{Figure 4. Fracture of periodic entangled networks.}
\textbf{(A)} Periodic entangle network configurations with crack with slidable node fractions $\varphi_s = 0\%$ (spring network), $25\%$, $50\%$, and $75\%$.
\textbf{(B)} Constitutive laws of a single chain, including both linear and nonlinear cases, with breaking force $f_f$ and breaking stretch $\Lambda_f$.
\textbf{(C)} Distributions of chain stretch ratios at crack initiation for entangled networks with different slidable node fractions, under both linear and nonlinear constitutive laws. 
\textbf{(D)} Energy distribution near the crack tip quantified by the ratio $U_n/U_\infty$, where $U_n$ is the average energy per unit chain length in the $n$-th layer away from the crack tip, and $U_{\infty}$ is the corresponding far-field value, showing that both nonlinearity and entanglements reduce stress concentration.
\textbf{(E)} Critical stretch $\lambda_c$ at crack initiation, obtained by loading a notched sample until the chain at the crack tip breaks.
\textbf{(F)} Intrinsic fracture energy $\Gamma_0$, computed by loading an unnotched sample to obtain the nominal stress-stretch curve $S(\lambda)$ and evaluating $\Gamma_0 = h_0 \int_{1}^{\lambda_c} S \, d\lambda$, where $h_0$ is the initial height of the sample.
\textbf{(G-I)} Three regimes of crack-tip stretch $\Lambda_{\textrm{tip}}$ as crack opens: (G) small deformation, where spring networks show no stress concentration ($\Lambda_{\textrm{tip}} \approx \Lambda_\infty$) and entangled networks exhibit stress deconcentration ($\Lambda_{\textrm{tip}} < \Lambda_\infty$); (H) intermediate deformation, where $\Lambda_{\textrm{tip}}$ is independent of constitutive law and determined by network topology; (I) large deformation, where $\Lambda_{\textrm{tip}}$ increases linearly with $\Lambda_{\infty}$.
\textbf{(J)} In-situ photoelastic characterizations measuring stress distribution of notched hydrogel fabrics with and without entanglements.
\textbf{(K)} Crack-tip stretch $\Lambda_{\textrm{tip}}$ versus network stretch ratio $\lambda$ from experiments, confirming stress-deconcentration at small deformation and linear scaling at large deformation.

\newpage
\vspace*{0pt}
\begin{figure}[H]
  \centering
  \includegraphics[trim={0cm 0cm 0cm 0cm},clip, width=1.0\textwidth]{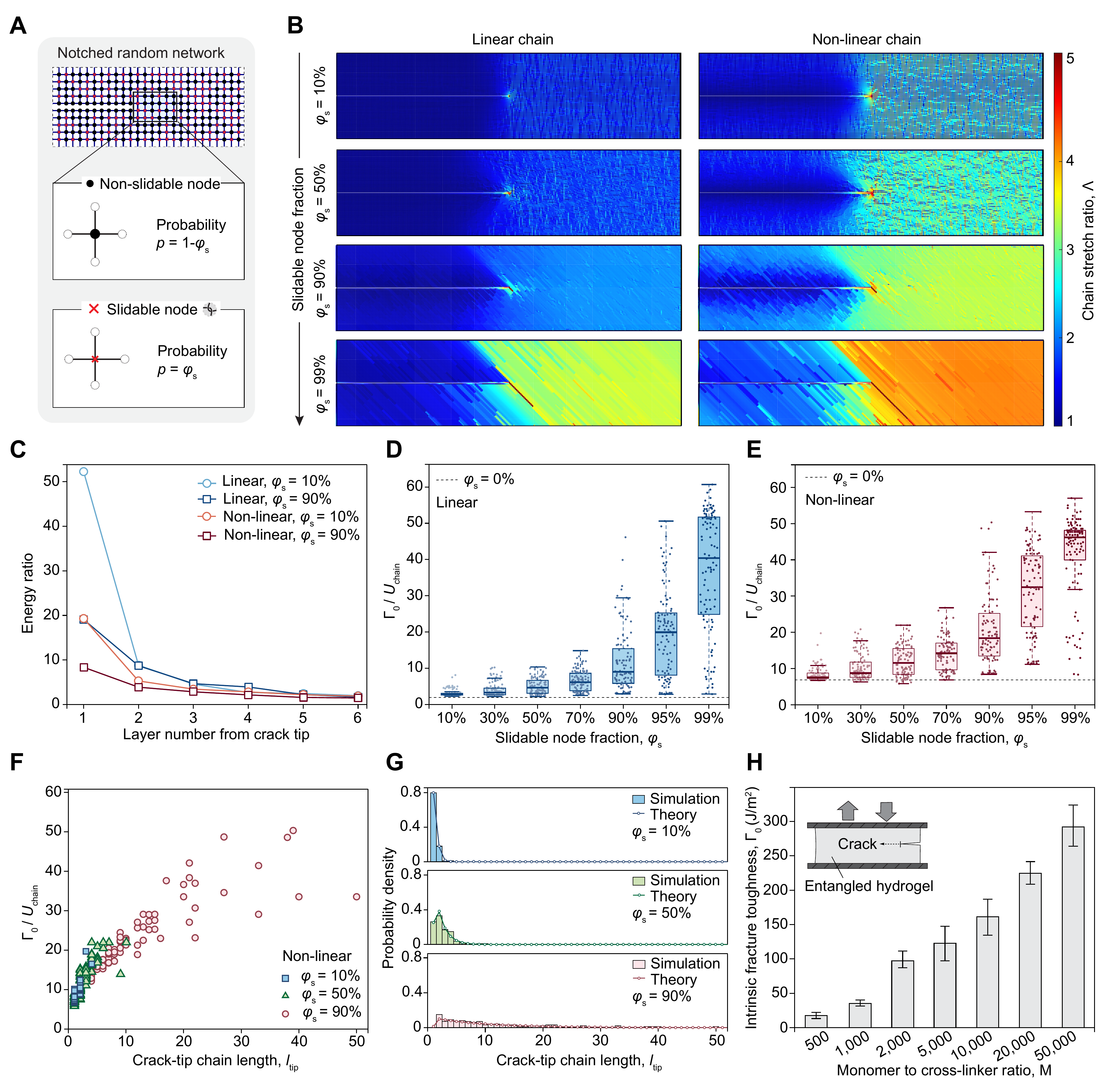}
\caption{See caption on the next page.}
  \label{fig-main-5}
\end{figure}

\newpage

\noindent\textbf{Figure 5. Fracture of random entangled networks.}
\textbf{(A)} Random entangled networks configurations with crack, generated by randomly assigning slidable and non-slidable nodes with a given slidable node fraction $\varphi_s$.
\textbf{(B)} Distributions of chain stretch ratios at crack initiation for representative random networks (median $\Gamma_0$) with $\varphi_s = 10\%$, $50\%$, $90\%$, $99\%$, under both linear and nonlinear constitutive laws.
\textbf{(C)} Energy distribution near the crack tip quantified by the ratio $U_n/U_\infty$, where $U_n$ as the average energy per unit chain length in the $n$-th layer away from the crack tip, and $U_{\infty}$ as the corresponding far-field value, showing that both nonlinearity and entanglements reduce stress concentration.
\textbf{(D-E)} Distributions of intrinsic fracture energy $\Gamma_0$ from 100 random realizations for each $\varphi_s$, under linear (D) and nonlinear (E) constitutive laws. On average, $\Gamma_0$ increases with $\varphi_s$; random networks may yield significantly higher $\Gamma_0$ than periodic ones; and the spread of $\Gamma_0$ widens as $\varphi_s$ increases.
\textbf{(F)} Dependence of $\Gamma_0$ on the crack-tip chain length $l_{\textrm{tip}}$, showing increase followed by saturation.
\textbf{(G)} Probability distribution of $l_{\textrm{tip}}$, comparing theoretical analysis with numerical simulations. The distribution broadens as $\varphi_s$ increases, explaining the wider spread of $\Gamma_0$ as $\varphi_s$ increases.
\textbf{(H)} Measured intrinsic fracture energy of hydrogels with tuned monomer to cross-linker ratio $M$.

\clearpage %

\bibliography{ref} %
\bibliographystyle{sciencemag}

\section*{Acknowledgments}
J.H. would like to thank Andrew J. Christlieb at Michigan State University for providing HPC resources that supported the research results reported in this paper. We also thank John C. Crocker at University of Pennsylvania for insightful discussions on network topology changes after locking during the Gordon Research Conference (GRC) on Soft Condensed Matter Physics.
\paragraph*{Funding:}
J.H. was partially funded by the National Science Foundation (DMS-2309655) and the startup fund of the College of Arts and Sciences at University of Delaware. J.L. and S.L. was partially funded by the National Science Foundation (CMMI-2338747), the National Science Foundation (CMMI-2423067), and the startup fund of the College of Engineering at  Michigan State University.
\paragraph*{Author contributions:}
J.L. identified the problem. J.H. and J.L. conceived the idea. J.H. developed the model, performed the theoretical analysis, designed the algorithm, and implemented the code. J.L. conducted the experiments and prepared the figures. J.H., J.L., and S.L. interpreted the results and wrote the manuscript. J.H. and S.L. acquired the funding and supervised the project.
\paragraph*{Competing interests:}
There are no competing interests to declare.
\paragraph*{Data and materials availability:}
The code used for numerical simulations will be made publicly accessible upon publication of the paper.

\subsection*{Supplementary materials}
Materials and Methods\\
Supplementary Text\\
Figures S1 to S78\\
Movies S1 to S4

\newpage

\renewcommand{\thefigure}{S\arabic{figure}}
\renewcommand{\thetable}{S\arabic{table}}
\renewcommand{\theequation}{S\arabic{equation}}
\renewcommand{\thepage}{S\arabic{page}}
\renewcommand{\thetheorem}{S\arabic{theorem}}
\renewcommand{\thedefinition}{S\arabic{definition}}
\renewcommand{\thealgorithm}{S\arabic{algorithm}}
\renewcommand{\theremark}{S\arabic{remark}}
\renewcommand{\theassumption}{S\arabic{assumption}}
\renewcommand{\thelemma}{S\arabic{lemma}}
\renewcommand{\thecorollary}{S\arabic{corollary}}
\setcounter{figure}{0}
\setcounter{table}{0}
\setcounter{equation}{0}
\setcounter{theorem}{0}
\setcounter{definition}{0}
\setcounter{page}{1} %
\renewcommand{\thesection}{S\arabic{section}}
\renewcommand{\thesubsection}{S\arabic{section}.\arabic{subsection}}

\begin{center}
\section*{Supplementary Materials for\\ \scititle}

Juntao~Huang$^{\ast\dagger}$,
Jiabin~Liu$^{\dagger}$,
Shaoting~Lin$^{\ast}$
\\ %
\small$^\ast$Corresponding authors. Email: huangjt@udel.edu, linst@msu.edu\\
\small$^\dagger$These authors contributed equally to this work.
\end{center}

\subsubsection*{This PDF file includes:} 
Materials and Methods\\
Supplementary Text\\
Figures S1 to S78\\
References

\subsubsection*{Other Supplementary Material for this manuscript includes the following:} 
Movies S1 to S4

\newpage

\section*{Materials and Methods}

\textbf{Materials}

\indent Acrylamide (AAm, Sigma-Aldrich A8887), N,N$'$-Methylenebisacrylamide (MBAA, Sigma-Aldrich 146072), ammonium persulfate (APS, Sigma-Aldrich A3678), and N,N,N$'$,N$'$-tetra\-methyl\-ethylene\-diamine(TEMED, Sigma-Aldrich T9281) used in this work were purchased from Sigma-Aldrich and used without modification. Dichloromethane (UN1593) was purchased from Aqua Solution and used without modification. Transparent acrylic sheets (8560K\-191, 8560K171) and black acrylic sheets (8505K113) for hydrogel molds, as well as PVC tubes (5231K124) for synthesizing hydrogel fibers, were purchased from McMaster-Carr. The linear polarizing film (XP44-40) and $\lambda/4$ retarder film (WP140HE) used in the homemade photoelastic setup were purchased from Edmund Optics. For the homemade photoelasticity test setup, a white LED panel light (4000\,K color temperature) was purchased from eBay, and a Nikon-D800 camera was used to capture images. The silicone oil used in the fatigue test to prevent dehydration of the hydrogel was purchased from Amazon.

\vspace{1em} 
\noindent\textbf{Synthesis of Hydrogel Fiber}

\indent The synthesis of hydrogel fibers follows the reported protocol \cite{liu2023fatigue}. Specifically, the pre-gel solution of PAAm hydrogel was prepared by dissolving 5\,g of acrylamide (AAm) monomer in 5\,g of deionized water. A 0.1\,M ammonium persulfate (APS) solution was used as the thermal initiator, 0.23\,wt\% N,N$'$-methylenebisacrylamide (MBAA) as the crosslinker, and N,N,N$'$,N$'$-tetra\-methyl\-ethylene\-diamine (TEMED) as the crosslinking accelerator. To synthesize the linear fiber, which exhibits a linear relationship between stretch ratio and force before fracture, 200\,$\mu$L of MBAA, 150\,$\mu$L of APS, and 10\,$\mu$L of TEMED were added to 10\,g of the pre-gel PAAm solution. To synthesize the nonlinear fiber, which exhibits a nonlinear relationship between stretch ratio and force before fracture, 50\,$\mu$L of MBAA, 150\,$\mu$L of APS, and 10\,$\mu$L of TEMED were added to 10\,g of the pre-gel PAAm solution. Each mixture was vortexed for 2 minutes and then degassed in a vacuum chamber for another 2 minutes. The precursor solution was carefully injected into a PVC tube using a syringe to avoid air bubbles. The filled tubes were placed in a chamber at 50\,°C for 2 hours for thermal curing. After curing, the tube was soaked in a dichloromethane solution for approximately 10 seconds to facilitate detachment of the hydrogel from the PVC tube. The hydrogel fibers were then rinsed about 10 times with deionized water to remove residual dichloromethane from the surface. Finally, the fibers were immersed in deionized water overnight to reach swelling equilibrium. 

\vspace{1em} 
\noindent\textbf{Fabrication of Hydrogel Fabric Network}

\indent To prepare the hydrogel fabric networks, fully swollen hydrogel fibers (linear and non-linear) were cut into lengths of 10\,mm and 20\,mm. Black acrylic sheets were cut into small cubic pieces (2\,mm~$\times$~2\,mm~$\times$~3\,mm) using a laser cutter (Epilog Fusion Laser Cutter). For hydrogel fabric networks without slidable nodes, 10\,mm-long nonlinear hydrogel fibers were glued to black acrylic cubes using super glue (Krazy Glue), following the topology shown in Fig.~S25. For networks with slidable node fraction of $\varphi_s = 50\%$, both 10\,mm and 20\,mm-long nonlinear fibers were assembled using the same method and topology shown in Fig.~S27. For notched hydrogel fabric networks without slidable nodes, 10\,mm-long fibers were glued to black acrylic cubes using super glue, following the topology shown in Fig.~S52. The same topology was used for networks composed of both linear and nonlinear hydrogel fibers. For notched networks with slidable node fraction of $\varphi_s = 50\%$, 10\,mm and 20\,mm-long fibers were used, again following the same topology shown in Fig.~S53, for both linear and nonlinear fiber networks. After fabrication, the hydrogel fiber networks were stored in sealed bags to prevent dehydration.

\vspace{1em} 
\noindent\textbf{The Circular Polariscope}

\indent The circular polariscope is a photoelastic setup used to visualize stress distribution in birefringent materials \cite{liu2024fatigue}.  The setup, shown in Fig. S13, consists of an LED light source (4000\,K), two linear polarizers, two quarter-wave plates, a universal testing machine with hydrogel fibers or hydrogel fiber networks, and a camera. The first polarizer and quarter-wave plate (collectively referred to as the polarizer) are placed before the sample, while the second set (the analyzer) is placed after. The two linear polarizers are installed orthogonally to each other to create the darkest possible observed light field. The first quarter-wave plate is installed at a 45-degree angle, and the second quarter-wave plate is orthogonal to the first to darken the observed light field again. The natural white light from the LED panel is considered light with arbitrary transverse vibrations. The linear polarizer allows light with a specific direction to pass through, and the quarter-wave plate introduces a phase difference of 90 degrees between the orthogonal components of light.

\vspace{1em} 
\noindent\textbf{Mechanical Characterizations of Hydrogel Fabrics}

\indent We first tested the force–stretch relationship of individual hydrogel fibers. Both linear and nonlinear fibers were cut to a length of 20\,mm and subjected to uniaxial tensile loading until fracture. The tests were conducted under a custom-built circular polariscope, and the photoelastic images were captured by the camera at one-second intervals. The resulting force–stretch curves and corresponding images are presented in Fig.~S13.

\indent For the two-chain model with a non-slidable node, a piece of black acrylic cube was used to glue together four 10\,mm-long fibers. Three ends of the fibers were fixed using an acrylic board, and a tensile test was applied to the fourth point until hydrogel fiber breaks. Before mechanical testing, the network sample was briefly dipped
in water to minimize friction between the fibers. The force was recorded using the tensile testing machine, and photoelastic images were taken every second. For the two-chain model with a slidable node (entangled node), two 20\,mm-long hydrogel fibers were entangled at their midpoints. Again, three ends were fixed, and uniaxial tension was applied to the fourth point until the fiber broke. The tensile force and photoelastic images were recorded every second. 

\indent For hydrogel fabric network, before mechanical testing, the network sample was briefly dipped in water to minimize friction between the fibers. For tensile testing, each sample (both notched and unnotched) was stretched at a constant rate of 0.5\,mm/s until fracture of one hydrogel fiber occurred. Photoelastic images were recorded by the camera under the same circular polariscope at one-second intervals, with the force recorded by the mechanical testing machine.

\vspace{1em} 
\noindent\textbf{Synthesis of Entangled Hydrogels}

\indent To prepare the entangled hydrogel, a pre-gel solution of PAAm hydrogel was first prepared by dissolving 7.1\,g of acrylamide (AAm) monomer in 5.4\,mL of deionized water. APS was used as the thermal initiator, MBAA as the crosslinker, and TEMED as the crosslinking accelerator. To synthesize hydrogels with different crosslinking densities, MBAA was added to the pre-gel solution at monomer-to-crosslinker molar ratios ranging from 500 to 50{,}000. For each 10\,g of pre-gel solution, 40\,$\mu$L of APS and 5\,$\mu$L of TEMED were added. The mixture was vortexed for 4 minutes, then poured into an acrylic mold sealed with an acrylic cover. Both the mold and the cover were fabricated using a laser cutter. The filled molds were placed in an oven at 50\,°C for 3 hours for thermal curing.

\indent After synthesis, the hydrogels were immersed in deionized water until they reached swelling equilibrium. Due to differences in crosslinker density, the polymer content of fully swollen hydrogels varied, with lower crosslinker ratios leading to higher water uptake. To ensure a consistent polymer content across samples, the swollen hydrogels were partially dried under ambient humidity. During drying, each sample was flipped every half hour to ensure uniform dehydration, and drying continued until the polymer content reached 15\%. The hydrogels were then sealed in containers to prevent further water loss prior to testing.

\vspace{1em} 
\noindent\textbf{Mechanical Characterizations of Entangled Hydrogels}

\indent For mechanical testing, the hydrogels were fabricated in a dog-bone shape. To obtain the stress–stretch curve, the samples were subjected to monotonic uniaxial loading until fracture. For fatigue testing, dog-bone-shaped hydrogel samples (with an length to width ratio of 1:4) were used. First, an unnotched sample was stretched monotonically until failure, with force and displacement recorded throughout the test. The energy release rate was calculated using the equation \( G = H \int_1^{\lambda} S \, d\lambda \), where \( H \) is the initial height of the sample, \( S \) is the nominal stress, and \( \lambda \) is the stretch ratio during the pure shear test of the unnotched sample. Next, a rectangular hydrogel sample of the same dimensions was prepared, with an initial crack (approximately one-fifth of the sample length) introduced along the horizontal centerline using a razor blade. The cracked sample was subjected to cyclic loading under different stretch ratios. To prevent dehydration during testing, the sample was immersed in a silicone oil bath. The crack length was measured every 100 cycles. To determine the crack extension rate per cycle (\( dc/dN_c \)), the total crack growth was divided by the number of cycles, assuming a linear relationship between crack extension and cycle number. By linearly extrapolating the crack extension rate per cycle ($dc/dN_c$) versus the energy release rate ($G$) to the intercept with the abscissa, we can experimentally measure the intrinsic fracture energy $\Gamma_0$.

\newpage
\section*{Supplementary Text}

\tableofcontents

\section{Modeling entangled networks}\label{sec:entangled-network}

In this part, we introduce our entangled network model. In Section \ref{sec:graph-entangle-network}, we first present the graph-based representation of entangled networks, which captures the topological structure of the network and the interactions between the chains through the entanglement. Then, in Section \ref{sec:energy-entangle-network}, we introduce the definition of the total elastic energy for the entangled network and define the equilibrium state as the solution to the corresponding energy minimization problem.

\subsection{Graph-based representation of entangled networks}\label{sec:graph-entangle-network}

In this section, we begin with the definition of the entangled network, which is motivated by the molecular entanglement of polymer chains. This definition captures the topological structure of the network and the interactions between the chains through the entanglements. To develop a general framework, we draw upon concepts from graph theory. For the reader's convenience, we restate below the same definition as in Definition~\ref{def:main-entangled-network} from the main text. We also note that, throughout this work, a graph refers to an undirected, unweighted graph.
\begin{definition}[Entangled network]\label{def:supp-entangled-network}
  An entangled network is an ordered pair $G_e = (V, C)$ consisting of:
  \begin{itemize}
    \item A set of nodes (also called vertices) $V=\{ v_i, \, i = 1, \dots, N \}$, where $N$ is the total number of nodes.
    
    \item A set of chains (also called edges) $C=\{ c_k = ( v_{i_{k,1}}, v_{i_{k,2}}, \dots, v_{i_{k,n_k}} ), \, k = 1,\dots, M \}$, where each chain $c_k$ is an ordered tuple of nodes, $M$ is the total number of chains, and $n_k\ge 2$ is the number of nodes in the $k$-th chain.
  \end{itemize}
\end{definition}

\begin{figure}
  \centering
  \includegraphics[trim={0cm 0cm 0cm 0},clip, width=1.0\textwidth]{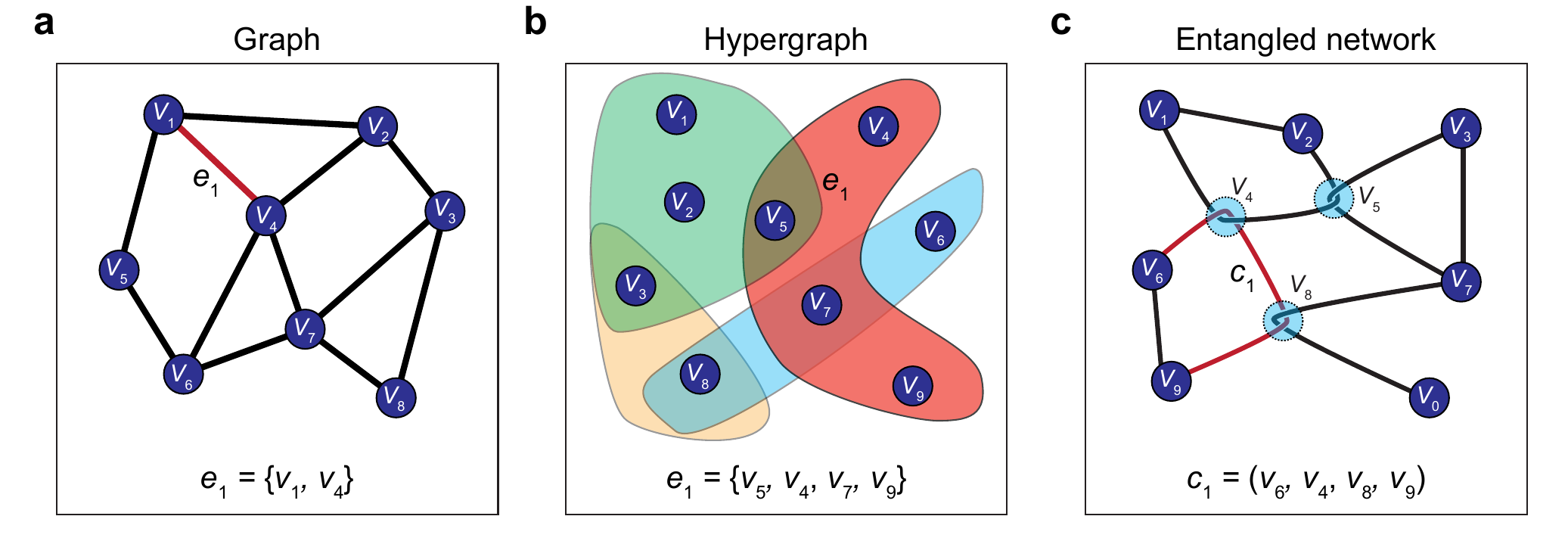}
  \caption{\textbf{Examples of graph, hypergraph, and entangled network.}
  \textbf{(a)} A classical graph consists of vertices and edges, where each edge connects only two vertices (e.g., $e_1=\{v_1, v_4\}$).
  \textbf{(b)} A hypergraph is a generalization of the graph, where each edge is an unordered set that can include arbitrary number of vertices (e.g., $e_1=\{v_5, v_4, v_7, v_9\}$).
  \textbf{(c)} An entangled network consists of nodes and chains, where each chain is an ordered tuple can include multiple nodes (e.g., $c_1=(v_6, v_4, v_8, v_9)$).
  }
  \label{fig-supp:graph-hypergraph-entangled-network}
\end{figure}

\begin{definition}[Slidable node and non-slidable node]
  In an entangled network $G_e = (V, C)$, for a chain $c_k = ( v_{i_{k,1}}, v_{i_{k,2}}, \dots, v_{i_{k,n_k}} )\in C$, we call the endpoints $v_{i_{k,1}}$ and $v_{i_{k,n_k}}$ the non-slidable nodes of the chain $c_k$, and call the remaining internal nodes $v_{i_{k,2}}, \dots, v_{i_{k,n_k-1}}$ the slidable nodes of the chain $c_k$.
\end{definition}
In Figure \ref{fig-supp:graph-hypergraph-entangled-network}(c), non-slidable nodes are marked by the dark blue color and the slidable nodes are marked by the light blue color. 
In this work, we also refer to a slidable node as an entangled node or simply an entanglement.

\begin{definition}[Adjacent node]
  In an entangled network $G_e = (V, C)$, two nodes $v_i, v_j\in V$ are called adjacent nodes if there exists a chain $c_k\in C$ such that $v_i$ and $v_j$ are consecutive along the chain $c_k$.
\end{definition} 

To construct a more realistic model of the entangled network and simplify subsequent analysis, we make the following assumption. This assumption reflects typical features observed in physical networks, where non-slidable nodes typically correspond to permanent network junctions as in classical spring networks, while slidable nodes represent entanglements.
\begin{assumption}\label{assumption-supp:internal-node-non-slidable}
  In an entangled network $G_e = (V, C)$, any node $v_i\in V$ must be consistently assigned as either slidable or non-slidable across all chains it belongs to. Specifically, if $v_i$ is a non-slidable node of one chain, it cannot serve as a slidable node in any other chains, and vice versa.
\end{assumption}

In the most general case, an entangled node may be formed at the intersection of more than two chains. However, interactions involving three or more chains are topologically complex \cite{patil2020topological} and difficult to analyze. To simplify the discussion, we restrict ourselves to entangled nodes that are connected by exactly two chains. 
Moreover, while \cite{assadi2025nonaffine} considers the entanglement creation and disentanglement, we restrict our analysis to a fixed network topology without such processes.
\begin{assumption}\label{assumption-supp:internal-node-two-chains}
  In an entangled network $G_e = (V, C)$, any slidable node belongs to exactly two chains.
\end{assumption}

To prevent overlap between any two chains in the network, we make the following assumption:
\begin{assumption}\label{assumption-supp:non-overlapping-chains}
  In an entangled network $G_e = (V, C)$, each pair of adjacent nodes is connected by only one chain.
\end{assumption}

\begin{proposition}\label{prop-supp:internal-node-four-adjacent-vertices}\label{prop:slidable-node-four-adjacent}
  In an entangled network $G_e = (V, C)$, any slidable node has four adjacent nodes: two from one chain and another two from another chain.
\end{proposition}
\begin{proof}
  By Assumption \ref{assumption-supp:internal-node-two-chains}, any slidable node $v_i\in V$ belongs to exactly two chains. Since $v_i$ is a slidable node, each of these two chains contributes two vertices adjacent to $v_i$; otherwise, $v_i$ is a non-slidable node. Thus, $v_i$ has four adjacent vertices in total. These four vertices must be distinct; otherwise, it will contradict Assumption \ref{assumption-supp:non-overlapping-chains}. 
\end{proof}

Based on the above proposition, we introduce the orientation of a slidable node to characterize how the chains intersect at that node. Specifically, the orientation of a slidable node describes the pairing of its adjacent vertices according to the two chains it belongs to.
\begin{definition}[Orientation of the slidable node]\label{def-supp:orientation-slidable-node}
  For an entangled network $G_e = (V, C)$, let the adjacent vertices of a slidable node $v_i\in V$ be denoted by $v_{i_1}, v_{i_2}, v_{i_3}, v_{i_4}$. If $v_i$, $v_{i_1}$, $v_{i_2}$ lie on one chain and $v_i$, $v_{i_3}$, $v_{i_4}$ belong to another, then the orientation of $v_i$ is defined as $\{\{v_{i_1}, v_{i_2}\}, \{v_{i_3}, v_{i_4}\}\}$ or equivalently $\{\{v_{i_3}, v_{i_4}\}, \{v_{i_1}, v_{i_2}\}\}$.
\end{definition}
For example, in Figure \ref{fig-supp:graph-hypergraph-entangled-network}c, the orientation of the slidable node $v_4$ is given by $\{\{v_1, v_5\}, \{v_6, v_8\}\}$. 

It is well-known that an undirected graph is uniquely determined by its set of nodes and the neighbors of each node. Our entangled network has a similar property but requires additional information beyond just the neighborhood structure. In particular, the entangled network is uniquely specified by the following components: (1) the set of non-slidable nodes; (2) the set of slidable nodes; (3) the set of adjacent nodes for each node; (4) the orientation of each slidable nodes.
This characterization is particularly useful for implementing algorithms to generate entangled networks, as we will describe in more details in Section~\ref{sec:numerical-method}.

To study the effects of the entanglements, we introduce the reference spring network, also known as classical spring network. This reference spring network is generated by replacing all the slidable nodes by non-slidable nodes. As a result, the reference spring network is a graph where each chain in the network consists of exactly two nodes. See, for example, Figure~\ref{fig-supp:entangleVSreference} for a visual comparison between an entangled network and its corresponding reference spring network. We give the formal mathematical definition as follows:
\begin{definition}[Reference spring network]\label{def-supp:reference-non-slidable-network}
  Given entangled network $G_e=(V, C)$, its reference spring network $G_s=(V, C_s)$ is defined as follows:
  \begin{itemize}
    \item The node set $V$ is the same as that of the entangled network $G_e$.
    
    \item The chain set $C_s$ consists of all unordered pairs of adjacent nodes in $G_e$. Each chain in $C_s$ thus consists of two nodes.
  \end{itemize}
\end{definition}
\begin{figure}
  \centering
  \includegraphics[width=0.7\textwidth]{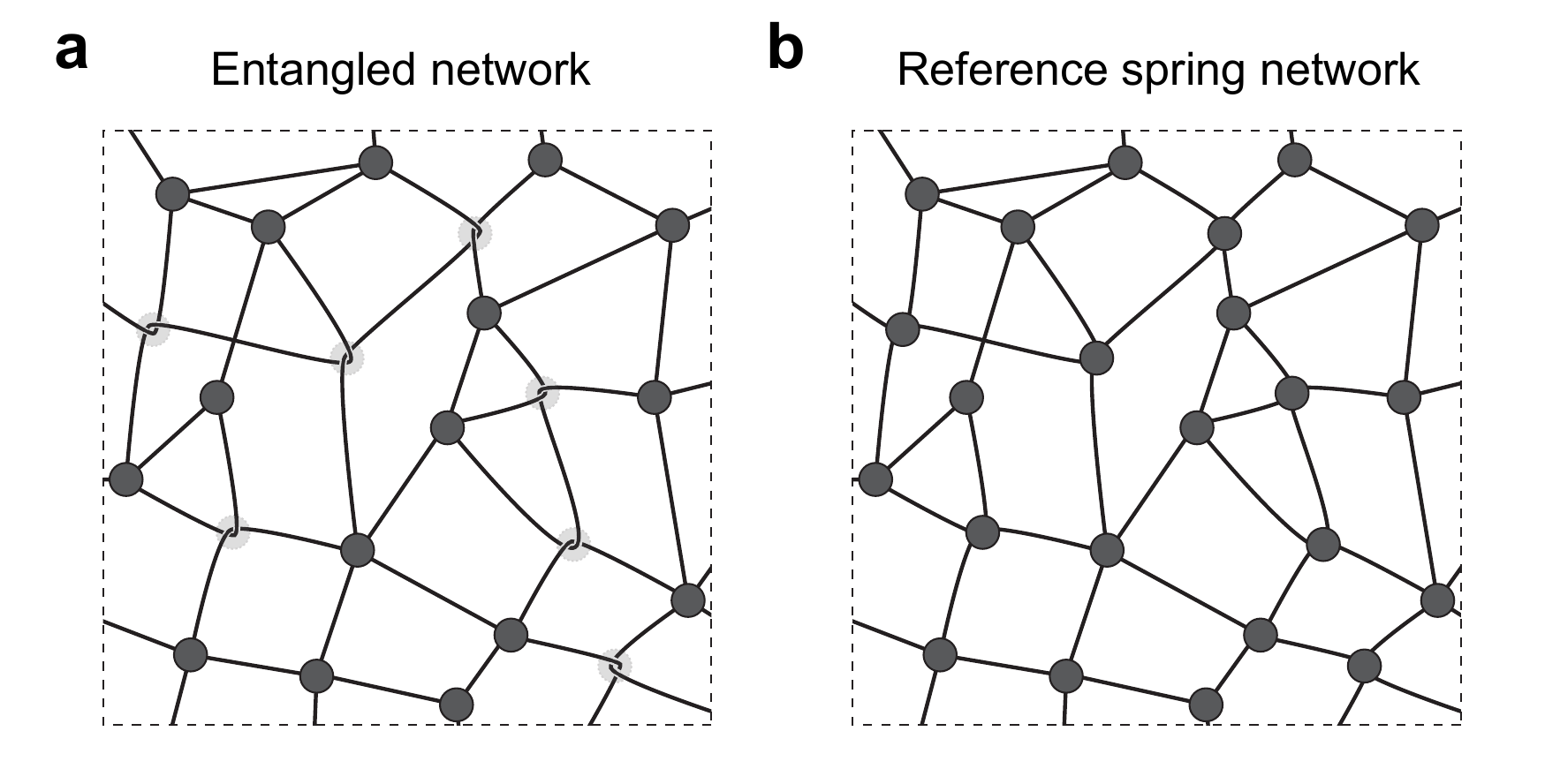}
  \caption{\textbf{An entangled network and its reference spring network.} \textbf{(a)} An entangled network. The non-slidable nodes are marked by the black color and the slidable nodes are marked by the gray color. \textbf{(b)} Its reference spring network is generated by replacing all the slidable nodes by non-slidable nodes.}
  \label{fig-supp:entangleVSreference}
\end{figure}

\subsection{Energy of entangled networks and the optimization problem }\label{sec:energy-entangle-network}

Since our primary interest lies in the mechanical properties of the entangled network, we extend the definition to include the length and force associated with each chain, and accordingly introduce the definition of the total energy. The equilibrium state is then defined as the solution to the  energy minimization problem.

For an entangled network $G_e = (V, C)$, for each node $v_i\in V$, we assign its initial position $\vect{X}_{i}\in\mathbb{R}^d$ in the undeformed state. Then, for each adjacent pair of nodes $v_i$ and $v_j$, we assign the initial length of the segment $L_{i,j} = \norm{\vect{X}_{i} - \vect{X}_{j}}$ which is the Euclidean distance between the two nodes.
We assume that all chains follow the same constitutive law given by the force-stretch function $f = f(\Lambda)$ where $\Lambda$ is the stretch ratio of the chain and $f$ is the force.

To find out the equilibrium state of the entangled network under some prescribed boundary conditions, we first define the energy of a single chain in the entangled network. For a chain $c_k = (v_{i_{k,1}}, v_{i_{k,2}}, \cdots, v_{i_{k,n_k}})\in C$, its energy is given by 
\begin{equation}
  U_{k} = L_{k} \int_{1}^{\Lambda_{k}} f(\Lambda) \, d\Lambda.
\end{equation}
Here $L_{k}$ is the initial length of the chain $c_k$, computed as the sum of the length for each segement: 
\begin{equation}
  L_{k} = \sum_{j=1}^{n_k-1} L_{i_{k,j},i_{k,j+1}}
\end{equation}
We denote the position of the node $v_i$ in the deformed state as $\vect{x}_{i}$. The stretch ratio of chain $c_k$ is defined as $\Lambda_k = {l_k}/{L_{k}}$ where $l_k$ is the length of the chain in the deformed state, given by
\begin{equation}
  l_{k} = \sum_{j=1}^{n_k-1} l_{i_{k,j},i_{k,j+1}}  
\end{equation}
with 
\begin{equation}
  l_{i_{k,j},i_{k,j+1}} = \norm{\vect{x}_{i_{k,j}} - \vect{x}_{i_{k,j+1}}}.
\end{equation}

The total energy of the entangled network is the sum of the energies of all the chains:
\begin{equation}\label{eq-supp:total-energy}
  U_e = \sum_{k=1}^M U_{k}.
\end{equation}
We note that the total energy of the entangled network is a function of the positions of all the vertices $\vect{x} = (\vect{x}_1, \vect{x}_2, \cdots, \vect{x}_N)\in\mathbb{R}^{dN}$.

To guarantee that the total energy function $U_e = U_e(\vect{x})$ in \eqref{eq-supp:total-energy} is mathematically well-defined for all $\vect{x}\in\mathbb{R}^{dN}$, the constitutive law $f = f(\Lambda)$ must be specified for all $\Lambda \ge 0$, which physically corresponds to covering both the stretching regime ($\Lambda > 1$) and the compressive regime ($0 \le \Lambda < 1$). 
For instance, if we set $f(\Lambda) = 0$ for $0 \le \Lambda < 1$, this corresponds to a rope that becomes slack under compression and carries no force. Alternatively, if we set $f(\Lambda) = \Lambda - 1$ for both $0 \le \Lambda < 1$ and $\Lambda \ge 1$, the chain in the model behaves as a linear spring that generates tensile force under stretching and compressive force under compression. We will later introduce additional reasonable assumptions on $f$ in Section~\ref{sec:analysis-energy-entangled-network} to help the theoretical analysis of the model.

In the simulation of mechanical behaviors of networks, it is often necessary to account for boundary conditions. In our entangled network model, we assume that a subset of the nodes is subject to prescribed displacement boundary conditions. Let $\mathcal{I}_{{b}}\subseteq\{1,2,\dots,N\}$ denote the index set of these boundary nodes. The equilibrium configuration $\vect{x}_e^*$ of the entangled network is then defined as the solution to the following constrained minimization problem:
\begin{equation}\label{eq-supp:min-constrain-entangled-network}
\begin{aligned}  
  {}& \min_{\vect{x}\in\mathbb{R}^{dN}} U_e(\vect{x}), \\
  {}& \textrm{subject to} \quad \vect{x}_i = \vect{x}_{i,b}, \quad i \in \mathcal{I}_{{b}}.
\end{aligned}
\end{equation}
Here $\vect{x}_{i,b}\in\mathbb{R}^d$ is the prescribed displacement of the boundary nodes $v_i$ with $i\in\mathcal{I}_{{b}}$. The minimum energy to the minimization problem \eqref{eq-supp:min-constrain-entangled-network} is denoted by $U_e^*$. We note that the above constrained optimization problem \eqref{eq-supp:min-constrain-entangled-network} can equivalently be reformulated as an unconstrained problem by substituting the displacement boundary conditions into the energy function.

Since the reference spring network $G_s=(V, C_s)$ is a special case of the entangled network where no slidable nodes are present, its energy function $U_s = U_s(\vect{x})$ can be defined in the same way. The equilibrium configuration $\vect{x}_s^*$ of the spring network minimizes the total energy subject to the same boundary conditions, and the minimum energy is denoted by $U_s^*$:
\begin{equation}
\begin{aligned}  
  {}& \min_{\vect{x}\in\mathbb{R}^{dN}} U_s(\vect{x}), \\
  {}& \textrm{subject to} \quad \vect{x}_i = \vect{x}_{i,b}, \quad i \in \mathcal{I}_{{b}}.
\end{aligned}
\end{equation}

We now conclude the introduction of our entangled network model, which adopts a graph-based representation. Its equilibrium state is determined by solving an optimization problem that minimizes the energy function.

\section{Theoretical analysis of entangled networks}\label{sec:analysis-energy-entangled-network}

In this part, we conduct mathematical analysis on the total energy of the entangled network. In Section \ref{sec:convexity-preliminaries}, we introduce the definition of the convexity and several related lemmas that will aid in analyzing the convexity of the energy function. In Section \ref{sec:energy-relation}, we show that the energy of the entangled network is no larger than the energy of its reference spring network with the same configuration. Lastly, we prove the convexity of the energy function of the entangled network in Section \ref{sec:convex-energy}.

\subsection{Preliminaries on convexity}\label{sec:convexity-preliminaries}

In this part, we introduce the definition of the convexity and several related lemmas that will aid in analyzing the convexity of a given function. The definitions and results presented here are standard in convex analysis, see, e.g., \cite{boyd2004convex}. For the reader's convenience, we provide proofs for selected lemmas.

We first introduce the definition of the convex function. A function is called convex if the line segment connecting any two points on the curve lies above the curve itself, see the following definition and Figure \ref{fig-supp:convex-function} for the graphical illustration of the convex function.
\begin{definition}[Convexity]\label{def:convex-function}
  A function $f:\mathbb{R}^n\rightarrow \mathbb{R}$ is convex if for any $\vect{x}_1, \vect{x}_2\in\mathbb{R}^n$ and $0\le w \le 1$,
  \begin{equation}
    f(w\vect{x}_1 + (1-w)\vect{x}_2) \le wf(\vect{x}_1) + (1-w)f(\vect{x}_2).
  \end{equation}
\end{definition}
\begin{figure}
  \centering
  \includegraphics[width=0.6\textwidth]{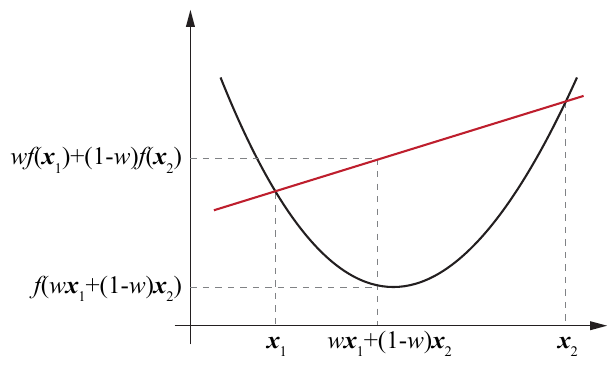}
  \caption{\textbf{A graphical illustration of a convex function.} The function $f$ is convex if the line segment (in red color) connecting any two points on the curve lies above the curve (in black color).}
  \label{fig-supp:convex-function}
\end{figure}

\begin{lemma}\label{lemma:convex-function-second-derivative}
  Assume $f:\mathbb{R}^n\rightarrow\mathbb{R}$ is twice continuously differentiable, i.e., second order partial derivatives exist and are continuous. Then $f$ is convex if and only if $\nabla^2 f(\vect{x})\ge0$ (i.e., the Hessian matrix is positive semidefinite) for all $\vect{x}\in\mathbb{R}^n$.
\end{lemma}

\begin{lemma}\label{lemma:convex-function-sum}
  If $f:\mathbb{R}^n\rightarrow \mathbb{R}$ and $g:\mathbb{R}^n\rightarrow \mathbb{R}$ are both convex, $w_1\ge0$, and $w_2\ge0$, then $h = w_1 f + w_2 g$ is also convex.
\end{lemma}

\begin{lemma}\label{lemma:convex-function-composition}
  If $f:\mathbb{R}^n\rightarrow \mathbb{R}$ and $g:\mathbb{R}\rightarrow \mathbb{R}$  are both convex functions and $g$ is increasing, then $h:\mathbb{R}^n\rightarrow \mathbb{R}$ given by $h(\vect{x}) = g(f(\vect{x}))$ is convex. 
\end{lemma}

\begin{lemma}\label{lemma:convex-2d-extension}
  If a function $f:\mathbb{R}^m\rightarrow \mathbb{R}$ is convex, then the function $g:\mathbb{R}^m\times \mathbb{R}^n$ with $g=g(\vect{x}, \vect{y}) := f(\vect{x})$ is also convex.
\end{lemma}
\begin{proof}
  For any $\vect{x}_1$, $\vect{x}_2\in\mathbb{R}^m$, $\vect{y}_1$, $\vect{y}_2\in\mathbb{R}^n$, and $0\le w \le 1$,
  \begin{equation}
  \begin{aligned}
    g(w\vect{x}_1 + (1-w)\vect{x}_2, w\vect{y}_1 + (1-w)\vect{y}_2) ={}& f(w\vect{x}_1 + (1-w)\vect{x}_2) \\
    \le{}& wf(\vect{x}_1) + (1-w)f(\vect{x}_2) \\
    ={}& wg(\vect{x}_1, \vect{y}_1) + (1-w)g(\vect{x}_2, \vect{y}_2).
  \end{aligned}
  \end{equation}
  Therefore, the function $g(\vect{x}, \vect{y})$ is convex.
\end{proof}

\begin{lemma}\label{lemma:convex-1d-restriction}
  If a function $f:\mathbb{R}^m\times \mathbb{R}^n \rightarrow \mathbb{R}$ is convex, then the function $g:\mathbb{R}^m\rightarrow \mathbb{R}$ with $g=g(\vect{x}) := f(\vect{x}, \vect{y}_0)$ with fixed $\vect{y_0}\in \mathbb{R}^n$ is also convex.
\end{lemma}
\begin{proof}
    For any $\vect{x}_1, \vect{x}_2 \in\mathbb{R}^n$ and $0\le w \le 1$, we have
\begin{equation}
\begin{aligned}
    g(w\vect{x}_1 + (1-w)\vect{x}_2) = {}& f(w\vect{x}_1 + (1-w)\vect{x}_2, \vect{y}_0)\\
    = {}& f(w\vect{x}_1 + (1-w)\vect{x}_2, w\vect{y}_0 + (1-w)\vect{y}_0)\\
    \le{}& w f(\vect{x}_1, \vect{y}_0) + (1-w) f(\vect{x}_2, \vect{y}_0) \\
    ={}& w g(\vect{x}_1) + (1-w) g(\vect{x}_2).
\end{aligned}
\end{equation}
Therefore, $g(\vect{x})$ is convex.
\end{proof}

\begin{lemma}\label{lemma:distance-function-two-points-convex}
  The distance function $d:\mathbb{R}^n\times\mathbb{R}^n\rightarrow\mathbb{R}$ with $d(\vect{x}, \vect{y}) = \norm{\vect{x} - \vect{y}}$ is convex where $\norm{\cdot}$ is the Euclidean norm.
\end{lemma}
\begin{proof}
  For any $\vect{x_1}, \vect{x_2}, \vect{y_1}, \vect{y_2}\in\mathbb{R}^n$ and $0\le w \le 1$, we have
  \begin{equation}
  \begin{aligned}
    d(w\vect{x_1} + (1-w)\vect{y_1}, w\vect{x_2} + (1-w)\vect{y_2}) ={}& \norm{w\vect{x_1} + (1-w)\vect{y_1} - w\vect{x_2} - (1-w)\vect{y_2}} \\
    ={}& \norm{w(\vect{x_1} - \vect{x_2}) + (1-w)(\vect{y_1} - \vect{y_2})} \\
    \le{}& w\norm{\vect{x_1} - \vect{x_2}} + (1-w)\norm{\vect{y_1} - \vect{y_2}} \\
    ={}& wd(\vect{x_1}, \vect{x_2}) + (1-w)d(\vect{y_1}, \vect{y_2}).
  \end{aligned}
  \end{equation}
  Therefore, the distance function $d(\vect{x}, \vect{y})$ is convex.
\end{proof}

\subsection{Energy relation between entangled and reference spring networks}\label{sec:energy-relation}

We start by considering the energy relationship of chains in one dimension. Consider two systems in the undeformed states: System I, consisting of two chains connected in series with the initial lengths $L_{1}$ and $L_{2}$; and System II, consisting of a single chain with initial length $L = L_{1} + L_{2}$. See Figure \ref{fig-supp:model-chain-one-dimension}a for an illustration of the undeformed configurations.

Suppose some displacement is applied such that System I is stretched to a new configuration, where two chains have the lengths $l_1$ and $l_2$, respectively, and System II is stretched to the length $l = l_1 + l_2$. The deformed state is shown in Figure \ref{fig-supp:model-chain-one-dimension}b. Note that we do not assume System I is in equilibrium in the deformed state in Figure \ref{fig-supp:model-chain-one-dimension}b; that is, the forces in the two chains are not necessarily equal. We assume that all the chains are governed by the same constitutive law, $f=f(\Lambda)$, where $\Lambda$ is the stretch ratio of the chain and $f$ is the force.
\begin{figure}[h]
  \centering
  \includegraphics[width=0.9\textwidth]{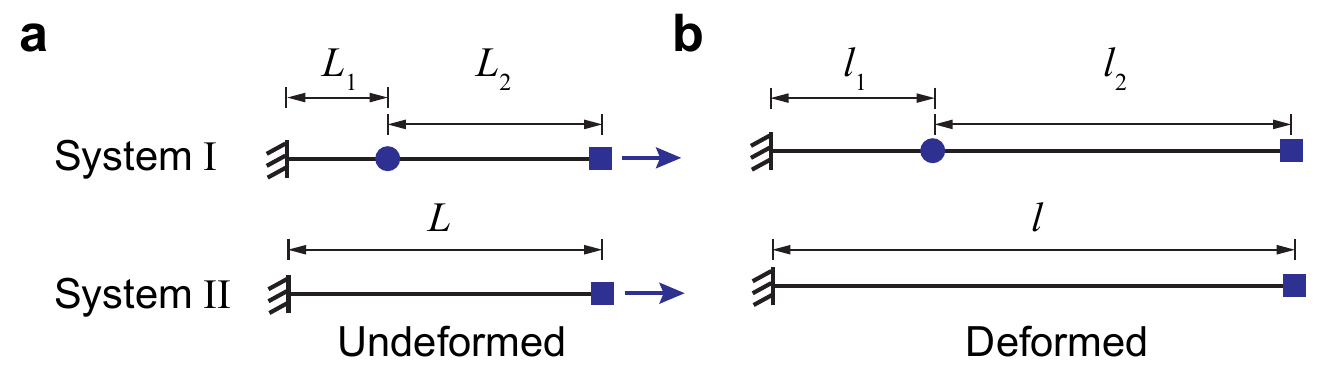}
  \caption{\textbf{One-dimensional chain models.} System I consists of two chains connected in series, and System II consists of a single chain. \textbf{(a)} The undeformed states of Systems I and II. \textbf{(b)} The deformed states of Systems I and II. We do not assume System I is in equilibrium in the deformed state; that is, the forces in the two chains are not necessarily equal.}
  \label{fig-supp:model-chain-one-dimension}
\end{figure}

Now we investigate the energy of two systems. The energy of System I is given by
\begin{equation}\label{eq-supp:energy-system-I-two-chain}
  U_{{1}} = L_{1} \int_{1}^{\Lambda_1} f(\Lambda) d\Lambda + L_{2} \int_{1}^{\Lambda_2} f(\Lambda) d\Lambda,
\end{equation}
where the stretch ratios $\Lambda_1$ and $\Lambda_2$ are defined as
\begin{equation}
  \Lambda_1 = \frac{l_1}{L_{1}}, \quad \Lambda_2 = \frac{l_2}{L_{2}}.
\end{equation}
The energy of System II is given by
\begin{equation}\label{eq-supp:energy-system-II-single-chain}
  U_2 = L \int_{1}^{\Lambda} f(\Lambda') d\Lambda',
\end{equation}
where the stretch ratio $\Lambda$ is
\begin{equation}
  \Lambda = \frac{l}{L}.
\end{equation}
We now establish the relationship between the energies of the two systems in the following theorem.
\begin{theorem}[Energy relation between one dimensional chains]\label{thm-supp:energy-relation-one-dimension-chains}
  Assume that the constitutive law $f=f(\Lambda)$ is an increasing function. Then the energy of System I (with two chains) is greater than or equal to the energy of System II (with a single chain); that is,
  \begin{equation}
    U_1 \ge U_2.
  \end{equation}
\end{theorem}
\begin{proof}
  We introduce the antiderivative of the function $f=f(\Lambda)$ denoted by $E=E(\Lambda)$:
  \begin{equation}
    E(\Lambda) := \int_1^{\Lambda} f(\Lambda') d\Lambda'
  \end{equation}
  Physically, $E(\Lambda)$ means the energy of the chain per unit length.
  Then the energy of System I in \eqref{eq-supp:energy-system-I-two-chain} can be written as
  \begin{equation}
    U_1 = L_{1} E(\Lambda_1) + L_{2} E(\Lambda_2)
  \end{equation}
  and the energy of System II in \eqref{eq-supp:energy-system-II-single-chain} as
  \begin{equation}
    U_2 = L E(\Lambda).
  \end{equation}

  Since $f=f(\Lambda)$ is an increasing function, we have
  \begin{equation}
    E''(\Lambda) = f'(\Lambda) \ge 0,
  \end{equation}
  which implies that $E(\Lambda)$ is a convex function by Lemma \ref{lemma:convex-function-second-derivative}. Then by Definition \ref{def:convex-function} and noticing that 
  \begin{equation}
    \frac{L_{1}}{L} + \frac{L_{2}}{L} = 1,
  \end{equation}
  we have
  \begin{equation}
    \frac{L_{1}}{L} E(\Lambda_1) + \frac{L_{2}}{L} E(\Lambda_2) \ge E\brac{\frac{L_{1}}{L} \Lambda_1 + \frac{L_{2}}{L} \Lambda_2},
  \end{equation}
  which is equivalently to
  \begin{equation}
    U_1 \ge U_2.
  \end{equation}
\end{proof}

\begin{remark}
  The condition in the above theorem only requires the constitutive law $f=f(\Lambda)$ to be an increasing function. 
  Physically, this means that the force increases as the stretch ratio grows. Consequently, the theorem applies to a broad class of constitutive laws, including the linear Hookean law as well as many nonlinear constitutive laws.
\end{remark}

We now generalize Theorem \ref{thm-supp:energy-relation-one-dimension-chains} on one dimensional chains to networks with arbitrary size and topology. Specifically, we show that, under the same configuration, the energy of the entangled network is always no larger than that of the reference spring network.
\begin{theorem}[Energy relation between entangled network and its reference spring network]\label{thm-supp:energy-relation-network}
  Let $G_e = (V, C)$ be an entangled network  and $G_s=(V, C_s)$ its reference spring network. Assume that the nodes in both networks have the same initial position in the undeformed state and the constitutive law $f=f(\Lambda)$ is an increasing function. Then, the total energy of the entangled network is less than or equal to that of the spring network with the same configuration in the deformed state:
  \begin{equation}\label{eq-supp:energy-order-network}
    U_e(\vect{x}) \le U_s(\vect{x}), \quad \forall \, \vect{x}\in \mathbb{R}^{dN}.
  \end{equation}
\end{theorem}
\begin{proof}
  We begin by considering the energy of a single chain in the entangled network, denoted by $c = (v_1, v_2, \cdots, v_n)\in C$. Its energy is given by
  \begin{equation}\label{eq-supp:energy-single-chain-entangled-network-theorem}
    U_{\textrm{chain}, e} = L \int_{1}^{\Lambda} f(\Lambda') d\Lambda',
  \end{equation}
  where $L$ is the initial length given by
  \begin{equation}
    L = L_{12} + L_{23} + \cdots + L_{n-1,n},
  \end{equation}
  with
  \begin{equation}
    L_{i,i+1} = \norm{\vect{X}_i - \vect{X}_{i+1}}, \quad 1\le i\le n-1.
  \end{equation}
  Here $\vect{X}_i$ denotes the initial position of the node $v_i$ in the undeformed state.
  The stretch ratio $\Lambda$ in \eqref{eq-supp:energy-single-chain-entangled-network-theorem} is given by
  \begin{equation}
    \Lambda = \frac{l}{L} = \frac{1}{L} \brac{\norm{\vect{x}_1 - \vect{x}_2} + \norm{\vect{x}_2 - \vect{x}_3} + \cdots + \norm{\vect{x}_{M-1} - \vect{x}_M}}.
  \end{equation}
  Here $\vect{x}_i$ denotes the position of the node $v_i$ in the deformed state.

  Next, we consider the energy of the same chain when modeled as a sequence of independent segments in the spring network:
  \begin{equation}
    U_{\textrm{chain},s} = \sum_{i=1}^{n-1} L_{i,i+1} \int_{1}^{\Lambda_{i,i+1}} f(\Lambda') d\Lambda'
  \end{equation}
  where $\Lambda_{i,i+1}$ the stretch ratio of the segement connecting $v_i$ and $v_{i+1}$:
  \begin{equation}
    \Lambda_{i,i+1} = \frac{1}{L_{i,i+1}}\norm{\vect{x}_i - \vect{x}_{i+1}},
  \end{equation}

  Following the same argument as in the proof of Theorem \ref{thm-supp:energy-relation-one-dimension-chains}, we introduce the antiderivative function $E = E(\Lambda)$:
  \begin{equation}
    E(\Lambda) := \int_{1}^{\Lambda} f(\Lambda') d\Lambda',
  \end{equation}
  which is convex due to the assumption that $f$ is increasing.
  
  Then the two energy expressions can be rewritten as
  \begin{equation}
    U_{\textrm{chain},e} = L E(\Lambda), \quad U_{\textrm{chain},s} = \sum_{i=1}^{n-1} L_{i,i+1} E(\Lambda_{i,i+1}).
  \end{equation}
  By Definition~\ref{def:convex-function}, we have
  \begin{equation}
    \sum_{i=1}^{n-1} \frac{L_{i,i+1}}{L} E(\Lambda_{i,i+1}) \ge E\brac{\sum_{i=1}^{n-1} \frac{L_{i,i+1}}{L} \Lambda_{i,i+1}} = E(\Lambda),
  \end{equation}
  which implies
  \begin{equation}
    U_{\textrm{chain},s} \ge U_{\textrm{chain},e}.
  \end{equation}
  Summing over all chains in the network completes the proof.
\end{proof}

From Theorem \ref{thm-supp:energy-relation-network}, we can derive a corollary that offers further physical insights. First, by evaluating $\vect{x}$ in \eqref{eq-supp:energy-order-network} at the equilibrium configuration $\vect{x}_s^*$ of the reference spring network, we obtain
\begin{equation}\label{eq-supp:energy-order-network-equilibrium}
  U_e(\vect{x}_s^*) \le U_s^*.
\end{equation}
Physically, this inequality implies that, for a spring network in the equilibrium state, the energy of the system can be reduced simply by replacing some non-slidable nodes with slidable ones, even if the slidable nodes are not allowed to slide (i.e., the entangled network need not be in equilibrium). This reduction is due to uniform tension along segments of each single chain
in the entangled network.

Next, we consider the relaxation of the entangled network from $\vect{x}_s^*$ (i.e., the equilibrium state of the spring network) to its own equilibrium state $\vect{x}_e^*$ and obtain
\begin{equation}\label{eq-supp:energy-order-network-relaxation}
  U_e^* \le U_e(\vect{x}_s^*).
\end{equation}
This reduction is due to stress redistribution enabled by the sliding of slidable nodes in the entangled network.

Combining \eqref{eq-supp:energy-order-network-equilibrium} and \eqref{eq-supp:energy-order-network-relaxation}, we obtain
\begin{equation}\label{eq-supp:energy-order-network-final}
  U_e^* \le U_s^*.
\end{equation}
This means that the energy of the entangled network in its equilibrium state is no larger than that of the reference spring network. We summarize these results in the following theorem, which is the main result of this section.
\begin{theorem}[Energy relation between entangled network and its reference spring network]
  Let $G_e = (V, C)$ be an entangled network  and $G_s=(V, C_s)$ its reference spring network. Assume that the nodes in both networks have the same initial position in the undeformed state and the constitutive law $f=f(\Lambda)$ is an increasing function. Then the following inequalities hold:
  \begin{equation}
    U_e(\vect{x}_s^*) \le U_s^*, \quad U_e^* \le U_e(\vect{x}_s^*),
  \end{equation}
  and consequently,
  \begin{equation}
    U_e^* \le U_s^*.
  \end{equation}
  Here $\vect{x}_e^*$ and $\vect{x}_s^*$ denote the equilibrium states of the entangled and spring networks, respectively, while $U_e^*$ and $U_s^*$ denote their corresponding total energies in the equilibrium state.
\end{theorem}

We note that the above theorem can be easily extended to the case where the same displacement boundary conditions are prescribed for both networks. The details are omitted here for saving space.

\subsection{Convexity of the energy function}\label{sec:convex-energy}

In the section, we will discuss the convexity of the energy function of the entangled network. We will show that the energy function is convex with respect to the position of the vertices. This property is crucial for the optimization problem of the entangled network, as it ensures that the minimization problem has a unique global minimum.

\begin{theorem}[Convexity of the energy function of entangled network]\label{thm:convexity-energy-entangled-network}
  Let $G_e = (V, C)$ be an entangled network. Assume that the constitutive law $f=f(\Lambda)$ is an increasing and non-negative function. Then the total energy of entangled network $U_e(\vect{x})$ is convex with respect to the positions of all nodes in the deformed state.
\end{theorem}
\begin{proof}
  We begin by considering the energy of a single chain $c = (v_1, v_2, \cdots, v_n)\in C$, which is given by
  \begin{equation}\label{eq-supp:energy-single-chain-entangled-network}
    U_{\textrm{chain}} = L \int_{1}^{\Lambda} f(\Lambda') d\Lambda',
  \end{equation}
  where $L$ is the initial length of the chain and the stretch ratio $\Lambda$ is given by
  \begin{equation}
    \Lambda = \frac{l}{L} = \frac{1}{L} \brac{\norm{\vect{x}_1 - \vect{x}_2} + \norm{\vect{x}_2 - \vect{x}_3} + \cdots + \norm{\vect{x}_{n-1} - \vect{x}_n}}.
  \end{equation}
  Here $\vect{x}_i$ denotes the position of the node $v_i$ in the deformed state.

  Let $\vect{x} = (\vect{x}_1, \vect{x}_2, \cdots, \vect{x}_n)$ represent the collection of positions for the nodes in chain $c$. Then the energy function in \eqref{eq-supp:energy-single-chain-entangled-network} can be expressed as a composition of two functions:
  \begin{equation}
    U_{\textrm{chain}}(\vect{x}) = G(F(\vect{x})),
  \end{equation}
  where 
  \begin{equation}
    F(\vect{x}) := \frac{1}{L}\brac{\norm{\vect{x}_1 - \vect{x}_2} + \norm{\vect{x}_2 - \vect{x}_3} + \cdots + \norm{\vect{x}_{n-1} - \vect{x}_n}},
  \end{equation}
  and
  \begin{equation}
    G(\Lambda) := L \int_{1}^{\Lambda} f(\Lambda') d\Lambda'.
  \end{equation}
By Lemma \ref{lemma:distance-function-two-points-convex}, the distance function $\norm{\vect{x}_i - \vect{x}_j}$ is convex with respect to $(\vect{x}_i, \vect{x}_j)$. Therefore, by Lemma \ref{lemma:convex-2d-extension} and Lemma \ref{lemma:convex-function-sum}, the function $F(\vect{x})$ is convex with respect to $\vect{x}$. Then we consider the function $G(\Lambda)$. Taking derivatives of $G(\Lambda)$ we obtain
  \begin{equation}
    G'(\Lambda) = L f(\Lambda), \quad G''(\Lambda) = L f'(\Lambda).
  \end{equation}
  Since $f$ is increasing and non-negative, we have that $G$ is convex and increasing. Therefore, the energy function $U_{\textrm{chain}}(\vect{x}) = G(F(\vect{x}))$ is convex with respect to $\vect{x}$ by Lemma \ref{lemma:convex-function-composition}.

  Since the energy function of the entangled network is the sum of the energy function of all the chains, we complete the proof by Lemma \ref{lemma:convex-2d-extension} and Lemma \ref{lemma:convex-function-sum}.
\end{proof}

By Lemma \ref{lemma:convex-1d-restriction}, it follows that the energy function remains convex with respect to the positions of the nodes when certain displacement boundary conditions are prescribed. For completeness, we state the following theorem without proof.
\begin{theorem}[Convexity of the energy function of entangled network with displacement boundary conditions]\label{thm:convexity-energy-entangled-network-displacement}
  Let $G_e = (V, C)$ be an entangled network. Assume that the constitutive law $f=f(\Lambda)$ is an increasing and non-negative function. Then under the displacement boundary condition $\vect{x}_i = \vect{x}_{i,b}$ with $i\in \mathcal{I}_b\subseteq \{1,2,\dots,N\}$, the total energy of entangle network is convex with respect to the positions of the free (non-boundary) nodes, i.e., $\vect{x}_i$ for all $i \in \{1, 2, \dots, N\} \setminus \mathcal{I}_{{b}}$.
\end{theorem}

\section{Two-chain model}\label{sec-supp:two-chain-model}

In this part, we study a simplified entangled network model consisting of two entangled chains, called two-chain model. In the setup, two chains intersect in a two-dimensional space with the intersection being either a slidable node, representing an entangled network (Figure \ref{fig-supp:entangled-two-chains}) or a non-slidable node, representing a spring network (Figure \ref{fig-supp:crosslink-two-chains}). 
This toy model serves as a concrete example to illustrate our key idea in constructing an entangled network. Furthermore, we analyze the regularity of the energy function near the equilibrium state and show that it admits a local minimum: it is differentiable under small displacements but becomes non-differentiable under large displacements. Physically, this transition corresponds to the locking phenomenon.

\subsection{Model}

We start by introducing the two-chain model with a slidable node. As shown in Figure~\ref{fig-supp:entangled-two-chains}a, in the undeformed state, the first chain has endpoints $P_1$ and $P_2$, located at $\vect{X}_1 = (0, 1)$ and $\vect{X}_2 = (1, 2)$. The second chain has endpoints $P_3$ and $P_4$, located at $\vect{X}_3 = (1, 0)$ and $\vect{X}_4 = (2, 1)$, respectively. The two chains intersect at a node $P$, whose position is given by $\vect{X} = (1, 1)$.
\begin{figure}
  \centering
  \includegraphics[width=0.8\textwidth]{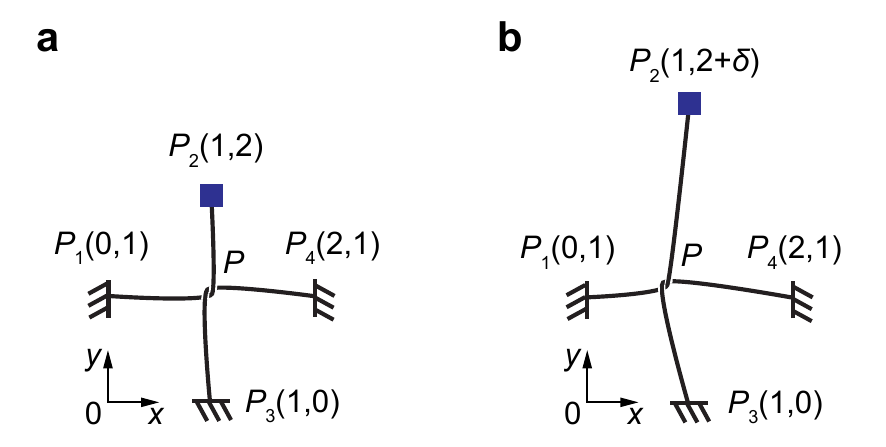}
  \caption{\textbf{The two-chain model with slidable node, representing an entangled network.} \textbf{(a)} The undeformed state: The two chains intersect at node $P(1,1)$, with the endpoints $P_1(0,1)$ and $P_2(1,2)$ for the first chain, and $P_3(1,0)$ and $P_4(2,1)$ for the second chain. \textbf{(b)} The deformed state: A displacement $\delta > 0$ in the vertical direction is applied to $P_2$, while the other endpoints ($P_1$, $P_3$, and $P_4$) remain fixed.}
  \label{fig-supp:entangled-two-chains}
\end{figure}

To introduce deformation, we impose a displacement boundary condition on $P_2$, while keeping $P_1$, $P_3$, and $P_4$ fixed ($\vect{x}_i=\vect{X}_i$ for $i=1,3,4$).  Specifically, $P_2$ is displaced in the $y$-direction, setting its position to $\vect{x}_2 = (1, 2 + \delta)$ where $\delta>0$ represents the applied displacement along $y$ direction. See Figure \ref{fig-supp:entangled-two-chains}b for the deformed state of the two chains. This setup is also consistent with our experiment design.

To determine the equilibrium configuration of the network, we first express its total elastic energy $U_e$ as the sum of the energies in the two chains:
\begin{equation}\label{eq:energy-two-fibers-entangle}
    U_e = U_{12} + U_{34},
\end{equation}
where $U_{12}$ is the energy of the chain connecting $\vect{x}_1$, $\vect{x}$, and $\vect{x}_2$, and $U_{34}$ is the energy of the chain connecting $\vect{x}_3$, $\vect{x}$, and $\vect{x}_4$. These are given by
\begin{equation}\label{eq:energy-two-fibers-entangle-each-chain}
    U_{12} = L_{12} \int_{1}^{\Lambda_{12}} f(\Lambda') d\Lambda', \quad U_{34} = L_{34} \int_{1}^{\Lambda_{34}} f(\Lambda') d\Lambda',
\end{equation}
Here, $L_{12}$ and $L_{34}$ are the initial (undeformed) length of two chains. In our case, both are set to $L_{12} = L_{34} = 2$. The function $f=f(\Lambda)$ represents the constitutive law of the chains, describing the force-stretch relation.
The stretch ratios of the two chains, $\Lambda_{12}$ and $\Lambda_{34}$, are defined by:
\begin{equation}\label{eq:stretch-ratio-two-fibers-entangle}
    \Lambda_{12} = \frac{d_{1} + d_{2}}{L_{12}}, \quad 
    \Lambda_{34} = \frac{d_{3} + d_{4}}{L_{34}},
\end{equation}
where $d_{i}$ for $j=1,2,3,4$ denotes the Euclidean distance between the points $\vect{x}=(x, y)$ and $\vect{x}_i=(x_{i}, y_{i})$, given by:
\begin{equation}\label{eq:distance-two-fibers-entangle}
    d_{i} = \norm{\vect{x} - \vect{x_i}} = \sqrt{(x - x_{i})^2 + (y - y_{i})^2}.
\end{equation}

The definitions above give the total elastic energy of the entangled two chains as a function of the intersection node $\vect{x}=(x, y)$ and the applied displacement $\delta$. 
We denote the energy as $U_e = U_e(\vect{x})$ to emphasize its dependence on the position of the intersection node $\vect{x}$. Since the boundary displacement $\delta$ also affects the energy, we sometimes write $U_e$ more explicitly as $U_e = U_e(\vect{x}, \delta)$ to indicate this parameter dependence.

To enhance the reader's understanding of the problem, we now provide an explicit expression for the energy function $U_e$ in \eqref{eq:energy-two-fibers-entangle} under a linear constitutive law $f(\Lambda) = \Lambda - 1$, which corresponds to the Hookean (linear elastic) model. In this case, the total energy of the network is given by:
\begin{equation}
\begin{aligned}
    U_e ={}& U_{12} + U_{34} \\
    ={}& (\Lambda_{12} - 1)^2 + (\Lambda_{34} - 1)^2 \\
    ={}& \brac{\frac{1}{2}(d_{1} + d_{2}) - 1}^2 + \brac{\frac{1}{2}(d_{3} + d_{4}) - 1}^2 \\
    ={}& \brac{\frac{1}{2}\brac{\sqrt{(x - x_{1})^2 + (y - y_{1})^2} + \sqrt{(x - x_{2})^2 + (y - y_{2})^2}} - 1}^2 \\
    {}& + \brac{\frac{1}{2}\brac{\sqrt{(x - x_{3})^2 + (y - y_{3})^2} + \sqrt{(x - x_{4})^2 + (y - y_{4})^2}} - 1}^2 \\
    ={}& \brac{\frac{1}{2}\brac{\sqrt{x^2 + (y - 1)^2} + \sqrt{(x - 1)^2 + (y - (2+\delta))^2}} - 1}^2 \\
    {}& + \brac{\frac{1}{2}\brac{\sqrt{(x - 1)^2 + y^2} + \sqrt{(x - 2)^2 + (y - 1)^2}} - 1}^2. \\
\end{aligned}    
\end{equation}
Here we have used the relations \eqref{eq:energy-two-fibers-entangle-each-chain}, \eqref{eq:stretch-ratio-two-fibers-entangle} and \eqref{eq:distance-two-fibers-entangle}.

To investigate the effect of the entanglement, we also introduce the two-chain model with a non-slidable node. In this case, each of the nodes $P_1$, $P_2$, $P_3$, and $P_4$ is directly connected to the intersection node $P$ with a single chain, see Figure \ref{fig-supp:crosslink-two-chains}. 
\begin{figure}
  \centering
  \includegraphics[width=0.8\textwidth]{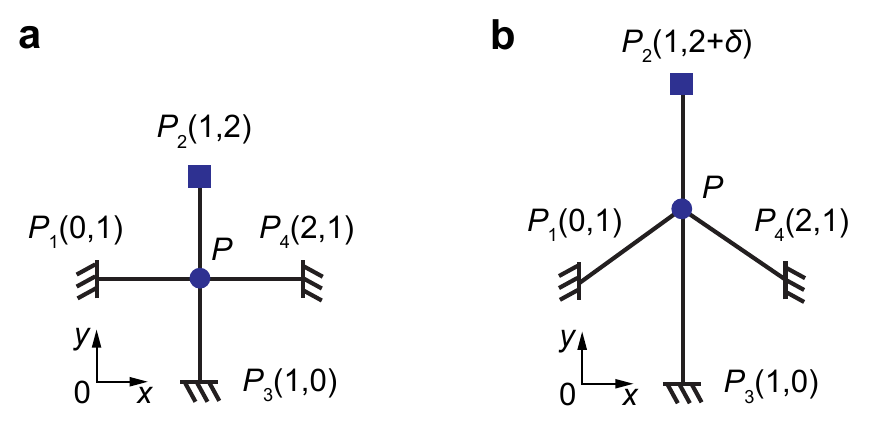}
  \caption{\textbf{The two-chain model with non-slidable node, representing a spring network.} \textbf{(a)} The undeformed state: Four chains are joined at a non-slidable node $P(1,1)$, with the endpoints $P_1(0,1)$, $P_2(1,2)$, $P_3(1,0)$ and $P_4(2,1)$ for four chains. \textbf{(b)} The deformed state: A vertical displacement $\delta > 0$ is applied to $P_2$, while the other endpoints ($P_1$, $P_3$, and $P_4$) remain fixed.}
  \label{fig-supp:crosslink-two-chains}
\end{figure}

The total elastic energy of the two-chain model with a non-slidable node is given by
\begin{equation}
    U_s = \sum_{i=1}^4 U_{i},
\end{equation}
where $U_{i}$ is the energy of the chain connecting $P$ and $P_i$, defined by
\begin{equation}
    U_{i} = L_{i} \int_{1}^{\Lambda_{i}} f(\Lambda') d\Lambda'.
\end{equation}
The stretch ratio $\Lambda_{i}$ is computed as
\begin{equation}
    \Lambda_{i} = \frac{d_{i}}{L_{i}}.
\end{equation}
Here $d_{i}$ denotes the Euclidean distance between the intersection node $\vect{x}$ and the node $\vect{x}_i$,
\begin{equation}
    d_{i} = \norm{\vect{x} - \vect{x}_i} = \sqrt{(x - x_i)^2 + (y - y_i)^2}.
\end{equation}
and $L_{i}$ is the initial (undeformed) length of the chain connecting $\vect{x}$ and $\vect{x}_i$. In this case, $L_{1} = L_{2} = L_{3} = L_{4} = 1$.

As an example, we explicitly write the energy function $U_s$ under the linear constitutive law $f(\Lambda) = \Lambda - 1$:
\begin{equation}
\begin{aligned}
    U_s ={}& \sum_{i=1}^4 \frac{1}{2} (\lambda_{i} - 1)^2 \\
    ={}& \sum_{i=1}^4 \frac{1}{2} \brac{ \sqrt{(x - x_i)^2 + (y - y_i)^2} - 1 }^2 \\
    ={}& \frac{1}{2} \brac{\sqrt{x^2 + (y - 1)^2} - 1}^2 + \frac{1}{2} \brac{\sqrt{(x - 1)^2 + (y - {2+\delta})^2} - 1}^2 \\
    {}& + \frac{1}{2} \brac{\sqrt{(x - 1)^2 + y^2} - 1}^2 + \frac{1}{2} \brac{\sqrt{(x - 2)^2 + (y - 1)^2} - 1}^2    
\end{aligned}    
\end{equation}

\subsection{Non-smoothness of the energy function}\label{sec-supp:non-smoothness-energy-two-chain}

In this part, we will investigate the regularity of the energy function of the slidable two-chain model near its minimum. Specifically, we will show that the energy function admits a local minimum that is differentiable at the local minimum for small applied displacements but becomes non-differentiable when the displacement exceeds a certain threshold. This analysis is crucial for the design of optimization algorithms aimed at minimizing the energy, as non-differentiability can significantly affect the convergence of the algorithm.

The following lemma in computing the derivative of the distance function will be used in the proof.
\begin{lemma}[Derivatives of distance function]\label{lemma:distance-function-derivative}
  Let $d: \mathbb{R}^n \times \mathbb{R}^n \rightarrow \mathbb{R}$ be the Euclidean distance function defined by
  \begin{equation}
    d(\vect{x}, \vect{y}) = \norm{\vect{x} - \vect{y}} = \brac{\sum_{i=1}^{n}(x_i - y_i)^2}^{\frac{1}{2}},
  \end{equation}
  for $\vect{x} = (x_1, x_2, \dots, x_n)\in\mathbb{R}^n$ and $\vect{y} = (y_1, y_2, \dots, y_n)\in\mathbb{R}^n$. Then, the partial derivatives of $d$ are given by
  \begin{equation}
    \frac{\partial d}{\partial \vect{x}} = \frac{\vect{x} - \vect{y}}{d(\vect{x}, \vect{y})}, \qquad \frac{\partial d}{\partial \vect{y}} = \frac{\vect{y} - \vect{x}}{d(\vect{x}, \vect{y})}.
  \end{equation}
\end{lemma}
\begin{proof}
  We prove by directly computing its derivatives. For $k=1,2,\cdots,n$,
  \begin{equation}
    \frac{\partial d}{\partial x_k} = 2(x_k - y_k) \frac{1}{2}\brac{\sum_{i=1}^{n}(x_i - y_i)^2}^{-1/2} = \frac{x_k - y_k}{d(\vect{x}, \vect{y})},
  \end{equation}
  and
  \begin{equation}
    \frac{\partial d}{\partial y_k} = 2(y_k - x_k) \frac{1}{2}\brac{\sum_{i=1}^{n}(x_i - y_i)^2}^{-1/2} = \frac{y_k - x_k}{d(\vect{x}, \vect{y})}.
  \end{equation}
  Therefore, in the vector notation, we have
  \begin{equation}
    \frac{\partial d}{\partial \vect{x}} = \frac{\vect{x} - \vect{y}}{d(\vect{x}, \vect{y})}, \quad \frac{\partial d}{\partial \vect{y}} = \frac{\vect{y} - \vect{x}}{d(\vect{x}, \vect{y})}.
  \end{equation}  
\end{proof}

\begin{theorem}[Non-differentiability of the two-chain model with a slidable node]\label{thm-supp:non-smoothness-energy-two-chain}
Assume that the constitutive law $f = f(\lambda)$ satisfies $\lim_{\lambda \to +\infty} f(\lambda) = +\infty$. Then, the energy function $U_e = U_e(\vect{x}, \delta)$, defined in \eqref{eq:energy-two-fibers-entangle} for the two-chain model with a slidable node, admits a local minimum that is non-differentiable and Lipschitz continuous for sufficiently large displacement. More precisely, there exists $\delta_0 > 0$ such that for all $\delta > \delta_0$, the function $U_e(\vect{x}, \delta)$ has a minimum at $\vect{x} = \vect{x}_1$, but is Lipschitz continuous and not differentiable at that point.
\end{theorem}
\begin{proof}
  For the energy function $U_e$ of the two-chain model with a slidable node defined in \eqref{eq:energy-two-fibers-entangle}, we have
  \begin{equation}
    U_e(\vect{x}) = L_{12} \int_{1}^{\Lambda_{12}} f(\Lambda') d\Lambda' + L_{34} \int_{1}^{\Lambda_{34}} f(\Lambda') d\Lambda'.
  \end{equation}
  Taking derivative of $U_e$ with respect to $\vect{x}$, we obtain
  \begin{equation}
  \begin{aligned}
    \frac{\partial U_e}{\partial \vect{x}} ={}& L_{12} \frac{\partial \Lambda_{12}}{\partial \vect{x}} f(\Lambda_{12}) + L_{34} \frac{\partial \Lambda_{34}}{\partial \vect{x}} f(\Lambda_{34}) \\
    ={}& L_{12} \frac{\partial}{\partial \vect{x}}\brac{\frac{1}{L_{12}}\brac{d(\vect{x}, \vect{x}_1) + d(\vect{x}, \vect{x}_2)}} f(\Lambda_{12}) \\
    {}& + L_{34} \frac{\partial}{\partial \vect{x}}\brac{\frac{1}{L_{34}}\brac{d(\vect{x}, \vect{x}_3) + d(\vect{x}, \vect{x}_4)}} f(\Lambda_{34}) \\    
    = {}& \brac{\frac{\vect{x} - \vect{x}_1}{d(\vect{x}, \vect{x}_1)} + \frac{\vect{x} - \vect{x}_2}{d(\vect{x}, \vect{x}_2)}} f(\Lambda_{12}) + \brac{\frac{\vect{x} - \vect{x}_3}{d(\vect{x}, \vect{x}_3)} + \frac{\vect{x} - \vect{x}_4}{d(\vect{x}, \vect{x}_4)}} f(\Lambda_{34}).    
  \end{aligned}        
  \end{equation}
  Here we use Lemma \ref{lemma:distance-function-derivative} to compute the derivative of the distance function.

  From the expression above, it is easy to see that $U_e = U_e(\vect{x})$ at has singularity at $\vect{x} = \vect{x}_i$ for $i=1,2,3,4$. While $U_e(\vect{x})$ is continuous at these points, its derivative is not well-defined due to the singularity of the gradient of the distance function at zero.
  To determine whether or not $\vect{x}_1$ is a local minimum, it suffices to verify the following equivalent condition: for any unit vector $\vect{n}\in\mathbb{R}^2$ with $\norm{\vect{n}} = 1$, the directional derivative satisfies
  \begin{equation}
    \brac{\vect{n} \cdot \frac{\partial U_e}{\partial \vect{x}}} \bigg|_{\vect{x} = \vect{x}_{\varepsilon}} \ge 0.
  \end{equation}
  where $\vect{x}_{\varepsilon} := \vect{x}_1 + \varepsilon \vect{n}$ for $0 < \varepsilon \ll 1$. If this condition holds true, then $\vect{x}_1$ is a local minimum of the energy function.

  Next, we compute the directional derivative of $U_e$ at $\vect{x}_1$ in the direction $\vect{n}$:
  \begin{equation}
    \brac{\vect{n} \cdot \frac{\partial U_e}{\partial \vect{x}}} \bigg|_{\vect{x} = \vect{x}_{\varepsilon}} = \vect{n} \cdot \brac{\frac{\vect{x}_{\varepsilon} - \vect{x}_1}{d(\vect{x}_{\varepsilon}, \vect{x}_1)} + \frac{\vect{x}_{\varepsilon} - \vect{x}_2}{d(\vect{x}_{\varepsilon}, \vect{x}_2)}} f(\Lambda_{12}) + \vect{n} \cdot \brac{\frac{\vect{x}_{\varepsilon} - \vect{x}_3}{d(\vect{x}_{\varepsilon}, \vect{x}_3)} + \frac{\vect{x}_{\varepsilon} - \vect{x}_4}{d(\vect{x}_{\varepsilon}, \vect{x}_4)}} f(\Lambda_{34})
  \end{equation}
  For the first term, we have:
  \begin{equation}
    \vect{n} \cdot \frac{\vect{x}_{\varepsilon} - \vect{x}_1}{d(\vect{x}_{\varepsilon}, \vect{x}_1)} = \vect{n} \cdot \frac{\varepsilon \vect{n}}{\norm{\varepsilon \vect{n}}} = 1.
  \end{equation}
  Therefore, for sufficitly small $0<\varepsilon\ll1$, we obtain the approximation:
  \begin{equation}\label{eq:directional-derivative-1}
  \begin{aligned}
    \brac{\vect{n} \cdot \frac{\partial U_e}{\partial \vect{x}}} \bigg|_{\vect{x} = \vect{x}_{\varepsilon}} \approx{}& f(\Lambda_{12}) + \vect{n} \cdot \brac{\frac{\vect{x}_1 - \vect{x}_2}{d(\vect{x}_1, \vect{x}_2)}} f(\Lambda_{12}) + \vect{n} \cdot \brac{\frac{\vect{x}_1 - \vect{x}_3}{d(\vect{x}_1, \vect{x}_3)} + \frac{\vect{x}_1 - \vect{x}_4}{d(\vect{x}_1, \vect{x}_4)}} f(\Lambda_{34}) \\
    ={}& f(\Lambda_{12}) + \vect{n} \cdot \brac{\frac{\vect{x}_1 - \vect{x}_2}{d(\vect{x}_1, \vect{x}_2)} f(\Lambda_{12}) + \frac{\vect{x}_1 - \vect{x}_3}{d(\vect{x}_1, \vect{x}_3)} f(\Lambda_{34}) + \frac{\vect{x}_1 - \vect{x}_4}{d(\vect{x}_1, \vect{x}_4)} f(\Lambda_{34})}.
  \end{aligned}
  \end{equation}
  Here, $\Lambda_{12}$ and $\Lambda_{34}$ denote the stretch ratio of two chains evaluated at $\vect{x} = \vect{x}_1$, given by
  \begin{equation}
    \Lambda_{12} = \frac{\norm{\vect{x}_1 - \vect{x}_2}}{L_{12}}, \quad \Lambda_{34} = \frac{\norm{\vect{x}_1 - \vect{x}_3} + \norm{\vect{x}_1 - \vect{x}_4}}{L_{34}}.
  \end{equation}
  We observe that $\brac{\vect{n} \cdot \frac{\partial U_e}{\partial \vect{x}}} \bigg|_{\vect{x} = \vect{x}_{\varepsilon}}$ is a function of the direction vector $\vect{n}$, and also depends implicitly on the displacement $\delta$ through the position of the node $\vect{x}_2$.
  
  To ensure that $\brac{\vect{n} \cdot \frac{\partial U_e}{\partial \vect{x}}} \bigg|_{\vect{x} = \vect{x}_{\varepsilon}}\ge0$ for all unit vector $\vect{n}\in\mathbb{R}^2$ and $\norm{\vect{n}}=1$, we consider the worst-case scenario by taking the minimum of expression \eqref{eq:directional-derivative-1} with respect to $\vect{n}$. This minimum occurs when $\vect{n}$ is parallel to and points in the inverse direction of the vector $\frac{\vect{x}_1 - \vect{x}_2}{d(\vect{x}_1, \vect{x}_2)} f(\lambda_{12}) + \frac{\vect{x}_1 - \vect{x}_3}{d(\vect{x}_1, \vect{x}_3)} f(\lambda_{34}) + \frac{\vect{x}_1 - \vect{x}_4}{d(\vect{x}_1, \vect{x}_4)} f(\lambda_{34})$. Therefore, the minimum value of $\brac{\vect{n} \cdot \frac{\partial U_e}{\partial \vect{x}}} \bigg|_{\vect{x} = \vect{x}_{\varepsilon}}$ is
  \begin{equation}
    \min_{\substack{\vect{n}\in\mathbb{R}^2 \\ \norm{\vect{n}}=1}} \brac{\vect{n} \cdot \frac{\partial U_e}{\partial \vect{x}}} \bigg|_{\vect{x} = \vect{x}_{\varepsilon}} = f(\lambda_{12}) - \norm{\frac{\vect{x}_1 - \vect{x}_2}{d(\vect{x}_1, \vect{x}_2)} f(\lambda_{12}) + \brac{\frac{\vect{x}_1 - \vect{x}_3}{d(\vect{x}_1, \vect{x}_3)} + \frac{\vect{x}_1 - \vect{x}_4}{d(\vect{x}_1, \vect{x}_4)}} f(\lambda_{34})}.
  \end{equation}
  Therefore, to determine whether $\vect{x}_1$ is a local minimum, we simply check whether this quantity is non-negative.

  We can write the above expression explicitly by substituting the coordinates of the nodes $P_i$ for $i=1,2,3,4$. Given
  \begin{equation}
    \vect{x}_1 = (0, 1), \quad \vect{x}_2 = (1, 2 + \delta), \quad \vect{x}_3 = (1, 0), \quad \vect{x}_4 = (2, 1),
  \end{equation}
  we compute:
  \begin{equation}
    \vect{x}_1 - \vect{x}_2 = (-1, -1-\delta), \quad \vect{x}_1 - \vect{x}_3 = (-1, 1), \quad \vect{x}_1 - \vect{x}_4 = (-2, 0),
  \end{equation}
  and the corresponding distances:
  \begin{equation}
    d(\vect{x}_1, \vect{x}_2) = \sqrt{1 + (1 + \delta)^2}, \quad d(\vect{x}_1, \vect{x}_3) = \sqrt{2}, \quad d(\vect{x}_1, \vect{x}_4) = 2.
  \end{equation}
  The stretch ratios are
  \begin{equation}
    \Lambda_{12} = \frac{\sqrt{1 + (1 + \delta)^2}}{2}, \quad \Lambda_{34} = \frac{\sqrt{2} + 2}{2}.
  \end{equation}
  Substituting into the expression for the minimum of the directional derivative, we have:
  \begin{equation}
  \begin{aligned}
    & \min_{\substack{\vect{n}\in\mathbb{R}^2 \\ \norm{\vect{n}}=1}} \brac{\vect{n} \cdot \frac{\partial U_e}{\partial \vect{x}}} \bigg|_{\vect{x} = \vect{x}_{\varepsilon}} \\
    ={}& f({\frac{\sqrt{1 + (1 + \delta)^2}}{2}}) - \norm{\frac{(-1, -1-\delta)}{\sqrt{1 + (1 + \delta)^2}} f(\frac{\sqrt{1 + (1 + \delta)^2}}{2}) + \brac{\frac{(-1, 1)}{\sqrt{2}} + \frac{(-2, 0)}{2}} f(\frac{\sqrt{2} + 2}{2})} \\
    ={}& f(\frac{\sqrt{1 + (1 + \delta)^2}}{2}) - \norm{\frac{(-1, -1-\delta)}{\sqrt{1 + (1 + \delta)^2}} f(\frac{\sqrt{1 + (1 + \delta)^2}}{2}) + \brac{-\frac{1}{\sqrt{2}}-1, \frac{1}{\sqrt{2}}} f(\frac{\sqrt{2} + 2}{2})} \\
    ={}& f(\frac{d_{\delta}}{2}) - \norm{\frac{(-1, -1-\delta)}{d_{\delta}} f(\frac{d_{\delta}}{2}) + (-\frac{1}{\sqrt{2}}-1, \frac{1}{\sqrt{2}}) f(\frac{\sqrt{2} + 2}{2})} \\
    ={}& f(\frac{d_{\delta}}{2}) - \norm{(-\frac{f(\frac{d_{\delta}}{2})}{d_{\delta}} + (-\frac{1}{\sqrt{2}}-1)f(\frac{\sqrt{2}+2}{2}), -\frac{(1+\delta)f(\frac{d_{\delta}}{2})}{d_{\delta}} + \frac{1}{\sqrt{2}}f(\frac{\sqrt{2}+2}{2}))} \\
    ={}& f(\frac{d_{\delta}}{2}) - \sqrt{\brac{\frac{f(\frac{d_{\delta}}{2})}{d_{\delta}} + (\frac{\sqrt{2}+2}{2})f(\frac{\sqrt{2}+2}{2})}^2 + \brac{\frac{(1+\delta)f(\frac{d_{\delta}}{2})}{d_{\delta}} - \frac{1}{\sqrt{2}}f(\frac{\sqrt{2}+2}{2})}^2}.
  \end{aligned}
  \end{equation}
  where we denote $d_{\delta} := \sqrt{1 + (1 + \delta)^2}$. 
  
  To analyze the sign of the above expression, we define the function $h = h(\delta)$ to be the minimum of the directional derivative and check if it is non-negative.
  \begin{equation}\label{eq:function-h-delta}
  \begin{aligned}      
    h(\delta) :={}& f(\frac{d_{\delta}}{2}) - \sqrt{\brac{\frac{f(\frac{d_{\delta}}{2})}{d_{\delta}} + (\frac{\sqrt{2}+2}{2})f(\frac{\sqrt{2}+2}{2})}^2 + \brac{\frac{(1+\delta)f(\frac{d_{\delta}}{2})}{d_{\delta}} - \frac{1}{\sqrt{2}}f(\frac{\sqrt{2}+2}{2})}^2} \\
    ={}& \frac{\brac{f(\frac{d_{\delta}}{2})}^2 - {\brac{\frac{f(\frac{d_{\delta}}{2})}{d_{\delta}} + (\frac{\sqrt{2}+2}{2})f(\frac{\sqrt{2}+2}{2})}^2 - \brac{\frac{(1+\delta)f(\frac{d_{\delta}}{2})}{d_{\delta}} - \frac{1}{\sqrt{2}}f(\frac{\sqrt{2}+2}{2})}^2}}{f(\frac{d_{\delta}}{2}) + \sqrt{\brac{\frac{f(\frac{d_{\delta}}{2})}{d_{\delta}} + (\frac{\sqrt{2}+2}{2})f(\frac{\sqrt{2}+2}{2})}^2 + \brac{\frac{(1+\delta)f(\frac{d_{\delta}}{2})}{d_{\delta}} - \frac{1}{\sqrt{2}}f(\frac{\sqrt{2}+2}{2})}^2}} \\
    ={}& \frac{- 2 \brac{\frac{f(\frac{d_{\delta}}{2})}{d_{\delta}}} (\frac{\sqrt{2}+2}{2})f(\frac{\sqrt{2}+2}{2}) + 2 \frac{(1+\delta)f(\frac{d_{\delta}}{2})}{d_{\delta}} \frac{1}{\sqrt{2}}f(\frac{\sqrt{2}+2}{2}) - ((\frac{\sqrt{2}+2}{2})^2 + (\frac{1}{\sqrt{2}})^2)(f(\frac{\sqrt{2}+2}{2}))^2 }{f(\frac{d_{\delta}}{2}) + \sqrt{\brac{\frac{f(\frac{d_{\delta}}{2})}{d_{\delta}} + (\frac{\sqrt{2}+2}{2})f(\frac{\sqrt{2}+2}{2})}^2 + \brac{\frac{(1+\delta)f(\frac{d_{\delta}}{2})}{d_{\delta}} - \frac{1}{\sqrt{2}}f(\frac{\sqrt{2}+2}{2})}^2}} \\
    ={}& \frac{(\sqrt{2}\delta - 2) \frac{f(\frac{d_{\delta}}{2})}{d_{\delta}} f(\frac{\sqrt{2}+2}{2}) - (\sqrt{2}+2)(f(\frac{\sqrt{2}+2}{2}))^2 }{f(\frac{d_{\delta}}{2}) + \sqrt{\brac{\frac{f(\frac{d_{\delta}}{2})}{d_{\delta}} + (\frac{\sqrt{2}+2}{2})f(\frac{\sqrt{2}+2}{2})}^2 + \brac{\frac{(1+\delta)f(\frac{d_{\delta}}{2})}{d_{\delta}} - \frac{1}{\sqrt{2}}f(\frac{\sqrt{2}+2}{2})}^2}}    
  \end{aligned}
  \end{equation}
  Now we analyze the sign of the numerator of $h(\delta)$. As $\delta\rightarrow+\infty$, we have 
  \begin{equation}
    \lim_{\delta\rightarrow+\infty}\frac{\sqrt{2}\delta - 2}{d_{\delta}} = \sqrt{2}
  \end{equation}
  and
  \begin{equation}
    \lim_{\delta\rightarrow+\infty}f\brac{\frac{d_{\delta}}{2}} = +\infty
  \end{equation}
  assuming that $f(\Lambda)\rightarrow+\infty$ as $\Lambda\rightarrow+\infty$.
  Therefore, for sufficitly large $\delta$, the numerator is positive. Notice that the denominator is always positive. Therefore, for sufficitly large $\delta$, we have $h(\delta) > 0$. 

  Thus, for sufficiently large displacement $\delta$, the directional derivative at $\vect{x}=\vect{x}_1$ remains positive in all directions, and hence $\vect{x}_1$ is a local minimum of the energy function. However, the energy function is non-differentiable at this point due to the singularity in the gradient. This demonstrates that $U_e$ has a non-differentiable local minimum for large displacement.

  We further prove that the energy function is Lipschitz continuous at this local minimum $\vect{x}_1$. It suffices to check if the directional derivative $\brac{\vect{n} \cdot \frac{\partial U_e}{\partial \vect{x}}} \bigg|_{\vect{x} = \vect{x}_{\varepsilon}}$ is finite for any unit vector $\vect{n}\in\mathbb{R}^2$ with $\norm{\vect{n}} = 1$, 
  where $\vect{x}_{\varepsilon} := \vect{x}_1 + \varepsilon \vect{n}$ for $0 < \varepsilon \ll 1$. This holds true by noticing the approximation \eqref{eq:directional-derivative-1}.
\end{proof}

We also numerically validate Theorem \ref{thm-supp:non-smoothness-energy-two-chain} by evaluating several specific choices of the constitutive laws. In Figure \ref{fig:function-h-delta}, we plot the minimum of the directional derivative, denoted by $h=h(\delta)$,  for the following cases: (1) the linear law: $f(\Lambda) = \Lambda - 1$; (2) the quadratic law: $f(\Lambda) = (\Lambda - 1)^2$; (3) the cubic law: $f(\Lambda) = (\Lambda - 1)^3$. The results show that the minimum of the directional derivative is negative for small $\delta$ and becomes positive as $\delta$ increases. Moreover, as the nonlinearity of the constitutive law increases, the transition point shifts to smaller values of $\delta$.
\begin{figure}[H]
    \centering
    \includegraphics[width=1.0\textwidth]{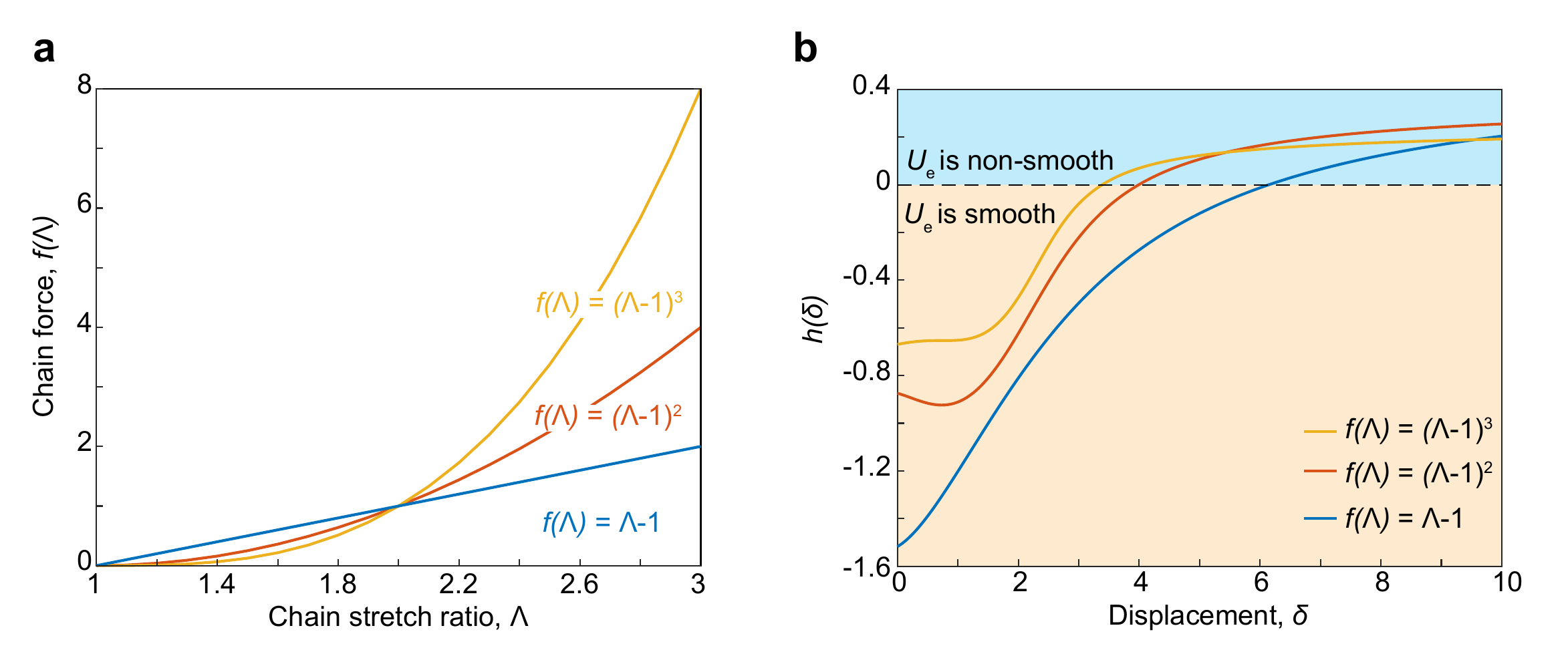}
    \caption{\textbf{The minimum of the directional derivative of the energy function for different constitutive laws.} \textbf{(a)} Different constitutive laws: (1) the linear case $f(\Lambda) = \Lambda - 1$; (2) the quadratic case $f(\Lambda) = (\Lambda - 1)^2$; (3) the cubic case $f(\Lambda) = (\Lambda - 1)^3$. \textbf{(b)} The minimum of the directional derivative, denoted by the function $h(\delta)$ for different constitutive laws. The minimum of the directional derivative is negative for small $\delta$ and becomes positive as $\delta$ increases. Moreover, as the nonlinearity of the constitutive law increases, the transition point shifts to smaller values of $\delta$.}
    \label{fig:function-h-delta}
\end{figure}  

Physically, the transition from smoothness to non-smoothness in the energy function is associated with the sliding-to-locking behavior of the entanglement. When the displacement $\delta$ is small, the entanglement allows for free sliding with a smooth energy landscape. As the displacement $\delta$ increases, the chains become locked, and the entanglement stops sliding, with a non-smooth energy landscape. From a topological perspective, locking transforms the system from slidable with two chains to non-slidable with three chains, thereby altering the network topology (see Figure \ref{fig-supp:topologychange}). This observation further suggests that, under large deformations, entangled networks may be analyzed through topologically equivalent reduced networks with fewer entanglements.

\begin{figure}
  \centering
  \includegraphics[width=0.9\textwidth]{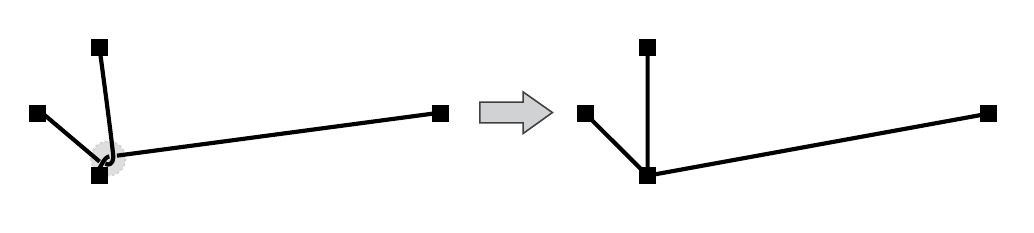}
  \caption{\textbf{Topology change in the two-chain model with slidable node.}
  Left: The two-chain model consists of two chains entangled at a slidable node.
  Right: After locking, the system can be equivalently represented as a spring network with three chains, where the slidable node becomes non-slidable.}
  \label{fig-supp:topologychange}
\end{figure}

We also note that, from a Newtonian mechanics perspective, the non-smoothness of the energy function may cause the force (i.e., the gradient of the energy) to be ill-defined at the slidable node when the system becomes locked. This, however, does not result in any inconsistency for our model, which is constructed from an energy perspective. Physically, at locking, the slidable node experiences a contact force that maintains equilibrium. While our model does not explicitly account for this contact force, it still correctly captures the equilibrium state from the perspective of the energy formulation. This highlights the elegance of our model in capturing the fundamental physics with conceptual simplicity.

Numerically, the non-smoothness of the energy function poses challenges for designing efficient optimization algorithms to compute its minimum. In particular, it prevents the direct application of Newton's method. To address this issue, we employ a specially designed gradient descent method to minimize the energy function. More details of the algorithm will be presented in the next section.

\section{Numerical algorithm}\label{sec:numerical-method}

In this section, we present the numerical algorithm to simulate the entangled network, i.e., numerically solve the optimization problem to find the minimum of the energy function.
In Section~\ref{sec:data-structure-entangle-network}, we introduces the data structure of the entangled network, including the generation of nodes and chains. In Section~\ref{sec:optimization-algorithm-two-chain}, we investigates optimization methods for the two-chain model, showing that the gradient descent method succeeds while Newton’s method fails. In Section~\ref{sec:optimization-algorithm-entangle-network}, we extends the gradient descent method to entangled networks and proposes strategies to accelerate convergence.  Section~\ref{sec:computation-energy-gradient} presents the computation of the gradient of the total energy, which is essential for implementing the optimization algorithm. Finally, we include the construction of the nonlinear constitutive law in Section~\ref{sec:construct-nonlinear-constitutive-law}.

\subsection{Data structure of the entangled network}\label{sec:data-structure-entangle-network}

For a classical graph, the most commonly used data structure is an adjacency list (or dictionary), where each node stores a list of its adjacent nodes, or a 2D adjacency matrix of size $N\times N$, where \texttt{matrix[i][j] = 1} if there is an edge connecting the $i$-th and $j$-th nodes. These representations are sufficient because the set of nodes and their adjacency relations uniquely determine the topology of a classical graph.

In contrast, representing an entangled network is significantly more complex. This complexity arises from the fact that each chain in the entangled network may consist of varying number of multiple nodes, and each node can be either slidable or non-slidable. Consequently, the topology of an entangled network cannot be fully described by simple pairwise adjacency, and additional structural information is essential to capture the configuration of the network.

Our data structure of the entangled network is motivated by two key considerations. First, as discussed in Section~\ref{sec:graph-entangle-network}, the topological structure of an entangled network is uniquely determined by four components: (1) the set of slidable nodes; (2) the set of non-slidable nodes; (3) the set of adjacent nodes for each node; and (4) the orientation of the slidable node. Second, our optimization algorithm requires the gradient information of the energy function, which in turn depends on the local connectivity of nodes, as detailed in Theorem~\ref{thm:gradient-energy-single-chain} in Section~\ref{sec:computation-energy-gradient}.

Guided by these considerations,  we use a three-dimensional tensor \texttt{chains} to efficiently store the chain connectivity. The first index of the tensor corresponds to the node index, the second to the index of the chain associated with that node (since multiple chains may connect to a single node), and \texttt{chains[i][j]} stores an array containing all nodes in the $j$-th chain connected to the $i$-th node.

Then the construction of the entangled network proceeds in two steps: (1) generation of nodes; (2) generation of chains. We also note that, although our current implementation focuses on 2D square lattices, the procedure can be easily adapted to random networks or 3D networks with minimal modification.

\vspace{1em}
\noindent\textbf{Step 1: Generation of nodes}

In this work, we restrict our simulations of entangled networks to a 2D square lattice. We choose the square lattice over other common lattice structures (e.g., the triangular or honeycomb lattices) because it is the simplest configuration in which each node is connected to exactly four adjacent nodes. This connectivity allows the potential for each node to serve as a slidable node shared by two chains, as discussed in Proposition~\ref{prop-supp:internal-node-four-adjacent-vertices} in Section~\ref{sec:entangled-network}. Extensions to random networks and 3D cases are possible and will be explored in future work.
In the 2D square lattice, the nodes are arranged in a uniform grid with $N_x$ nodes in the $x$ direction and $N_y$ nodes in the $y$ direction. Therefore, the total number of nodes is $N = N_x N_y$.

Based on our previous discussion,  the structure of an entangled network is uniquely determined by four components: (1) the set of slidable nodes; (2) the set of non-slidable nodes; (3) the set of adjacent nodes for each node; and (4) the orientation of the slidable node. For the square lattice, the adjacent relations are fixed. In this work, we further simplify the setting by fixing the orientation of slidable nodes: the left and top neighbors belong to one chain, while the right and bottom neighbors belong to another chain (as shown in Figure \ref{fig-supp:orientation}a). The effect of alternative orientations (i.e., using mixed orientations in Figures \ref{fig-supp:orientation}a and \ref{fig-supp:orientation}b) will be explored in future work. Under this simplification, the only degrees of freedom are the choice of slidable nodes and non-slidable nodes in the present setting.

\begin{figure}
  \centering
  \includegraphics[width=0.75\textwidth]{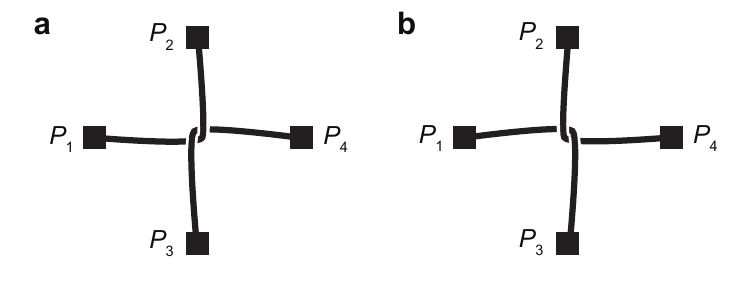}
  \caption{\textbf{Different orientations of slidable node.} 
  \textbf{(a)} The left ($P_1$) and top ($P_2$) neighboring nodes belong to one chain, while the right ($P_4$) and bottom ($P_3$) neighboring nodes belong to another chain.
  \textbf{(b)} The left ($P_1$) and bottom ($P_3$) neighboring nodes belong to one chain, while the right ($P_4$) and top ($P_2$) neighboring nodes belong to another chain.}
  \label{fig-supp:orientation}
\end{figure}

We consider two types of networks: periodic networks and random networks. The periodic network is constructed from the square lattice with a repeating periodic structures. For simplicity, we take a $2\times2$ lattice as a unit cell and replicate it to form the entire network. Other periodic lattice structures can be considered in the future work. For this $2\times2$ unit cell containing four nodes, we consider the following cases:
\begin{itemize}
  \item \textbf{Case I ($\varphi_s = 0\%$):} All four nodes are non-slidable nodes, corresponding to a spring network with $\varphi_s = 0\%$. See Figure \ref{fig-supp:SI_Fig8_0_Period_Network} for the simulation setup of elasticity and Figure \ref{fig-supp:SI_Fig27_CrackNetworkConfiguration}a for the simulation setup of fracture.
  
  \item \textbf{Case II ($\varphi_s = 25\%$):} Three nodes are non-slidable and one node is slidable, corresponding to an entangled network with $\varphi_s = 25\%$. See Figure \ref{fig-supp:SI_Fig9_25_Period_Network} for the simulation setup of elasticity and Figure \ref{fig-supp:SI_Fig27_CrackNetworkConfiguration}b for the simulation setup of fracture.
  
  \item \textbf{Case III ($\varphi_s = 50\%$):} Two nodes are non-slidable and the other two are slidable, corresponding to an entangled network with $\varphi_s = 50\%$. See Figure \ref{fig-supp:SI_Fig10_50_Period_Network} for the simulation setup of elasticity and Figure \ref{fig-supp:SI_Fig27_CrackNetworkConfiguration}c for the simulation setup of fracture.
  
  \item \textbf{Case IV ($\varphi_s = 75\%$):} One node is non-slidable and the other three nodes are slidable, corresponding to an entangled network with $\varphi_s = 75\%$. See Figure \ref{fig-supp:SI_Fig11_75_Period_Network} for the simulation setup of elasticity and Figure \ref{fig-supp:SI_Fig27_CrackNetworkConfiguration}d for the simulationsetup  of fracture.
\end{itemize}

To generate the nodes in a random network, each node is assigned as either slidable or non-slidable according to a prescribed probability $0 \le \varphi_s \le 1$, which represents the fraction of slidable nodes in the network. The procedure is as follows: all boundary nodes are fixed as non-slidable. For each internal nodes (inside the domain), we independently generate a random number $r$ from the uniform distribution $U[0, 1]$. If $r < \varphi_s$,  the node is assigned as slidable; otherwise, it is treated as non-slidable. In the implementation, an array \texttt{isSlidable} of size $N$ stores the node types, with \texttt{1} denoting a slidable node and \texttt{0} denoting a non-slidable node. The algorithm for generating the random network is given in Algorithm \ref{alg:generate-random-network}.
\begin{algorithm}[H]
\caption{Algorithm for generating nodes in a random network}\label{alg:generate-random-network}
\begin{algorithmic}[1]
\For{every node $i$ in the network}
\If{$i$ is a boundary node} \State \texttt{isSlidable[$i$] $\leftarrow$ 0}
\Else
\State \texttt{Generate a random number $r\in U(0, 1)$}
\If{$r < \varphi_s$}
\State \texttt{isSlidable[$i$] $\leftarrow$ 1}
\Else
\State \texttt{isSlidable[$i$] $\leftarrow$ 0}
\EndIf
\EndIf
\EndFor
\end{algorithmic}
\end{algorithm}

\vspace{1em}
\noindent\textbf{Step 2: Generation of chains}

After generating the nodes with their assigned types (non-slidable or slidable) and specifying orientations to the slidable nodes, we proceed to construct the chains that connect them. To efficiently store the chain connectivity, we use a three-dimensional tensor \texttt{chains}. The first index of the tensor corresponds to the node index, the second to the index of the chain associated with that node, and \texttt{chains[i][j]} stores an array containing all nodes in the $j$-th chain connected to the $i$-th node.

The construction of the tensor \texttt{chains} relies on two key subroutines:
\begin{itemize}
\item
\texttt{getNextNodeInChain}: Given a node index $i$ and one of its neighboring slidable nodes $j$ (if such a neighbor exists), it returns the next neighboring node $k$ of $j$ along the same chain that contains both $i$ and $j$. Since all slidable node orientations are specified in advance, this step is always well-defined.

\item
\texttt{getNodeListInChain}: Given a non-slidable node $i$ and one of its neighboring slidable nodes $j$, it returns the complete ordered list of nodes in the chain containing $i$ and $j$ by recursively calling \texttt{getNextNodeInChain} until a terminating non-slidable node is reached.
\end{itemize}

With these two subroutines in place, we are ready to present the algorithm for generating the tensor \texttt{chains}, which constructs the chain connectivity for each node in the network based on its type (slidable or non-slidable). The algorithm iterates over all nodes and applies the following procedure:
\begin{itemize}
    \item For non-slidable node, the algorithm examines each of its neighboring node and constructs the list of nodes along the chain starting from the current non-slidable node and passing through its neighbor. This list is then stored in the corresponding entry of \texttt{chains}.
    
    \item For slidable node, its neighboring nodes are grouped into two parts based on the prescribed orientation. For each group, the algorithm traces the sequence of nodes starting from the slidable node and extending in both directions until it reaches non-slidable nodes at both ends. These lists are merged and stored in \texttt{chains[i][0]} and \texttt{chains[i][1]}, respectively.
\end{itemize}
The details are summarized in Algorithm \ref{alg:generate-network-chain-info}.
\begin{algorithm}[H]
\caption{The algorithm for generating chain information}\label{alg:generate-network-chain-info}
\begin{algorithmic}[1]
\For{each node $i$ in the network}
\If{$i$ is non-slidable node}
  \For{each neighbouring node $j$ of $i$}
  \State Get the node list in the chain starting from $i$ and $j$ using \texttt{getNodeListInChain}
  \State Add this list to \texttt{chains[i][j]}
  \EndFor
\ElsIf{$i$ is slidable node}
\State Get four neighbouring nodes of $i$ and partition them into two chains $\{v_{i_1}, v_{i_2}\}$ and $\{v_{i_3}, v_{i_4}\}$ based on orientation
\State Get the node list $l_1$ in the chain starting from $i$ and $i_1$ using \texttt{getNodeListInChain}
\State Get the node list $l_2$ in the chain starting from $i$ and $i_2$ using \texttt{getNodeListInChain}
\State Merge $l_1$ and $l_2$, and add to \texttt{chains[i][0]}
\State Get the node list $l_3$ in the chain starting from $i$ and $i_3$ using \texttt{getNodeListInChain}
\State Get the node list $l_4$ in the chain starting from $i$ and $i_4$ using \texttt{getNodeListInChain}
\State Merge $l_3$ and $l_4$, and add to \texttt{chains[i][1]}
\EndIf
\EndFor
\end{algorithmic}
\end{algorithm}

\subsection{Optimization algorithm for two-chain model}\label{sec:optimization-algorithm-two-chain}

In this part, we consider the optimization problem for the two-chain model introduced in Section~\ref{sec-supp:two-chain-model}. The goal is to find the equilibrium state of the network by determining the minimum of the energy function. We begin by applying Newton's method (see e.g., \cite{boyd2004convex}) to solve this optimization problem and demonstrate that it fails, along with an explanation of the underlying reasons. We then proceed to solve the problem using the gradient descent method.

To compute the minimum of the energy function $U_e = U_e(\vect{x})$ given in \eqref{eq:energy-two-fibers-entangle}, Netwon's method is given by
\begin{equation}\label{eq-supp:newton-method}
    \vect{x}^{(n+1)} = \vect{x}^{(n)} - \brac{\nabla^2 U_e(\vect{x}^{(n)})}^{-1} \nabla U_e(\vect{x}^{(n)}), \quad n = 0,1,2,\dots,
\end{equation}
equipped with some initial guess $\vect{x}=\vect{x}^{(0)}$. Here $n$ denotes the iteration step, $\nabla U_e(\vect{x}^{(n)})$ and $\nabla^2 U_e(\vect{x}^{(n)})$ denote the gradient and Hessian matrix of the energy $U_e$ at the $n$-th iteration step.

Next, we derive the gradient and Hessian matrix of the energy $U_e$ given in \eqref{eq:energy-two-fibers-entangle} for the two-chain model. The energy $U_e$ is given by
\begin{equation}
  U_e(\vect{x}) = L_{12} \int_{1}^{\Lambda_{12}} f(\Lambda') d\Lambda' + L_{34} \int_{1}^{\Lambda_{34}} f(\Lambda') d\Lambda'.
\end{equation}
Taking derivative of $U_e(\vect{x})$ with respect to $\vect{x}$, we have
\begin{equation}\label{eq:gradient-entangle-two-chain}
\begin{aligned}
  \nabla U_e ={}& L_{12} \frac{\partial \Lambda_{12}}{\partial \vect{x}} f(\Lambda_{12}) + L_{34} \frac{\partial \Lambda_{34}}{\partial \vect{x}} f(\Lambda_{34}) \\
  ={}& L_{12} \frac{\partial}{\partial \vect{x}}\brac{\frac{1}{L_{12}}\brac{d(\vect{x}, \vect{x}_1) + d(\vect{x}, \vect{x}_2)}} f(\Lambda_{12}) \\
  {}& + L_{34} \frac{\partial}{\partial \vect{x}}\brac{\frac{1}{L_{34}}\brac{d(\vect{x}, \vect{x}_3) + d(\vect{x}, \vect{x}_4)}} f(\Lambda_{34}) \\    
  = {}& \brac{\frac{\vect{x} - \vect{x}_1}{d(\vect{x}, \vect{x}_1)} + \frac{\vect{x} - \vect{x}_2}{d(\vect{x}, \vect{x}_2)}} f(\Lambda_{12}) + \brac{\frac{\vect{x} - \vect{x}_3}{d(\vect{x}, \vect{x}_3)} + \frac{\vect{x} - \vect{x}_4}{d(\vect{x}, \vect{x}_4)}} f(\Lambda_{34}).    
\end{aligned}        
\end{equation}
Then we compute the Hessian matrix of the energy $U_e$. For the first term $\brac{\frac{\vect{x} - \vect{x}_1}{d(\vect{x}, \vect{x}_1)} + \frac{\vect{x} - \vect{x}_2}{d(\vect{x}, \vect{x}_2)}} f(\Lambda_{12})$ in the gradient \eqref{eq:gradient-entangle-two-chain}, its gradient is given by
\begin{equation}
\begin{aligned}
    {}& \frac{\partial}{\partial \vect{x}} \brac{\brac{\frac{\vect{x} - \vect{x}_1}{d(\vect{x}, \vect{x}_1)} + \frac{\vect{x} - \vect{x}_2}{d(\vect{x}, \vect{x}_2)}} f(\Lambda_{12})} \\
    ={}& \brac{\frac{\vect{I} d(\vect{x}, \vect{x}_1) - (\vect{x} - \vect{x}_1) \otimes\frac{\vect{x} - \vect{x}_1}{d(\vect{x}, \vect{x}_1)}}{(d(\vect{x}, \vect{x}_1))^2} + \frac{\vect{I} d(\vect{x}, \vect{x}_2) - (\vect{x} - \vect{x}_2)\otimes \frac{\vect{x} - \vect{x}_2}{d(\vect{x}, \vect{x}_2)}}{(d(\vect{x}, \vect{x}_2))^2}} f(\Lambda_{12}) \\
    & + \brac{\frac{\vect{x} - \vect{x}_1}{d(\vect{x}, \vect{x}_1)} + \frac{\vect{x} - \vect{x}_2}{d(\vect{x}, \vect{x}_2)}} \frac{\partial \Lambda_{12}}{\partial \vect{x}} f'(\Lambda_{12}) \\
    ={}& \brac{\frac{\vect{I} d(\vect{x}, \vect{x}_1) - (\vect{x} - \vect{x}_1) \otimes \frac{\vect{x} - \vect{x}_1}{d(\vect{x}, \vect{x}_1)}}{(d(\vect{x}, \vect{x}_1))^2} + \frac{\vect{I} d(\vect{x}, \vect{x}_2) - (\vect{x} - \vect{x}_2) \otimes \frac{\vect{x} - \vect{x}_2}{d(\vect{x}, \vect{x}_2)}}{(d(\vect{x}, \vect{x}_2))^2}} f(\Lambda_{12}) \\
    & + \brac{\frac{\vect{x} - \vect{x}_1}{d(\vect{x}, \vect{x}_1)} + \frac{\vect{x} - \vect{x}_2}{d(\vect{x}, \vect{x}_2)}} \otimes \brac{\frac{\vect{x} - \vect{x}_1}{d(\vect{x}, \vect{x}_1)} + \frac{\vect{x} - \vect{x}_2}{d(\vect{x}, \vect{x}_2)}} \frac{1}{L_{12}} f'(\Lambda_{12}) \\    
    ={}& \brac{\frac{\vect{I}(d(\vect{x}, \vect{x}_1))^2 - (\vect{x} - \vect{x}_1)\otimes(\vect{x} - \vect{x}_1)}{(d(\vect{x}, \vect{x}_1))^3} + \frac{\vect{I}(d(\vect{x}, \vect{x}_2))^2 - (\vect{x} - \vect{x}_2)\otimes(\vect{x} - \vect{x}_2)}{(d(\vect{x}, \vect{x}_2))^3}} f(\Lambda_{12}) \\
    & + \brac{\frac{\vect{x} - \vect{x}_1}{d(\vect{x}, \vect{x}_1)} + \frac{\vect{x} - \vect{x}_2}{d(\vect{x}, \vect{x}_2)}} \otimes \brac{\frac{\vect{x} - \vect{x}_1}{d(\vect{x}, \vect{x}_1)} + \frac{\vect{x} - \vect{x}_2}{d(\vect{x}, \vect{x}_2)}} \frac{1}{L_{12}} f'(\Lambda_{12}).
\end{aligned}
\end{equation}
Here $\vect{I}$ denotes the identity tensor and $\otimes$ denotes the tensor product. Similarly, we can derive the gradient of the second term $\brac{\frac{\vect{x} - \vect{x}_3}{d(\vect{x}, \vect{x}_3)} + \frac{\vect{x} - \vect{x}_4}{d(\vect{x}, \vect{x}_4)}} f(\Lambda_{34})$ in \eqref{eq:gradient-entangle-two-chain}. Consequently, the Hessian matrix of the energy $U_e$ is given by
\begin{equation}\label{eq:hessian-matrix-entangle-two-chain}
\begin{aligned}
    & \nabla^2 U_e \\
    ={}& \brac{\frac{\vect{I}(d(\vect{x}, \vect{x}_1))^2 - (\vect{x} - \vect{x}_1) \otimes (\vect{x} - \vect{x}_1)}{(d(\vect{x}, \vect{x}_1))^3} + \frac{\vect{I}(d(\vect{x}, \vect{x}_2))^2 - (\vect{x} - \vect{x}_2)\otimes(\vect{x} - \vect{x}_2)}{(d(\vect{x}, \vect{x}_2))^3}} f(\Lambda_{12}) \\
    & + \brac{\frac{\vect{x} - \vect{x}_1}{d(\vect{x}, \vect{x}_1)} + \frac{\vect{x} - \vect{x}_2}{d(\vect{x}, \vect{x}_2)}}  \otimes \brac{\frac{\vect{x} - \vect{x}_1}{d(\vect{x}, \vect{x}_1)} + \frac{\vect{x} - \vect{x}_2}{d(\vect{x}, \vect{x}_2)}} \frac{1}{L_{12}} f'(\Lambda_{12}) \\
    & + \brac{\frac{\vect{I}(d(\vect{x}, \vect{x}_3))^2 - (\vect{x} - \vect{x}_3)\otimes(\vect{x} - \vect{x}_3)}{(d(\vect{x}, \vect{x}_3))^3} + \frac{\vect{I}(d(\vect{x}, \vect{x}_4))^2 - (\vect{x} - \vect{x}_4)\otimes(\vect{x} - \vect{x}_4)}{(d(\vect{x}, \vect{x}_4))^3}} f(\Lambda_{34}) \\
    & + \brac{\frac{\vect{x} - \vect{x}_3}{d(\vect{x}, \vect{x}_3)} + \frac{\vect{x} - \vect{x}_4}{d(\vect{x}, \vect{x}_4)}}\otimes \brac{\frac{\vect{x} - \vect{x}_3}{d(\vect{x}, \vect{x}_3)} + \frac{\vect{x} - \vect{x}_4}{d(\vect{x}, \vect{x}_4)}} \frac{1}{L_{34}} f'(\Lambda_{34}).
\end{aligned}
\end{equation}

With the gradient in \eqref{eq:gradient-entangle-two-chain} and the Hessian matrix in \eqref{eq:hessian-matrix-entangle-two-chain}, we use Newton's method \eqref{eq-supp:newton-method} to solve the energy minimization problem. Numerically, we observe that it converges to the equilibrium state when the displacement $\delta$ is small (see Figure  \ref{fig-supp:SI_Fig6_Slidable_linear_simulation}f). However, for large values of $\delta$, the method fails to converge (see Figure \ref{fig-supp:SI_Fig6_Slidable_linear_simulation}i). By examining the energy landscape (Figures \ref{fig-supp:SI_Fig6_Slidable_linear_simulation}g and \ref{fig-supp:SI_Fig6_Slidable_linear_simulation}h), we find that the failure is due to the non-smoothness of the energy surface near the minimum. 

In general, a necessary condition for the convergence of Newton's method when solving $g(x) = 0$ for a nonlinear function $g$ is that the second-order derivative $g''(x)$ should be continuous in a neighborhood of the root. Therefore, it is not surprising that Newton's method fails, as the energy function is not even differentiable at the minimum point. More precisely, the energy function becomes only Lipschitz continuous (see e.g., \cite{boyd2004convex}) and non-differentiable near the equilibrium state when the displacement $\delta$ is large.

To address the non-smoothness of the energy landscape, we use the gradient descent method (see e.g., \cite{boyd2004convex}) to solve the optimization problem:
\begin{equation}
    \vect{x}^{(n+1)} = \vect{x}^{(n)} - \alpha \nabla U_e(\vect{x}^{(n)}),
\end{equation}
where $\alpha>0$ is the step size. Our numerical experiments show that the gradient descent method successfully converges to the equilibrium state for both small and large displacements $\delta$ (see Figures \ref{fig-supp:SI_Fig6_Slidable_linear_simulation}c, \ref{fig-supp:SI_Fig6_Slidable_linear_simulation}f, and \ref{fig-supp:SI_Fig6_Slidable_linear_simulation}i). From an optimization perspective, the gradient descent method as a first-order method is more robust than Newton’s method as a second-order method, when applied to non-smooth objective functions.

In addition, for the energy minimization problem in entangled networks, the gradient descent method is favored over Newton's method due to the difficulty of computing the Hessian matrix of the energy function, particularly for entangled networks with long chains and multiple entangled nodes. As demonstrated in \eqref{eq:hessian-matrix-entangle-two-chain}, the derivation of the Hessian matrix for the simplified two-chain model is already highly complicated.

For completeness, we also present numerical results for the two-chain model with non-slidable node. To save space, we omit the derivation of the gradient and Hessian matrix of the energy function for the non-slidable model, as the procedure is similar to that of the slidable case. Numerical experiments show that both Newton's method and the gradient descent method converge to the equilibrium state (see Figure \ref{fig-supp:SI_Fig5_Nonslidable_linear_simulation}c). Furthermore, the energy landscape is smooth in this case (see Figures \ref{fig-supp:SI_Fig5_Nonslidable_linear_simulation}e and \ref{fig-supp:SI_Fig5_Nonslidable_linear_simulation}f), indicating that the corresponding optimization problem is significantly easier to solve than that of the entangled network.

We also note that, for the two-chain model with a slidable node under small displacement, where both the gradient descent method and Newton’s method are applicable, gradient descent converges significantly more slowly, requiring many more iterations to reach equilibrium (see Figure \ref{fig-supp:SI_Fig6_Slidable_linear_simulation}f). This slower convergence is an intrinsic nature to first-order optimization methods. While this is not a serious issue for the two-chain toy model, it becomes a major challenge for entangled networks, where the degrees of freedom grow rapidly with increasing network size. In the next section, when extending to network systems, we will introduce several strategies to accelerate convergence.

\subsection{Optimization algorithm for entangled network}\label{sec:optimization-algorithm-entangle-network}

In the previous section, we have demonstrated that the gradient descent method can successfully capture the equilibrium state of the two-chain model while Newton's method fails. In this section, we extend the gradient descent method to general entangled networks and introduce several strategies to accelerate its convergence.

We first test the gradient descent method on entangled networks and observe numerically that it converges to the equilibrium state. However, it requires a large number of iterations, particularly for networks of large size or highly nonlinear large deformations such as fracture. Therefore, improving the efficiency of the algorithm is crucial for its practical application in mechanics.

To improve the efficiency, we first try different variants of the gradient descent method including the gradient descent with momentum \cite{rumelhart1986learning}, the Nesterov accelerated gradient \cite{nesterov1983method}, the Adam method \cite{kingma2014adam}, etc. Numerically, we find that the gradient descent with momentum is the most efficient method for our problem. The gradient descent with momentum is designed to improve the convergence speed and stability of standard gradient descent. It achieves this by incorporating a momentum term that helps the optimization process move more smoothly toward the minimum. Instead of updating parameters using only the current gradient, momentum accumulates past gradients in an exponentially weighted moving average, which helps accelerating convergence and 
smoothing updates. The gradient descent method with momentum is given by:
\begin{equation}\label{eq:gradient-descent-momentum}
\begin{aligned}
  \vect{v}^{(n+1)} ={}& \beta \vect{v}^{(n)} + \nabla U_e(\vect{x}^{(n)}), \\
  \vect{x}^{(n+1)} ={}& \vect{x}^{(n)} - \alpha \vect{v}^{(n+1)}.
\end{aligned}
\end{equation}
Here $\beta>0$ is the momentum parameter and we typically take $\beta=0.9$. $\vect{v}^{(n)}$ is the momentum and initialized to be zero vector.

There are several important considerations in implementing the gradient descent method with momentum in \eqref{eq:gradient-descent-momentum}. The first key factor is the choice of the step size $\alpha$, which controls the magnitude of the update in the direction of the gradient. If $\alpha$ is too small, the convergence is slow; if too large, the algorithm may oscillate around the minimum or even diverge. To address this, we adopt a step size decay strategy, commonly known as learning rate decay in machine learning, where the step size is gradually reduced during the optimization process. The idea is to use a relatively large step size initially for rapid exploration of the energy landscape, followed by smaller steps to refine the solution and ensure convergence. In our implementation, we set the initial step size to $\alpha_0 = 0.1$ and reduce it by a factor of 10 whenever the objective function (i.e., the total energy) fails to decrease. Additionally, to prevent the numerical instability caused by overly large step sizes, we monitor the solution for invalid values (e.g., NaN or Inf). If such values are detected, the step size is reduced by a factor of 10 and the optimization is restarted.

The second important issue is the choice of the initial guess. In our implementation, we initialize the network using the uniformly stretched configuration for both elasticity and fracture simulations (see Section~\ref{sec:application-elasticity} and Section~\ref{sec:application-fracture} for details). We find this initialization approach to be a robust and efficient choice for both elasticity and fracture.
We also test an alternative initialization strategy inspired by~\cite{deng2023nonlocal}, where the loading process is divided into incremental steps, and the solution from the previous step is used as the initial guess for the current step to accelerate convergence. While this method improves efficiency for elasticity problems, it performs poorly in fracture simulations of entangled networks. We suspect this is because the entangled nodes may become locked in previous configurations, preventing them from sliding freely in subsequent steps.

The third issue concerns the stopping criterion of the optimization algorithm. As demonstrated in Theorem~\ref{thm-supp:non-smoothness-energy-two-chain}, the energy function of the entangled network may be non-differentiable at the minimum, especially under large deformations. Consequently, using the norm of the gradient as a stopping condition is inappropriate, as the gradient may not vanish near the minimum. 
Instead, we monitor the total energy during the iterations: if the energy increases, the step size is too large and is reduced accordingly. The iteration is terminated once the time step becomes sufficiently small ($10^{-5}$ in our numerical experiments). This turns out to be a more robust criterion in the numerical experiments.

We summarize our optimization algorithm in Algorithm \ref{alg:gradient-descent-entangle-network}.
\begin{algorithm}
\caption{The algorithm for gradient descent with momentum}\label{alg:gradient-descent-entangle-network}
\begin{algorithmic}[1]
\State \texttt{Initialize $\alpha_0 = 0.1$, $\beta=0.9$, $N_t=10000$, $\varepsilon=10^{-3}$}
\State \texttt{Initialize $\alpha \leftarrow \alpha_0$}
\For{$i=0$ to 5}
  \For{$n=0$ to $N_t$ }
  \State \texttt{Initialize $\vect{v}^{(0)} \leftarrow \vect{0}$ and $\vect{x}^{(0)}$ as the configuration in uniform stretch}
  \State \texttt{Compute the gradient $\nabla U_e(\vect{x}^{(n)})$}
  \State \texttt{$\vect{v}^{(n+1)} \leftarrow \beta \vect{v}^{(n)} + \nabla U_e(\vect{x}^{(n)})$}
  \State \texttt{$\vect{x}^{(n+1)} \leftarrow \vect{x}^{(n)} - \alpha \vect{v}^{(n+1)}$}
    \If{$n\Mod{100}=0$}
    \State \texttt{Compute the energy $U_e(\vect{x}^{(n)})$}
      \If{$U_e(\vect{x}^{(n)}) > U_e(\vect{x}^{(n-100)})$}
        \State use smaller time step \texttt{$\alpha \leftarrow 0.1 \alpha$}
        \State \texttt{Break}
      \EndIf
    \EndIf
  \EndFor
  \State \texttt{$\alpha \leftarrow 0.1 \alpha$}
\EndFor
\end{algorithmic}
\end{algorithm}

\subsection{Computation of the energy gradient}\label{sec:computation-energy-gradient}

In this part, we present the computation of the gradient of the total energy of the entangled network, $\nabla U_e$, which is essential for implementing the optimization algorithm.

We start by considering a single chain in the network. We assume that the chain consists of $n$ vertices $(v_1, v_2, \cdots, v_n)$ with the coordinates in the deformed state denoted by $(\vect{x}_1, \vect{x}_2, \cdots, \vect{x}_n)$. The energy of the single chain is given by
\begin{equation}\label{eq:energy-single-chain}
  U_{\textrm{chain}} = L \int_{1}^{\Lambda} f(\Lambda') d\Lambda',
\end{equation}
where $L$ is the initial (undeformed) length of the chain and $\Lambda$ is the stretch ratio defined as
\begin{equation}
  \Lambda = \frac{1}{L} \brac{\sum_{i=1}^{n-1} d_{i,i+1}}
\end{equation}
with $d_{i,i+1}$ denotes the distance between the $i$-th and $(i+1)$-th node
\begin{equation}
  d_{i,i+1} = \norm{\vect{x}_i - \vect{x}_{i+1}} = \sqrt{(x_i - x_{i+1})^2 + (y_i - y_{i+1})^2}.
\end{equation}

We now compute the gradient of the energy $U_{\textrm{chain}}$ in \eqref{eq:energy-single-chain} with respect to the $i$-th node $\vect{x}_i$:
\begin{equation}\label{eq-supp:compute-gradient-1}
  \frac{\partial U_\textrm{chain}}{\partial \vect{x}_i} = L \frac{\partial \Lambda}{\partial \vect{x}_i} f(\Lambda) = L \frac{1}{L} \frac{\partial}{\partial \vect{x}_i} \brac{\sum_{j=1}^{n-1} d_{j,j+1}} f(\Lambda) = \frac{\partial}{\partial \vect{x}_i} \brac{\sum_{j=1}^{n-1} d_{j,j+1}} f(\Lambda).
\end{equation}
To compute $\frac{\partial}{\partial \vect{x}_i} \brac{\sum_{j=1}^{n-1} d_{j,j+1}}$ in \eqref{eq-supp:compute-gradient-1}, we consider two cases:
\begin{itemize}
  \item Case 1: If $v_i$ is a non-slidable node, i.e. $i=1$ or $i=n$, we have
  \begin{equation}
    \frac{\partial}{\partial \vect{x}_1} \brac{\sum_{j=1}^{n-1} d_{j,j+1}} = \frac{\partial d_{12}}{\partial \vect{x}_1} = \frac{\vect{x}_1 - \vect{x}_2}{d_{12}},
  \end{equation}
  or
  \begin{equation}
    \frac{\partial}{\partial \vect{x}_n} \brac{\sum_{j=1}^{n-1} d_{j,j+1}} = \frac{\partial d_{n-1,n}}{\partial \vect{x}_n} = \frac{\vect{x}_{n} - \vect{x}_{n-1}}{d_{n-1,n}}.
  \end{equation}

  \item Case 2: If $v_i$ is a slidable node, i.e. $1 < i < n$, we have
  \begin{equation}
    \frac{\partial}{\partial \vect{x}_i} \brac{\sum_{j=1}^{n-1} d_{j,j+1}} = \frac{\partial d_{i-1,i}}{\partial \vect{x}_i} + \frac{\partial d_{i,i+1}}{\partial \vect{x}_i} = \frac{\vect{x}_{i} - \vect{x}_{i-1}}{d_{i-1,i}} + \frac{\vect{x}_i - \vect{x}_{i+1}}{d_{i,i+1}}.
  \end{equation}
\end{itemize}
Here, we use Lemma \ref{lemma:distance-function-derivative} to compute the derivative of the distance function.
We summarize the above results in the following theorem.
\begin{theorem}[Gradient of the energy function for a single chain]\label{thm:gradient-energy-single-chain}
  The gradient of the energy function $U_{\mathrm{chain}}$ in \eqref{eq:energy-single-chain} with respect to the $i$-th node $\vect{x}_i$ is given by  
    \begin{equation}
    \frac{\partial U_{\textrm{chain}}}{\partial \vect{x}_i} = 
    \begin{cases}
      \frac{\vect{x}_1 - \vect{x}_2}{d_{12}} f(\Lambda), & i = 1, \\
      \brac{\frac{\vect{x}_{i} - \vect{x}_{i-1}}{d_{i-1,i}} + \frac{\vect{x}_i - \vect{x}_{i+1}}{d_{i,i+1}}} f(\Lambda), & 1 < i < n, \\
      \frac{\vect{x}_{n} - \vect{x}_{n-1}}{d_{n-1,n}} f(\Lambda), & i = n.
    \end{cases}
    \end{equation}
    Here, $f$ is the constitutive law of the chain, $\Lambda$ is the stretch ratio of the chain, and $d_{ij}$ is the distance between the $i$-th and $j$-th node.
    Alternatively, the above expression can be written compactly as
    \begin{equation}
      \frac{\partial U_{\textrm{chain}}}{\partial \vect{x}_i} = \sum_{j\in \mathcal{N}_i} \frac{\vect{x}_i - \vect{x}_j}{d_{ij}} f(\Lambda),
    \end{equation}
    where $\mathcal{N}_i$ denote the set of nodes adjacent to the $i$-th node in the current chain.
\end{theorem}

Once the gradient of the energy function  is computed for each individual chain, the gradient of the total energy function $U_e$ can be computed by summing up the contributions from all the chains. The pesudocode for the computation of the gradient is presented in Algorithm \ref{alg:compute-gradient-energy}. We use the two-dimensional array \texttt{grad} of size  $N\times 2$ to store the gradient of the energy function at each node, where $N$ is the total number of nodes in the network and $2$ is the spatial dimension. Note that this algorithm relies on the array \texttt{chains}, which contains the chain information generated by Algorithm \ref{alg:generate-network-chain-info} in Section \ref{sec:data-structure-entangle-network}.
\begin{algorithm}
\caption{Computation of the gradient of the total energy function}\label{alg:compute-gradient-energy}
\begin{algorithmic}[1]
\State \textbf{Input:} Node coordinates $\{\vect{x}_i\}_{i=1}^N$, chain information \texttt{chains}
\State \textbf{Output:} Gradient array \texttt{grad} of size $N \times 2$
\State Initialize \texttt{grad} to be a zero array of size $N \times 2$
\For{each node $i$ in the entangled network}
  \For{each chain $c$ containing the node $i$}
    \State Compute the length $L$ of chain $c$ in the undeformed state
    \State {Compute the length $l$ of the chain $c$ in the deformed state}
    \State {Compute the stretch ratio $\Lambda = l/L$}
    \For{each node $j$ in the chain $c$}
      \If{node $j$ is adjacent to node $i$}
        \State {Compute the distance $d_{ij} = \norm{\vect{x}_i - \vect{x}_j}$}        
        \State Update \texttt{grad[i]} by adding $\frac{\vect{x}_i - \vect{x}_j}{d_{ij}} f(\Lambda)$        
      \EndIf
    \EndFor
  \EndFor
\EndFor
\end{algorithmic}
\end{algorithm}

\subsection{Construction of nonlinear constitutive laws}\label{sec:construct-nonlinear-constitutive-law}

For the nonlinear constitutive law, it is natural to consider the inverse Langevin function, widely used in polymer physics \cite{rubinstein2003polymer}:
\begin{equation}
    f = L^{-1}\brac{\frac{\Lambda}{\Lambda_{\lim}}},
\end{equation}
where $\Lambda$ denotes the stretch ratio of a single chain, $\Lambda_{\lim}$ represents the entropic stretch limit, $f$ represents the stretching force along the chain. The Langevin function $L$ is given by
\begin{equation}
    L(x) = \coth(x) - \frac{1}{x}.
\end{equation}
However, the inverse function $L^{-1}(x)$ lacks a closed-form expression, which introduces challenges in numerical simulations. To address this issue, we adopt the Cohen approximation \cite{cohen1991pade}, which provides an accurate closed-form approximation with a maximum relative error less than 5\%:
\begin{equation}
    L^{-1}(x) \approx C(x) = x\frac{3-x^2}{1-x^2}.
\end{equation}
See the comparison of the inverse Langevin function $L^{-1}(x)$ with the Cohen approximation $C(x)$ in Figure \ref{fig-supp:langevin_cohen}. This gives us the approximated constitutive law:
\begin{equation}
    f = C\brac{\frac{\Lambda}{\Lambda_{\lim}}}.
\end{equation}
\begin{figure}
  \centering
  \includegraphics[width=0.9\textwidth]{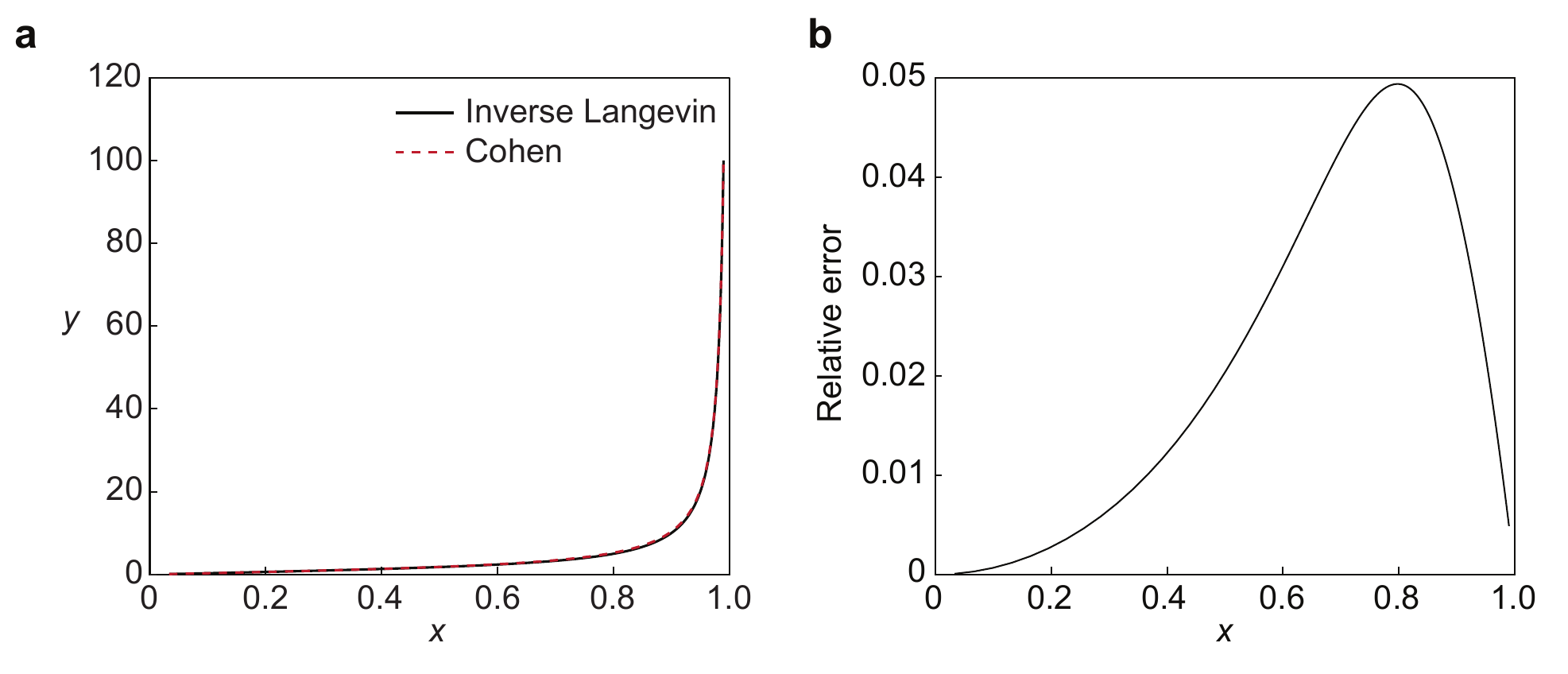}
  \caption{\textbf{The inverse Langevin function and the Cohen approximation.} 
  \textbf{(a)} Comparison of the inverse Langevin function $L^{-1}(x)$ with the Cohen approximation $C(x)$.
  \textbf{(b)} The relative error of the Cohen approximation $C(x)$ in approximating the inverse Langevin function $L^{-1}(x)$.}
  \label{fig-supp:langevin_cohen}
\end{figure}

The second issue is that both the the inverse Langevin function $L^{-1}(x)$ and the Cohen approximation $C(x)$ exhibits a singularity at $x=1$ and thus result in singularity for the constitutive law at $\Lambda = \Lambda_{\lim}$ and is not defined for $\Lambda > \Lambda_{\lim}$. In our simulations, due to the imposed boundary conditions, the chain stretch near the boundary may exceed this critical value. To fix this issue, we take the Cohen approximation for median stretch ratio and construct a smooth extension for large stretch ratio.
To be more specific, we first introduce an extension ratio $0 < \beta < 1$ and define the nonlinear constitutive law for $1\le \Lambda\le \beta\Lambda_{f}$ using the Cohen approximation:
\begin{equation}
  f_{\textrm{C}}(\Lambda) = C\brac{\frac{\Lambda}{\Lambda_{f}}} = \frac{\Lambda}{\Lambda_{f}} \frac{3 - (\frac{\Lambda}{\Lambda_{f}})^2}{1 - (\frac{\Lambda}{\Lambda_{f}})^2} = \frac{\Lambda(3\Lambda_{f}^2 - \Lambda^2)}{\Lambda_{f}(\Lambda_{f}^2 - \Lambda^2)}.
\end{equation}
Since this approximation does not satisfy the condition $f(1) = 0$ (i.e., zero tension in the undeformed state), we shift it to ensure zero force at the undeformed state:
\begin{equation}
  f_{\textrm{median}}(\Lambda) = f_{\textrm{C}}(\Lambda) - f_{\textrm{C}}(1), \quad 1 \le \Lambda \le \beta\Lambda_{f}.
\end{equation}

To smoothly extend the function $f_{\textrm{median}}$ beyond $\Lambda = \beta \Lambda_{f}$, we evaluate the function and its derivatives up to third order at that point:
\begin{equation}\label{eq:const-cohen-approximation-derivative}
  a_0 := f_{\textrm{median}}(\beta\Lambda_{f}), \quad a_1 := f_{\textrm{median}}'(\beta\Lambda_{f}), \quad a_2 := f_{\textrm{median}}''(\beta\Lambda_{f}), \quad a_3 := f_{\textrm{median}}^{(3)}(\beta\Lambda_{f}),
\end{equation}
and then construct a cubic extrapolation using a Taylor expansion for $\Lambda > \beta \Lambda_{f}$:
\begin{equation}
  f(\Lambda) = a_0 + a_1(\Lambda - \beta\Lambda_{f}) + \frac{a_2}{2}(\Lambda - \beta\Lambda_{f})^2 + \frac{a_3}{6}(\Lambda - \beta\Lambda_{f})^3, \quad \Lambda > \beta\Lambda_{f}.
\end{equation} 
In summary, the constitutive law used in our simulation is given by
\begin{equation}\label{eq-supp:nonlinear-constitutive-law}
  f(\Lambda) = 
  \begin{cases}    
    f_{\textrm{C}}(\Lambda) - f_{\textrm{C}}(1), & 1 \le \Lambda \le \beta\Lambda_{f}, \\
    a_0 + a_1(\Lambda - \beta\Lambda_{f}) + \frac{a_2}{2}(\Lambda - \beta\Lambda_{f})^2 + \frac{a_3}{6}(\Lambda - \beta\Lambda_{f})^3, & \Lambda > \beta\Lambda_{f}.
  \end{cases}
\end{equation}
where the constants $a_0, a_1, a_2, a_3$ are given in \eqref{eq:const-cohen-approximation-derivative}.

Throughout this work, we set $\Lambda_f=5$ and $\beta = 0.7$ for nonlinear constitutive law (see Figure \ref{fig-supp:extrapolation}c), unless otherwise specified. We also present the cases with $\beta = 0.5$, $0.6$, and $0.8$ in Figure \ref{fig-supp:extrapolation}. As expected, larger values of $\beta$ lead to stronger nonlinear behavior.
In addition, we sometimes compare the nonlinear law with the linear one $f = \Lambda - 1$. To ensure that both constitutive laws exhibit failure at the same stretch ratio and force, the nonlinear law is rescaled in the $y$-direction.
\begin{figure}[H]
  \centering
  \includegraphics[width=0.9\textwidth]{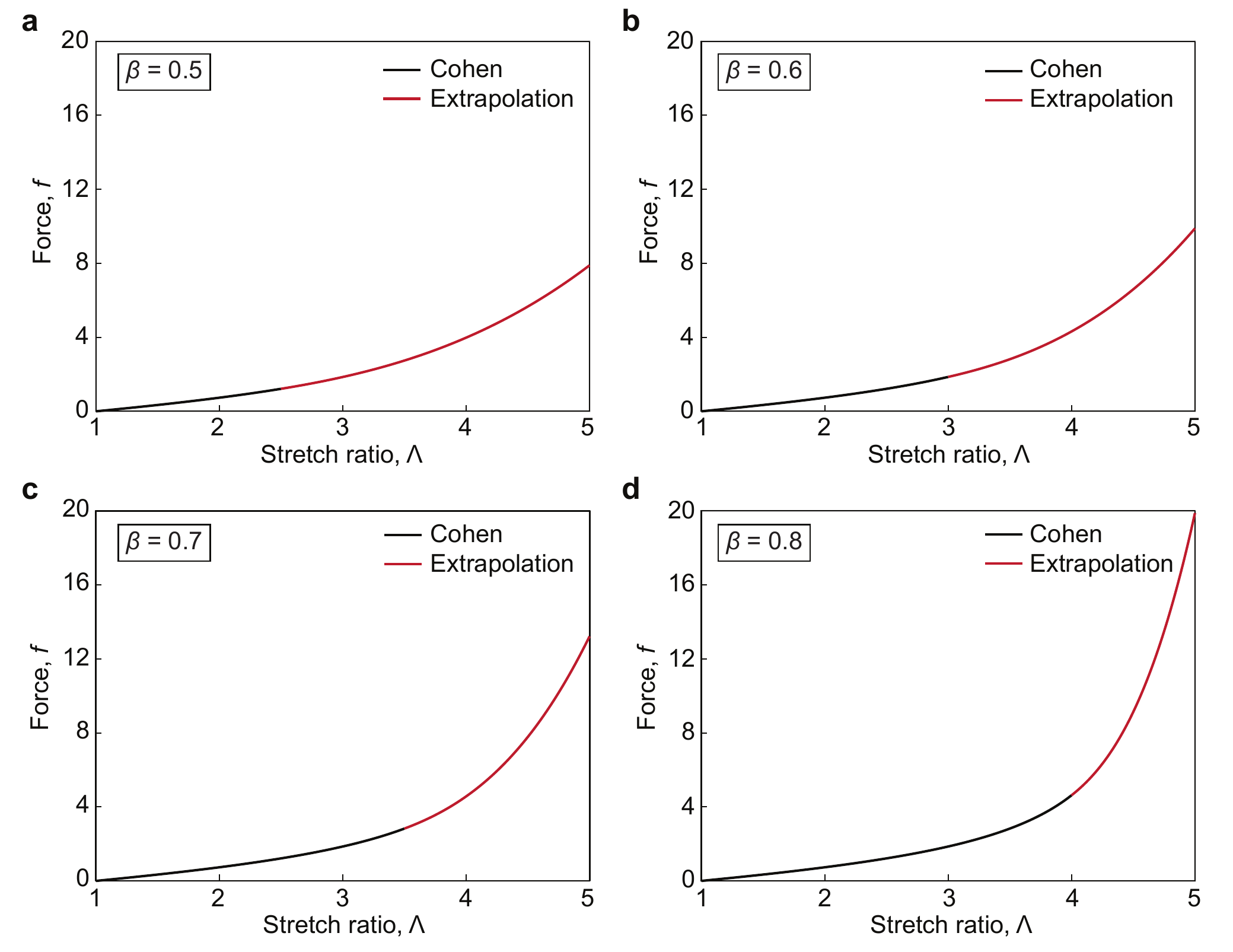}
  \caption{\textbf{Nonlinear constitutive law for a single chain.} 
  \textbf{(a-d)} The profiles of constitutive law for a single chain with $\beta = 0.5$, $0.6$, $0.7$, and $0.8$. The Cohen approximation is applied for moderate stretch ratios (in black color), while extrapolation is used for large stretch ratios (in red color).}
  \label{fig-supp:extrapolation}
\end{figure}

\section{Application to elasticity}\label{sec:application-elasticity}

In this part, we present the detailed setup of the simulation of the elasticity of the entangled networks in Section \ref{sec:setup-entangle-network-elasticity} and the analytical solution for periodic networks in Section \ref{sec:analytical-solution-elasticity}.

\subsection{Simulation setup of the entangled network}\label{sec:setup-entangle-network-elasticity}

We construct the network as a uniform square lattice of size $N_x \times N_y$, where $N_x$ and $N_y$ denote the number of nodes in the $x$-direction and $y$-direction, respectively. The total number of nodes is therefore $N = N_x N_y$, and the initial length for each segment is $l_0$. In our simulation, we take $l_0=1$ without loss of generality. The index of the nodes is denoted by $(i,j)$ for $i=1,2,\dots,N_x$ and $j=1,2,\dots,N_y$. 

The network is subjected to uniaxial tension in the $y$-direction. Boundary conditions are applied to the top and bottom boundaries to quasi-statically stretch the sample, while the displacements of nodes on the left and right boundaries are fixed in the $x$-direction. To be more specifically, we set the following boundary nodes as the non-slidable nodes: 
\begin{itemize}
    \item
    \textbf{Upper boundary:} The nodes $(i,j)$ with $i=1,2,\cdots,N_x$ and $j=N_y$ are fixed with a prescribed displacement boundary condition.
    
    \item 
    \textbf{Lower boundary:} The nodes $(i,j)$ with $i=1,2,\cdots,N_x$ and $j=1$ are fixed as in the undeformed state.

    \item
    \textbf{Left boundary:} The nodes $(i,j)$ with $i=1$ and $j=1,2,\cdots,N_y$ are fixed for the displacement in the $x$ direction.

    \item
    \textbf{Right boundary:} The nodes $(i,j)$ with $i=N_x$ and $j=1,2,\cdots,N_y$ are fixed for the displacement in the $x$ direction. 
\end{itemize}
All other nodes follow either a periodic lattice arrangement or a random distribution of slidable and non-slidable nodes with a given probability.

\subsection{Analytical solutions for periodic entangled networks}\label{sec:analytical-solution-elasticity}

In this part, we derive the analytical solution in the strength at failure, the modulus (i.e. the initial slope of stress-stretch curve), and the stretch at failure for the periodic entangled networks. We assume that the network is subjected to uniaxial tension in the $y$-direction with the stretch ratio $\lambda$. 
\begin{itemize}
    \item \textbf{Case I ($\varphi_s = 0\%$):}
    
    For the spring network with $\varphi_s = 0\%$, it is clear from Figure~\ref{fig-supp:SI_Fig8_0_Period_Network} that only the vertical chains are stretched with stretch ratio $\Lambda = \lambda$, while the horizontal chains remain unstretched ($\Lambda = 1$). The energy in a unit cell is
    \begin{equation}
        U_{\textrm{unit}} = 2 l_0 \int_{1}^{\lambda} f(\Lambda) d\Lambda.
    \end{equation}
    Therefore, the total energy of the network is
    \begin{equation}
        U_{\textrm{total}} = \frac{2 N_x N_y}{4} U_{\textrm{unit}} = N_x N_y l_0 \int_{1}^{\lambda} f(\Lambda) d\Lambda,
    \end{equation}
    where $2 N_x N_y$ is the total number of chains in the network, and each unit cell contains 4 chains.
    
    The total network force is
    \begin{equation}\label{eq-supp:nominal-stress-phis-0}
        F_{\textrm{total}} = N_x f(\lambda),
    \end{equation}
    where each vertical chain contributes a force of $f(\lambda)$ and $N_x$ denotes the number of vertical chains along the top boundary of the network. The nominal stress $S$, defined as the total force normalized by the number of vertical layers, is computed as
    \begin{equation}
        S = \frac{F_{\textrm{total}}}{N_x} = f(\lambda).
    \end{equation}
    The stretch of the network at failure coincides with the stretch of a single chain at break:
    \begin{equation}
        \lambda_f = \Lambda_f.
    \end{equation}
    Differentiating the nominal stress-stretch relation in \eqref{eq-supp:nominal-stress-phis-0} at $\lambda=1$ gives the modulus of the network:
    \begin{equation}
        E = F'(1) = f'(1),
    \end{equation}
    and the strength of the network:
    \begin{equation}
        S_{th} = F(\lambda_f) = f(\Lambda_f).
    \end{equation}    
    
    \item \textbf{Case II ($\varphi_s = 25\%$):}
    
    In a unit cell of the entangled network with $\varphi_s = 25\%$ (Figure~\ref{fig-supp:SI_Fig9_25_Period_Network}), the left vertical chain is stretched with $\Lambda = \lambda$, while the bottom horizontal chain remains unstretched. The right and top chains merge into a single chain with stretch ratio $\Lambda = (\lambda + 1)/2$. The energy of a unit cell is
    \begin{equation}
        U_{\textrm{unit}} = l_0 \int_{1}^{\lambda} f(\Lambda) d\Lambda + 2 l_0 \int_{1}^{(\lambda+1)/2} f(\Lambda) d\Lambda.
    \end{equation}
    Therefore, the total energy of the network is
    \begin{equation}
        U_{\textrm{total}} = \frac{2 N_x N_y}{4} U_{\textrm{unit}} = N_x N_y l_0 \brac{\frac{1}{2}\int_{1}^{\lambda} f(\Lambda) d\Lambda + \int_{1}^{(\lambda+1)/2} f(\Lambda) d\Lambda}.
    \end{equation}
    The total force carried by the network is
    \begin{equation}
        F_{\textrm{total}} = \frac{N_x}{2} \brac{f(\lambda) + f\brac{\frac{\lambda+1}{2}}},
    \end{equation}
    since half of the boundary chains carry force $f(\lambda)$ and the other half $f((\lambda+1)/2)$. Then the nominal stress is
    \begin{equation}
        S = \frac{F_{\textrm{total}}}{N_x} = \frac{1}{2} \brac{f(\lambda) + f\brac{\frac{\lambda+1}{2}}}.
    \end{equation}
    The stretch of the network at failure is 
    \begin{equation}
        \lambda_f = \Lambda_f.
    \end{equation}
    The modulus of the network is
    \begin{equation}
        E = F'(1) = \frac{3}{4} f'(1),
    \end{equation}
    and the strength of the network is
    \begin{equation}
        S_{th} = F(\lambda_f) = \frac{1}{2} \brac{f(\Lambda_f) + f\brac{\frac{\Lambda_f+1}{2}}}.
    \end{equation}

    \item \textbf{Case III ($\varphi_s = 50\%$) and Case IV ($\varphi_s = 75\%$):}  
    
    For the entangled network with $\varphi_s = 50\%$ and $75\%$, the corresponding unit cells are shown in Figure \ref{fig-supp:SI_Fig10_50_Period_Network} and Figure \ref{fig-supp:SI_Fig11_75_Period_Network}. Since these two cases are identical in the deformed state, we combine the discussion here. In this configuration, all chain segments have the same stretch ratio $\Lambda = (\lambda + 1)/2$. Then the energy in a unit cell is 
    \begin{equation}
        U_{\textrm{unit}} = 4 l_0 \int_{1}^{(\lambda+1)/2} f(\Lambda) d\Lambda.
    \end{equation}
    Therefore, the total energy of the network is
    \begin{equation}
        U_{\textrm{total}} = \frac{2 N_x N_y}{4} U_{\textrm{unit}} = 2 N_x N_y l_0 \int_{1}^{(\lambda+1)/2} f(\Lambda) d\Lambda.
    \end{equation}
    The total force carried by the network is
    \begin{equation}
        F_{\textrm{total}} = N_x f\brac{\frac{\lambda+1}{2}},
    \end{equation}
    Then the nominal stress is
    \begin{equation}
        S = \frac{F_{\textrm{total}}}{N_x} = f\brac{\frac{\lambda+1}{2}}.
    \end{equation}
    Moreover, the stretch for the network at failure is 
    \begin{equation}
        \lambda_f = 2\Lambda_f - 1.
    \end{equation}
    The modulus of the network is
    \begin{equation}
        E = F'(1) = \frac{1}{2} f'(1),
    \end{equation}
    and the strength of the network is
    \begin{equation}
        S_{th} = F(\lambda_f) = f(\Lambda_f).
    \end{equation}
\end{itemize}

We note that the analytical solution is derived solely from the unit-cell solution and is therefore expected to exhibit slight deviations from the numerical simulation. In fact, in the simulations, all boundaries are  fixed, which is not taken into account for the unit-cell solution. To enable a fair comparison between the analytical solution and the numerical simulation, we evaluate the total energy using only the inner portion of the computational domain. Specifically, for the whole computational domain with the node indices $(i, j)$ with $i=1,\dots,N_x$ and $j=1,\dots,N_y$, we only calculate the energy of the chains connected to the nodes in the interior domain, defined by $(i, j)$ with $i=N_x/2,\dots,3N_x/4$ and $j=N_y/2,\dots,3N_y/4$.

\section{Application to fracture}\label{sec:application-fracture}

In this part, we apply the entangled network model to study the fracture of the network. In Section \ref{sec:setup-entangle-network-fracture}, we illustrate the setup of the entangle network. In Section \ref{sec:bisection-method-critical-stretch}, we present a bisection method for efficient computation of the critical stretch of the network in fracture. In Section \ref{sec:three-regime-crack-opening}, we give the theoretical explanation for the nontrivial relationship between the crack-tip stretch and far-field stretch in three deformation regimes. Finally, we derive the probability distribution of the crack tip chain length in random entangled networks in Section \ref{sec:prob-crack-tip-length-rand-network}.

\subsection{Simulation setup of the entangle network}\label{sec:setup-entangle-network-fracture}

When investigating the fracture of the network, we set up the network as a uniform square lattice of size $N_x \times N_y$ where $N_x$ and $N_y$ are the number of nodes in the $x$ and $y$ directions. The total number of nodes is $N = N_x N_y$. The index of the nodes is denoted by $(i,j)$ for $i=1,2,\dots,N_x$ and $j=1,2,\dots,N_y$.

To set up a notch along with some boundary conditions in the sample, we set the following nodes as the non-slidable nodes: 
\begin{itemize}
    \item
    \textbf{Upper boundary:} The nodes $(i,j)$ with $i=1,2,\cdots,N_x$ and $j=N_y$ are fixed with a prescribed displacement boundary condition.
    
    \item 
    \textbf{Lower boundary:} The nodes $(i,j)$ with $i=1,2,\cdots,N_x$ and $j=1$ are fixed as in the undeformed state.

    \item
    \textbf{Left boundary:} The nodes $(i,j)$ with $i=1$ and $j=1,2,\cdots,N_y$ are fixed for the displacement in the $x$ direction.

    \item
    \textbf{Right boundary:} The nodes $(i,j)$ with $i=N_x$ and $j=1,2,\cdots,N_y$ are fixed for the displacement in the $x$ direction. 
    
    \item
    \textbf{Notch:} The crack is introduced at the left middle of the sample by setting nodes  $(i,j)$ with $i=1,2,\cdots,N_x/2$, $j=N_y/2-1$ and $j=N_y/2$ as non-slidable.
\end{itemize}
All other nodes follow either a periodic lattice arrangement or a random distribution of slidable and non-slidable nodes with a given probability, as in the elasticity problem in Section \ref{sec:application-elasticity}.

To compute the intrinsic fracture energy $\Gamma_0$ of a given network, we conduct the pure shear test consisting of two steps: (i) load an unnotched sample and record the nominal stress $S$ as a function of sample stretch $\lambda$; (ii) load a notched sample to the point where the chain around crack tip breaks and record the critical stretch of the sample $\lambda_c$, as illustrated in Figure~4E and Figure~4F. Then $\Gamma_0$ is calculated as 
\begin{equation}
    \Gamma_0 = h_0 \int_{1}^{\lambda_c} S \, d\lambda,    
\end{equation}
where $h_0$ is the initial height of the sample. The value of $\Gamma_0$ is an intrinsic property of the network and is size independent, provided that the network is sufficiently large. 
Following \cite{deng2023nonlocal}, we increase the number of vertical layers and compute $\Gamma_0$ until it converges. Figure~\ref{fig-supp:SI_Fig33_Converge} shows the convergence behavior of $\Gamma_0$ for periodic entangled networks with both linear and nonlinear constitutive relations, evaluated at different entangled node fractions $\varphi_s$ and various numbers of vertical layers $N_y$.
For networks composed of linear chains, $\Gamma_0$ converges within approximately 10 layers in the non-slidable case ($\varphi_s = 0\%$), while around 100 layers are required for convergence in slidable cases ($\varphi_s = 25\%$, $50\%$, $75\%$). For networks composed of nonlinear chains, $\Gamma_0$ converges within approximately 50 layers for spring networks and around 240 layers for slidable networks. These numerical observations are also consistent with the numerical results reported in \cite{deng2023nonlocal} for spring networks, which show that as the nonlinearity of the chain constitutive behavior increases, a larger number of vertical layers is required to achieve convergence of the intrinsic fracture energy.

We further perform numerical simulations for networks with different aspect ratios. As shown in Figure~\ref{fig-supp:SI_Fig65_AspectRatio}, for spring networks ($\varphi_s = 0\%$) with either linear or nonlinear chains, the critical stretch $\lambda_c$ converges at an aspect ratio of 1:1. In contrast, for entangled network ($\varphi_s = 50\%$), $\lambda_c$ converges with aspect ratio of 4:1. Based on this observation, we fix the network aspect ratio at 4:1 in the numerical simulations of fracture.

\subsection{Bisection method for efficient computation of critical stretch}\label{sec:bisection-method-critical-stretch}

Fracture simulations exhibit significantly stronger stress localization than the elasticity simulations, leading to much slower convergence of the gradient descent method when solving the optimization problem. Moreover, our primary interest lies in identifying the onset of fracture, which occurs when a chain around the crack tip reaches the failure stretch $\Lambda_{f}$. To reduce the computational cost, instead of simulating the full loading process that incrementally stretches the network until the fracture occurs, we adopt a more efficient bisection method (see e.g., \cite{boyd2004convex}) to directly compute the critical stretch $\lambda_c$ of the network.

Let $\lambda$ denote the overall applied stretch of the network, and let $\Lambda_{\max}$ be the maximum stretch around the crack tip in the deformed configuration. Then $\Lambda_{\max}$ can be taken as an increasing function of $\lambda$, denoted by $g$:
\begin{equation}
  \Lambda_{\max} = g(\lambda).
\end{equation}
Note that the function $g$ does not have an explicit expression but can be evaluated by solving the optimization problem: for a given applied network stretch $\lambda$, one solves the energy optimization problem to obtain the equilibrium state of the network and then compute $\Lambda_{\max}$ from this state. This evaluation is computationally expensive, since the optimization problem is challenging to solve, as demonstrated in Section~\ref{sec:numerical-method}.

Our objective is to find the critical stretch $\lambda_c$, such that 
\begin{equation}\label{eq-supp:root-bisection}
  g(\lambda_c) = \Lambda_{f}. 
\end{equation}
In the undeformed state, $\lambda = 1$ and $\Lambda_{\max} = 1$ and therefore $g(1) = 1$. For sufficiently large applied stretch (e.g., $\lambda = 10$), the maximum chain stretch $\Lambda_{\max}$ exceeds $\Lambda_{f}$, i.e. $g(10) > \Lambda_{f}$. Since $g(\lambda)$ is a continuous and increasing function, the critical stretch $\lambda_c$ can be found as the root of \eqref{eq-supp:root-bisection} in the interval $[1, 10]$ by the intermediate value theorem.

We apply the bisection method to solve \eqref{eq-supp:root-bisection} as follows. Starting with the interval $[a, b] = [1, 10]$, we evaluate the midpoint $c = (a + b)/2$ and compute $\Lambda_{\max}$ by solving the optimization problem for equilibrium state under the overall applied stretch $\lambda = c$. If $g(c) > \Lambda_{f}$, then the critical stretch lies in the left half of the interval, i.e., $[a, c]$; otherwise, it lies in the right half $[c, b]$. This process is repeated until the interval length falls below a prescribed tolerance $\varepsilon$ or $\Lambda_{\max}$ is close enough to $\Lambda_{f}$. The interval is halved in each iteration, so the number of iterations required is at most $\log_2((b-a)/\varepsilon)$, where $\varepsilon$ is the desired error tolerance for the critical stretch. For example, with $\varepsilon = 10^{-2}$ and $b-a = 9$, the maximum number of iterations is approximately 10. This means that the optimization problem needs to be solved at most 10 times to compute $\lambda_c$ with an accuracy of $10^{-2}$. This makes our approach significantly more efficient than simulating the full incremental loading process in \cite{deng2023nonlocal}.

To further accelerate the computation, we numerically observe that, during the gradient descent iterations with a given overall stretch ratio $\lambda$, the maximum chain stretch $\Lambda_{\max}$ typically increases monotonically as the system relaxes toward equilibrium. This behavior is consistent with the interpretation of gradient descent as a dynamic relaxation procedure. Consequently, if $\Lambda_{\max}$ exceeds the critical value $\Lambda_f$ at any point during the gradient descent iteration, it is guaranteed that the final equilibrium value will also exceed $\Lambda_f$. Based on this observation, we introduce an early stopping criterion: once $\Lambda_{\max} > \Lambda_f$ during the gradient descent iteration, we terminate the optimization and jump to the bisection interval accordingly. This strategy avoids unnecessary computation and significantly reduces the overall cost.
We summarize the above bisection method in Algorithm~\ref{alg:bisection-method}.
\begin{algorithm}
\caption{Bisection method for computing the critical stretch $\lambda_c$}\label{alg:bisection-method}
\begin{algorithmic}
\State \textbf{Input:} \texttt{Initial interval $[a, b] = [1, 10]$, error tolerance $\varepsilon$}
\State \textbf{Output:} \texttt{Estimated critical stretch $\lambda_c$}
\While{\texttt{True}}
\State \texttt{$c \leftarrow (a+b)/2$}
\State \texttt{Compute $\Lambda_{\max}(c)$ by solving the optimization problem:}
\For{$i$ = 0 to $N_t$}
\State \texttt{Compute the energy gradient}
\State \texttt{Update the node positions}
\State \texttt{Compute the maximum chain stretch ratio $\Lambda_{\max}$}
\If{\texttt{$\Lambda_{\max} > \Lambda_{f}$}}
\State \texttt{$b \leftarrow c$}
\State \texttt{Break}
\EndIf
\EndFor
\If{$\abs{\Lambda_{\max} - \Lambda_{f}}\le \varepsilon$ \texttt{or} $\abs{b-a} \le \varepsilon$}
\State \texttt{$\lambda_c \leftarrow c$}
\State \texttt{Break}
\EndIf
\If{$\Lambda_{\max} > \Lambda_{f}$}
\State \texttt{$b \leftarrow c$}
\Else
\State \texttt{$a \leftarrow c$}
\EndIf
\EndWhile
\end{algorithmic}
\end{algorithm}

\subsection{Three regimes of crack opening process}\label{sec:three-regime-crack-opening}

In this part, we present the theoretical explanations for the nontrivial relationship between the crack-tip stretch and the far-field stretch. The numerical results are presented for small deformation in Figure~4G, intermediate deformation in Figure~4H, and large deformation in Figure~4I.

\vspace{1em}
\noindent\textbf{Small deformation regime}

\noindent Here we give the explanation of the small deformation regime for the spring network. We first consider the unnotched sample. Under uniaxial tension with small stretch ratio $\lambda = 1 + \varepsilon$ with $0 < \varepsilon \ll 1$. In this case, each vertical chain is stretched by a factor of $1+\varepsilon$, while each horizontal chain remains unstretched.
Next, we introduce a crack by artificially removing the chains in the left-middle portion of the sample. At the instant of the removal, the nodes at the crack tip remain in their original equilibrium state. However, the nodes in front of the crack tip are in the non-equilibrium state due to the missing chains along the crack. Assume a linear constitutive law $f = k(\Lambda-1)$, and without loss of generality, take $k=1$. In this case, the non-equilibrium force on these nodes is $f=\varepsilon$. If we now artificially add a vertical displacement on these nodes by $\varepsilon$, they are almost in equilibrium, because the chain adjacent to the crack tip will then have length $(1+\varepsilon^2)^{1/2}$. Using the Taylor expansion, $(1+\varepsilon^2)^{1/2} \approx 1 + \frac{1}{2}\varepsilon^2$. This higher order term $\frac{1}{2}\varepsilon^2$ has no leading-order effect on the crack-tip nodes. Therefore, there is no stress concentration in the small-deformation regime.

\vspace{1em}
\noindent\textbf{Intermediate deformation regime}

\noindent In the intermediate deformation regime, we observe from Figure~4H that the values of $\Lambda_{\textrm{tip}}$ for both linear and nonlinear constitutive laws coincide. This indicates that $\Lambda_{\textrm{tip}}$ is independent of the specific form of the force-stretch relationship of chains and is instead determined solely by the lattice topology of the network. This behavior can be explained by the fact that any smooth nonlinear constitutive law can be locally approximated by a linear function via Taylor expansion around the unstretched state:
\begin{equation}
\begin{aligned}
    f(\Lambda) ={}& f(1) + f'(1)(\Lambda - 1) + \mathcal{O}((\Lambda - 1)^2) \\
    \approx{}& f(1) + f'(1)(\Lambda - 1) \\
    ={}& f'(1)(\Lambda - 1).
\end{aligned}  
\end{equation}
where we have used the assumption that the chain force vanishes at $\Lambda = 1$, i.e., $f(1) = 0$. Furthermore, for any linear constitutive law, the crack-tip stretch $\Lambda_{\textrm{tip}}$ depends solely on the overall stretch $\lambda$, and is independent of the slope $f'(1)$. While the linear slope $f'(1)$ determines the magnitude of force in the network, it does not affect the distribution of stretch, which is governed by the network's geometry and connectivity.

\vspace{1em}
\noindent\textbf{Large deformation regime}

\noindent We first consider the spring network. The intrinsic fracture energy is given by
\begin{equation}\label{eq-supp:fracture-energy-def}
  \Gamma_0 = h_0 \int_{1}^{\lambda_c} f(\lambda) d\lambda,
\end{equation}
where $h_0$ is the initial height of the network and $\lambda_c$ is the critical stretch. According to the scaling law proposed in \cite{hartquist2025scaling}, the intrinsic fracture energy satisfies
\begin{equation}\label{eq-supp:scaling-law-cross-link}
  \Gamma_0 = \alpha l_0 f_f (\Lambda_f - 1),
\end{equation}
for any constitutive law $f$, and any $f_f$ and $\Lambda_f$. Here $\alpha$ is a constant depending on the network topology, $l_0$ is the unit length in the network, and $f_f$ and $\Lambda_f$ are the breaking force and breaking stretch ratio of a single chain, respectively. We note that the original work \cite{hartquist2025scaling} uses the notation of chain length at failure, $L_f$, in the scaling law. In our formulation, we interpret this more precisely as the stretch ratio at failure, $\Lambda_f$, and correspondingly replace $\Lambda_f$ by $(\Lambda_f - 1)$, to reflect the stretch measured from the rest length.

Let $\Lambda_{\textrm{tip}}$ denote the crack tip stretch ratio. Since $\Lambda_{\textrm{tip}}$ increases monotonically with the overall network stretch ratio $\lambda$, we write
\begin{equation}
  \Lambda_{\textrm{tip}} = \phi(\lambda),
\end{equation}
where $\phi$ is an increasing function.
At the critical stretch $\lambda_c$, the stretch at crack tip reaches the failure threshold $\Lambda_f$, i.e.,
\begin{equation}
  \Lambda_f = \phi(\lambda_c).
\end{equation}
Substituting this into the fracture energy definition \eqref{eq-supp:fracture-energy-def} and the scaling law \eqref{eq-supp:scaling-law-cross-link}, we obtain
\begin{equation}\label{eq-supp:scaling-law-derivation-1}
  h_0 \int_{1}^{\phi^{-1}(\Lambda_f)} f(\lambda) d\lambda = \alpha l_0 f(\Lambda_f) (\Lambda_f - 1).
\end{equation}
The above relation holds for any choice of constitutive law $f$. More importantly, it builds the relationship between the functions $\phi$ and $f$.

While it is difficult to derive an explicit expression for $\phi$ from \eqref{eq-supp:scaling-law-derivation-1} in general, we can solve it for some specific forms of $f$. Consider the monomial constitutive law:
\begin{equation}
  f(\Lambda) = (\Lambda - 1)^p, \quad p>0.
\end{equation}
Then the relation \eqref{eq-supp:scaling-law-derivation-1} becomes
\begin{equation}
  h_0 \int_{1}^{\phi^{-1}(\Lambda_f)} (\lambda - 1)^p d\lambda = \alpha l_0 (\Lambda_f - 1)^p (\Lambda_f - 1),
\end{equation}
which simplifies to
\begin{equation}
  \frac{h_0}{l_0 (p+1)} \brac{\phi^{-1}(\Lambda_f) - 1}^{p+1} = \alpha (\Lambda_f - 1)^{p+1}
\end{equation}
Solving for $\phi^{-1}(\Lambda_f)$, we obtain
\begin{equation}
  \phi^{-1}(\Lambda_f) = \brac{\frac{\alpha l_0(p+1)}{h_0}}^{\frac{1}{p+1}} (\Lambda_f - 1) + 1,
\end{equation}
Taking the inverse yields
\begin{equation}
  \phi(\lambda) = \brac{\frac{h_0}{\alpha l_0(p+1)}}^{\frac{1}{p+1}} (\lambda - 1) + 1.
\end{equation}
This analysis leads to the conclusion that the crack tip stretch $\Lambda_{\textrm{tip}}$ linearly increase with the overall stretch $\lambda$ with the increase rate
\begin{equation}\label{eq-supp:large-deform-slope-cross-link}
    k_s = \brac{\frac{h_0}{\alpha l_0(p+1)}}^{\frac{1}{p+1}}.
\end{equation}

The slope $k_s$ in \eqref{eq-supp:large-deform-slope-cross-link} determines when the fracture initiates. Its dependence on $h_0/l_0$, $\alpha$, and $p$ is analyzed as follows:
\begin{itemize}
    \item \textbf{Dependence on $h_0/l_0$:} \\ $k_s$ increases as $h_0/l_0$ (i.e., number of layer in vertical direction) increases. This implies that  networks with larger size are easier for fracture initiation (see the numerical validations in Figures \ref{fig-supp:SI_Fig35_Nonslidable_linearlaw}b and \ref{fig-supp:SI_Fig35_Nonslidable_linearlaw}d).

    \item \textbf{Dependence on $\alpha$:} \\ $k_s$ decreases as $\alpha$ increases. Since $\alpha$ represents the  factor in the scaling law \eqref{eq-supp:scaling-law-cross-link}, a larger $\alpha$ corresponds to a higher intrinsic fracture energy, delaying the fracture initiation.

    \item \textbf{Dependence on $p$:} \\ We introduce the function $k_s = k_s(p)$ and study its dependence on $p$,
\begin{equation}
  k_s(p) = \brac{\frac{h_0}{\alpha l_0(p+1)}}^{\frac{1}{p+1}}.
\end{equation}
Taking logarithms gives
\begin{equation}
  \log k_s = \frac{1}{p+1} \brac{\log\brac{\frac{h_0}{\alpha l_0}} - \log(p+1)}.
\end{equation}
Taking derivative with respect to $p$ gives
\begin{equation}
\begin{aligned}  
  \frac{1}{k_s} k_s'(p) ={}& -\frac{1}{(p+1)^2} \brac{\log \brac{\frac{h_0}{\alpha l_0}} - \log(p+1)} - \frac{1}{(p+1)^2} \\
  ={}& -\frac{1}{(p+1)^2} \brac{\log \brac{\frac{h_0}{\alpha l_0}} - \log(p+1) + 1} \\
  ={}& \frac{1}{(p+1)^2} \brac{\log(p+1) - \log \brac{\frac{e h_0}{\alpha l_0}}}.
\end{aligned}
\end{equation}
Setting $k_s'(p) = 0$ yields
\begin{equation}
  p^* = \frac{e h_0}{\alpha l_0} - 1.
\end{equation}
Therefore, $k_s$ decreases with $p$ if $p < p^*$ and increases with $p$ if $p > p^*$.
In typical networks, $p^* \gg 1$ (e.g., layer numbers $\sim$10 for linear constitutive laws, and hundreds or thousands for highly nonlinear laws \cite{deng2023nonlocal}). Thus, in the common nonlinear regime, $k_s$ decreases with increasing $p$, implying that nonlinearity tends to increase the intrinsic fracture energy.

\end{itemize}

For the periodic entangled network with $\varphi_s = 50\%$, the derivation proceeds in the same way as for the spring network. In this case, the intrinsic fracture energy is
\begin{equation}
  \Gamma_0 = h_0 \int_{1}^{\lambda_c} f\brac{\frac{1}{2}(\lambda+1)} d\lambda,
\end{equation}
where the integrated function $f(\frac{1}{2}(\lambda + 1))$ arises from the analytical form of the stress–stretch relation in periodic entangled networks, see Section \ref{sec:analytical-solution-elasticity}.
Following the similar derivation, we obtain
\begin{equation}
  h_0 \int_{1}^{\phi^{-1}(\Lambda_f)} f\brac{\frac{1}{2}(\lambda+1)} d\lambda = \alpha l_0 f(\Lambda_f) (\Lambda_f - 1).
\end{equation}
Assume the same monomial constitutive law for a single chain
\begin{equation}
    f(\Lambda) = (\Lambda - 1)^p, \quad p > 0, 
\end{equation}
we have
\begin{equation}
  h_0 \int_{1}^{\phi^{-1}(\lambda_f)} \brac{\frac{1}{2}(\lambda+1) - 1}^p d\lambda = \alpha l_0 (\lambda_f - 1)^{p+1}.
\end{equation}
Evaluating the integral yields
\begin{equation}
  h_0 \frac{1}{2^p(p+1)} \brac{\phi^{-1}(\lambda_f) - 1}^{p+1} = \alpha l_0 (\lambda_f - 1)^{p+1}.
\end{equation}
Solving for $\phi(\lambda)$, we find
\begin{equation}
  \phi(\lambda) = \brac{\frac{h_0}{\alpha l_0 2^p(p+1)}}^{\frac{1}{p+1}} (\lambda - 1) + 1.
\end{equation}
Therefore, the crack-tip stretch increases linearly with overall stretch at the rate
\begin{equation}\label{eq-supp:large-deformation-linear-slope-entangle}
    k_e = \brac{\frac{h_0}{\alpha l_0 2^p(p+1)}}^{\frac{1}{p+1}}.
\end{equation}
The dependence of $k_e$ on the parameters $h_0$, $\alpha$, and $p$ is similar to that for the spring network and is therefore omitted here. Compared with \eqref{eq-supp:large-deform-slope-cross-link}, the additional factor of $2^p$ in \eqref{eq-supp:large-deformation-linear-slope-entangle} further reduces the slope, demonstrating that entangled networks are even more effective in mitigating stress concentration.

\subsection{Probability distribution of crack-tip chain length for random networks}\label{sec:prob-crack-tip-length-rand-network}

In this section, we derive the probability distribution of the crack-tip chain length in random entangled networks. We assume that each node is independently designated as a slidable node with probability $\varphi_s$, or a non-slidable node with probability $(1 - \varphi_s)$. We also assume that the entanglement orientation is fixed in the network: the left and lower neighboring nodes belong to one chain, while the right and upper neighboring nodes belong to another chain (Figure \ref{fig-supp:orientation}a). 

\begin{figure}[H]
  \centering
  \includegraphics[width=0.9\textwidth]{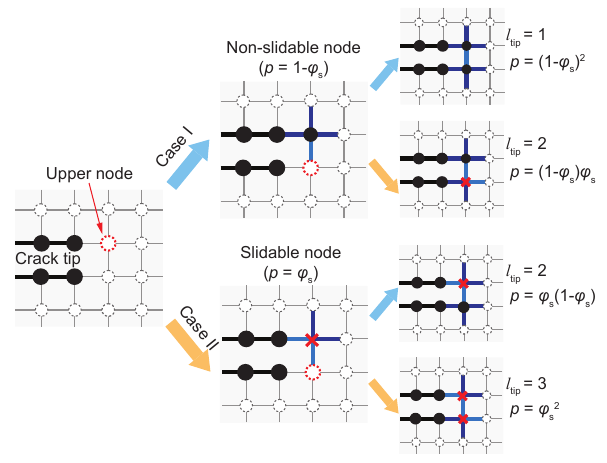}
  \caption{\textbf{Schematic illustration of the derivation of probability distribution of crack-tip chain length.} To compute the probability distribution, first consider two cases: (I) the upper crack-tip node is non-slidable, with probability $(1-\varphi_s)$, and (II) the upper crack-tip node is slidable, with probability $\varphi_s$. Then, for each case, the analysis proceeds by checking the next node along the chain path.}
  \label{fig-supp:randomprobability}
\end{figure}
To analyze the chain length at the crack tip, denoted by $l_{\mathrm{tip}}$, we consider two cases based on the configuration of the node at the upper crack tip (Figure~\ref{fig-supp:randomprobability}):
\begin{itemize}
  \item \textbf{Case I:} The upper crack tip node is a non-slidable node, which occurs with the probability  $(1-\varphi_s)$.
  
  \item \textbf{Case II:} The upper crack tip node is a slidable node, which occurs with the probability $\varphi_s$.
\end{itemize}
In Case I, the chain at the crack tip starts from the non-slidable node and continues through slidable nodes until it terminates at a non-slidable node. Since each node is independently slidable with probability $\varphi_s$, the probability that the chain has length $k$ follows a geometric distribution (Figure~\ref{fig-supp:randomprobability}):
\begin{equation}
  P(l_{\textrm{tip}} = k ~|~ \textrm{Case I}) = \varphi_s^{k-1} (1-\varphi_s), \quad k = 1, 2, \dots
\end{equation}
Here, we assume that the network is sufficiently large so that the crack tip chain can achieve any length $k \ge 1$ without boundary effects. We also use the assumption that the entanglement orientation is fixed, ensuring that as the chain progresses along slidable nodes, it does not intersect itself.

In Case II, the upper crack tip node is slidable, so the chain must include at least one more slidable node before reaching a non-slidable node. Therefore, the minimum chain length is $k = 2$, and the distribution becomes (Figure~\ref{fig-supp:randomprobability}):
\begin{equation}
  P(l_{\textrm{tip}} = k ~|~ \textrm{Case II}) = \varphi_s^{k-2} (1-\varphi_s), \quad k = 2, 3, \dots
\end{equation}
We now compute the probability by combining both cases, weighted by their respective probabilities:
\begin{itemize}
  \item For $k=1$:
  \begin{equation}
    P(l_{\textrm{tip}} = 1) = P(\textrm{Case I}) P(l_{\textrm{tip}} = 1 ~|~ \textrm{Case I})= (1-\varphi_s)^2.
  \end{equation}  
      
  \item For $k\ge2$:
  \begin{equation}
  \begin{aligned}
    P(l_{\textrm{tip}} = k) ={}& P(\textrm{Case I}) P(l_{\textrm{tip}} = k ~|~ \textrm{Case I}) + P(\textrm{Case II}) P(l_{\textrm{tip}} = k ~|~ \textrm{Case II}) \\
    ={}& \varphi_s^{k-1} (1-\varphi_s)(1-\varphi_s) + \varphi_s^{k-2} (1-\varphi_s)\varphi_s \\
    ={}& \varphi_s^{k-1} (1 - \varphi_s) (2 - \varphi_s).
  \end{aligned}
  \end{equation}
\end{itemize}

Therefore, The probability distribution of the crack tip chain length is given by:
\begin{equation}
  P(l_{\textrm{tip}} = k) = 
  \begin{cases}
    (1-\varphi_s)^2, & k = 1, \\
    \varphi_s^{k-1} (1 - \varphi_s) (2 - \varphi_s), & k = 2, 3, \dots
  \end{cases}
\end{equation}

\vspace*{0pt}
\begin{figure}[H]
  \centering
  \includegraphics[trim={0cm 0.5cm 0cm 1cm},clip, width=1.0\textwidth]{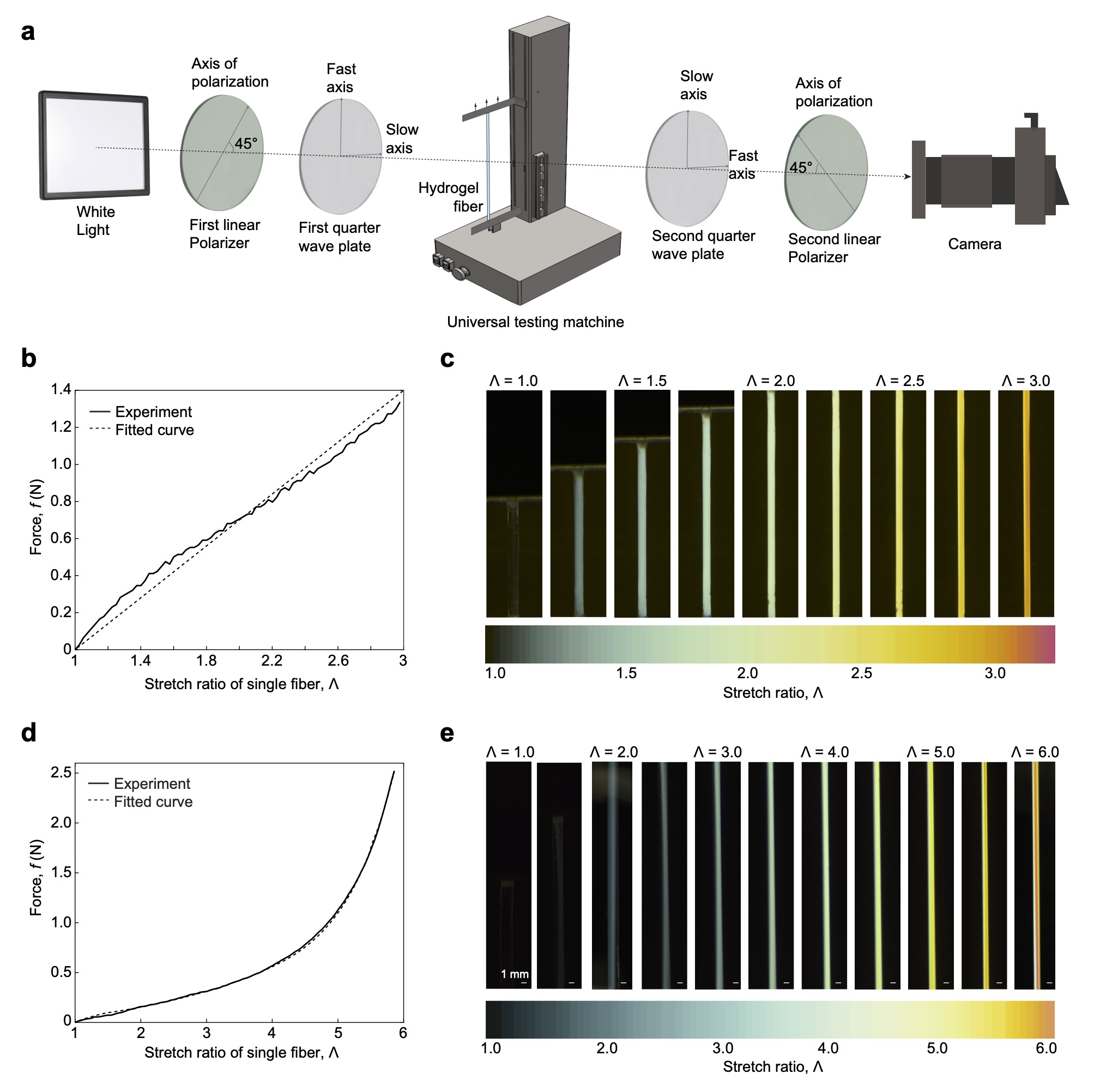}
  \caption{\textbf{Homemade photoelastic setup and mechanical and photoelastic performance of hydrogel fibers.} \textbf{(a)} Schematic of the homemade photoelastic setup, which consists of a white light, a camera, two linear polarizers, two quarter-wave plates, a universal testing machine, and a hydrogel fiber. \textbf{(b)} Force-stretch relationship of an approximately linear hydrogel fiber. \textbf{(c)} Photoelastic images of a linear fiber subjected to various stretch levels under a circular polariscope. \textbf{(d)} Force-stretch relationship of a non-linear hydrogel fiber. \textbf{(e)} Photoelastic images of a non-linear fiber subjected to various stretch levels under a circular polariscope.
  }
  \label{fig-supp:SI_Fig1_SingleFiber}
\end{figure}

\vspace*{0pt}
\begin{figure}[H]
  \centering
  \includegraphics[trim={0cm 0cm 0cm 0cm},clip, width=0.98\textwidth]{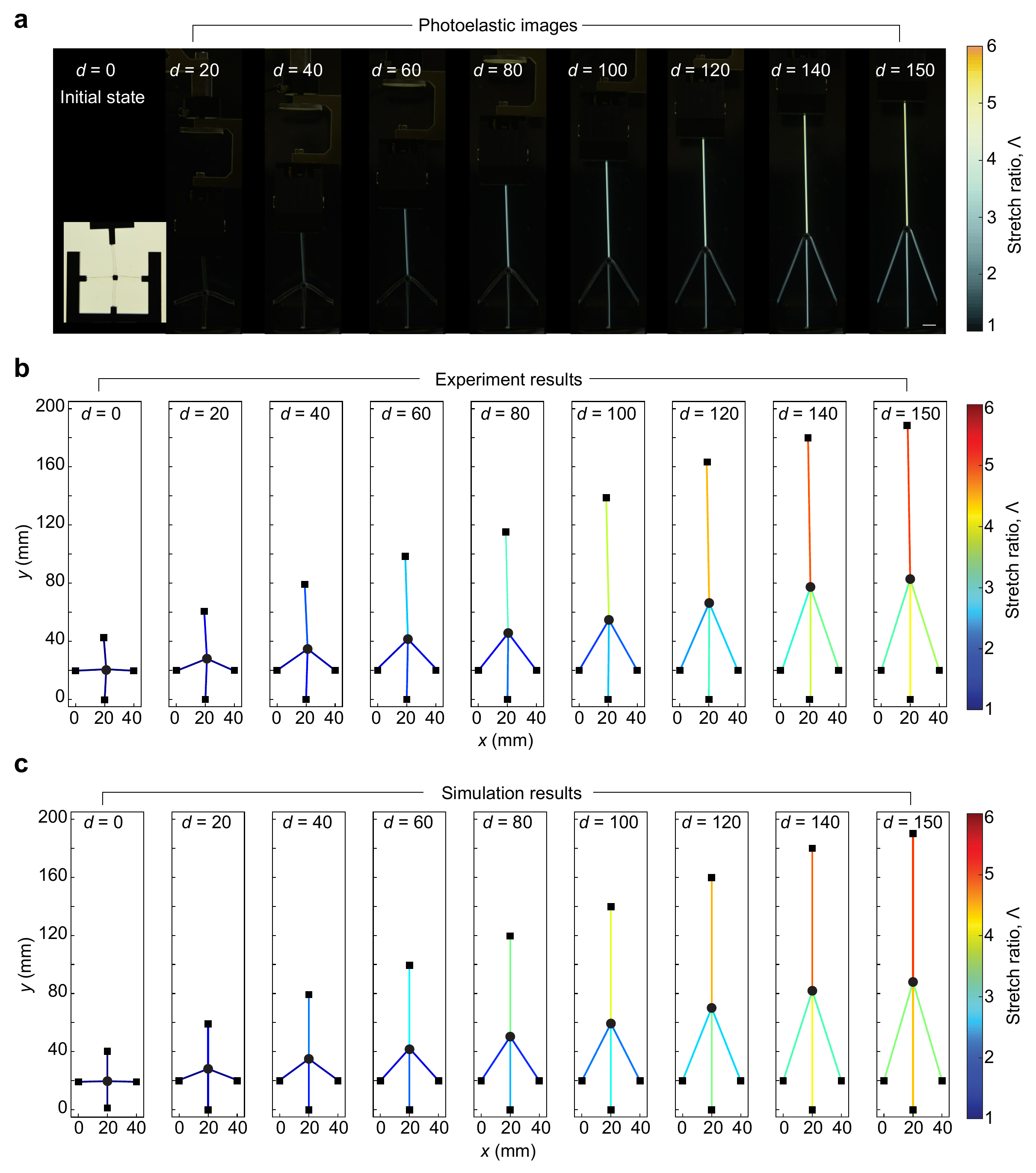}
  \caption{\textbf{Photoelastic images, experimental results, and simulation results of the two-chain model with a non-slidable node (spring network).} 
\textbf{(a)} Photoelastic images. Scale bar: 10 mm. 
\textbf{(b)} Experimental results obtained from the photoelastic images. 
\textbf{(c)} Simulation results.
 The unit of displacement $d$ is mm. Black nodes represent non-slidable nodes.
  }
  \label{fig-supp:SI_Fig2_Nonslidable_TwoFiber_Exp_Sim}
\end{figure}

\vspace*{0pt}
\begin{figure}[H]
  \centering
  \includegraphics[trim={0cm 0.5cm 0cm 0cm},clip, width=0.98\textwidth]{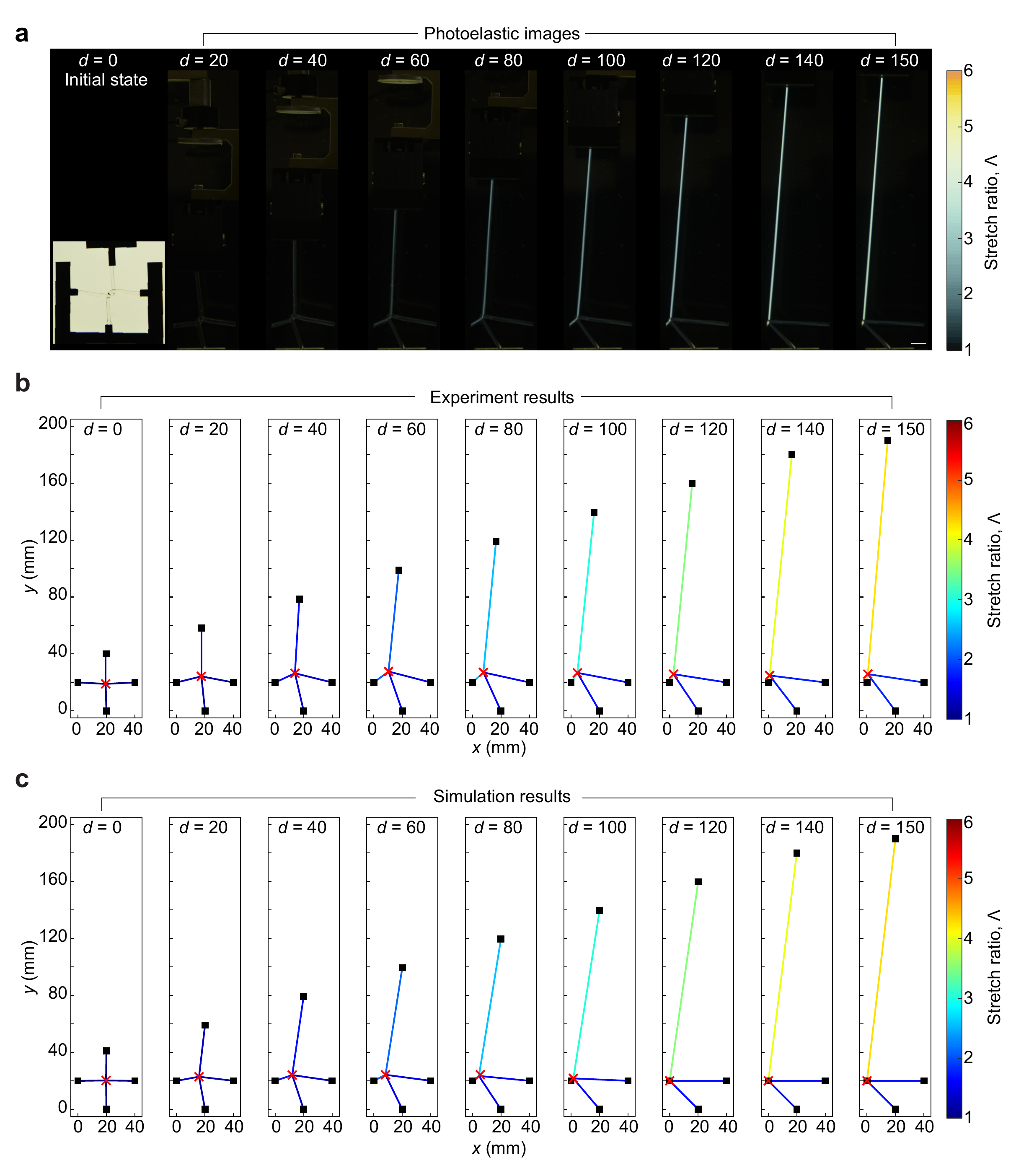}
  \caption{\textbf{Photoelastic images, experimental results, and simulation results of the two-chain model with a slidable node (entangled network).} 
\textbf{(a)} Photoelastic images. Scale bar: 10 mm. 
\textbf{(b)} Experimental results obtained from the photoelastic images. 
\textbf{(c)} Simulation results. The unit of displacement $d$ is mm. Black nodes represent non-slidable nodes and red crosses represent slidable nodes. 
  }
  \label{fig-supp:SI_Fig3_SlidableTwoFiber_Exp_Sim}
\end{figure}

\vspace*{0pt}
\begin{figure}[H]
  \centering
  \includegraphics[trim={0cm 0cm 0cm 0cm},clip, width=1.0\textwidth]{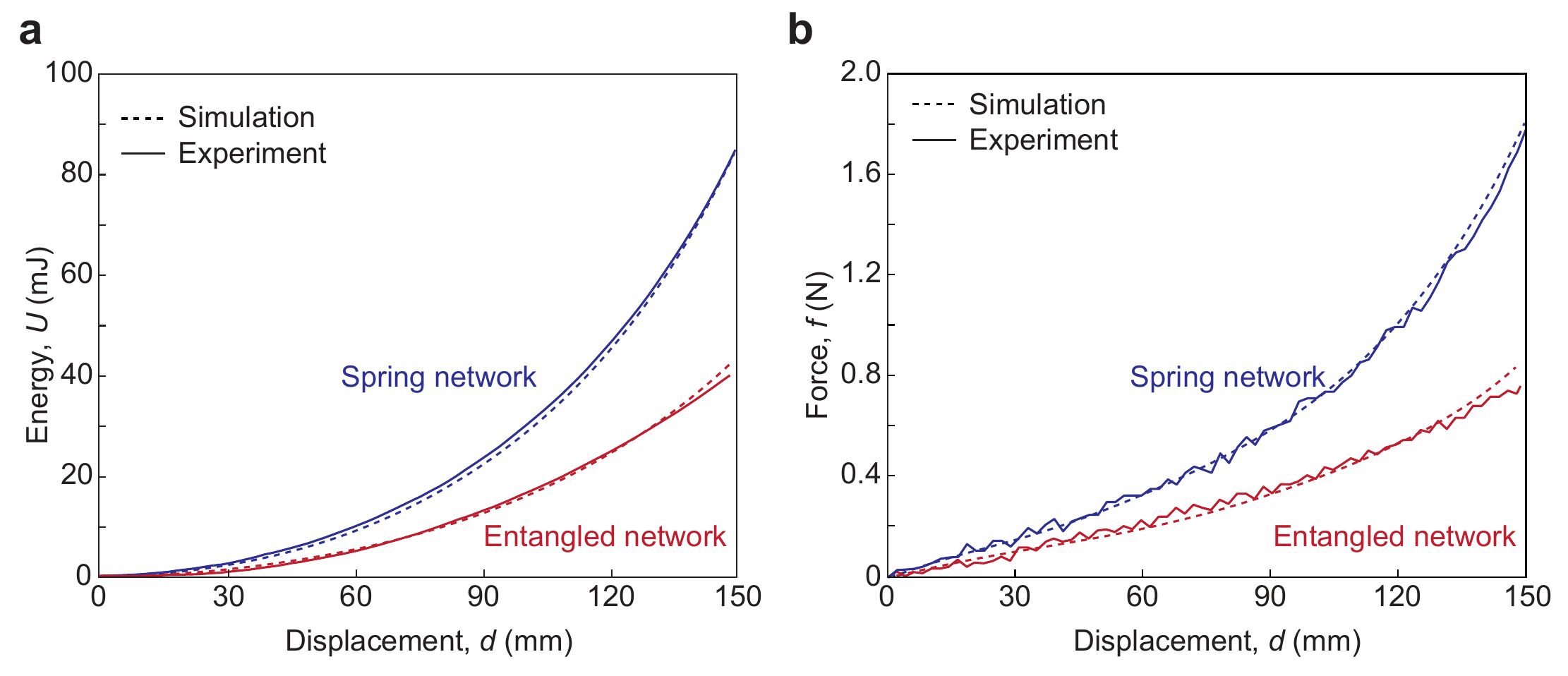}
  \caption{\textbf{Experimental and simulation results of the energy and force curves for the two-chain model with slidable nodes (entangled network) and non-slidable nodes (spring network).} 
\textbf{(a)} Experimental and simulation results of the energy-displacement curves of the two-chain models.
\textbf{(b)} Experimental and simulation results of the force-displacement curves of the two-chain models. In both curves, the experimental and simulation results show good agreement.
  }
  \label{fig-supp:SI_Fig4_SlidableTwoFiber_Exp_Sim}
\end{figure}

\vspace*{0pt}
\begin{figure}[H]
  \centering
  \includegraphics[trim={0cm 0cm 0cm 0cm},clip, width=0.9\textwidth]{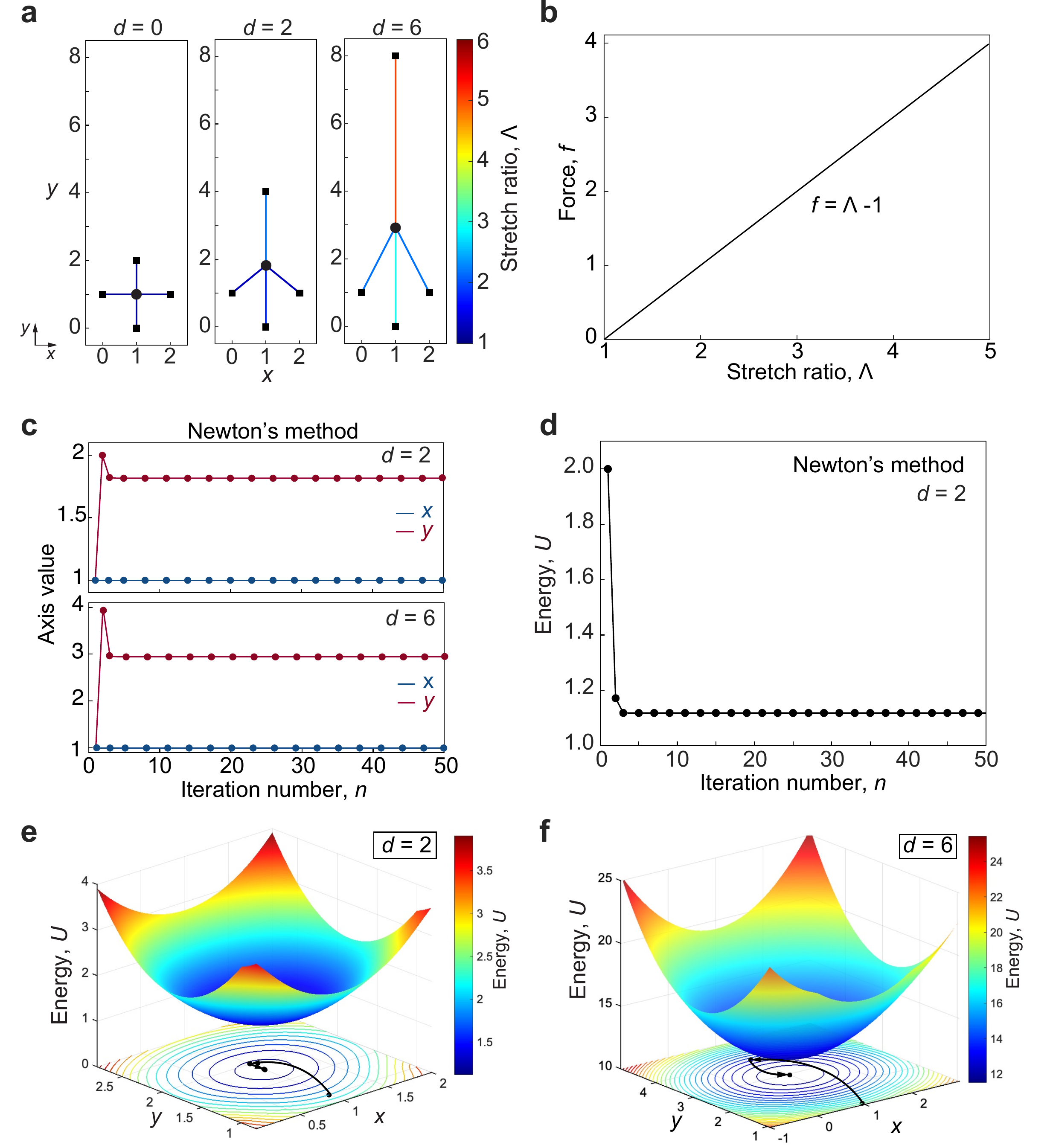}
  \caption{See caption on the next page.}
  \label{fig-supp:SI_Fig5_Nonslidable_linear_simulation}
\end{figure}

\clearpage
\noindent\textbf{Figure S17:} \textbf{Simulation of the two-chain model with a non-slidable node (spring network).}
\textbf{(a)} 
The simulation results of the two-chain model with a non-slidable node subjected to different displacements: $d=0$ (initial undeformed state), $d=2$ (small displacement), and $d=6$ (large displacement).
\textbf{(b)} 
The constitutive law for a single chain, $f = \Lambda - 1$, where $f$ is force and $\Lambda$ is stretch ratio. 
\textbf{(c)} 
The values of $x$ and $y$ (the coordinates of the non-slidable node) versus the iteration number $n$ in Newton’s method. The method converges in both cases.
\textbf{(d)} 
The energy $U$ at $d = 2$ versus the iteration number $n$ in Newton’s method.
\textbf{(e)} 
The energy landscape at $d = 2$, where $x$ and $y$ denote the coordinates of the non-slidable node. The trajectory of $(x, y)$ in Newton’s method is shown as black dots.
\textbf{(f)}
The energy landscape at $d = 6$, where $x$ and $y$ denote the coordinates of the non-slidable node. The iteration trajectory of $(x, y)$ in Newton’s method is shown as black dots.
\clearpage

\vspace*{0pt}
\begin{figure}[H]
  \centering
  \includegraphics[trim={0cm 0cm 0cm 0cm},clip, width=1.0\textwidth]{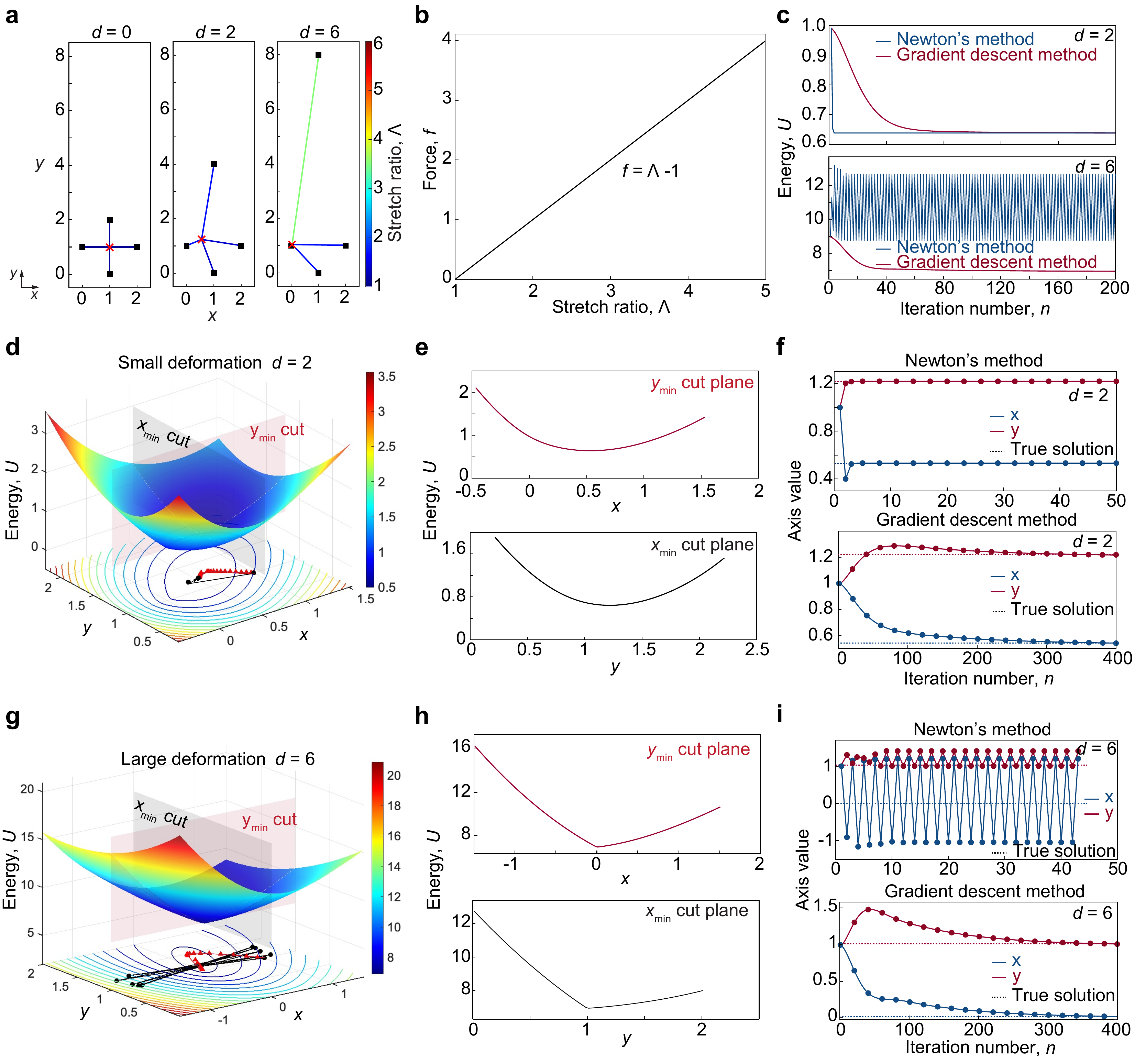}
  \caption{See caption on the next page.}
  \label{fig-supp:SI_Fig6_Slidable_linear_simulation}
\end{figure}

\clearpage
\noindent\textbf{Figure S18:} \textbf{Simulation of the two-chain model with a slidable node (entangled network).}
\textbf{(a)} 
The simulation results of the two-chain model with a slidable node subjected to different displacements: $d=0$ (initial undeformed state), $d=2$ (small displacement), and $d=6$ (large displacement).
\textbf{(b)} 
The constitutive law for a single chain, $f = \Lambda - 1$, where $f$ is force and $\Lambda$ is stretch ratio. 
\textbf{(c)} 
The energy $U$ at $d = 2$ (small) and $d = 6$ (large) versus the iteration number $n$ in Newton’s method and the gradient descent method. Newton’s method converges for $d=2$ but exhibits oscillatory behavior for $d=6$. The gradient descent method converges in both cases.
\textbf{(d)} 
The energy landscape at displacement $d = 2$, where $x$ and $y$ denote the coordinates of the slidable node. The trajectories of $(x, y)$ during the iterations of Newton’s method (black dots) and the gradient descent method (red rectangles) are shown. Both trajectories converge to the minimum of the energy function.
\textbf{(e)} 
The cross-sectional views of the energy function $U$ at displacement $d = 2$, taken at the coordinates corresponding to the energy minimum. The top panel shows a cross-section at $y = y_{\textrm{min}}$ (i.e., the $y$-coordinate of the minimum), while the bottom panel shows a cross-section at $x = x_{\textrm{min}}$ (i.e., the $x$-coordinate of the minimum).
\textbf{(f)} 
The values of $x$ and $y$ (the coordinates of the slidable node) versus the iteration number $n$ during the iterations of Newton’s method and the gradient descent method. Both methods converge to the true solution.
\textbf{(g)}
The energy landscape at displacement $d = 6$, where $x$ and $y$ denote the coordinates of the slidable node. The trajectories of $(x, y)$ during the iterations of Newton’s method (black dots) and the gradient descent method (red rectangles) are shown. The gradient descent method converges, while Newton’s method exhibits oscillatory behavior.
\textbf{(h)} 
The cross-sectional views of the energy function $U$ at displacement $d = 6$, taken at the coordinates corresponding to the energy minimum. The top panel shows a cross-section at $y = y_{\textrm{min}}$ (i.e., the $y$-coordinate of the minimum), while the bottom panel shows a cross-section at $x = x_{\textrm{min}}$ (i.e., the $x$-coordinate of the minimum).
\textbf{(i)} 
The values of $x$ and $y$ (the coordinates of the slidable node) versus the iteration number $n$ during the iterations of Newton’s method and the gradient descent method. The gradient descent method converges, while Newton’s method exhibits oscillatory behavior.
\clearpage

\vspace*{0pt}
\begin{figure}[H]
  \centering
  \includegraphics[trim={0cm 0cm 0cm 0cm},clip, width=1.0\textwidth]{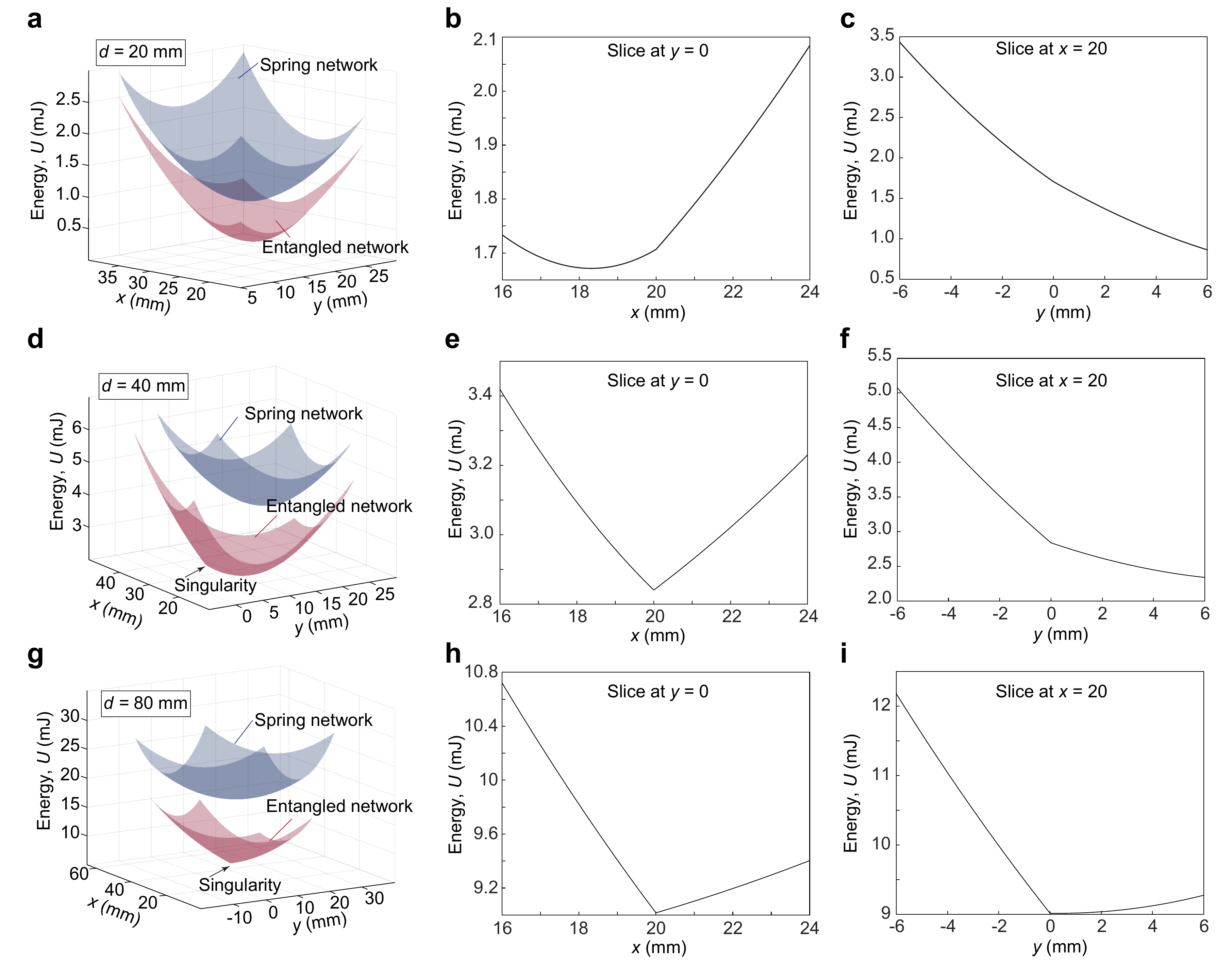}
  \caption{\textbf{Energy landscapes and their cross-sectional view for the two-chain model with slidable and non-slidable nodes.} Here we use the setup and the constitutive law from the two-hydrogel-fiber experiment. 
  \textbf{(a, d, g)} Energy landscapes at displacements $d = 20$, $40$, and $80$ mm, respectively. Here $x$ and $y$ denote the coordinates of either a slidable or non-slidable node. Systems with non-slidable nodes exhibit higher energy than those with slidable nodes. There is a singularity point in the energy landscape of the system with slidable node. The singular point shifts toward the minimum as the displacement increases.
  \textbf{(b, e, h)} Cross-sectional views of the energy landscape of the slidable-node system in (a, d, g) at $y = 0$ mm. The singular point shifts toward the minimum as the displacement increases.
  \textbf{(c, f, i)} Cross-sectional views of the energy landscape of the slidable-node system in (a, d, g) at $x = 20$ mm. The singular point shifts toward the minimum as the displacement increases.
}
  \label{fig-supp:SI_Fig7_EnergylandscapeOfExperiment}
\end{figure}

\vspace*{0pt}
\begin{figure}[H]
  \centering
  \includegraphics[trim={0cm 0cm 0cm 0cm},clip, width=1.0\textwidth]{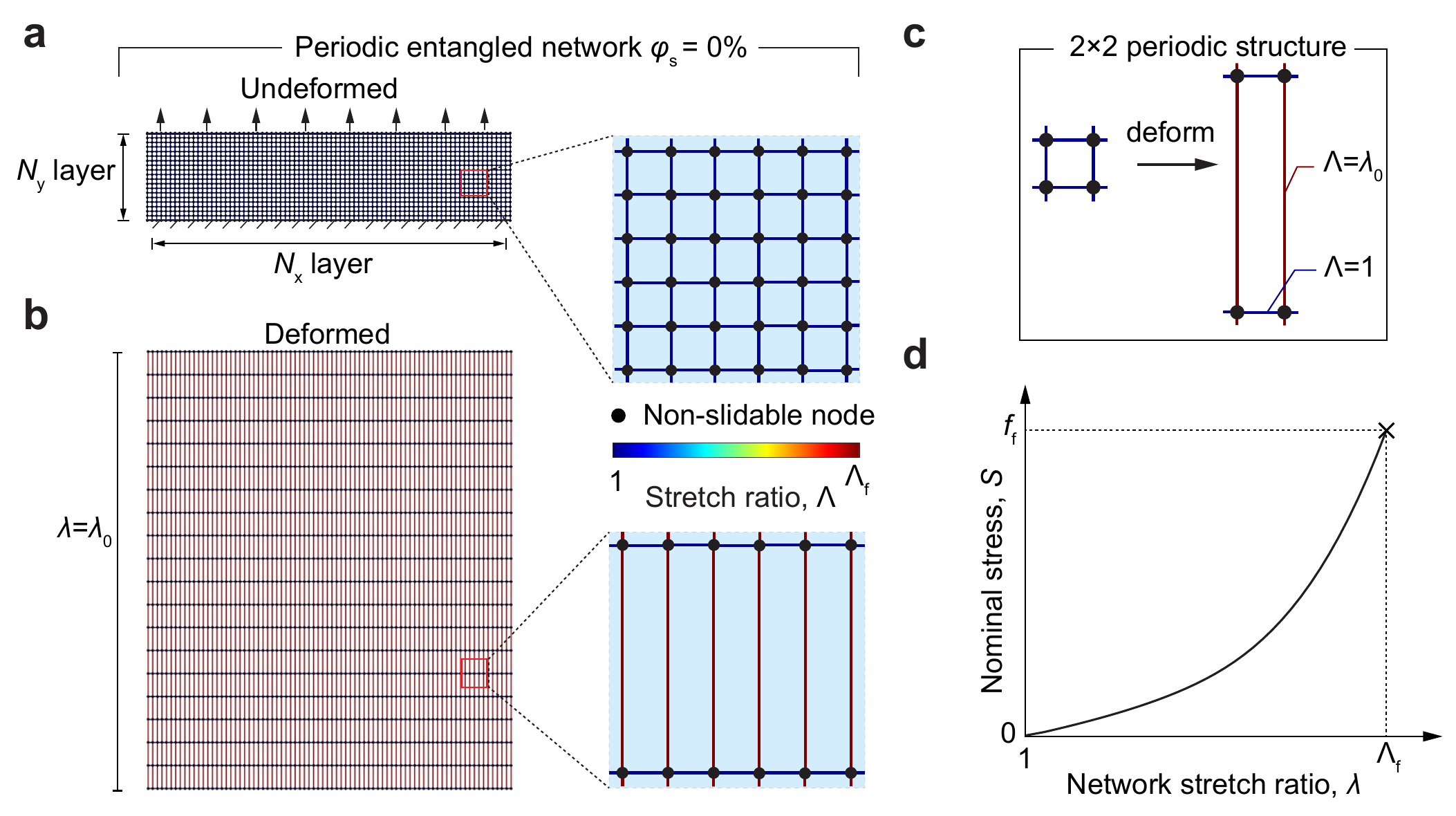}
  \caption{\textbf{Simulation of spring network with $\varphi_s = 0\%$.}
\textbf{(a)} Network configuration in the undeformed state. All nodes in the network are non-slidable nodes. 
The network has $N_x$ horizontal layers and $N_y$ vertical layers.
\textbf{(b)} Network subjected to uniaxial tension with stretch ratio $\lambda_0$. Chains aligned with the stretching direction are stretched by $\lambda_0$, while chains perpendicular to the stretching direction remain unstretched. 
\textbf{(c)} Undeformed and deformed states of the $2 \times 2$ periodic structure. 
\textbf{(d)} Nominal stress of network $S$ versus network stretch ratio $\lambda$. The nominal stress $S$ is defined as the total network force divided by the layer number $N_x$. $\Lambda_f$ is the rupture stretch ratio of an individual chain, and $f_f$ is the corresponding breakage force.
}
  \label{fig-supp:SI_Fig8_0_Period_Network}
\end{figure}

\vspace*{0pt}
\begin{figure}[H]
  \centering
  \includegraphics[trim={0cm 0cm 0cm 0cm},clip, width=1.0\textwidth]{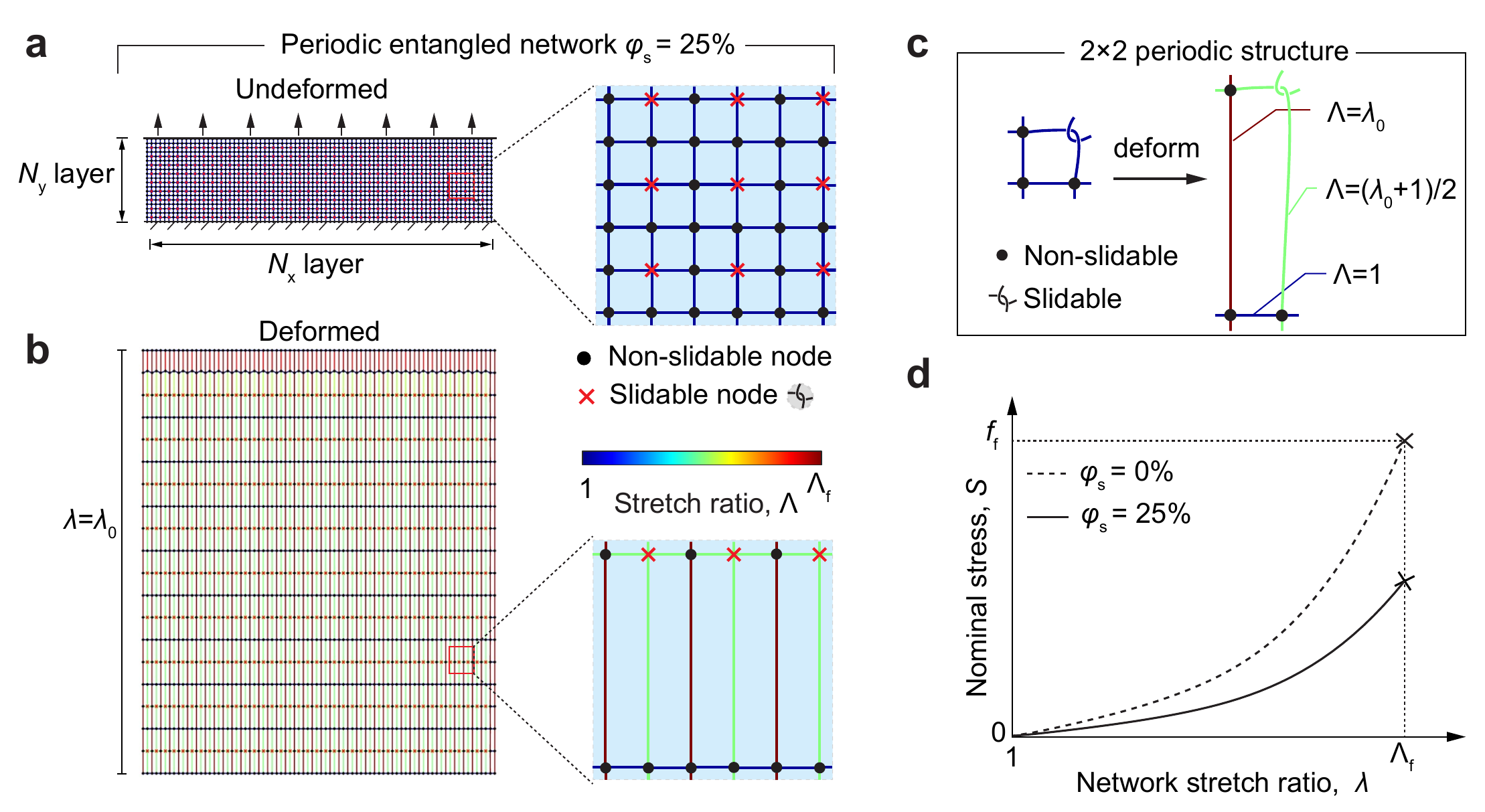}
  \caption{\textbf{Simulation of periodic entangled network with $\varphi_s = 25\%$.} 
\textbf{(a)} Network configuration in the undeformed state. Black nodes represent non-slidable nodes, and red crosses indicate slidable nodes with fixed orientation. 
\textbf{(b)} Network subjected to uniaxial tension with stretch ratio $\lambda_0$. 
\textbf{(c)} Undeformed and deformed states of the $2 \times 2$ periodic structure. Due to the heterogeneity of the network, chains exhibit three distinct stretch ratios under tension: $\Lambda = \lambda_0$, $\Lambda = 1$, and $\Lambda = (\lambda_0 + 1)/2$. $\lambda_0$ denotes the network stretch ratio.
\textbf{(d)} Nominal stress $S$ versus network stretch ratio $\lambda$. The nominal stress $S$ is defined as the total network force divided by the layer number $N_x$. $\Lambda_f$ is the rupture stretch ratio of an individual chain, and $f_f$ is the corresponding breakage force.
}
  \label{fig-supp:SI_Fig9_25_Period_Network}
\end{figure}

\vspace*{0pt}
\begin{figure}[H]
  \centering
  \includegraphics[trim={0cm 0cm 0cm 0cm},clip, width=1.0\textwidth]{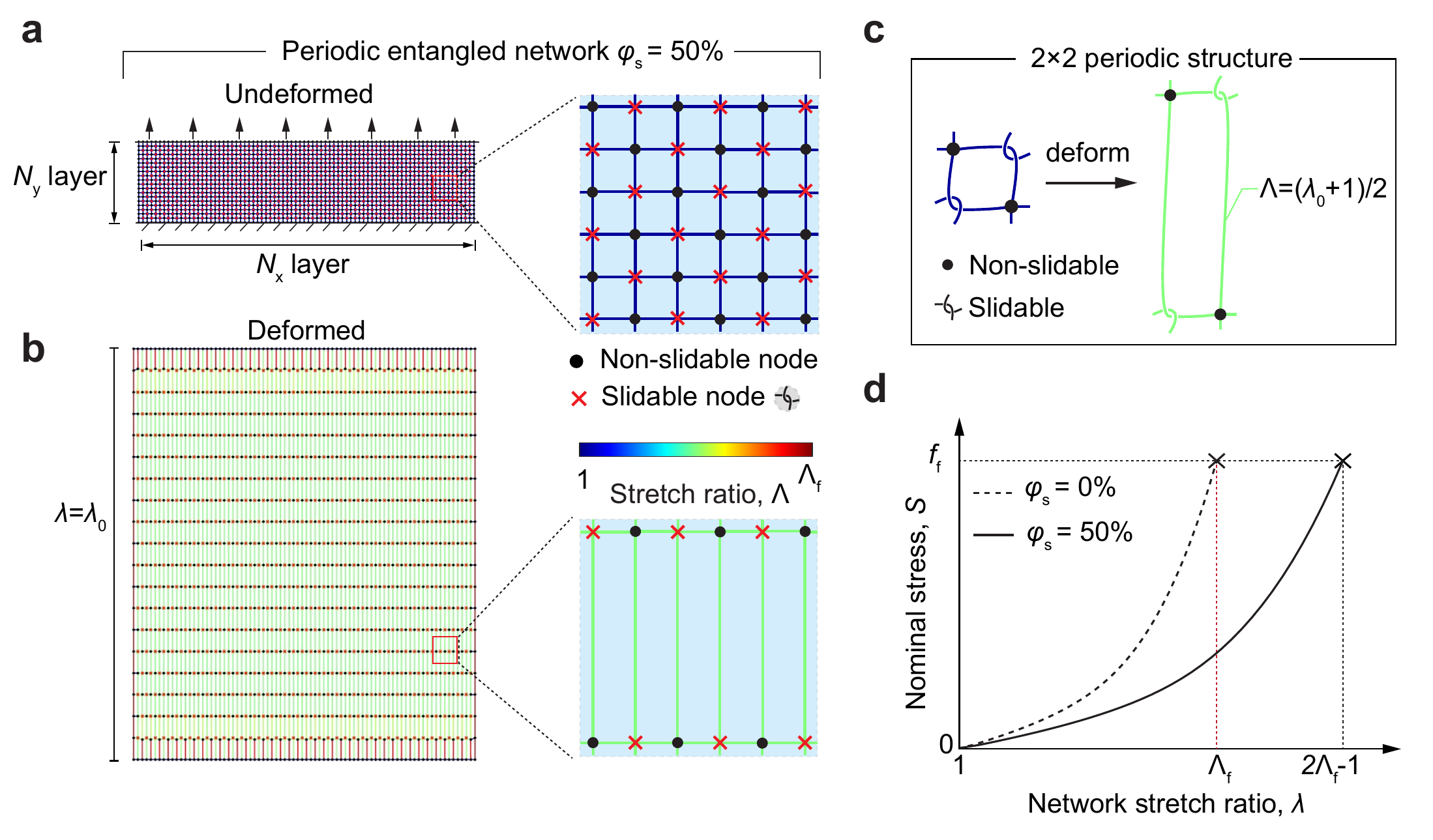}
  \caption{\textbf{Simulation of periodic entangled network with $\varphi_s = 50\%$.} 
\textbf{(a)} Network configuration in the undeformed state. 
\textbf{(b)} Network subjected to uniaxial tension with stretch ratio $\lambda_0$. Black nodes represent non-slidable nodes, and red crosses indicate slidable nodes with fixed orientation. 
\textbf{(c)} Undeformed and deformed states of the $2 \times 2$ periodic structure. Each chain in the network is stretched uniformly with an identical stretch ratio $\Lambda = (\lambda_0 + 1)/2$. 
\textbf{(d)} Nominal stress $S$ versus network stretch ratio $\lambda$. The nominal stress $S$ is defined as the total network force divided by the layer number $N_x$. $\Lambda_f$ is the rupture stretch ratio of an individual chain, and $f_f$ is the corresponding breakage force.
}
  \label{fig-supp:SI_Fig10_50_Period_Network}
\end{figure}

\vspace*{0pt}
\begin{figure}[H]
  \centering
  \includegraphics[trim={0cm 0cm 0cm 0cm},clip, width=1.0\textwidth]{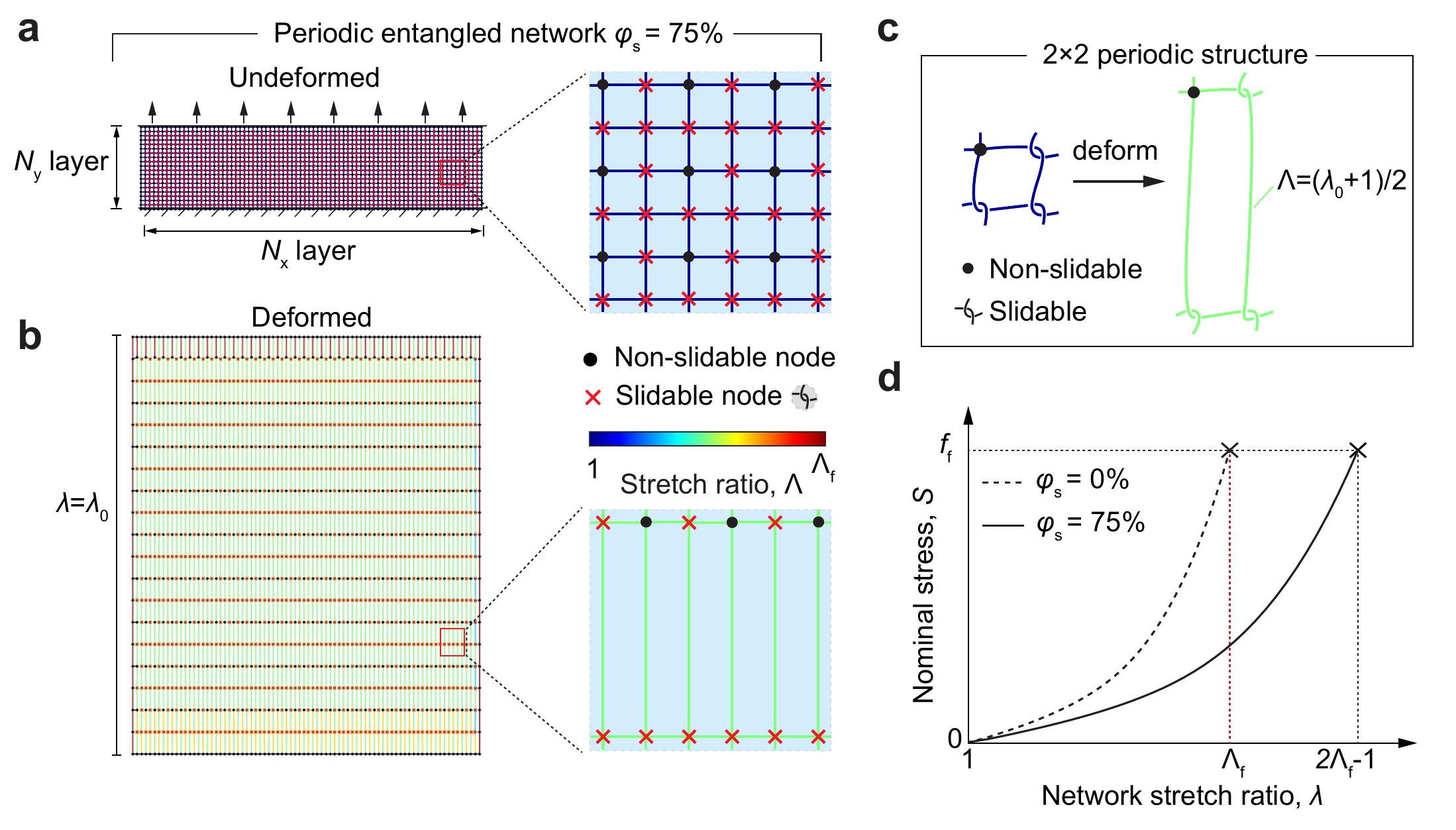}
  \caption{\textbf{Simulation of periodic entangled network with $\varphi_s = 75\%$.} 
\textbf{(a)} Network configuration in the undeformed state. 
\textbf{(b)} Network subjected to uniaxial tension with stretch ratio $\lambda_0$. Black nodes represent non-slidable nodes, and red crosses indicate slidable nodes with fixed orientation. 
\textbf{(c)} Undeformed and deformed states of the $2 \times 2$ periodic structure. Each chain in the network is stretched uniformly with an identical stretch ratio $\Lambda = (\lambda_0 + 1)/2$. 
\textbf{(d)} Nominal stress $S$ versus network stretch ratio $\lambda$. The nominal stress $S$ is defined as the total network force divided by the layer number $N_x$. $\Lambda_f$ is the rupture stretch ratio of an individual chain, and $f_f$ is the corresponding breakage force.
}
  \label{fig-supp:SI_Fig11_75_Period_Network}
\end{figure}

\vspace*{0pt}
\begin{figure}[H]
  \centering
  \includegraphics[trim={0cm 0cm 0cm 0cm},clip, width=1.0\textwidth]{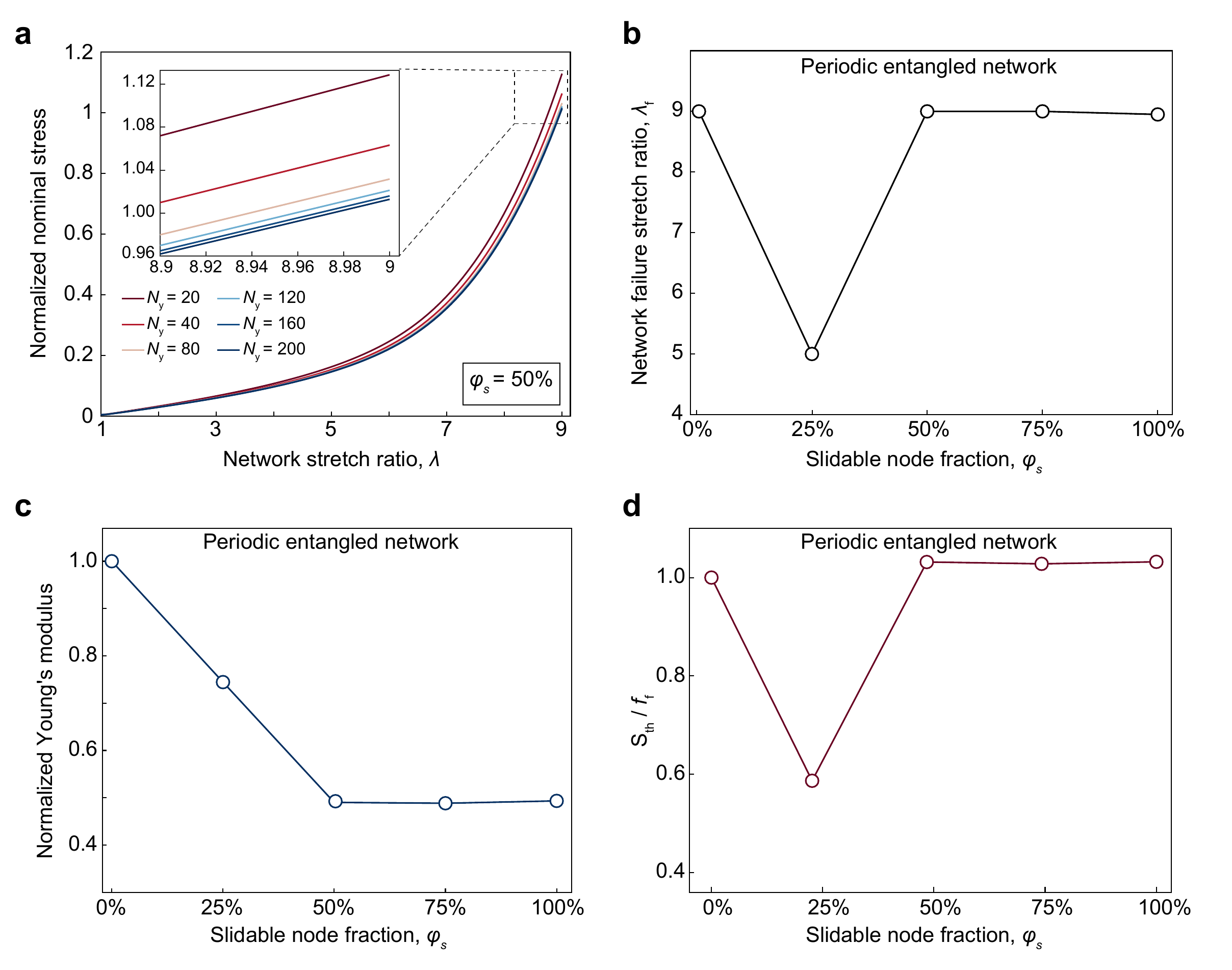}
  \caption{\textbf{Summarized simulation results for elasticity of periodic entangled networks.} 
\textbf{(a)} 
Normalized nominal stress of periodic entangled networks with $\varphi_s = 50\%$ at varying numbers of layers. The normalized nominal stress is defined as nominal stress $S$ divided by chain breakage force $f_f$. The network has a length-to-width ratio of 4:1, and $N_y$ is the vertical layer number.
\textbf{(b)} Network failure stretch ratio $\lambda_f$ of periodic entangled networks with various slidable node fraction $\varphi_s$. 
\textbf{(c)} Normalized Young’s modulus of periodic entangled networks with various slidable node fraction $\varphi_s$. The normalized Young’s modulus is defined as the network modulus divided by the modulus of spring network ($\varphi_s = 0\%$). 
\textbf{(d)} Normalized strength of periodic entangled networks with various slidable node fraction $\varphi_s$. The normalized strength is defined as the network failure strength $S_{th}$ divided by the failure strength of spring network ($\varphi_s = 0\%$).
}
  \label{fig-supp:SI_Fig12_SimulationOfPeriodNetwork}
\end{figure}

\vspace*{0pt}
\begin{figure}[H]
  \centering
  \includegraphics[trim={0cm 0cm 0cm 0cm},clip, width=1.0\textwidth]{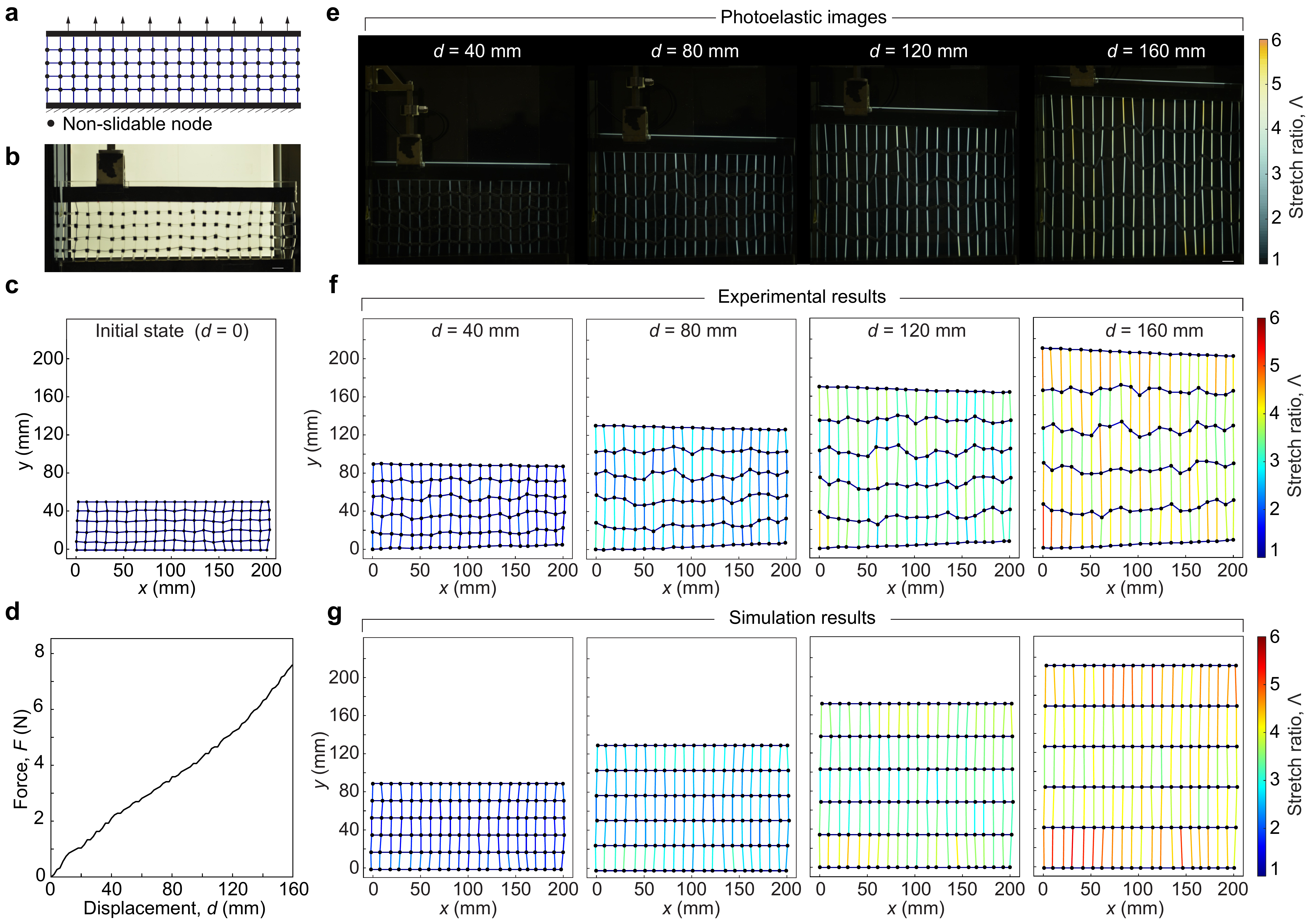}
  \caption{\textbf{Experimental and simulation results of a hydrogel fabric without slidable nodes ($\varphi_s = 0\%$) under various displacements.}
\textbf{(a)} Topology of the network.
\textbf{(b)} Image of the hydrogel fabric prepared based on the topology in (a). Scale bar: 10 mm. 
\textbf{(c)} Initial undeformed state of the hydrogel fabric network. 
\textbf{(d)} Force-displacement curve obtained from mechanical testing. 
\textbf{(e)} Photoelastic images of the hydrogel network under various displacements. Scale bar: 10 mm. 
\textbf{(f)} Stretch ratio distribution of hydrogel fibers extracted from experimental images in (e). Black nodes represent non-slidable nodes. 
\textbf{(g)} Stretch ratio distribution of hydrogel fibers extracted from simulation results of the corresponding network. Black nodes represent non-slidable nodes.
}
  \label{fig-supp:SI_Fig13_Elasticity_Experiment_Nonslidable}
\end{figure}

\vspace*{0pt}
\begin{figure}[H]
  \centering
  \includegraphics[trim={0cm 0cm 0cm 0cm},clip, width=1.0\textwidth]{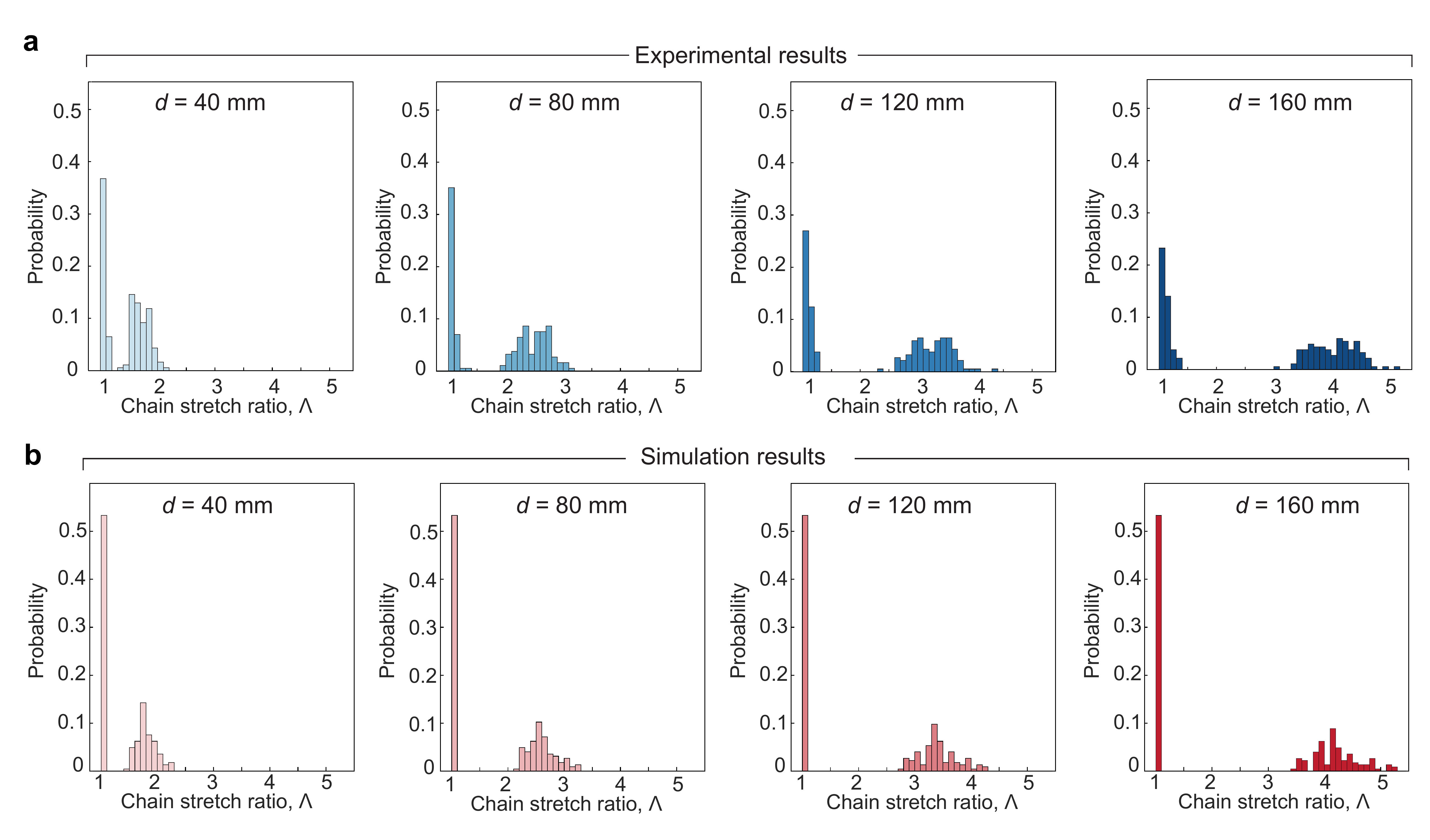}
  \caption{\textbf{Chain stretch ratios distribution of the spring network ($\varphi_s = 0\%$) } 
\textbf{(a)} Experimental results of chain stretch ratio distribution of spring network under various displacements.
\textbf{(b)} Simulation results of chain stretch ratio distribution of spring network under various displacements.
}
  \label{fig-supp:SI_Fig14_DistributionOfNonslidableNetwork}
\end{figure}

\vspace*{0pt}
\begin{figure}[H]
  \centering
  \includegraphics[trim={0cm 0cm 0cm 0cm},clip, width=1.0\textwidth]{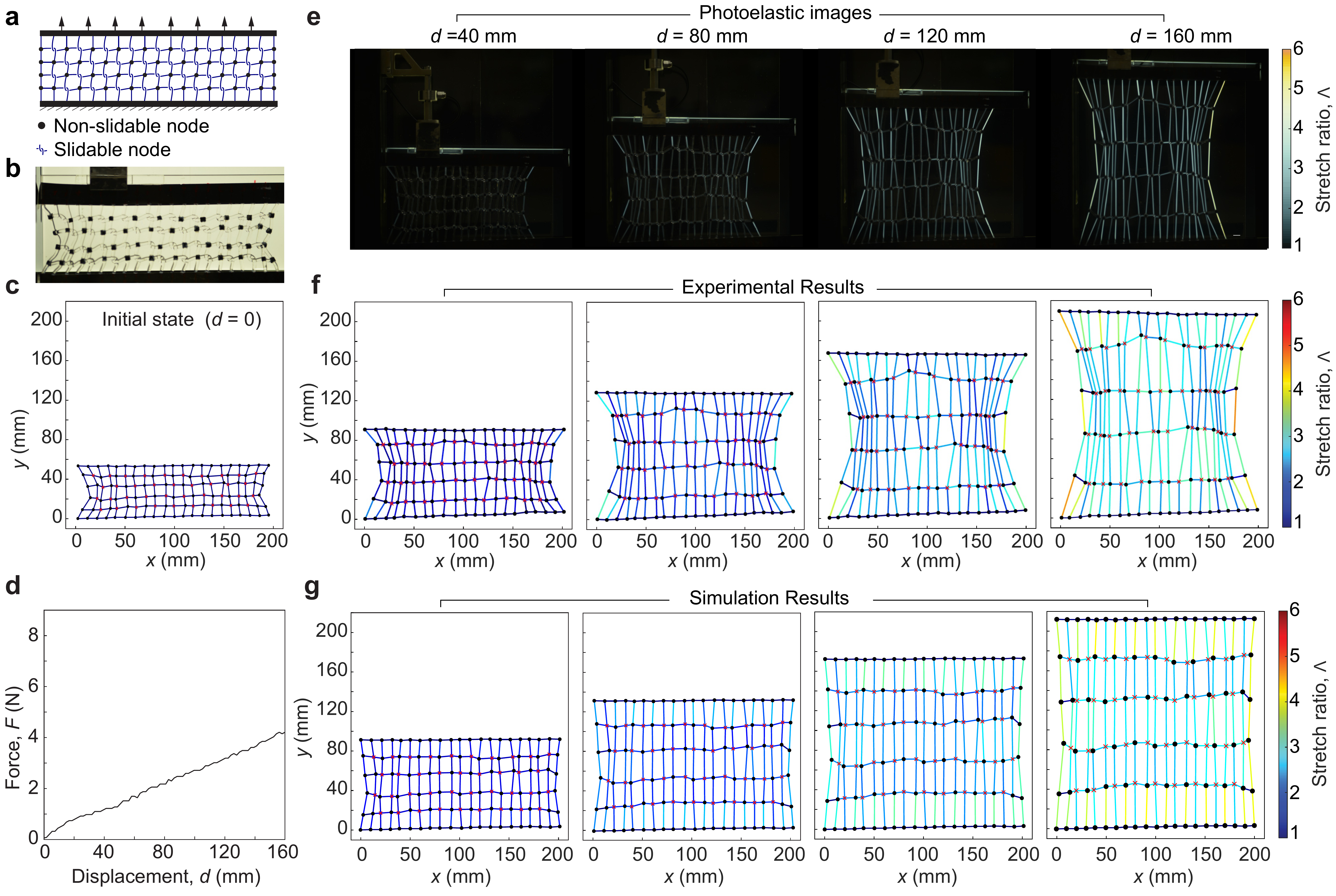}
  \caption{\textbf{Experimental and simulation results of an entangled hydrogel fabric with slidable node fraction ($\varphi_s = 50\%$) under various displacements.}
\textbf{(a)} Topology of the network.
\textbf{(b)} Image of the hydrogel fabric entangled network fabricated based on the topology in (a). Scale bar: 10 mm. 
\textbf{(c)} Initial undeformed state of the hydrogel fabric entangled network. 
\textbf{(d)} Force-displacement curve obtained from mechanical testing. 
\textbf{(e)} Photoelastic images of the hydrogel fiber network under increasing displacement. Scale bar: 10 mm. 
\textbf{(f)} Stretch ratio distribution of hydrogel fibers extracted from experimental images in (e). Black nodes denote non-slidable nodes, and red crosses denote slidable nodes. 
\textbf{(g)} Stretch ratio distribution of hydrogel fibers extracted from simulation results of the corresponding network. Black nodes denote non-slidable nodes, and red crosses denote slidable nodes.
}
  \label{fig-supp:SI_Fig15_Elasticity_Experiment_50_slidable}
\end{figure}

\vspace*{0pt}
\begin{figure}[H]
  \centering
  \includegraphics[trim={0cm 0cm 0cm 0cm},clip, width=1.0\textwidth]{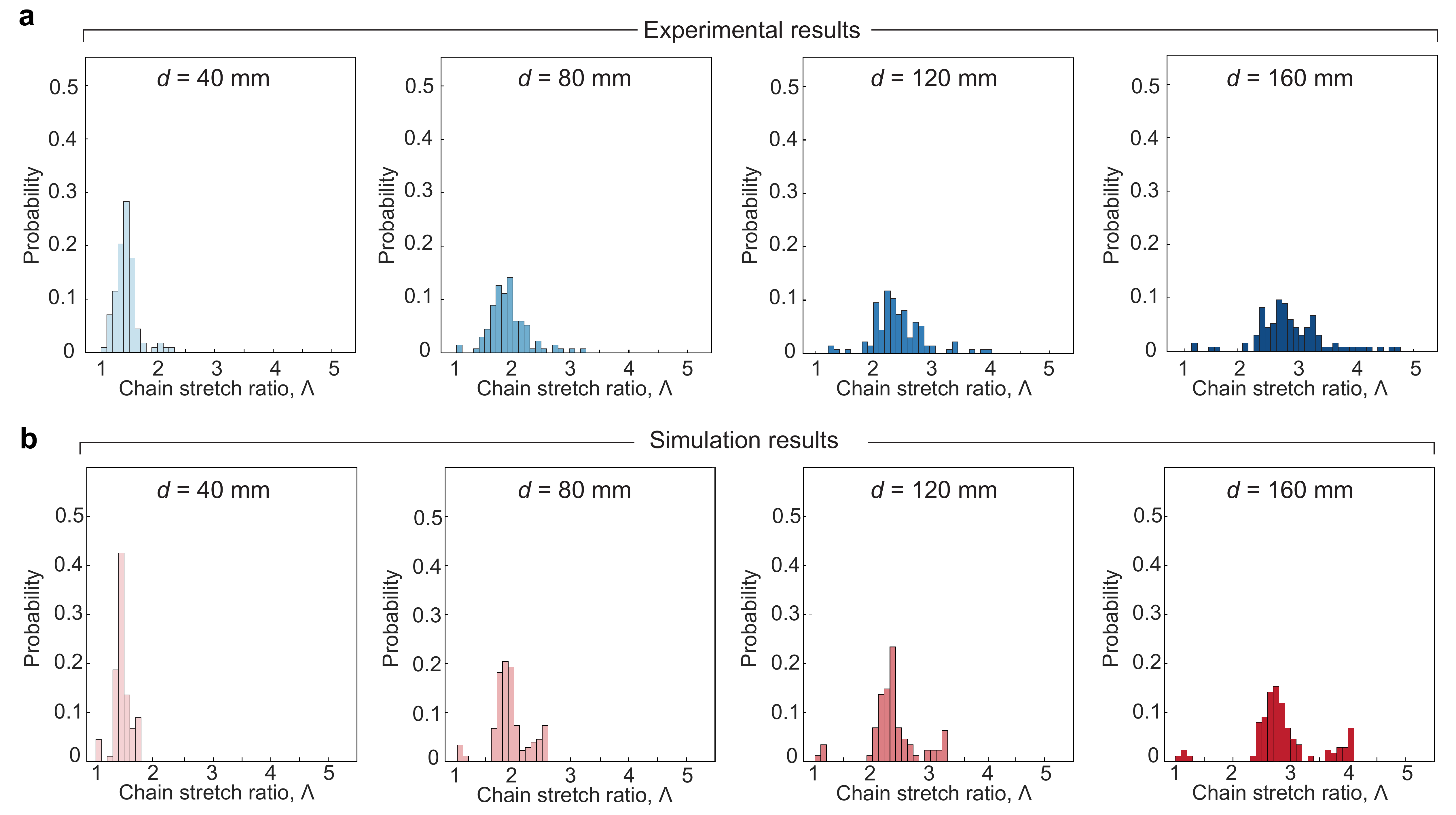}
  \caption{\textbf{Chain stretch ratios distribution of the periodic entangled network ($\varphi_s = 50\%$).}
  \textbf{(a)} Experimental results of chain stretch ratio distribution of the periodic entangled network under various displacements. 
  \textbf{(b)} Simulation results of chain stretch ratio distribution of the periodic entangled network under various displacements.
}
  \label{fig-supp:SI_Fig16_DistributionOfSlidableNetwork}
\end{figure}

\vspace*{0pt}
\begin{figure}[H]
  \centering
  \includegraphics[trim={0cm 0cm 0cm 0cm},clip, width=1.0\textwidth]{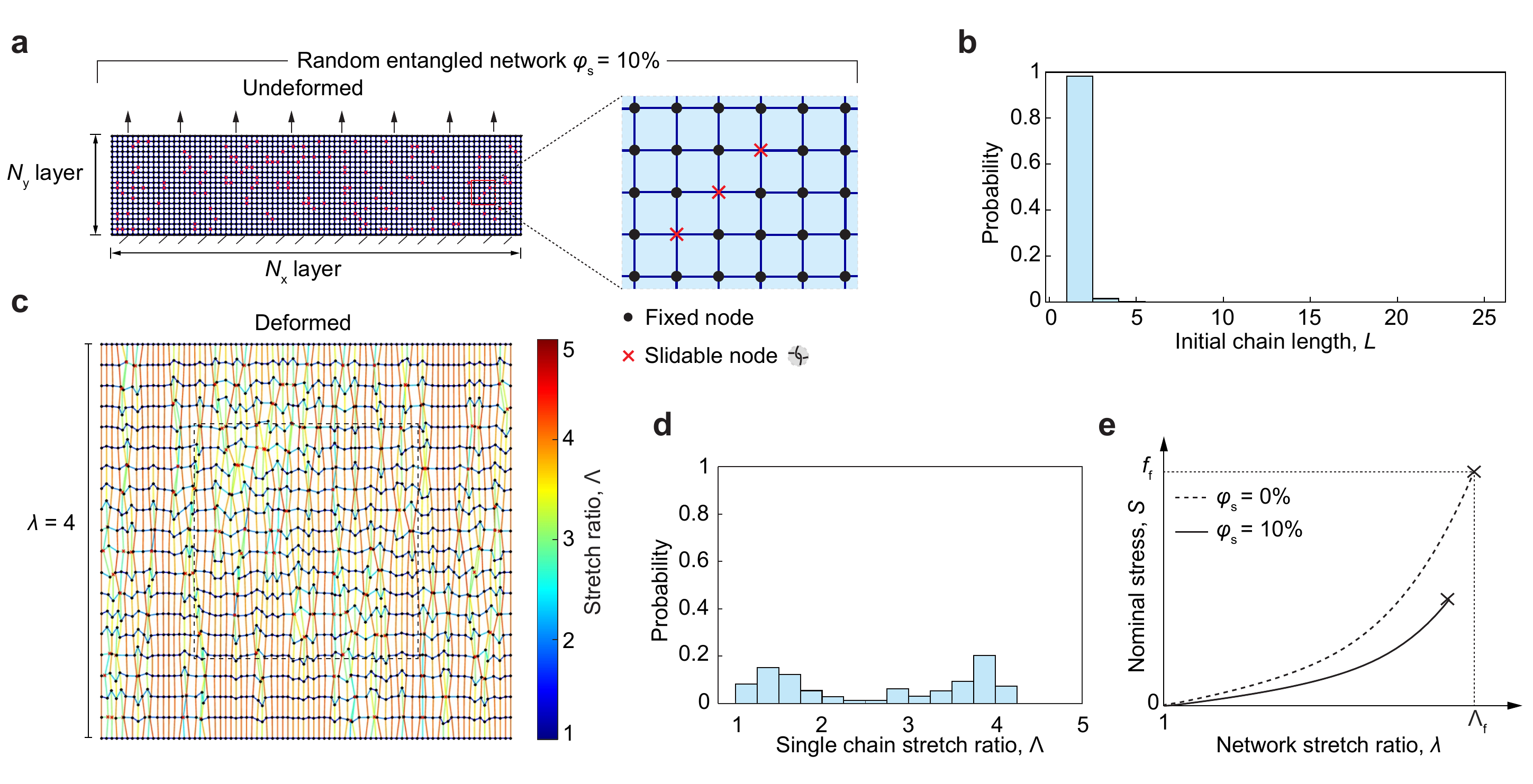}
  \caption{\textbf{Simulation of random entangled network with a slidable node fraction of $\varphi_s = 10\%$.}
\textbf{(a)} Network configuration in the undeformed state. 
\textbf{(b)} Initial chain length $L$ distribution. Most chains in the network are concentrated in the short-length range. 
\textbf{(c)} Network subjected to uniaxial tension at a stretch ratio of 4. 
\textbf{(d)} Distribution of chain stretch ratios at a network stretch ratio of 4. 
\textbf{(e)} Nominal stress $S$ versus network stretch ratio $\lambda$. The nominal stress $S$ is defined as the total network force divided by the horizontal layer number $N_x$. $\Lambda_f$ is the rupture stretch ratio of a single chain, and $f_f$ is the corresponding breakage force.
}
  \label{fig-supp:SI_Fig17_10_Random_Network}
\end{figure}

\vspace*{0pt}
\begin{figure}[H]
  \centering
  \includegraphics[trim={0cm 0cm 0cm 0cm},clip, width=1.0\textwidth]{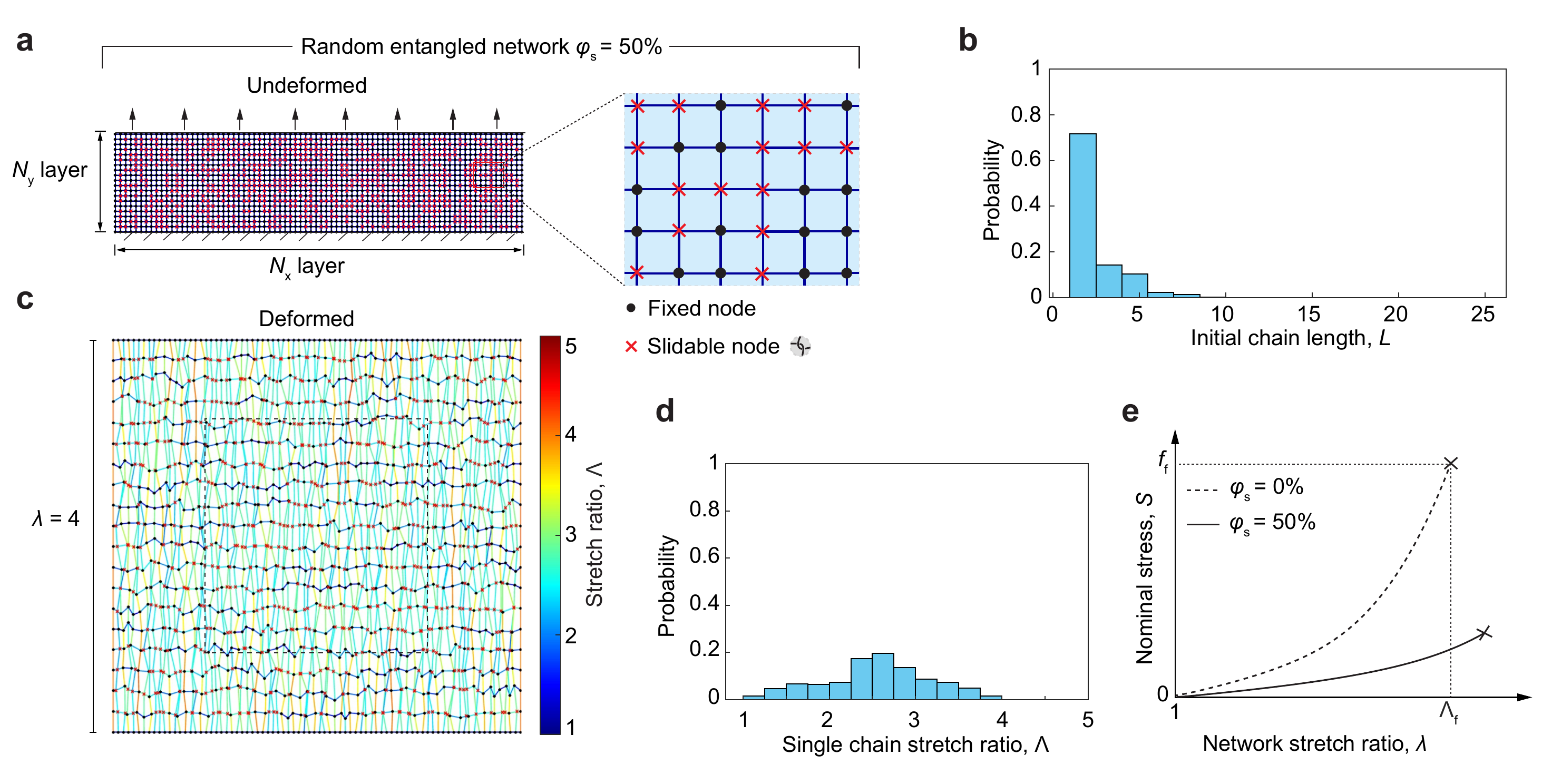}
  \caption{\textbf{Simulation of random entangled network with a slidable node fraction of $\varphi_s = 50\%$.}
\textbf{(a)} Network configuration in the undeformed state. 
\textbf{(b)} Initial chain length $L$ distribution. 
\textbf{(c)} Network subjected to uniaxial tension at a stretch ratio of 4. 
\textbf{(d)} Distribution of chain stretch ratios at a network stretch ratio of 4. 
\textbf{(e)} Nominal stress $S$ versus network stretch ratio $\lambda$. The nominal stress $S$ is defined as the total network force divided by the layer number $N_x$. $\Lambda_f$ is the rupture stretch ratio of a single chain, and $f_f$ is the corresponding breakage force.
}
  \label{fig-supp:SI_Fig18_50_Random_Network}
\end{figure}

\vspace*{0pt}
\begin{figure}[H]
  \centering
  \includegraphics[trim={0cm 0cm 0cm 0cm},clip, width=1.0\textwidth]{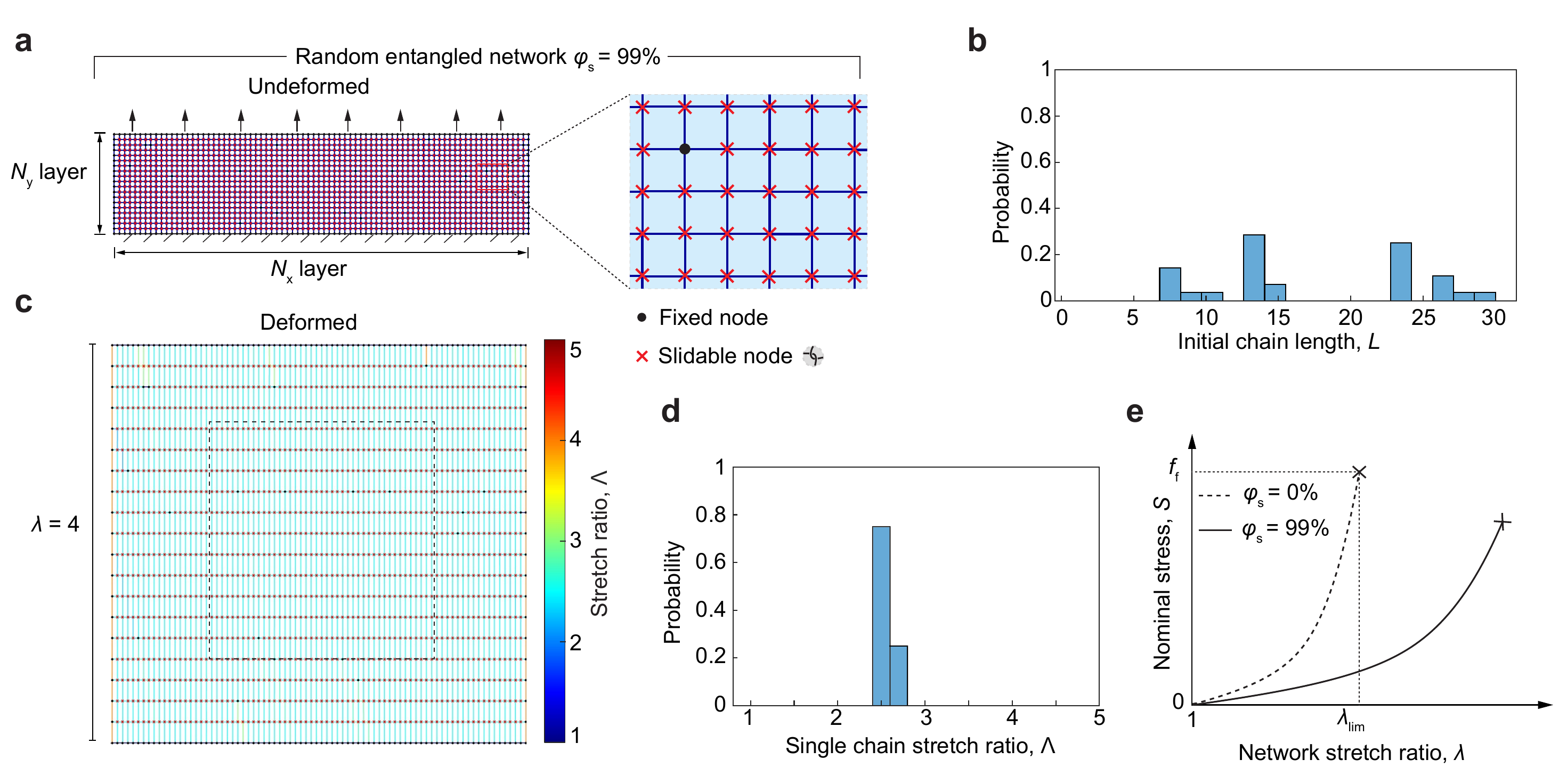}
\caption{\textbf{Simulation of random entangled network with a slidable node fraction of $\varphi_s = 99\%$.}
\textbf{(a)} Network configuration in the undeformed state. 
\textbf{(b)} Initial chain length $L$ distribution. Most chains in this network are concentrated in the long-length range. 
\textbf{(c)} Network subjected to uniaxial tension at a stretch ratio of 4. 
\textbf{(d)} Distribution of chain stretch ratios at a network stretch ratio of 4. Chains in this network undergo uniform deformation. 
\textbf{(e)} Nominal stress $S$ versus network stretch ratio $\lambda$. The nominal stress $S$ is defined as the total network force divided by the layer number $N_x$. $\Lambda_f$ is the rupture stretch ratio of a single chain, and $f_f$ is the corresponding breakage force.
}
  \label{fig-supp:SI_Fig19_99_Random_Network}
\end{figure}

\vspace*{0pt}
\begin{figure}[H]
  \centering
  \includegraphics[trim={0cm 0cm 0cm 0cm},clip, width=1.0\textwidth]{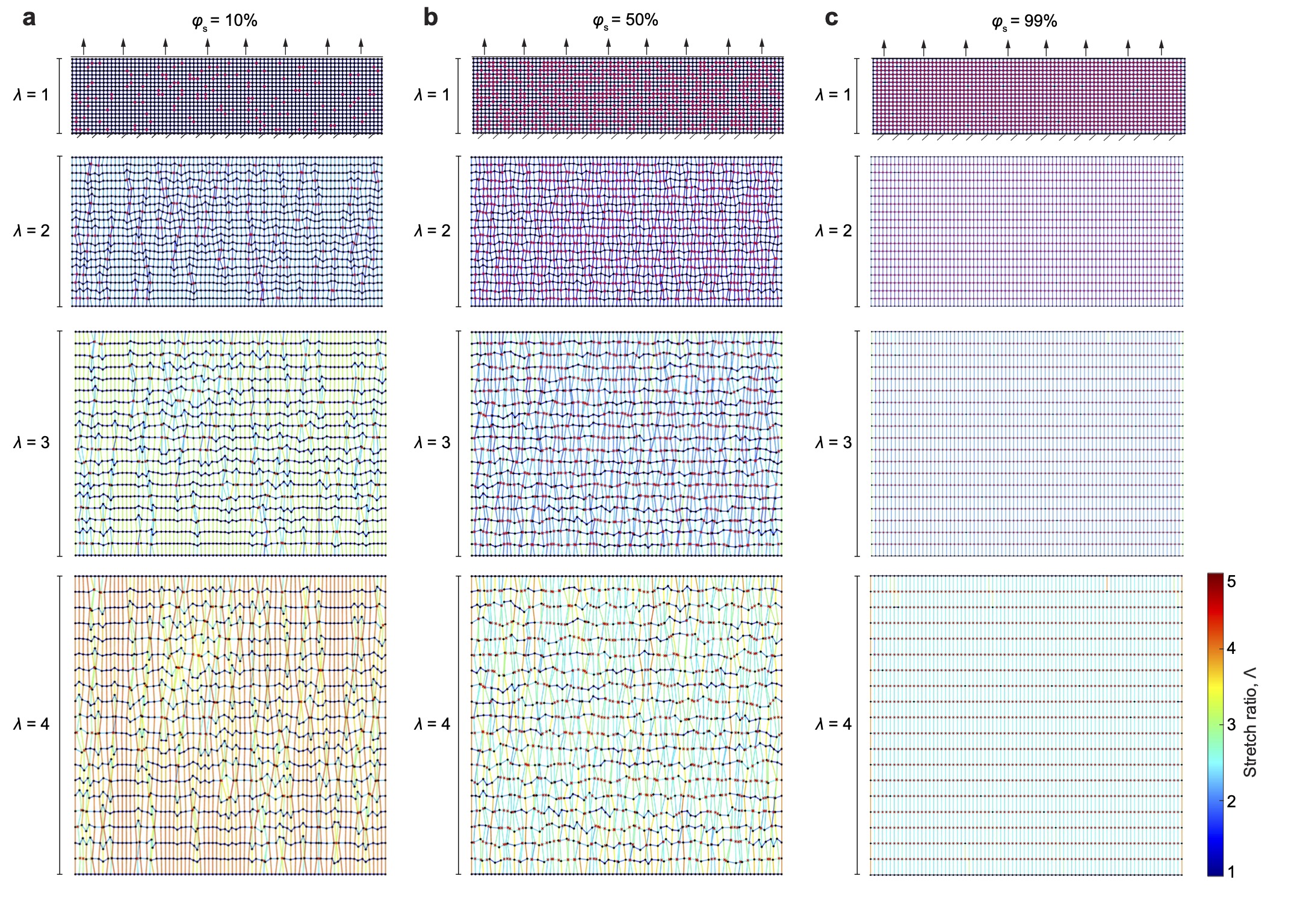}
  \caption{\textbf{Random entangled networks with different entangled node ratios subjected to various stretch ratio.} The network stretch ratios are $\lambda=2$, 3, 4.
\textbf{(a)} $\varphi_s = 10\%$. 
\textbf{(b)} $\varphi_s = 50\%$. 
\textbf{(c)} $\varphi_s = 99\%$.
}
  \label{fig-supp:SI_Fig20_10_50_99_networkstretching}
\end{figure}

\vspace*{0pt}
\begin{figure}[H]
  \centering
  \includegraphics[trim={0cm 0cm 0cm 0cm},clip, width=1.0\textwidth]{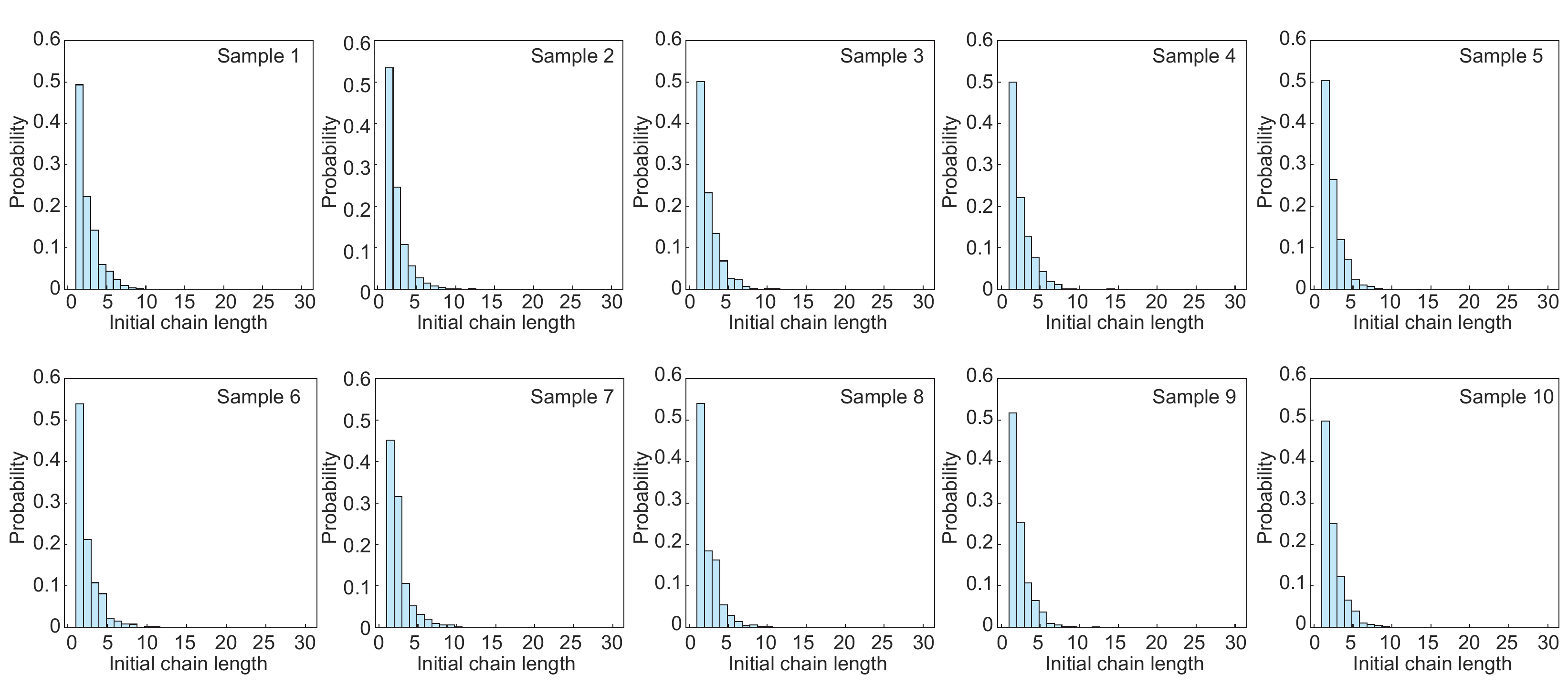}
  \caption{\textbf{Initial chain length distributions of random entangled network with a slidable node raction of $\varphi_s = 50\%$.} 10 independent random samples are listed to show the distribution of the initial chain length.
}
  \label{fig-supp:SI_Fig21_50_ChainlengthDistribution}
\end{figure}

\vspace*{0pt}
\begin{figure}[H]
  \centering
  \includegraphics[trim={0cm 0cm 0cm 0cm},clip, width=1.0\textwidth]{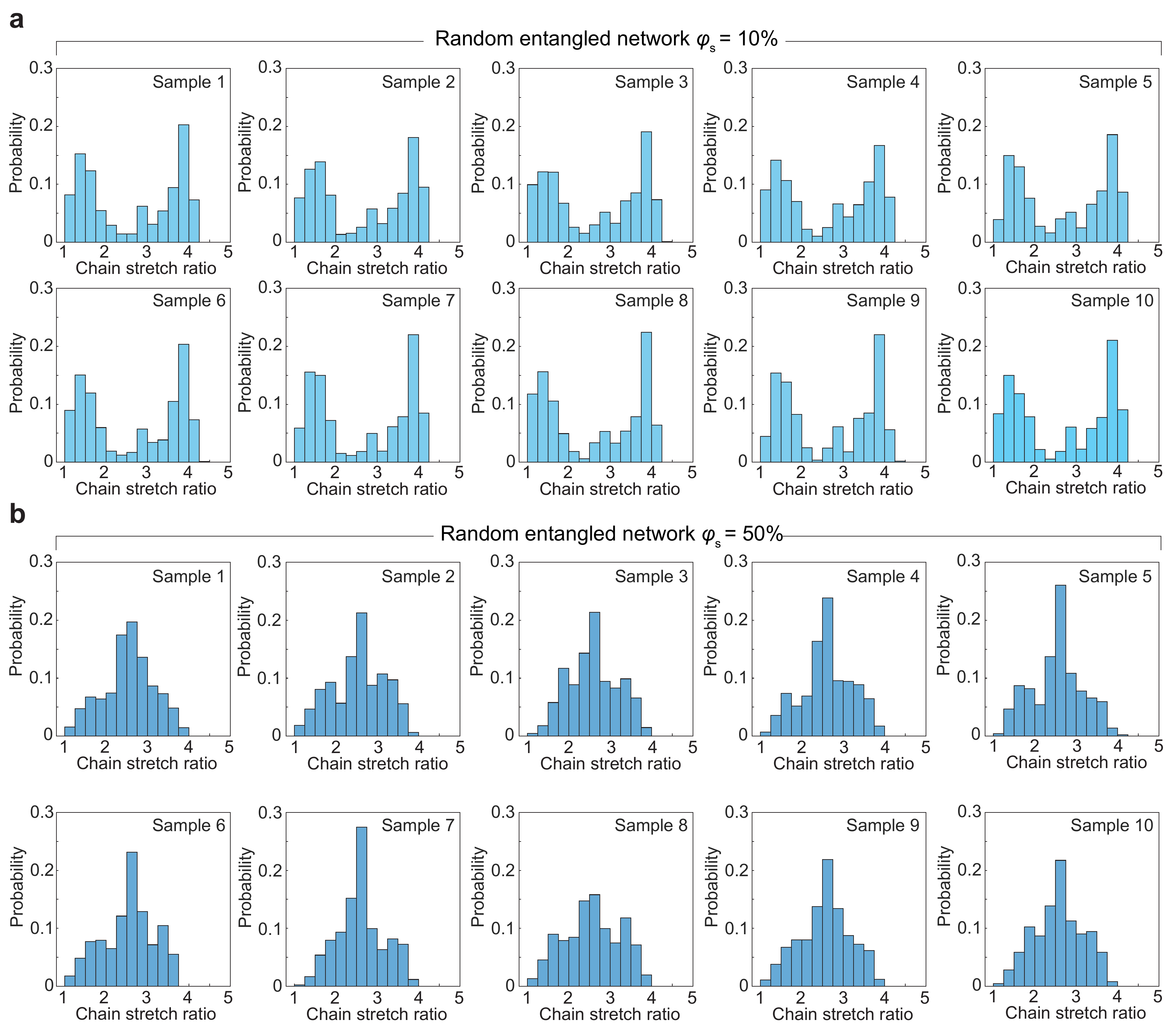}
  \caption{\textbf{Chain stretch ratio of random entangled networks at deformed state with network stretch ratio $\lambda =4$.} For each network with a different slidable node fraction, 10 independent random samples are shown to illustrate the distribution of chain stretch ratios.
\textbf{(a)} $\varphi_s = 10\%$. 
\textbf{(b)} $\varphi_s = 50\%$. 
}
  \label{fig-supp:SI_Fig22_10_50_stretchratio}
\end{figure}

\vspace*{0pt}
\begin{figure}[H]
  \centering
  \includegraphics[trim={0cm 0cm 0cm 0cm},clip, width=1.0\textwidth]{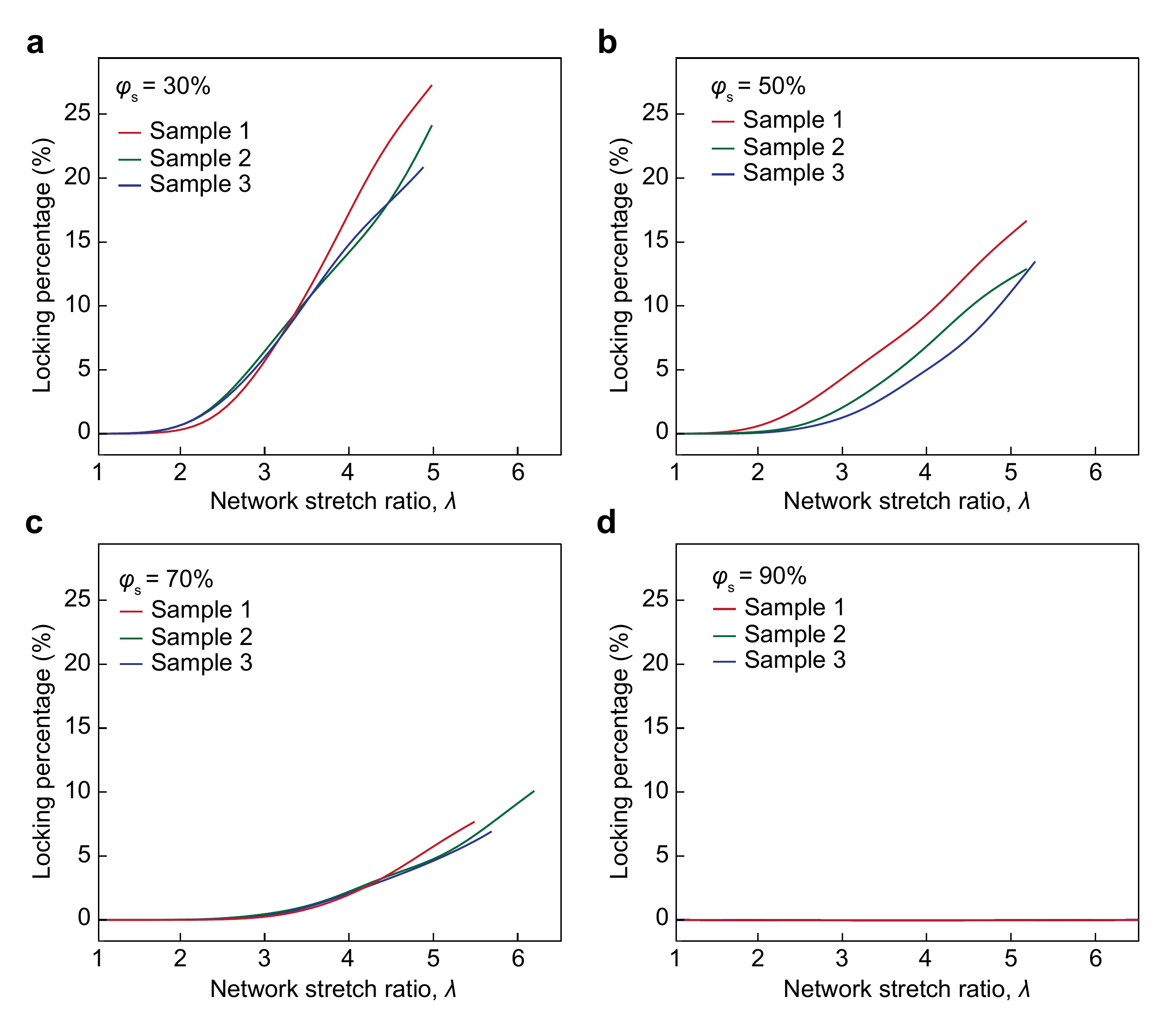}
  \caption{\textbf{Curves of locking percentage of the slidable nodes versus network stretch ratio.} Each figure shows 3 different samples. The locking percentage of the slidable nodes is computed as the number of locked slidable nodes divided by the number of all slidable nodes.
\textbf{(a)} Random entangled network with a slidable node fraction of $\varphi_s = 30\%$. 
\textbf{(b)} Random entangled network with a slidable node fraction of $\varphi_s = 50\%$. 
\textbf{(c)} Random entangled network with a slidable node fraction of $\varphi_s = 70\%$. 
\textbf{(d)} Random entangled network with a slidable node fraction of $\varphi_s = 90\%$. 
}
  \label{fig-supp:SI_Fig23_LockingEffectSample}
\end{figure}

\vspace*{0pt}
\begin{figure}[H]
  \centering
  \includegraphics[trim={0cm 0cm 0cm 0cm},clip, width=0.72\textwidth]{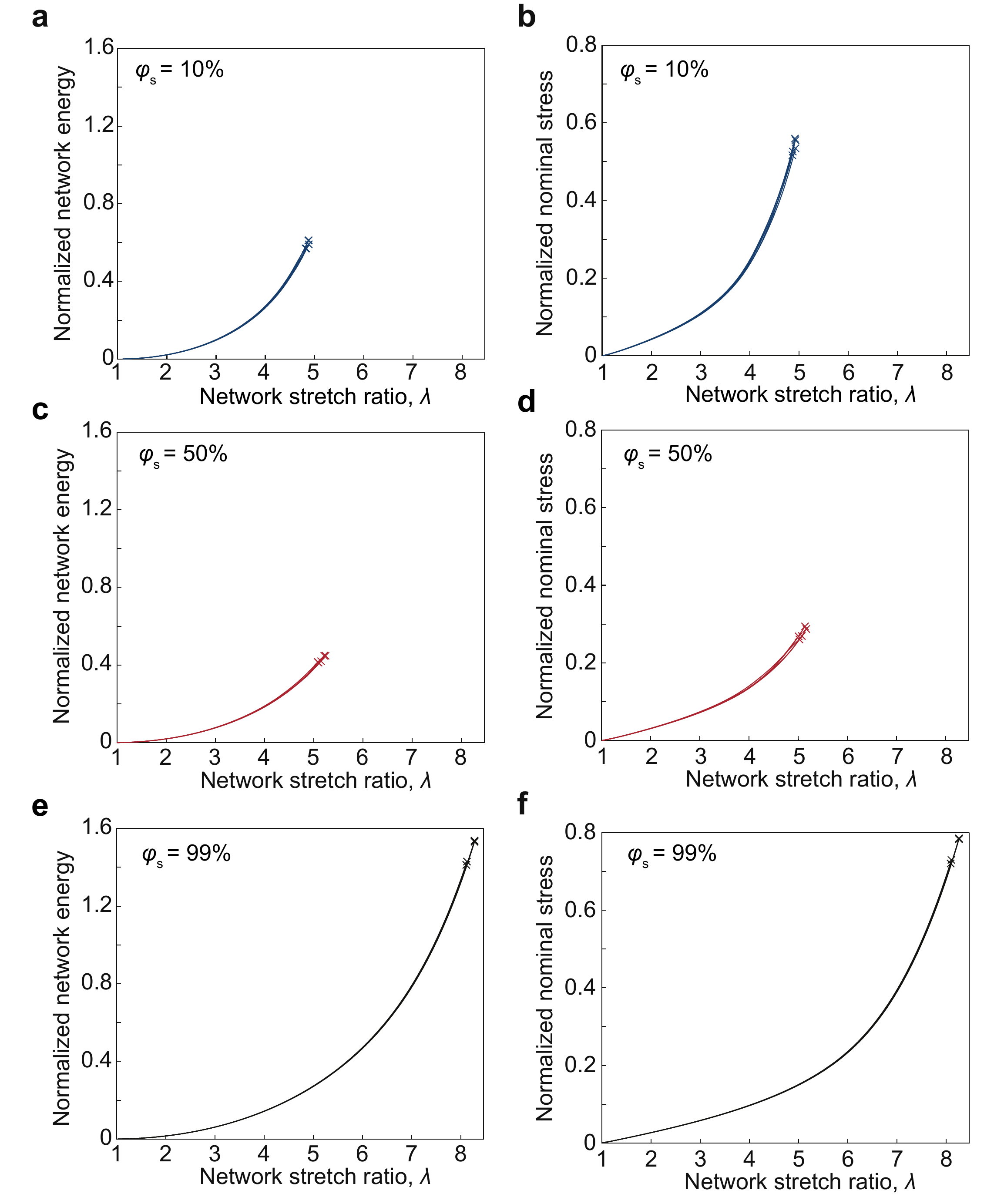}
  \caption{\textbf{Simulation results of random entangled networks with various slidable node fractions $\varphi_s$.}
\textbf{(a, c, e)} Energy-stretch curves of random networks with $\varphi_s = 10\%$, $50\%$, and $99\%$. The network energy $U$ is normalized by the rupture energy of the spring network ($\varphi_s=0\%$). 
\textbf{(b, d, f)} Energy-stretch curves of random networks with $\varphi_s = 10\%$, $50\%$, and $99\%$.
 The nominal stress of network $S$ is normalized by the breakage force of an individual chain $f_f$.
For each $\varphi_s$, we take 5 random samples. Despite the variability introduced by randomness, the energy-stretch curve and force-stretch curve of random networks are surprisingly consistent across different samples.
}
  \label{fig-supp:SI_Fig24_Simulation_ForceEnergyCurve}
\end{figure}

\vspace*{0pt}
\begin{figure}[H]
  \centering
  \includegraphics[trim={0cm 0cm 0cm 0cm},clip, width=1.0\textwidth]{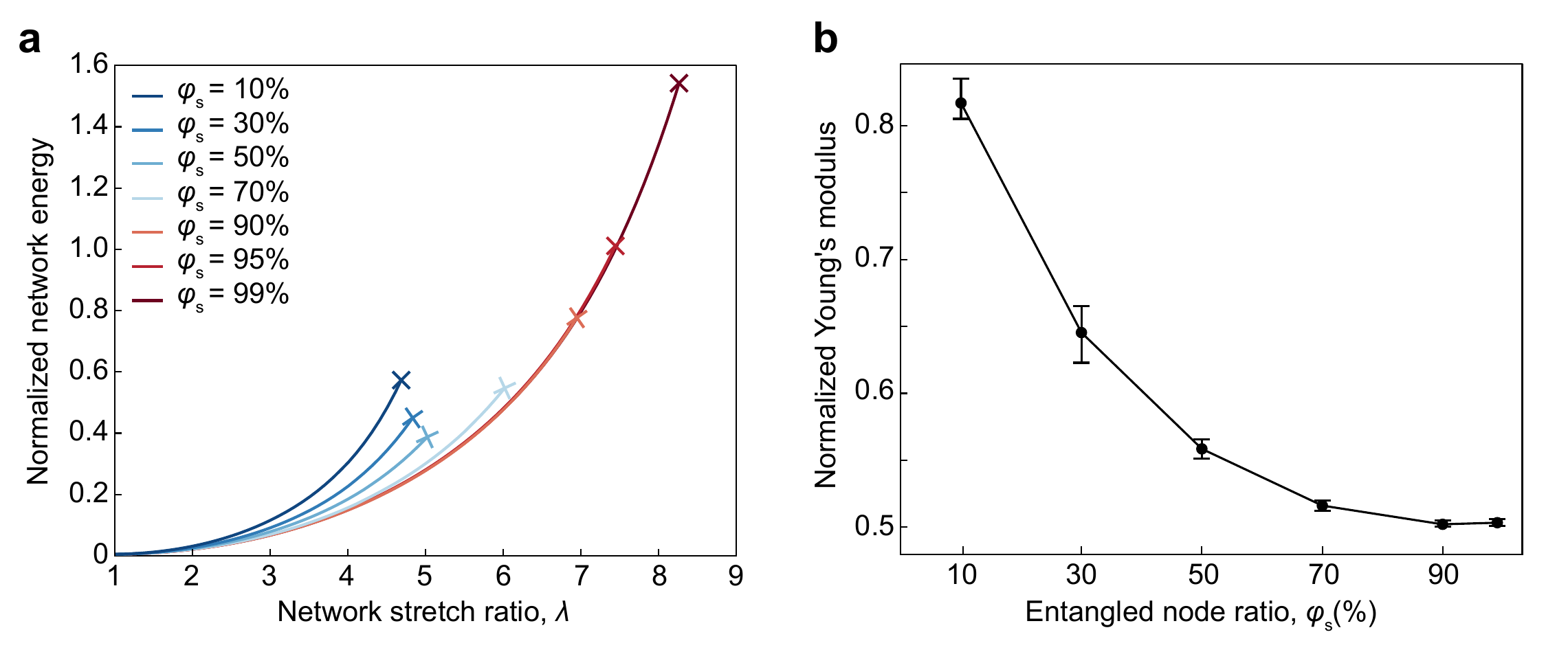}
  \caption{\textbf{Energy and modulus of random entangled networks with different slidable node fractions $\varphi_s$.} 
\textbf{(a)} Energy-stretch curves of random network with various slidable node fractions $\varphi_s$. The network energy $U$ is normalized by the rupture energy of an individual chain $U_{chain}$.
\textbf{(b)} Young's modulus of random entangled network decreases with the increase of slidable node fraction $\varphi_s$ in random entangled network. The Young's modulus is normalized by the modulus of spring network ($\varphi_s = 0\%$). }
  \label{fig-supp:SI_Fig25_Energy_ModulusOfRandomNetwork}
\end{figure}

\vspace*{0pt}
\begin{figure}[H]
  \centering
  \includegraphics[trim={0cm 0cm 0cm 0cm},clip, width=0.8\textwidth]{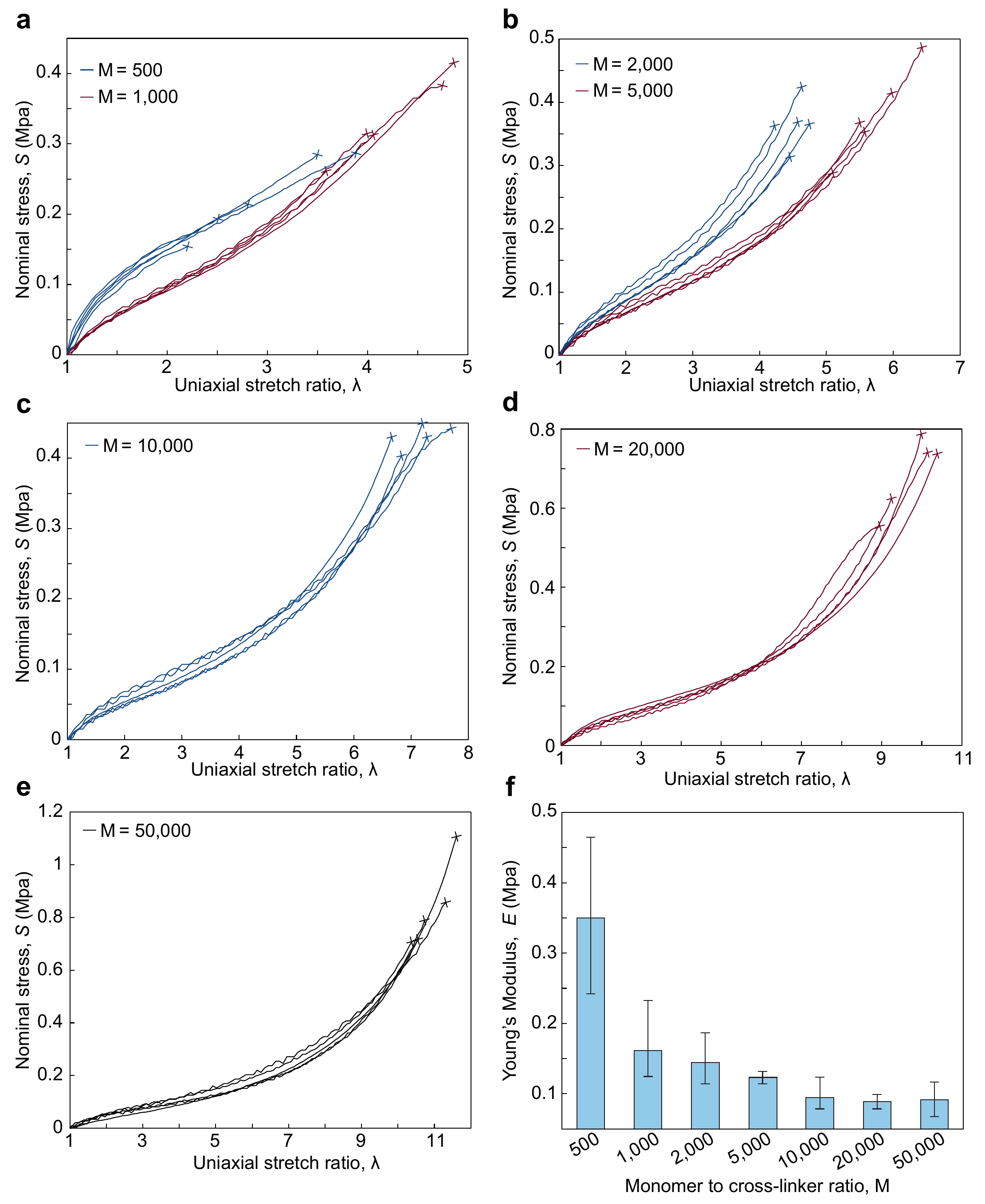}
  \caption{\textbf{Mechanical characterizations of entangled hydrogels with varying monomer-to-crosslinker molar ratios $M$.}
\textbf{(a)} Stress-stretch curves of entangled hydrogel with monomer-to-crosslinker molar ratios of 500 and 1,000 (5 samples each). 
\textbf{(b)} Ratios of 2,000 and 5,000 (5 samples each). 
\textbf{(c)} Ratio of 10,000 (5 samples). 
\textbf{(d)} Ratio of 20,000 (5 samples). 
\textbf{(e)} Ratio of 50,000 (5 samples). 
\textbf{(f)} Young’s modulus of entangled hydrogel with different monomer-to-crosslinker ratios M. The modulus $E$ decreases and then reaches a plateau as the monomer-to-crosslinker ratio M increases.
}
  \label{fig-supp:SI_Fig26_HydrogelExperimentData}
\end{figure}

\vspace*{0pt}
\begin{figure}[H]
  \centering
  \includegraphics[trim={0cm 0cm 0cm 0cm},clip, width=1.0\textwidth]{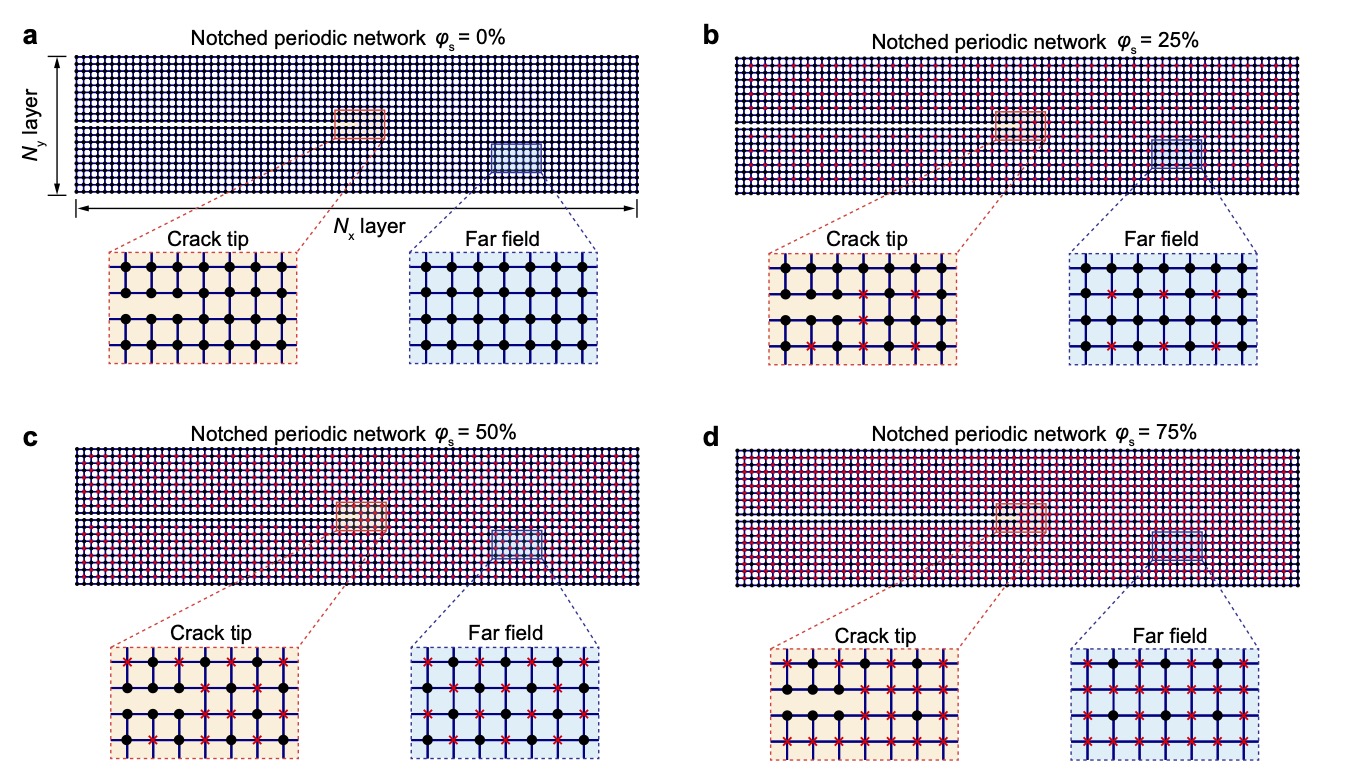}
  \caption{\textbf{Configurations of notched periodic entangled networks with different slidable node fractions $\varphi_s$.} 
\textbf{(a)} $\varphi_s = 0\%$. 
\textbf{(b)} $\varphi_s = 25\%$. 
\textbf{(c)} $\varphi_s = 50\%$. 
\textbf{(d)} $\varphi_s = 75\%$.
}
  \label{fig-supp:SI_Fig27_CrackNetworkConfiguration}
\end{figure}

\vspace*{0pt}
\begin{figure}[H]
  \centering
  \includegraphics[trim={0cm 0cm 0cm 0cm},clip, width=1.0\textwidth]{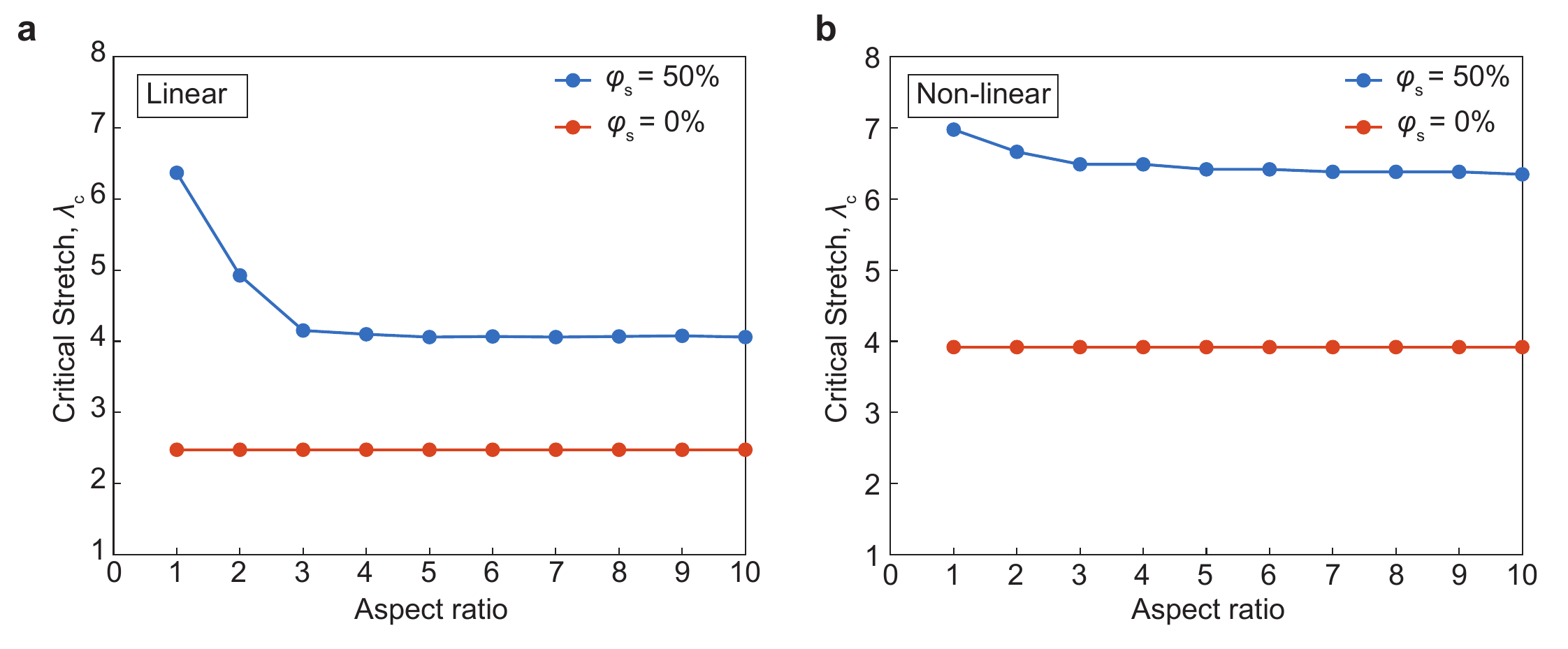}
  \caption{\textbf{Convergence behavior of the critical stretch $\lambda_c$ of notched periodic entangled network with different aspect ratios.} The aspect ratio is defined as $N_x/N_y$, where $N_x$ and $N_y$ denote the number of horizontal layers and the number of vertical layers, respectively. $N_y=20$ is fixed.
  \textbf{(a)} Periodic entangled networks with linear chains. For spring networks ($\varphi_s = 0\%$), $\lambda_c$ converges at an aspect ratio of 1. For entangled networks ($\varphi_s=50\%$), $\lambda_c$ converges at an aspect ratio of 3.
    \textbf{(b)} Periodic entangled networks with nonlinear chains. For spring networks ($\varphi_s = 0\%$), $\lambda_c$ converges at an aspect ratio of 1. For entangled networks ($\varphi_s=50\%$), $\lambda_c$ converges at an aspect ratio of 4. This indicates that an aspect ratio of $N_x/N_y = 4$ provides reliable results for fracture simulations.}
  \label{fig-supp:SI_Fig65_AspectRatio}
\end{figure}

\vspace*{0pt}
\begin{figure}[H]
  \centering
  \includegraphics[trim={0cm 0cm 0cm 0cm},clip, width=1.0\textwidth]{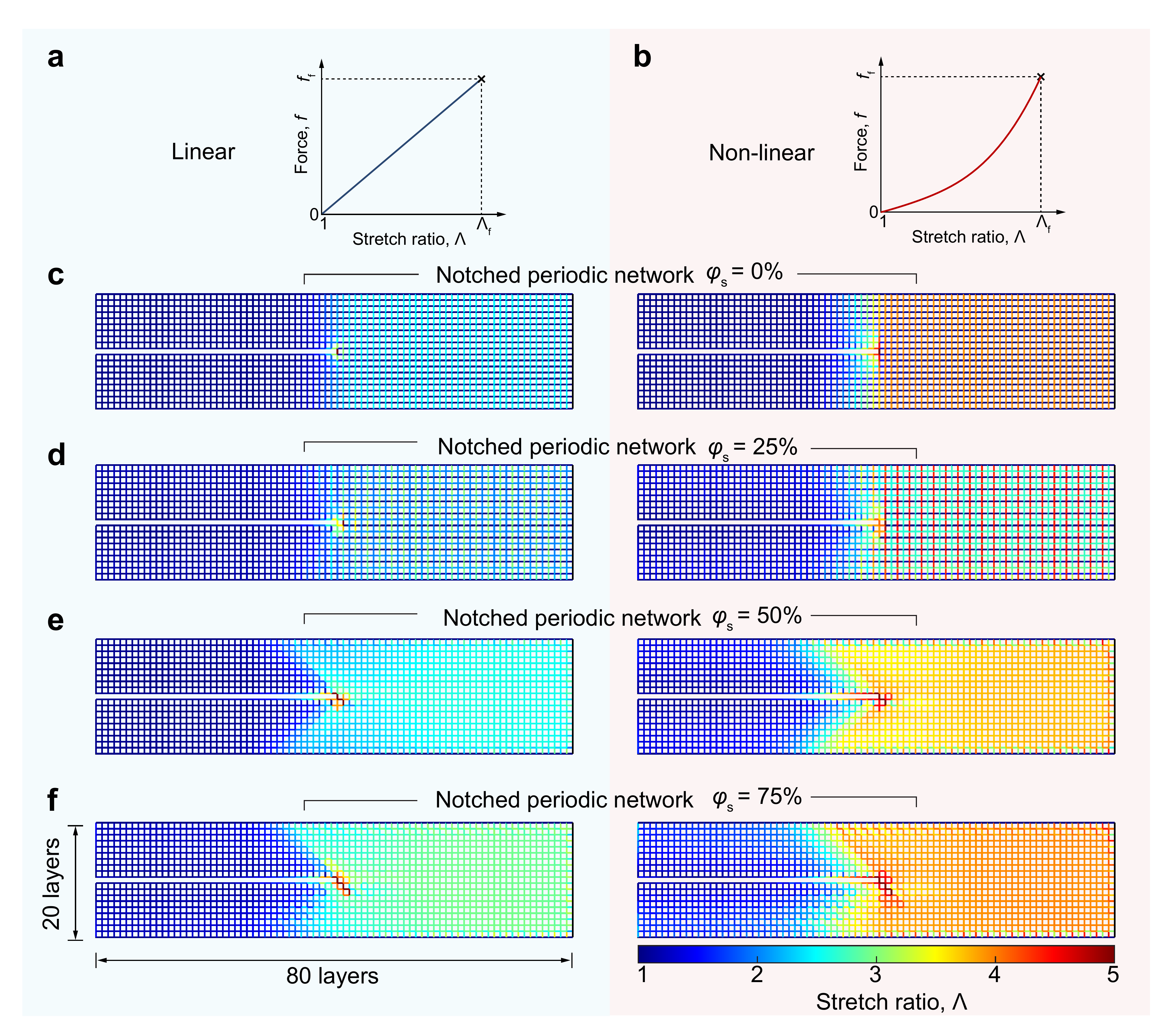}
  \caption{\textbf{Chain stretch ratio distribution of notched periodic entangled networks at crack initialization, shown in the undeformed state.}
\textbf{(a)} Force-stretch relationship of a linear chain. 
\textbf{(b)} Force-stretch relationship of a non-linear chain. 
\textbf{(c)} Periodic entangled network with a slidable node fraction $\varphi_s = 0\%$. 
\textbf{(d)} Periodic entangled network with a slidable node fraction $\varphi_s = 25\%$. 
\textbf{(e)} Periodic entangled network with a slidable node fraction $\varphi_s = 50\%$. 
\textbf{(f)} Periodic entangled network with a slidable node fraction $\varphi_s = 75\%$. The networks consist of 80 horizontal layers and 20 vertical layers.
}
  \label{fig-supp:SI_Fig28_Undeformed_20_80}
\end{figure}

\vspace*{0pt}
\begin{figure}[H]
  \centering
  \includegraphics[trim={0cm 0cm 0cm 0cm},clip, width=1.0\textwidth]{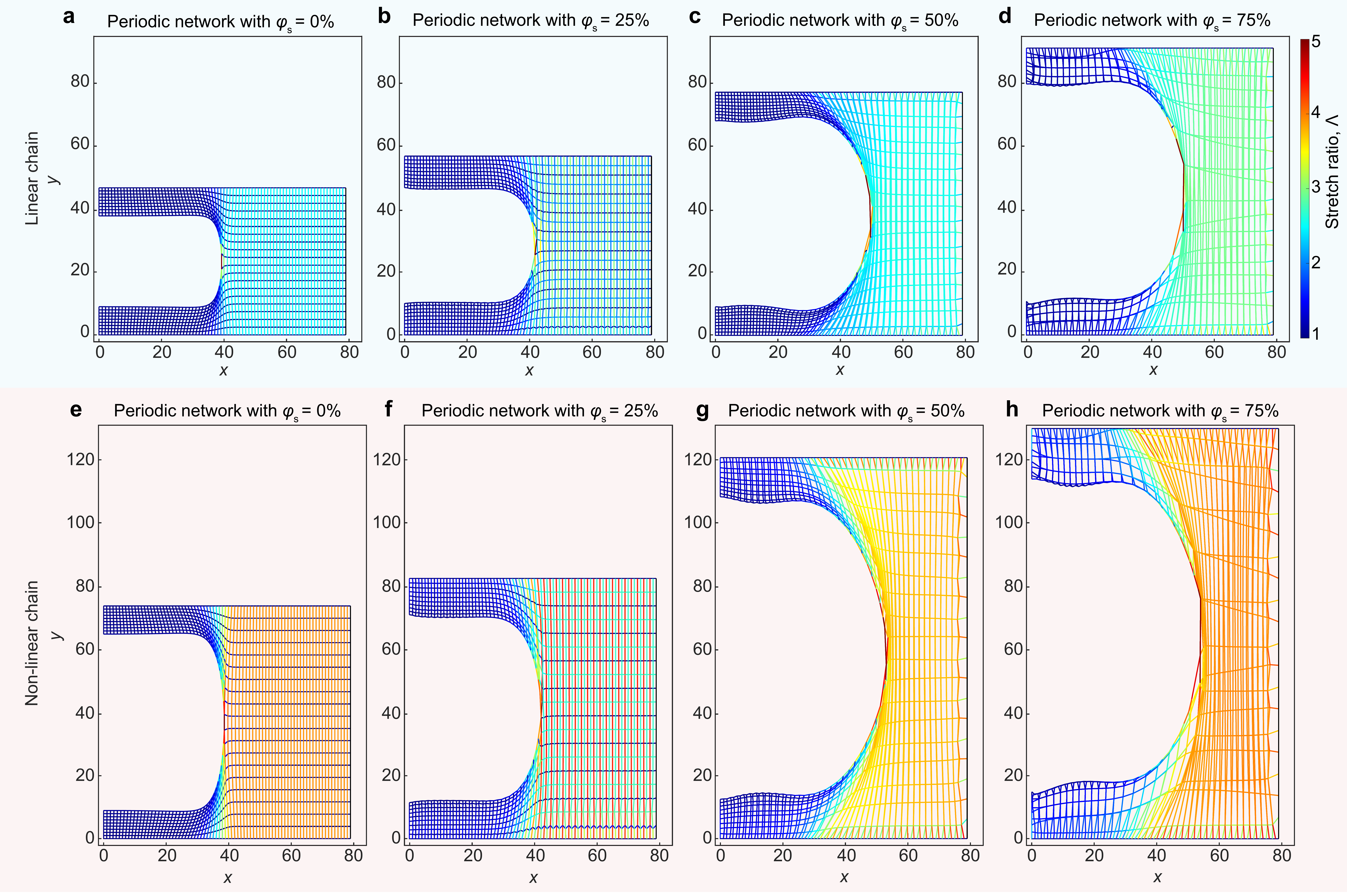}
  \caption{\textbf{Chain stretch ratio distribution of notched periodic entangled networks at crack initialization, shown in the deformed state.}
\textbf{(a-d)} The periodic entangled networks are composed of linear chains with slidable node fractions $\phi_s = 0\%, 25\%, 50\%, 75\%$, respectively. 
\textbf{(e-h)} The periodic entangled networks are composed of non-linear chains with slidable node fractions $\phi_s = 0\%, 25\%, 50\%, 75\%$, respectively. The networks consist of 80 horizontal layers and 20 vertical layers.
}
  \label{fig-supp:SI_Fig29_Deformed_20_80}
\end{figure}

\vspace*{0pt}
\begin{figure}[H]
  \centering
  \includegraphics[trim={0cm 0cm 0cm 0cm},clip, width=1.0\textwidth]{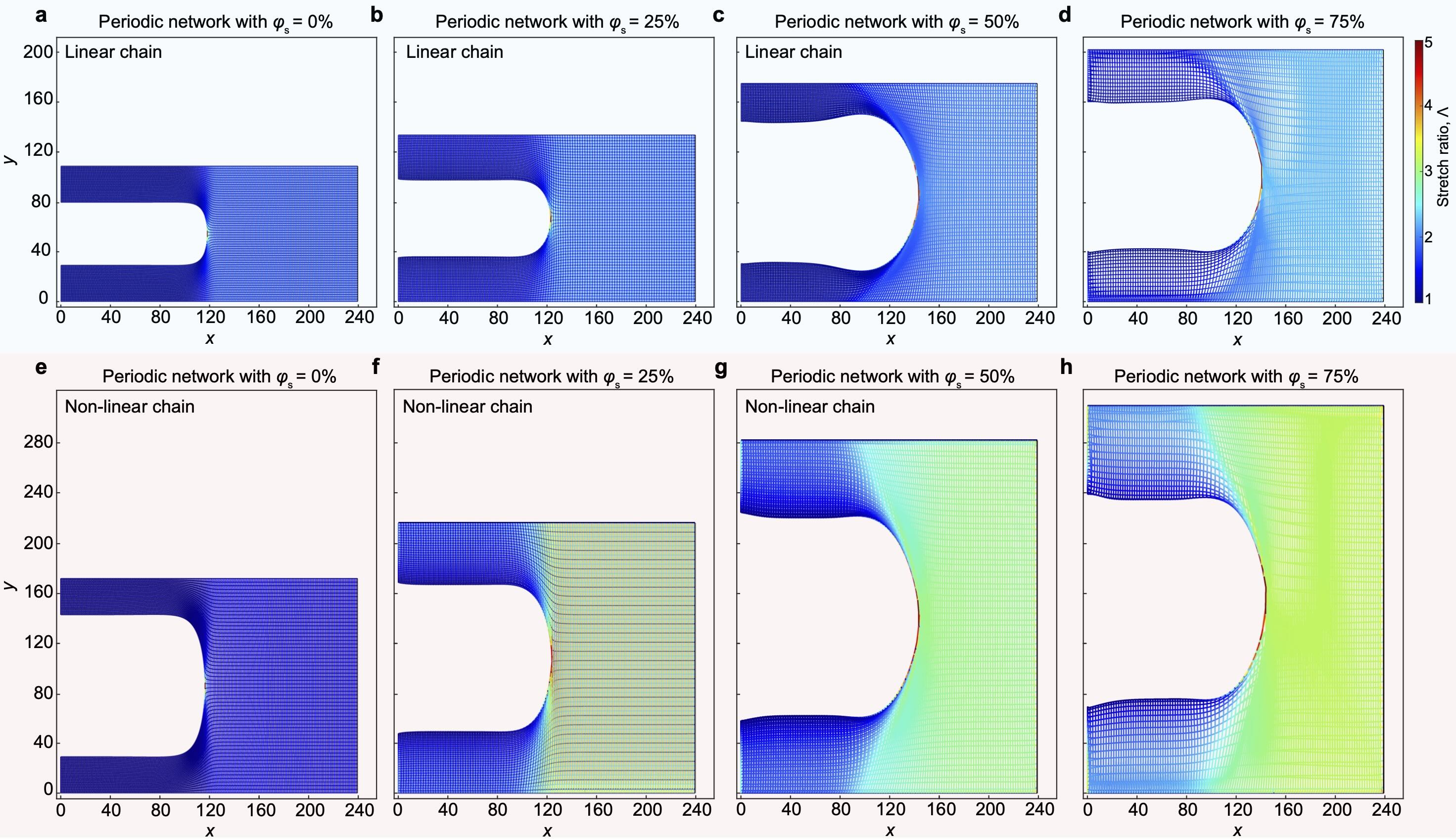}
  \caption{\textbf{Chain stretch ratio distribution of notched periodic entangled networks at crack initialization, shown in the deformed state.}
\textbf{(a–d)} The periodic entangled networks are composed of linear chains with slidable node fractions $\phi_s = 0\%, 25\%, 50\%, 75\%$, respectively. 
\textbf{(e–h)} The periodic entangled networks are composed of non-linear chains with slidable node fractions $\phi_s = 0\%, 25\%, 50\%, 75\%$, respectively.
The networks consist of 240 horizontal layers and 60 vertical layers.
}
  \label{fig-supp:SI_Fig30_Deformed_60_240}
\end{figure}

\vspace*{0pt}
\begin{figure}[H]
  \centering
  \includegraphics[trim={0cm 0cm 0cm 0cm},clip, width=1.0\textwidth]{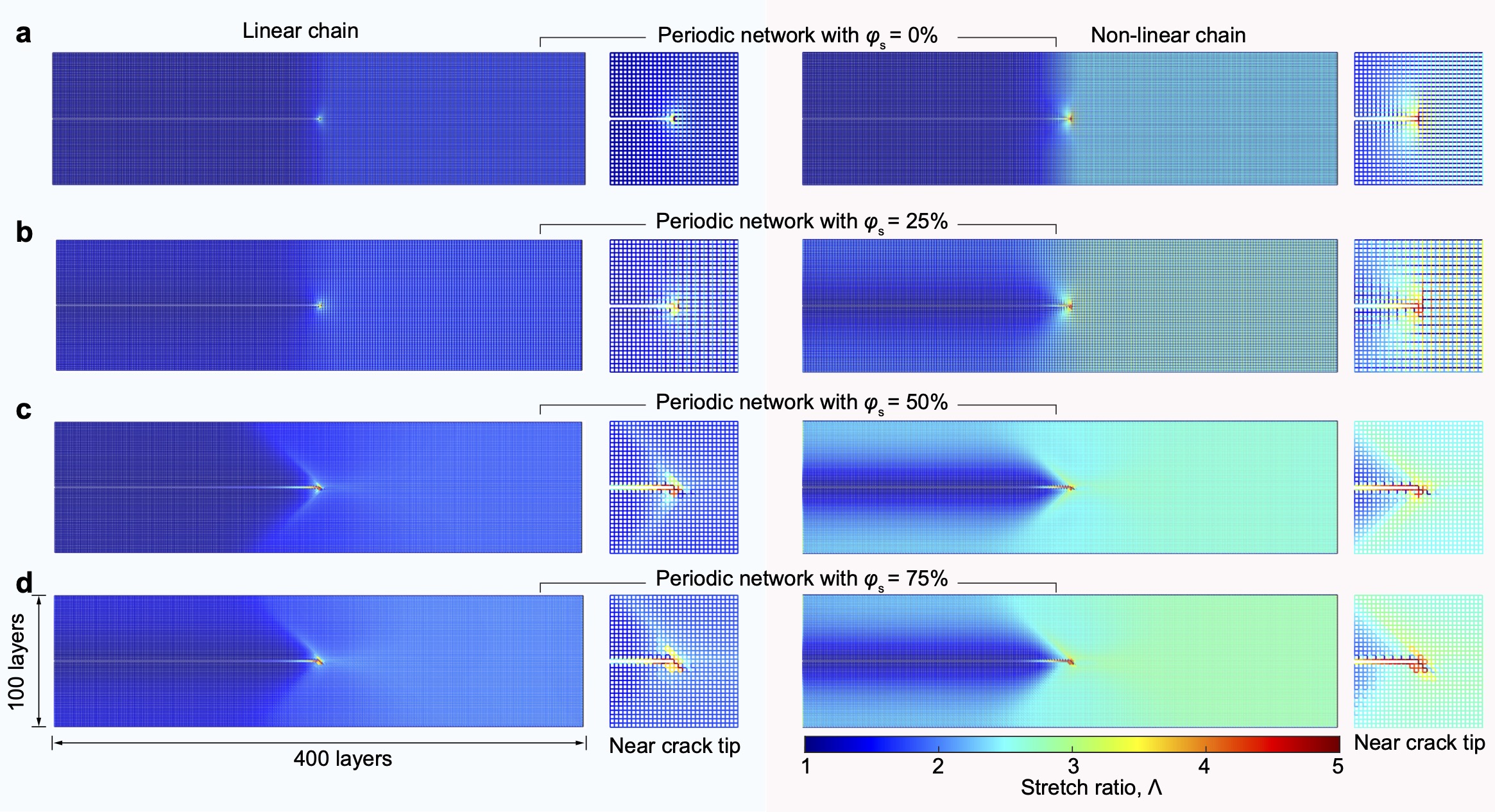}
  \caption{\textbf{Chain stretch ratio distribution of notched periodic entangled networks at crack initialization, shown in the undeformed state.} Left: linear chains; right: nonlinear chains.
\textbf{(a)} Periodic entangled network with a slidable node fraction $\varphi_s = 0\%$. 
\textbf{(b)} Periodic entangled network with a slidable node fraction $\varphi_s = 25\%$. 
\textbf{(c)} Periodic entangled network with a slidable node fraction $\varphi_s = 50\%$. 
\textbf{(d)} Periodic entangled network with a slidable node fraction $\varphi_s = 75\%$.
The networks consist of 400 horizontal layers and 100 vertical layers.
}
  \label{fig-supp:SI_Fig31_Undeformed_100_400}
\end{figure}

\vspace*{0pt}
\begin{figure}[H]
  \centering
  \includegraphics[trim={0cm 0cm 0cm 0cm},clip, width=1.0\textwidth]{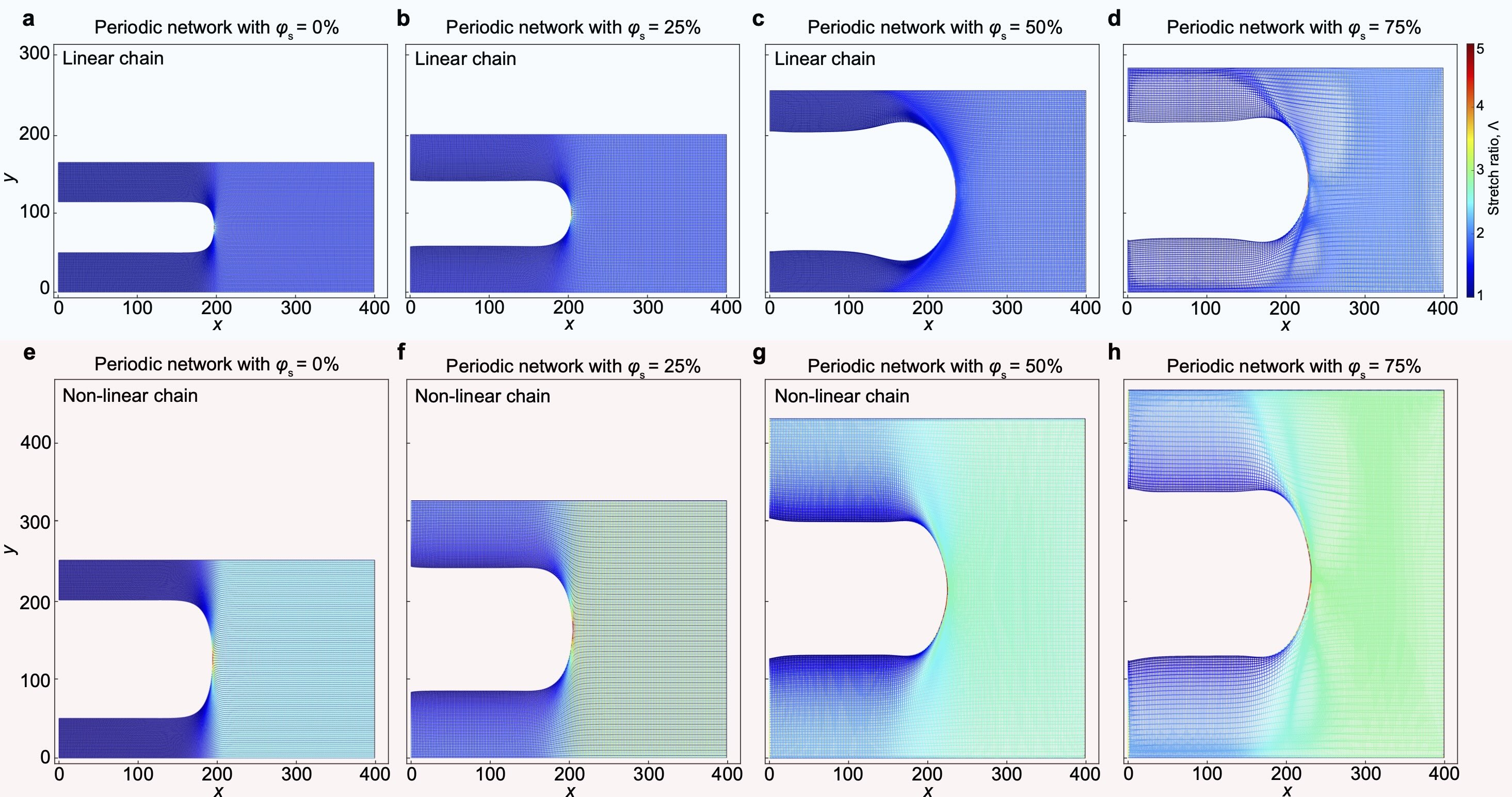}
  \caption{\textbf{Chain stretch ratio distribution of notched periodic entangled networks at crack initialization, shown in the deformed state.}
\textbf{(a-d)} The periodic entangled networks are composed of linear chains with slidable node fractions $\phi_s = 0\%, 25\%, 50\%, 75\%$, respectively. 
\textbf{(e-h)} The periodic entangled networks are composed of non-linear chains with slidable node fractions $\phi_s = 0\%, 25\%, 50\%, 75\%$, respectively.
The networks consist of 400 horizontal layers and 100 vertical layers.
}
  \label{fig-supp:SI_Fig32_Deformed_100_400}
\end{figure}

\vspace*{0pt}
\begin{figure}[H]
  \centering
  \includegraphics[trim={0cm 0cm 0cm 0cm},clip, width=0.6\textwidth]{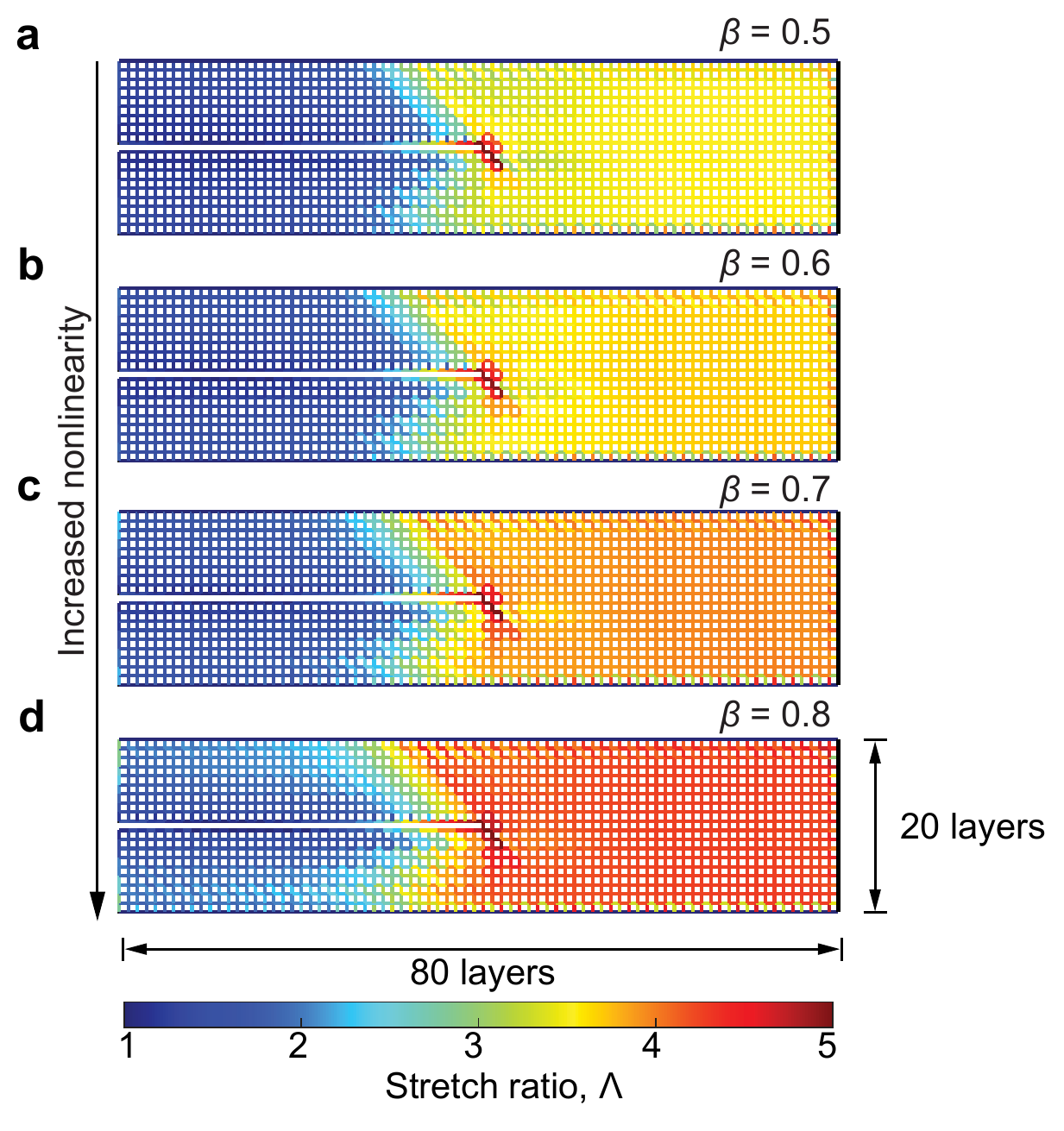}
  \caption{\textbf{Chain stretch ratio distribution of notched periodic entangled networks at crack initialization, with different nonlinearity for constitutive laws, shown in the undeformed state.} The networks consist of 80 horizontal layers and 20 vertical layers.
\textbf{(a-d)} The periodic entangled networks with slidable node fractions $\phi_s = 75\%$ are composed of nonlinear chains with $\beta = 0.5$, $0.6$, $0.7$, and $0.8$, respectively. Here $\beta$ denotes the extension ratio in the chain nonlinear constitutive law given in Eq.~\eqref{eq-supp:nonlinear-constitutive-law}. Larger values of $\beta$ correspond to stronger nonlinearity in the chain constitutive law (see Fig.~\ref{fig-supp:extrapolation}). As nonlinearity increases, more chains near the crack tip experience large stretch ratios, and the far-field stretch ratio also increases. This indicates the nonlinearity can mitigate stress concentration and enable extreme toughening.
}
  \label{fig-supp:SI_Fig66_Nonlinearity_20layers}
\end{figure}

\vspace*{0pt}
\begin{figure}[H]
  \centering
  \includegraphics[trim={0cm 0cm 0cm 0cm},clip, width=0.6\textwidth]{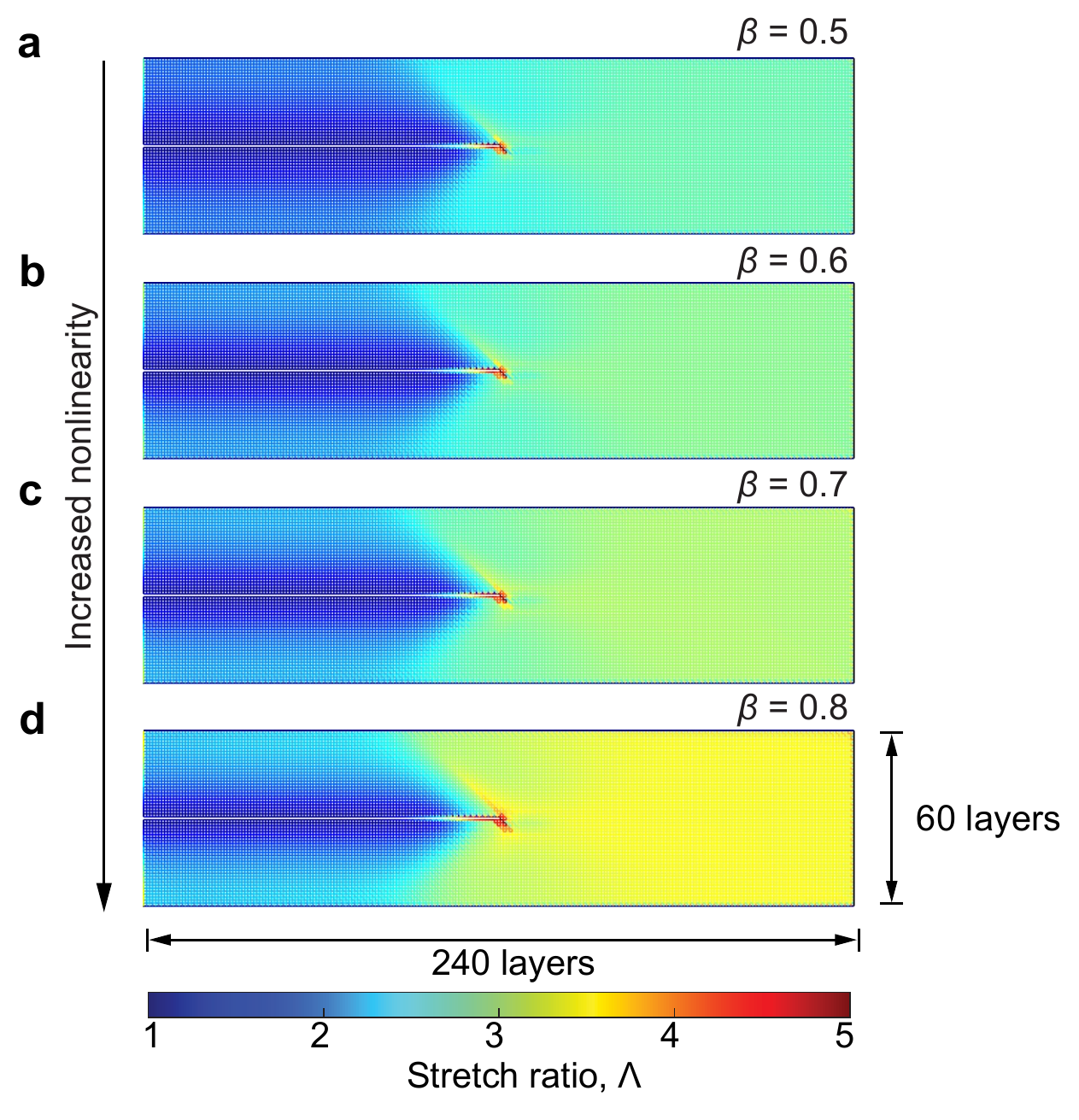}
  \caption{\textbf{Chain stretch ratio distribution of notched periodic networks at crack initialization, with different nonlinearity for constitutive laws, shown in the undeformed state.} The networks consist of 240 horizontal layers and 60 vertical layers.
\textbf{(a-d)} The periodic entangled networks with slidable node fractions $\phi_s = 75\%$ are composed of nonlinear chains with $\beta = 0.5$, $0.6$, $0.7$, and $0.8$, respectively. Here $\beta$ denotes the extension ratio in the chain nonlinear constitutive law given in Eq.~\eqref{eq-supp:nonlinear-constitutive-law}. Larger values of $\beta$ correspond to stronger nonlinearity in the chain constitutive law (see Fig.~\ref{fig-supp:extrapolation}). As nonlinearity increases, more chains near the crack tip experience large stretch ratios, and the far-field stretch ratio also increases. This indicates the nonlinearity can mitigate stress concentration and enable extreme toughening.
}
  \label{fig-supp:SI_Fig67_Nonlinearity_60layers}
\end{figure}

\vspace*{0pt}
\begin{figure}[H]
  \centering
  \includegraphics[trim={0cm 0cm 0cm 0cm},clip, width=1.0\textwidth]{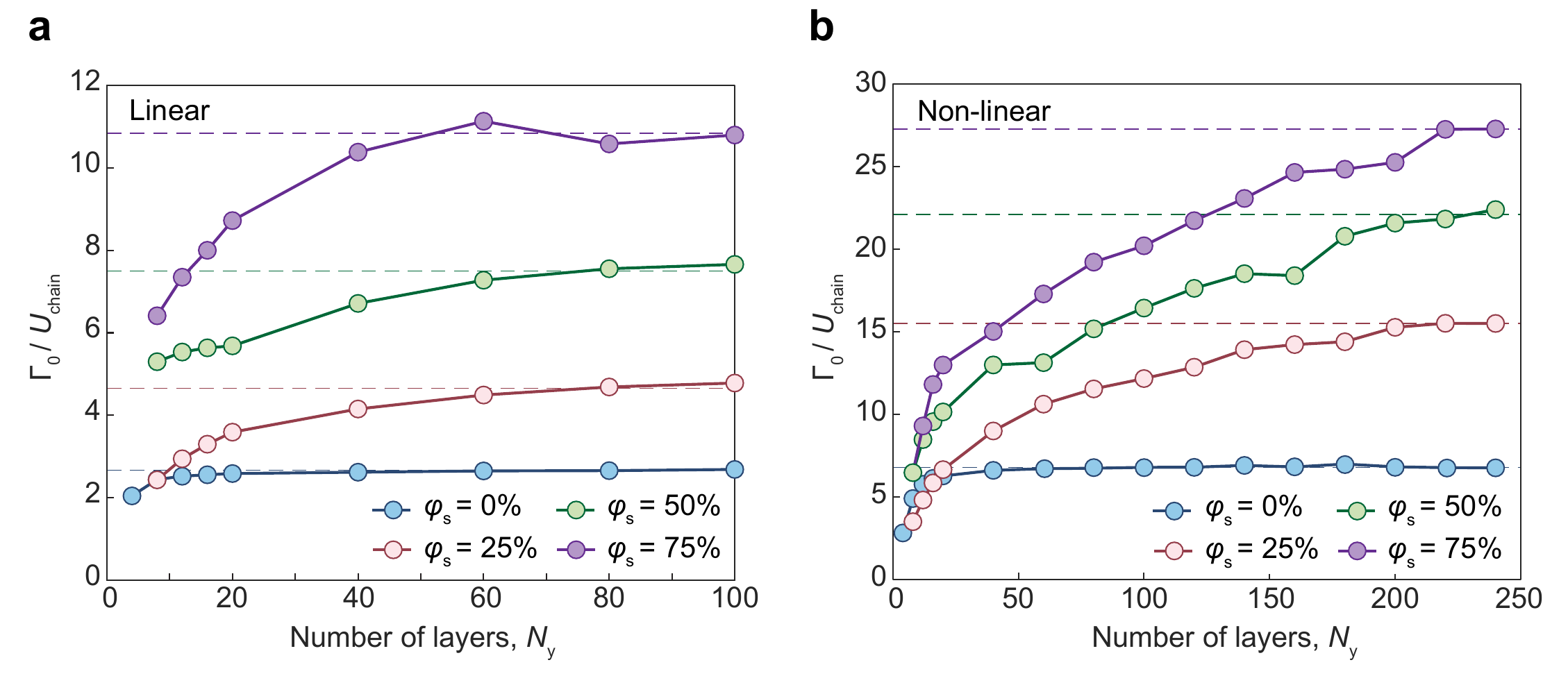}
  \caption{\textbf{Convergence behavior of the intrinsic fracture energy $\Gamma_0$ of notched periodic entangled network with different number of layers in height.} 
    \textbf{(a)} Periodic entangled networks with linear chains. For spring networks ($\varphi_s = 0\%$), $\Gamma_0$ converges in approximately 10 layers. For slidable networks ($\varphi_s=25\%$, $50\%$, and $75\%$), $\Gamma_0$ converges in approximately 100 layers.
    \textbf{(b)} Periodic entangled networks with nonlinear chains. For spring networks ($\varphi_s = 0\%$), $\Gamma_0$ converges in approximately 50 layers. For slidable networks ($\varphi_s=25\%$, $50\%$, and $75\%$), $\Gamma_0$ converges in approximately 240 layers. $U_{\textrm{chain}}$ is the rupture energy of an individual chain. $N_y$ is the number of vertical layers. The periodic entangled network has an aspect ratio (i.e., $N_x/N_y$) of 4:1.}
  \label{fig-supp:SI_Fig33_Converge}
\end{figure}

\vspace*{0pt}
\begin{figure}[H]
  \centering
  \includegraphics[trim={0cm 0cm 0cm 0cm},clip, width=1.0\textwidth]{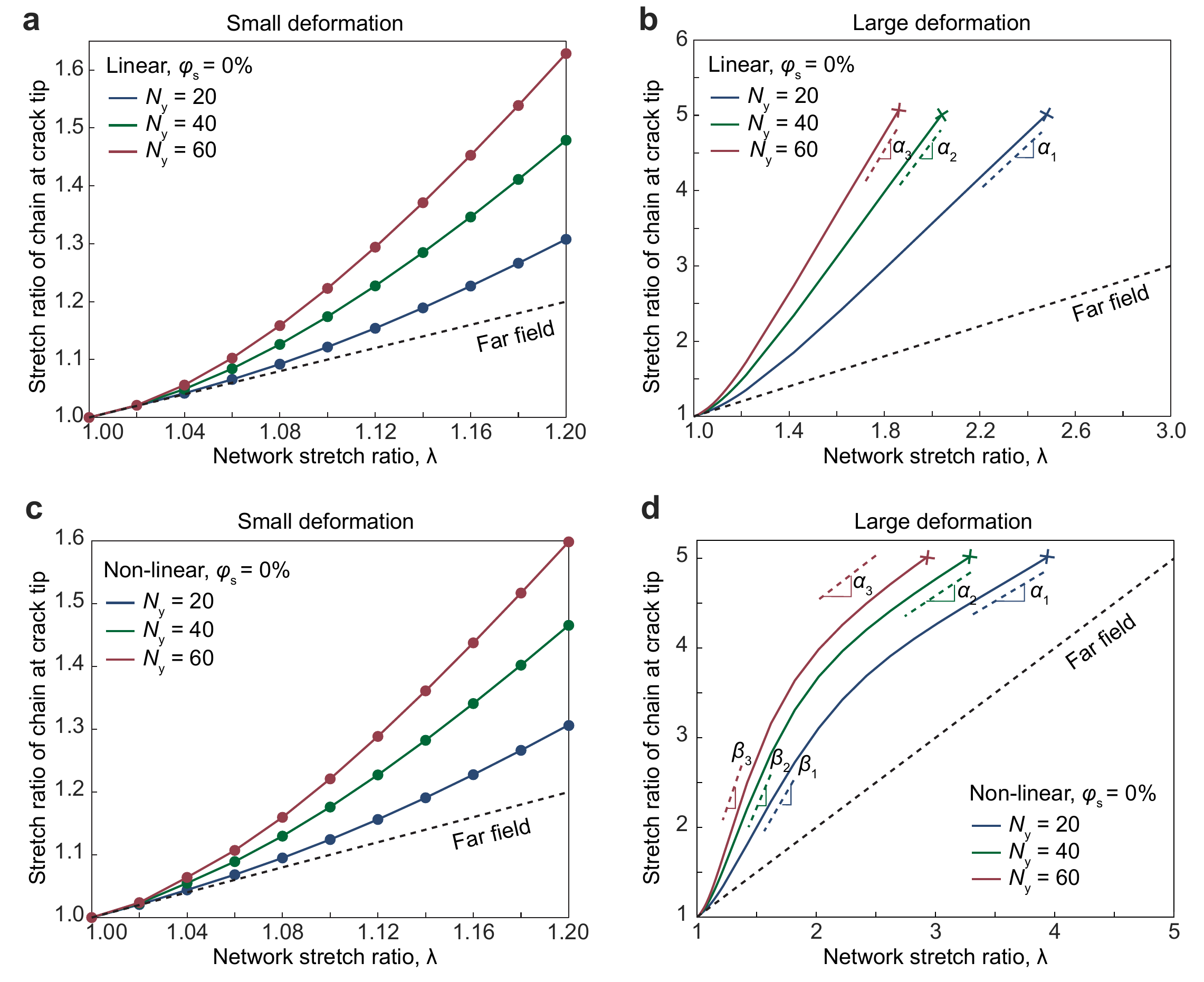}
  \caption{\textbf{Stretch ratio of crack-tip chain versus network stretch ratio in spring networks ($\varphi_s = 0\%$).} 
\textbf{(a)} Networks with linear chains and varying numbers of vertical layers under small deformation. The crack-tip chain stretch ratio is the same as that in the far field under very small deformation.
\textbf{(b)} Networks with linear chains and varying numbers of vertical layers under large deformation. The slope $\alpha$ increases with the layer number.
\textbf{(c)} Networks with non-linear chains and varying numbers of vertical layers under small deformation. The crack-tip chain stretch ratio is the same as that in the far field under very small deformation.
\textbf{(d)} Networks with non-linear chains and varying numbers of vertical layers under large deformation. The slopes $\alpha$ and $\beta$ both increase with the number of layers under large deformation. $N_y$ is the number of vertical layers. The spring network has an aspect ratio (i.e., $N_x/N_y$) of 4:1.
}
  \label{fig-supp:SI_Fig35_Nonslidable_linearlaw}
\end{figure}

\vspace*{0pt}
\begin{figure}[H]
  \centering
  \includegraphics[trim={0cm 0cm 0cm 0cm},clip, width=1.0\textwidth]{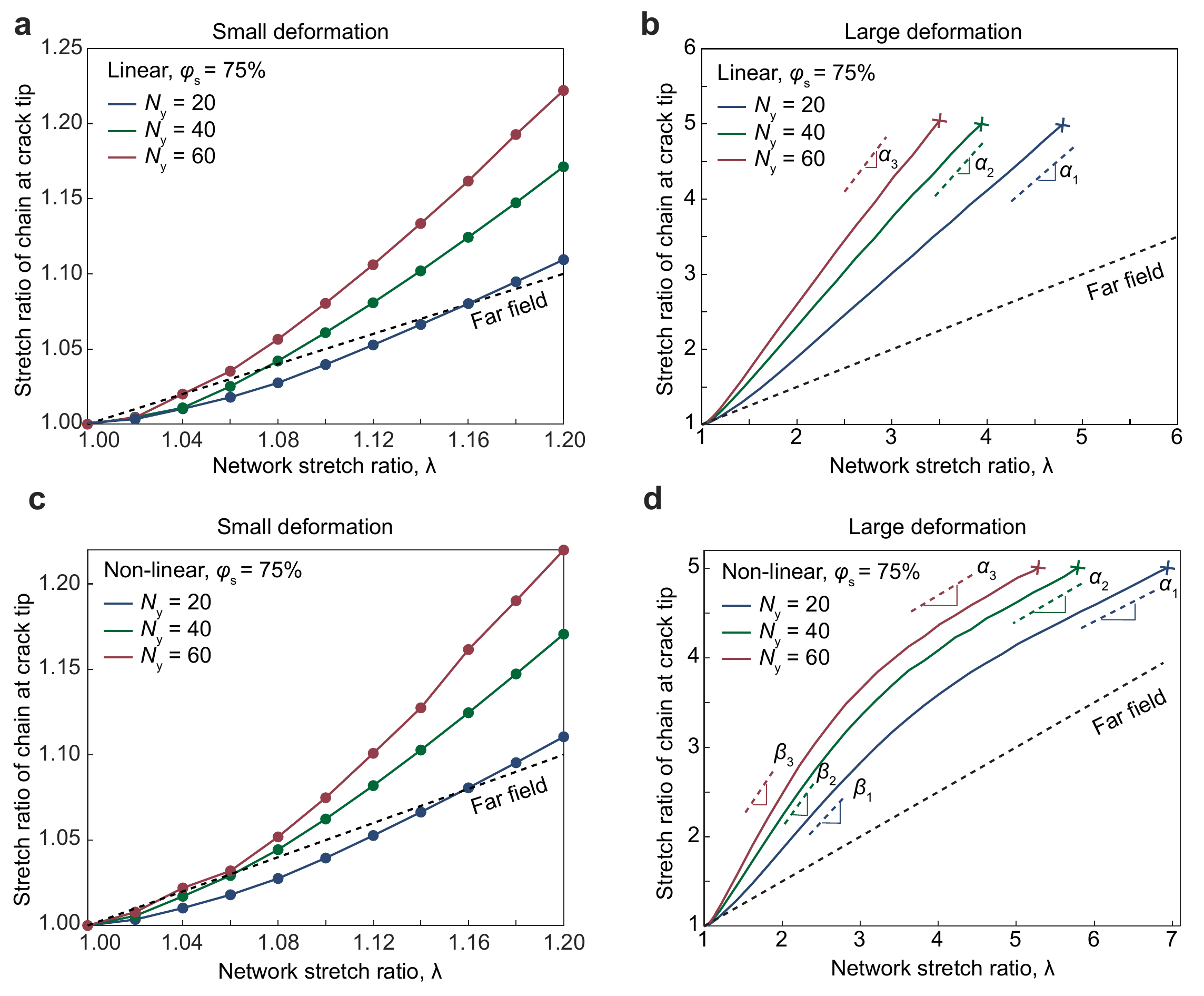}
  \caption{\textbf{Stretch ratio of crack-tip chain versus network stretch ratio in entangled periodic networks ($\varphi_s = 75\%$).} 
\textbf{(a)} Networks with linear chains and varying numbers of layers under small deformation. At very small deformation, the crack-tip chain stretch ratio is lower than that in the far field.
\textbf{(b)} Networks with linear chains and varying numbers of layers under large deformation. The slope $\alpha$ increases with the number of layers under large deformation.
\textbf{(c)} Networks with non-linear chains and varying numbers of layers under small deformation. At very small deformation, the crack-tip chain stretch ratio is lower than that in the far field. 
\textbf{(d)} Networks with non-linear chains and varying numbers of layers under large deformation. The slopes $\alpha$ and $\beta$ both increase with the number of layers under large deformation. $N_y$ is the number of vertical layers. The entangled network has an aspect ratio (i.e., $N_x/N_y$) of 4:1.
}
  \label{fig-supp:SI_Fig36_Slidable_linearlaw}
\end{figure}

\vspace*{0pt}
\begin{figure}[H]
  \centering
  \includegraphics[trim={0cm 0cm 0cm 0cm},clip, width=1.0\textwidth]{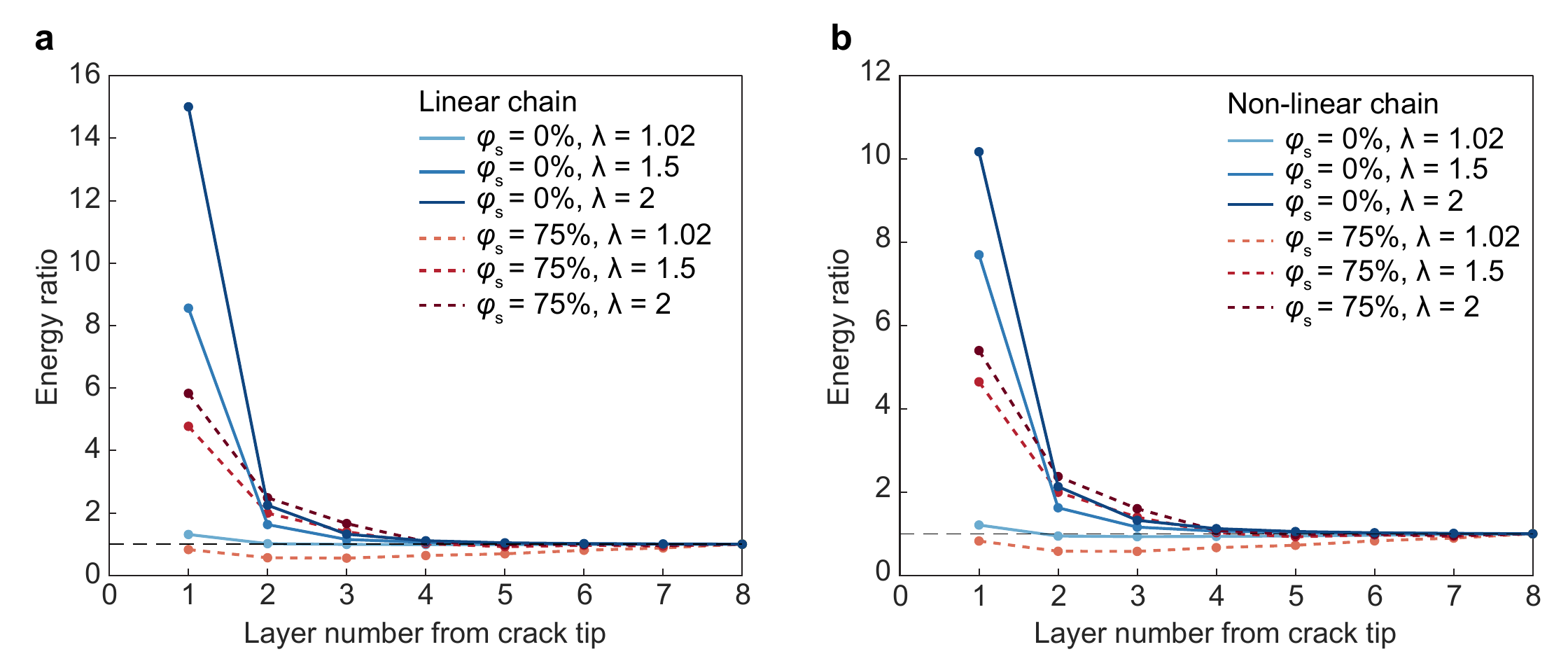}
  \caption{\textbf{Energy ratio between chains near the crack tip and chains in the far field of periodic entangled networks at various network stretch ratio $\lambda$.}
\textbf{(a)} Periodic entangled network with linear chains with slidable node fractions $\varphi_s = 0\%$ and $\varphi_s = 75\%$. 
\textbf{(b)} Periodic entangled network with non-linear chains with slidable node fractions $\varphi_s = 0\%$ and $\varphi_s = 75\%$. 
}
  \label{fig-supp:SI_Fig37_EnergySmallDeform}
\end{figure}

\vspace*{0pt}
\begin{figure}[H]
  \centering
  \includegraphics[trim={0cm 0cm 0cm 0cm},clip, width=1.0\textwidth]{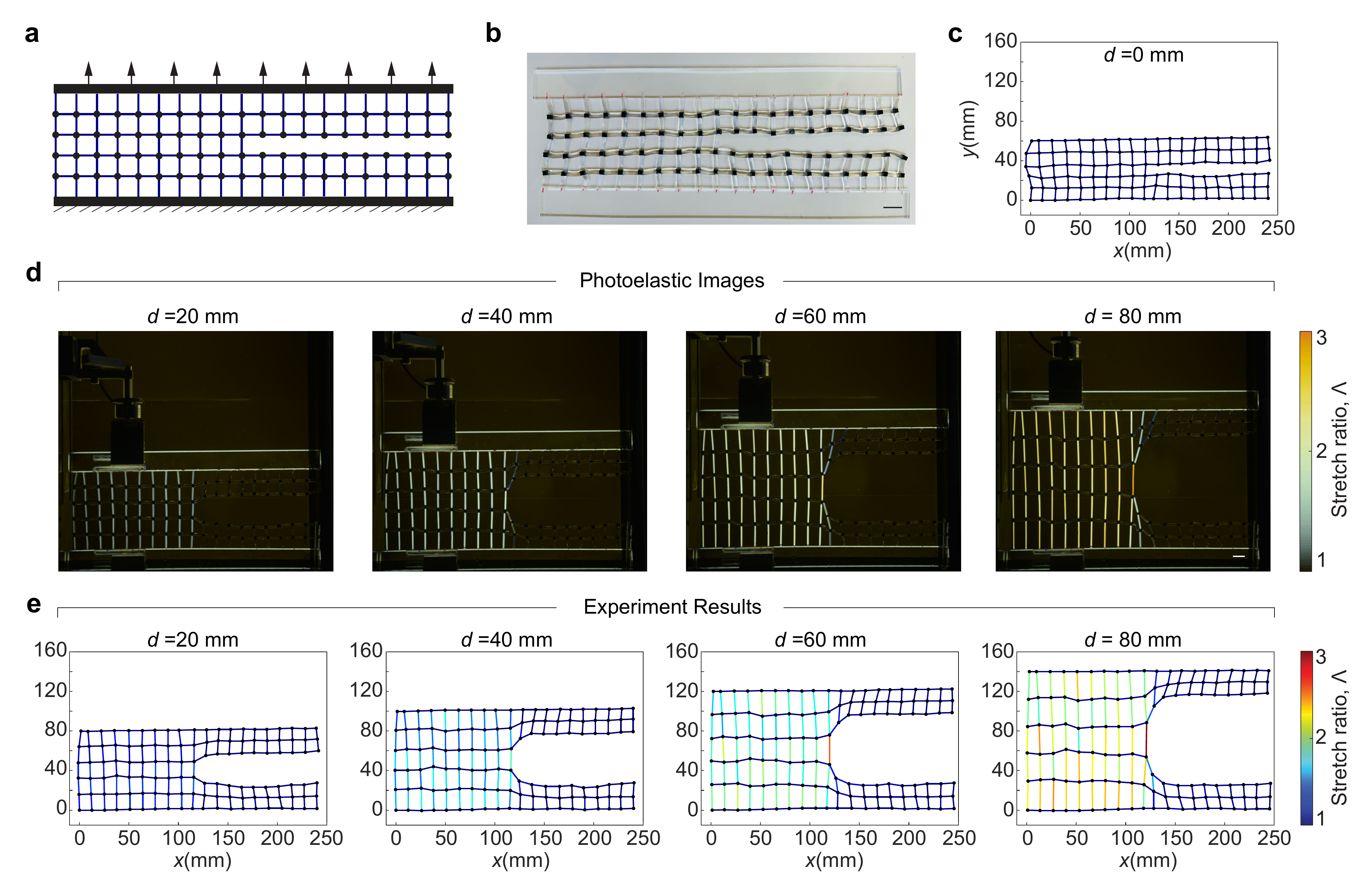}
  \caption{\textbf{Crack opening characterization of a notched hydrogel fabric network without slidable nodes ($\varphi_s = 0\%$).} The hydrogel fiber exhibits a linear constitutive law (i.e. force-stretch relationship).
\textbf{(a)} Topology of the notched network composed without slidable nodes. 
\textbf{(b)} Hydrogel fabric network fabricated based on the topology in (a). Scale bar: 10 mm. 
\textbf{(c)} Initial state of the hydrogel fabric network. 
\textbf{(d)} Photoelastic images of the hydrogel fabric network subjected to increasing displacement. Scale bar: 10 mm. 
\textbf{(e)} Stretch ratio distribution of hydrogel fibers. The result is extracted from the photoelastic images in (d). Black nodes denote non-slidable nodes.
}
  \label{fig-supp:SI_Fig38_Nonslidable_linear_Fiber}
\end{figure}

\vspace*{0pt}
\begin{figure}[H]
  \centering
  \includegraphics[trim={0cm 0cm 0cm 0cm},clip, width=1.0\textwidth]{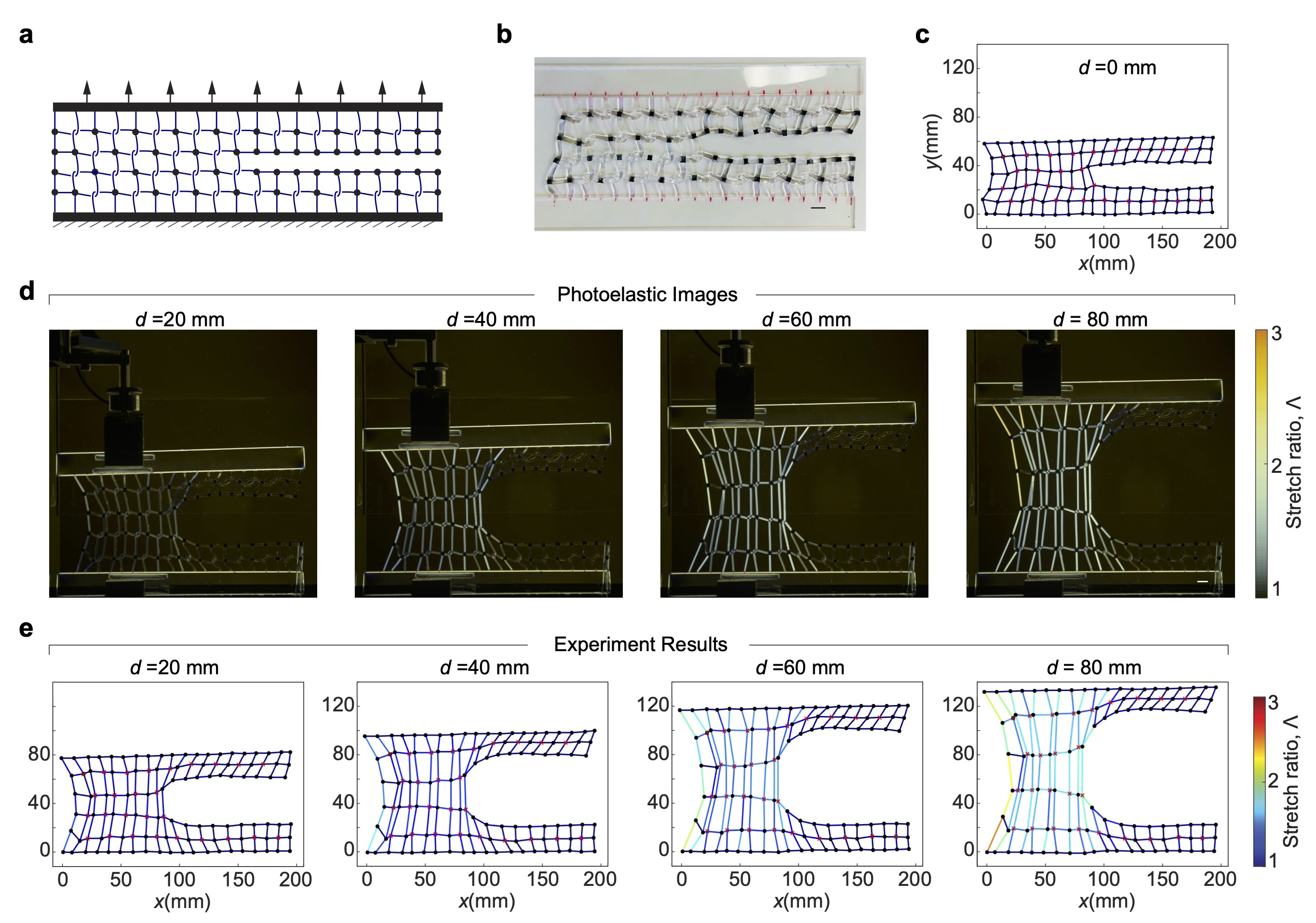}
\caption{\textbf{Crack opening characterization of a notched hydrogel fabric network with slidable nodes fraction $\varphi_s = 50\%$.} The hydrogel fiber exhibits a linear constitutive law (i.e. force-stretch relationship).
\textbf{(a)} Topology of the notched network with the slidable node fraction ($\varphi_s = 50\%$). 
\textbf{(b)} Hydrogel fabric network fabricated based on the topology in (a). Scale bar: 10 mm. 
\textbf{(c)} Initial state of the hydrogel fabric network. 
\textbf{(d)} Photoelastic images of the hydrogel fabric network subjected to increasing displacement. Scale bar: 10 mm. 
\textbf{(e)} Stretch ratio distribution of hydrogel fibers. The result is extracted from the photoelastic images in (d). Black nodes denote non-slidable nodes and red crosses denote slidable nodes.}
  \label{fig-supp:SI_Fig39_Slidable_linear_Fiber}
\end{figure}

\vspace*{0pt}
\begin{figure}[H]
  \centering
  \includegraphics[trim={0cm 0cm 0cm 0cm},clip, width=1.0\textwidth]{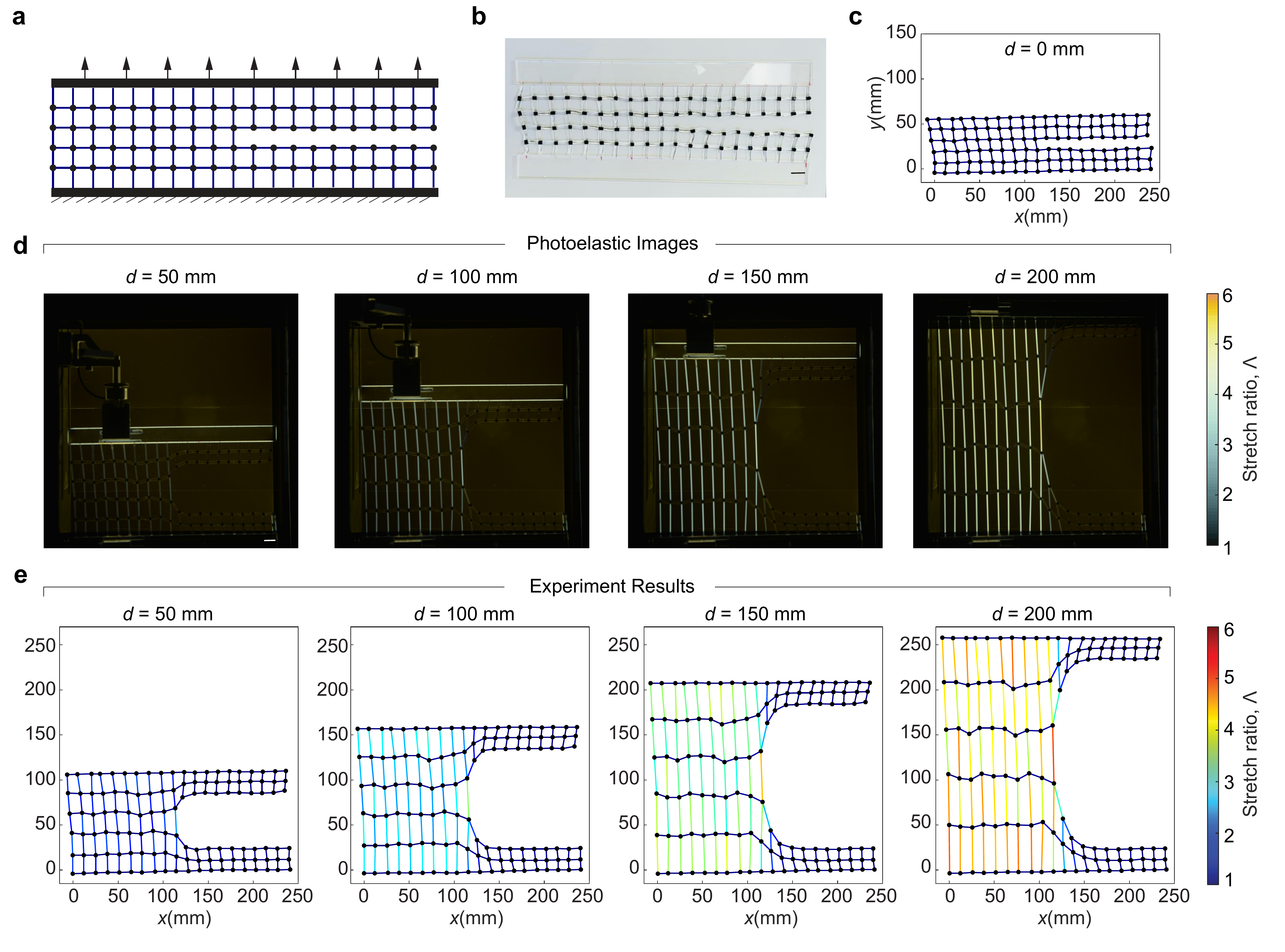}
  \caption{\textbf{Crack opening characterization of a notched hydrogel fabric network without slidable nodes ($\varphi_s = 0\%$).} The hydrogel fiber exhibits a non-linear linear constitutive law (i.e. force-stretch relationship).
\textbf{(a)} Topology of the notched network composed without slidable nodes. 
\textbf{(b)} Hydrogel fabric network fabricated based on the topology in (a). Scale bar: 10 mm. 
\textbf{(c)} Initial state of the hydrogel fabric network. 
\textbf{(d)} Photoelastic images of the hydrogel fabric network subjected to increasing displacement. Scale bar: 10 mm. 
\textbf{(e)} Stretch ratio distribution of hydrogel fibers. The result is extracted from the photoelastic images in (d). Black nodes denote non-slidable nodes.}
  \label{fig-supp:SI_Fig40_Nonslidable_nonlinear_Fiber}
\end{figure}

\vspace*{0pt}
\begin{figure}[H]
  \centering
  \includegraphics[trim={0cm 0cm 0cm 0cm},clip, width=1.0\textwidth]{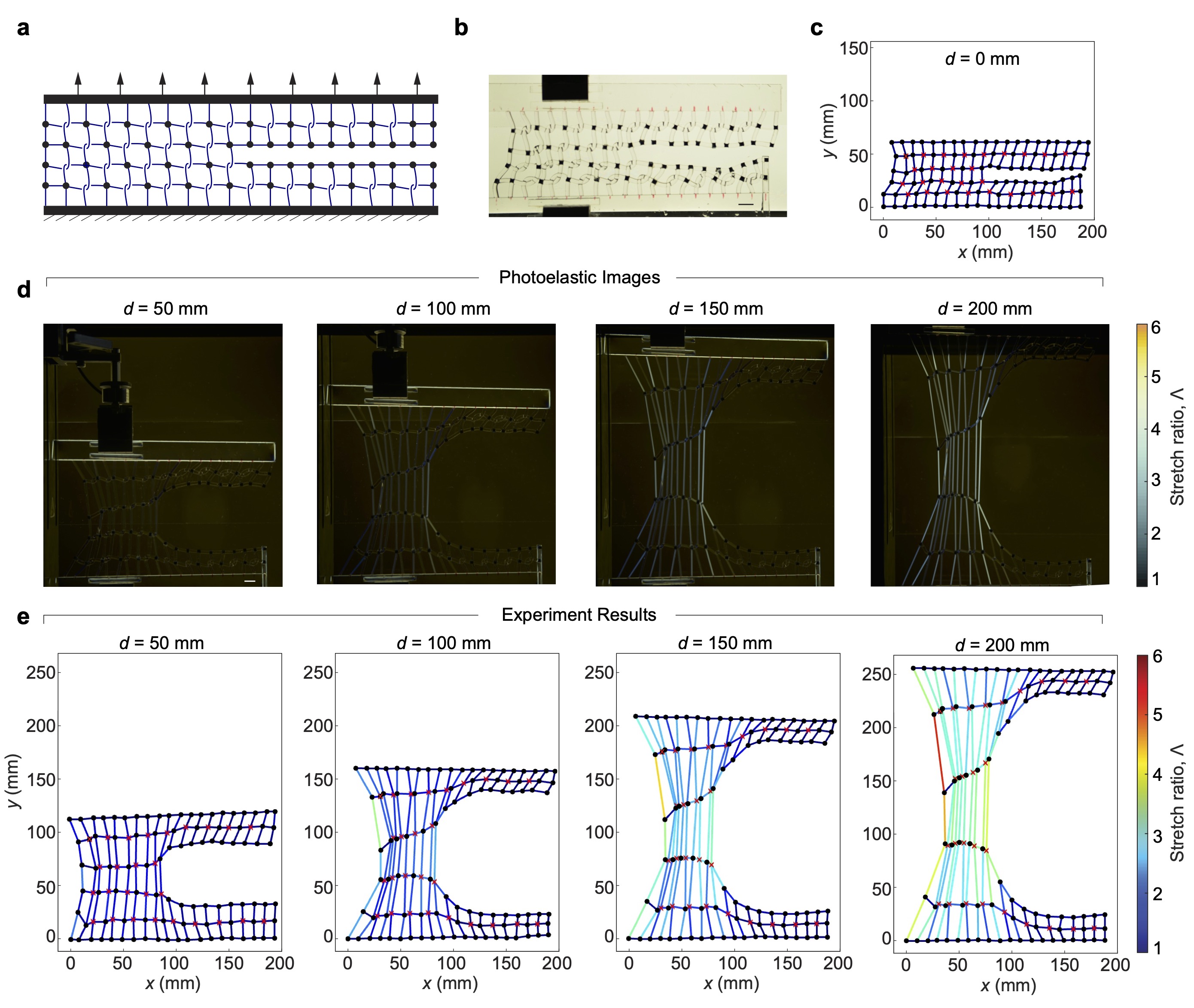}
  \caption{\textbf{Crack opening characterization of a notched hydrogel fabric network with slidable nodes fraction $\varphi_s = 50\%$.} The hydrogel fiber exhibits a non-linear constitutive law (i.e. force-stretch relationship).
\textbf{(a)} Topology of the notched network with the slidable node fraction ($\varphi_s = 50\%$). 
\textbf{(b)} Hydrogel fabric network fabricated based on the topology in (a). Scale bar: 10 mm. 
\textbf{(c)} Initial state of the hydrogel fabric network. 
\textbf{(d)} Photoelastic images of the hydrogel fabric network subjected to increasing displacement. Scale bar: 10 mm. 
\textbf{(e)} Stretch ratio distribution of hydrogel fibers. The result is extracted from the photoelastic images in (d). Black nodes denote non-slidable nodes and red crosses denote slidable nodes.}
  \label{fig-supp:SI_Fig41_Slidable_nonlinear_Fiber}
\end{figure}

\vspace*{0pt}
\begin{figure}[H]
  \centering
  \includegraphics[trim={0cm 0cm 0cm 0cm},clip, width=1.0\textwidth]{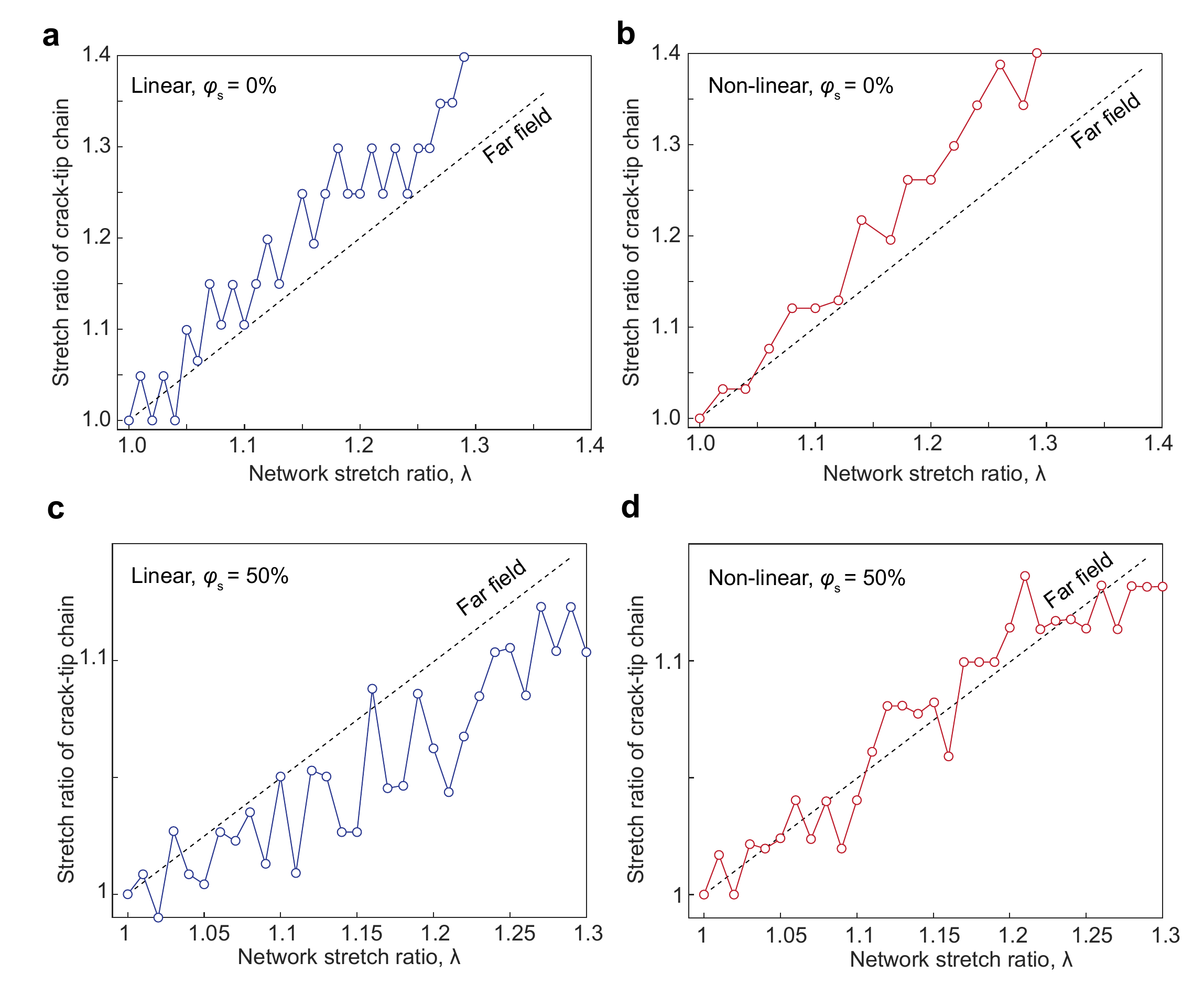}
  \caption{\textbf{Stretch ratio of crack-tip chain versus network stretch ratio under small deformation in hydrogel fabric experiments.} 
\textbf{(a)} Network consist of linear chains with slidable node fraction $\varphi_s = 0\%$. 
\textbf{(b)} Network consist of non-linear chains with slidable node fraction $\varphi_s = 0\%$. 
\textbf{(c)} Network consist of linear chains with slidable node fraction $\varphi_s = 50\%$. 
\textbf{(d)} Network consist of non-linear chains with slidable node fraction $\varphi_s = 50\%$. Dashed lines indicate far-field chain stretch ratio with a slope of 1 in (a–b) and 0.5 in (c–d). The stretch ratios of crack-tip chain are extracted from experimental photoelastic images. An initial stage of stress de-concentration is observed.}
  \label{fig-supp:SI_Fig42_SmallDeform_Experiment}
\end{figure}

\vspace*{0pt}
\begin{figure}[H]
  \centering
  \includegraphics[trim={0cm 0cm 0cm 0cm},clip, width=1.0\textwidth]{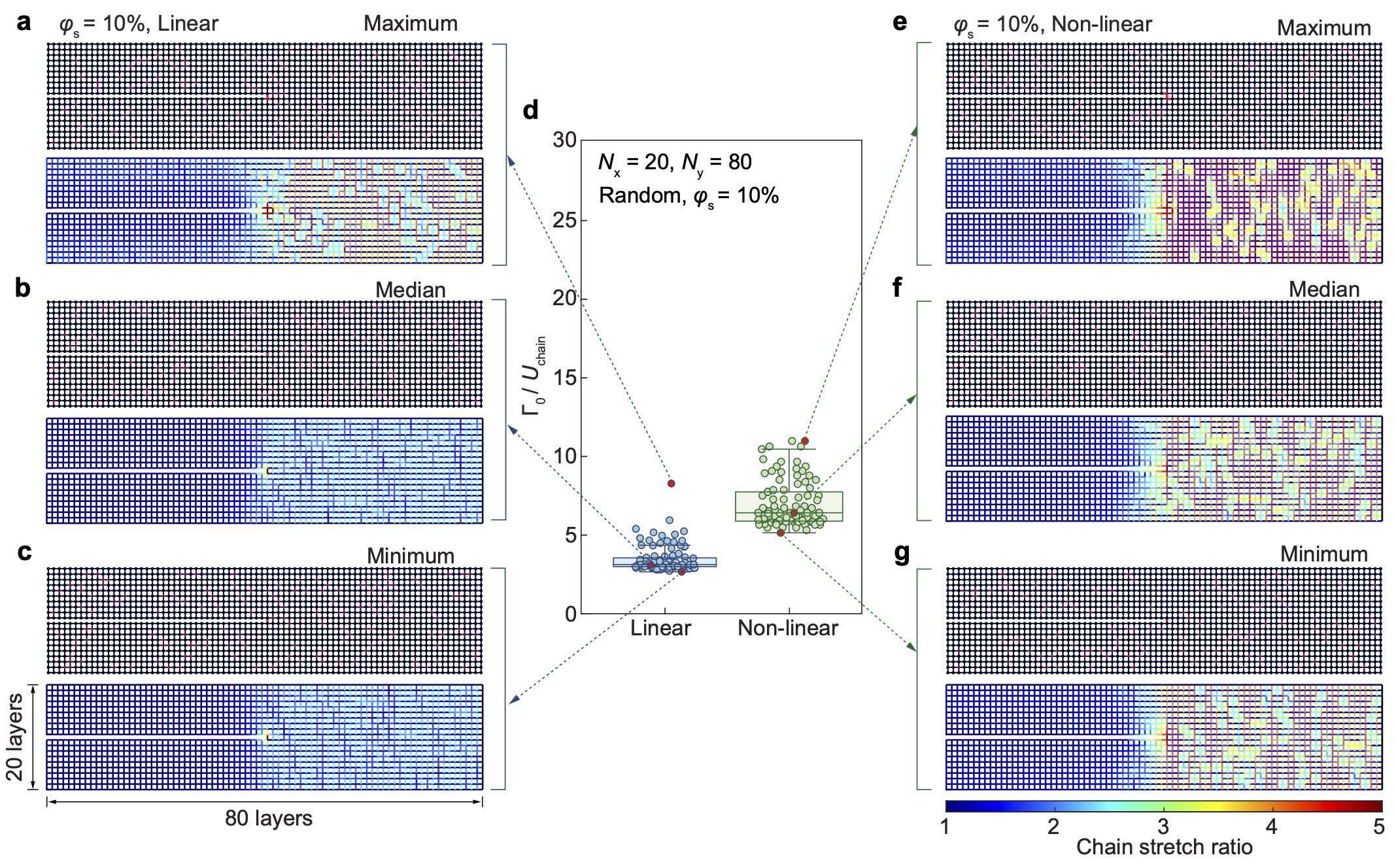}
  \caption{\textbf{Initial configuration and stretch ratio distribution of notched random entangled networks with $\varphi_s = 10\%$ at crack initialization, shown in undeformed state.}
\textbf{(a)} Sample with the largest intrinsic fracture energy, formed by linear chains (crack-tip chain length $l_{\textrm{tip}}$ = 2). 
\textbf{(b)} Sample with median intrinsic fracture energy ($l_{\textrm{tip}}$ = 1). 
\textbf{(c)} Sample with the lowest intrinsic fracture energy ($l_{\textrm{tip}}$ = 1). 
\textbf{(d)} Box plot of $\Gamma_0 / U_{\text{chain}}$ for notched random entangled networks. Networks are formed by linear and non-linear chains, with 100 samples in each case. Each box shows the interquartile range, with the line inside the box indicating the median. Whiskers extend to the most extreme data points within 1.5 times the interquartile range from the lower and upper quartiles.
\textbf{(e)} Sample with the largest intrinsic fracture energy, formed by non-linear chains ($l_{\textrm{tip}}$ = 3). 
\textbf{(f)} Sample with median intrinsic fracture energy ($l_{\textrm{tip}}$ = 1). 
\textbf{(g)} Sample with the lowest intrinsic fracture energy ($l_{\textrm{tip}}$ = 1). The networks consist of 80 horizontal layers and 20 vertical layers.}
  \label{fig-supp:SI_Fig43_Random_10_20_undeformed}
\end{figure}

\vspace*{0pt}
\begin{figure}[H]
  \centering
  \includegraphics[trim={0cm 0cm 0cm 0cm},clip, width=1.0\textwidth]{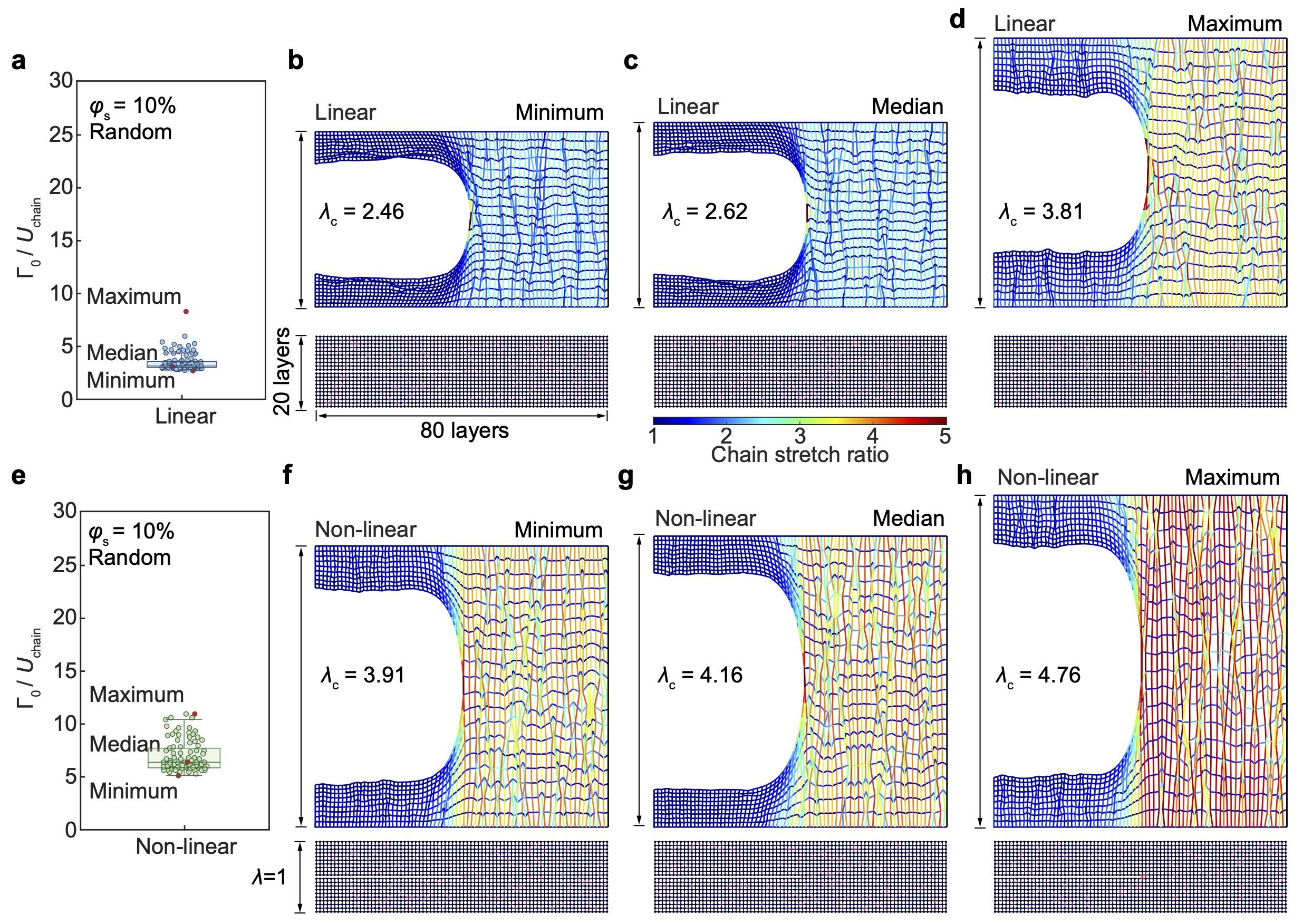}
  \caption{\textbf{Initial configuration and stretch ratio distribution of notched random entangled networks with $\varphi_s = 10\%$ at crack initialization, shown in deformed state.} 
\textbf{(a)} Box plot of $\Gamma_0 / U_{\text{chain}}$ for notched networks formed by linear chains (100 samples). 
\textbf{(b)} Sample with the largest intrinsic fracture energy.
\textbf{(c)} Sample with median intrinsic fracture energy. 
\textbf{(d)} Sample with the smallest intrinsic fracture energy. 
\textbf{(e)} Box plot of $\Gamma_0 / U_{\text{chain}}$ for notched networks formed by non-linear chains (100 samples). 
\textbf{(f)} Sample with the largest intrinsic fracture energy.
\textbf{(g)} Sample with median intrinsic fracture energy. 
\textbf{(h)} Sample with the smallest intrinsic fracture energy. The networks consist of 80 horizontal layers and 20 vertical layers.}
  \label{fig-supp:SI_Fig44_Random_10_20_deformed}
\end{figure}

\vspace*{0pt}
\begin{figure}[H]
  \centering
  \includegraphics[trim={0cm 0cm 0cm 0cm},clip, width=1.0\textwidth]{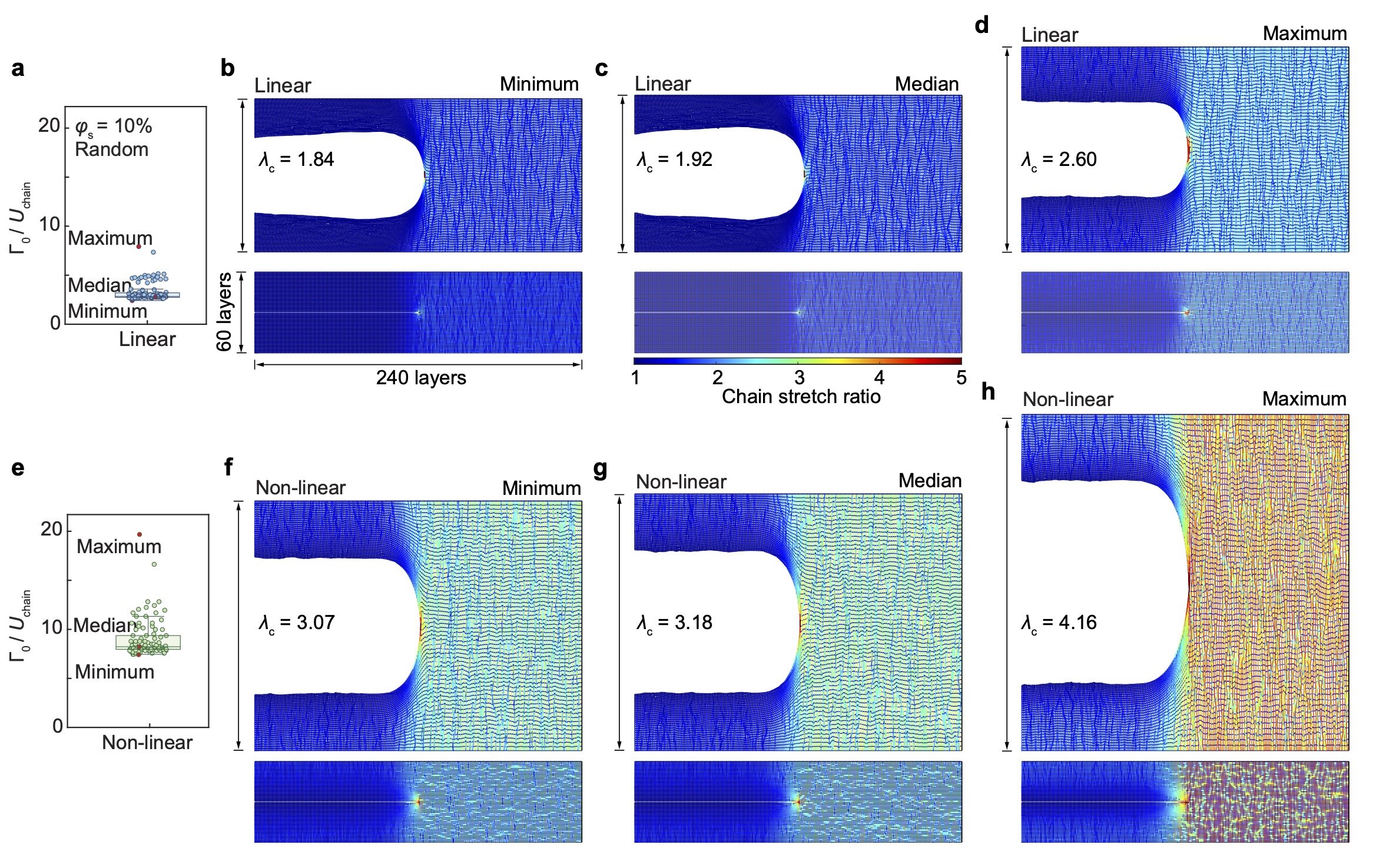}
  \caption{\textbf{Stretch ratio distribution of notched random entangled networks ($\varphi_s = 10\%$) at crack initialization, shown in deformed and undeformed states.}
\textbf{(a)} Box plot of $\Gamma_0 / U_{\text{chain}}$ for notched networks formed by linear chains (100 samples). 
\textbf{(b)} Stretch ratio distribution of the sample with the highest intrinsic fracture energy. 
\textbf{(c)} Sample with median intrinsic fracture energy. 
\textbf{(d)} Sample with the lowest intrinsic fracture energy. 
\textbf{(e)} Box plot of $\Gamma_0 / U_{\text{chain}}$ for notched networks formed by non-linear chains (100 samples). 
\textbf{(f)} Stretch ratio distribution of the sample with the highest intrinsic fracture energy in both undeformed and deformed states. 
\textbf{(g)} Sample with median intrinsic fracture energy. 
\textbf{(h)} Sample with the lowest intrinsic fracture energy. The networks consist of 240 horizontal layers and 60 vertical layers.}
  \label{fig-supp:SI_Fig45_Random_10_60_deformed}
\end{figure}

\vspace*{0pt}
\begin{figure}[H]
  \centering
  \includegraphics[trim={0cm 0cm 0cm 0cm},clip, width=1.0\textwidth]{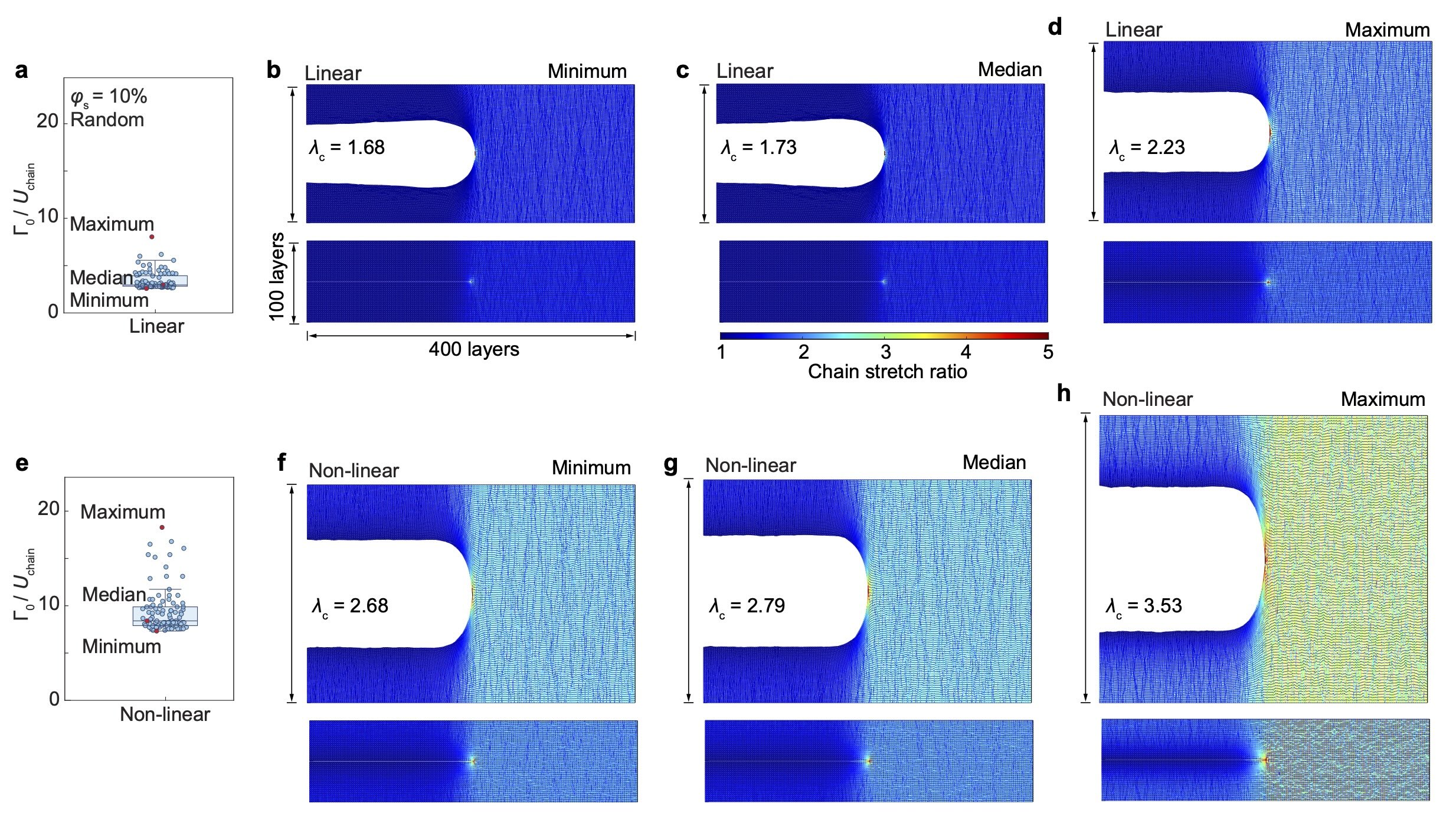}
  \caption{\textbf{Stretch ratio distribution of notched random entangled networks ($\varphi_s = 10\%$) at crack initialization, shown in deformed and undeformed states.}
\textbf{(a)} Box plot of $\Gamma_0 / U_{\text{chain}}$ for notched networks formed by linear chains (100 samples). 
\textbf{(b)} Stretch ratio distribution of the sample with the highest intrinsic fracture energy.
\textbf{(c)} Sample with median intrinsic fracture energy. 
\textbf{(d)} Sample with the lowest intrinsic fracture energy. 
\textbf{(e)} Box plot of $\Gamma_0 / U_{\text{chain}}$ for notched networks formed by non-linear chains (100 samples). 
\textbf{(f)} Stretch ratio distribution of the sample with the highest intrinsic fracture energy. 
\textbf{(g)} Sample with median intrinsic fracture energy. 
\textbf{(h)} Sample with the lowest intrinsic fracture energy. The networks consist of 400 horizontal layers and 100 vertical layers.}
  \label{fig-supp:SI_Fig46_Random_10_100_deformed}
\end{figure}

\vspace*{0pt}
\begin{figure}[H]
  \centering
  \includegraphics[trim={0cm 0cm 0cm 0cm},clip, width=1.0\textwidth]{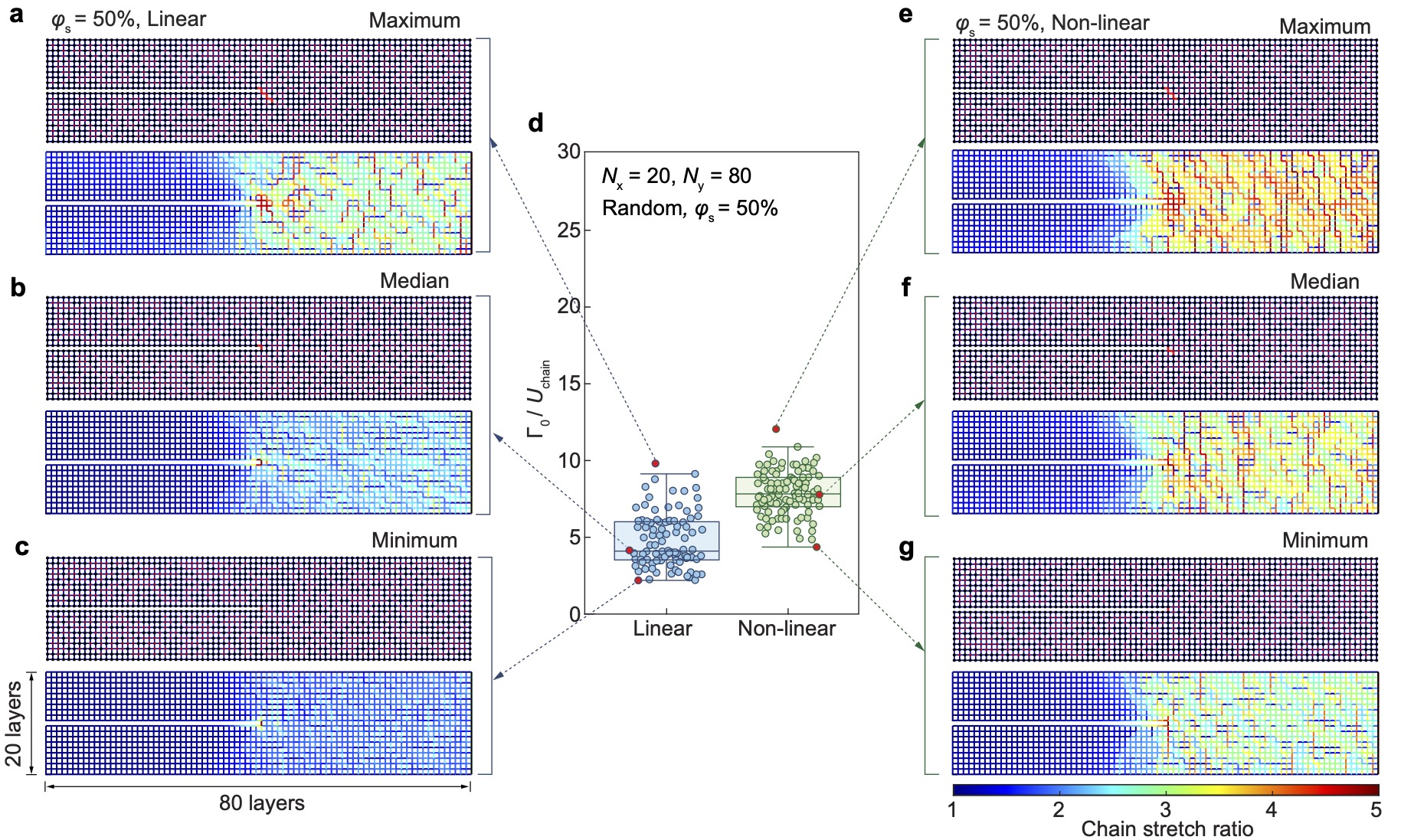}
  \caption{
  \textbf{Initial configuration and stretch ratio distribution of notched random entangled networks with $\varphi_s = 50\%$ at crack initialization, shown in undeformed state.} 
\textbf{(a)} Sample with the highest intrinsic fracture energy in a network formed by linear chains (crack-tip chain length $l_{\textrm{tip}}$ = 6). 
\textbf{(b)} Sample with median intrinsic fracture energy ($l_{\textrm{tip}}$ = 2). 
\textbf{(c)} Sample with the lowest intrinsic fracture energy ($l_{\textrm{tip}}$ = 1). 
\textbf{(d)} Box plot of $\Gamma_0 / U_{\text{chain}}$ for notched networks formed by linear and non-linear chains (100 samples each). Each box shows the interquartile range, with the line inside the box indicating the median. Whiskers extend to the most extreme data points within 1.5 times the interquartile range from the lower and upper quartiles.
\textbf{(e)} Sample with the highest intrinsic fracture energy in a network formed by non-linear chains ($l_{\textrm{tip}}$ = 5). 
\textbf{(f)} Sample with median intrinsic fracture energy ($l_{\textrm{tip}}$ = 3). 
\textbf{(g)} Sample with the lowest intrinsic fracture energy ($l_{\textrm{tip}}$ = 1). The networks consist of 80 horizontal layers and 20 vertical layers.}
  \label{fig-supp:SI_Fig47_Random_50_20_undeformed}
\end{figure}

\vspace*{0pt}
\begin{figure}[H]
  \centering
  \includegraphics[trim={0cm 0cm 0cm 0cm},clip, width=1.0\textwidth]{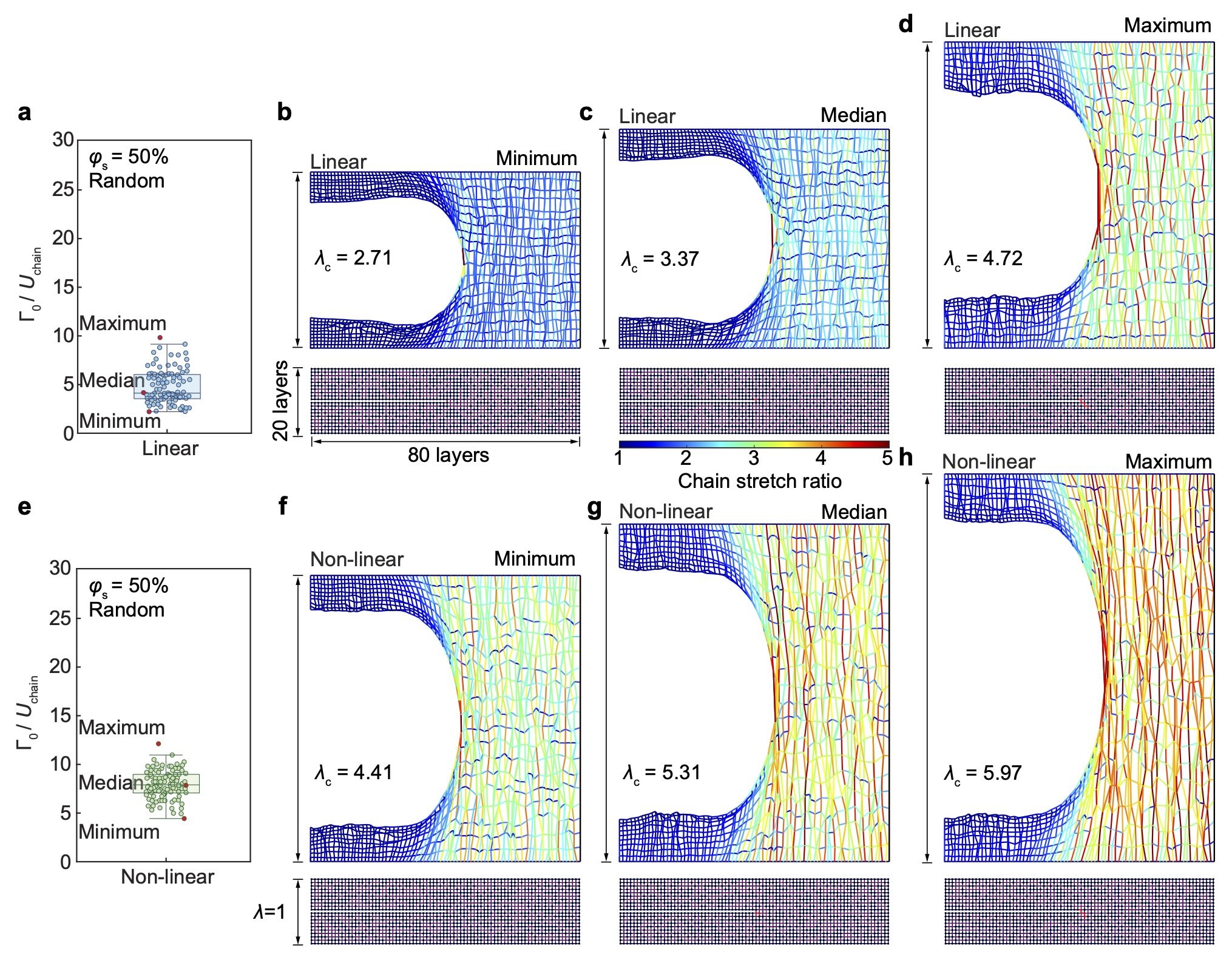}
  \caption{\textbf{Initial configuration and stretch ratio distribution of notched random entangled networks with $\varphi_s = 50\%$ at crack initialization, shown in deformed state.} 
\textbf{(a)} Box plot of $\Gamma_0 / U_{\text{chain}}$ for notched networks formed by linear chains (100 samples). 
\textbf{(b)} Sample with the highest intrinsic fracture energy. 
\textbf{(c)} Sample with median intrinsic fracture energy. 
\textbf{(d)} Sample with the lowest intrinsic fracture energy. 
\textbf{(e)} Box plot of $\Gamma_0 / U_{\text{chain}}$ for notched networks formed by non-linear chains (100 samples). 
\textbf{(f)} Sample with the largest intrinsic fracture energy in the deformed state. 
\textbf{(g)} Sample with median intrinsic fracture energy. 
\textbf{(h)} Sample with the smallest intrinsic fracture energy. The networks consist of 80 horizontal layers and 20 vertical layers.}
  \label{fig-supp:SI_Fig48_Random_50_20_deformed}
\end{figure}

\vspace*{0pt}
\begin{figure}[H]
  \centering
  \includegraphics[trim={0cm 0cm 0cm 0cm},clip, width=1.0\textwidth]{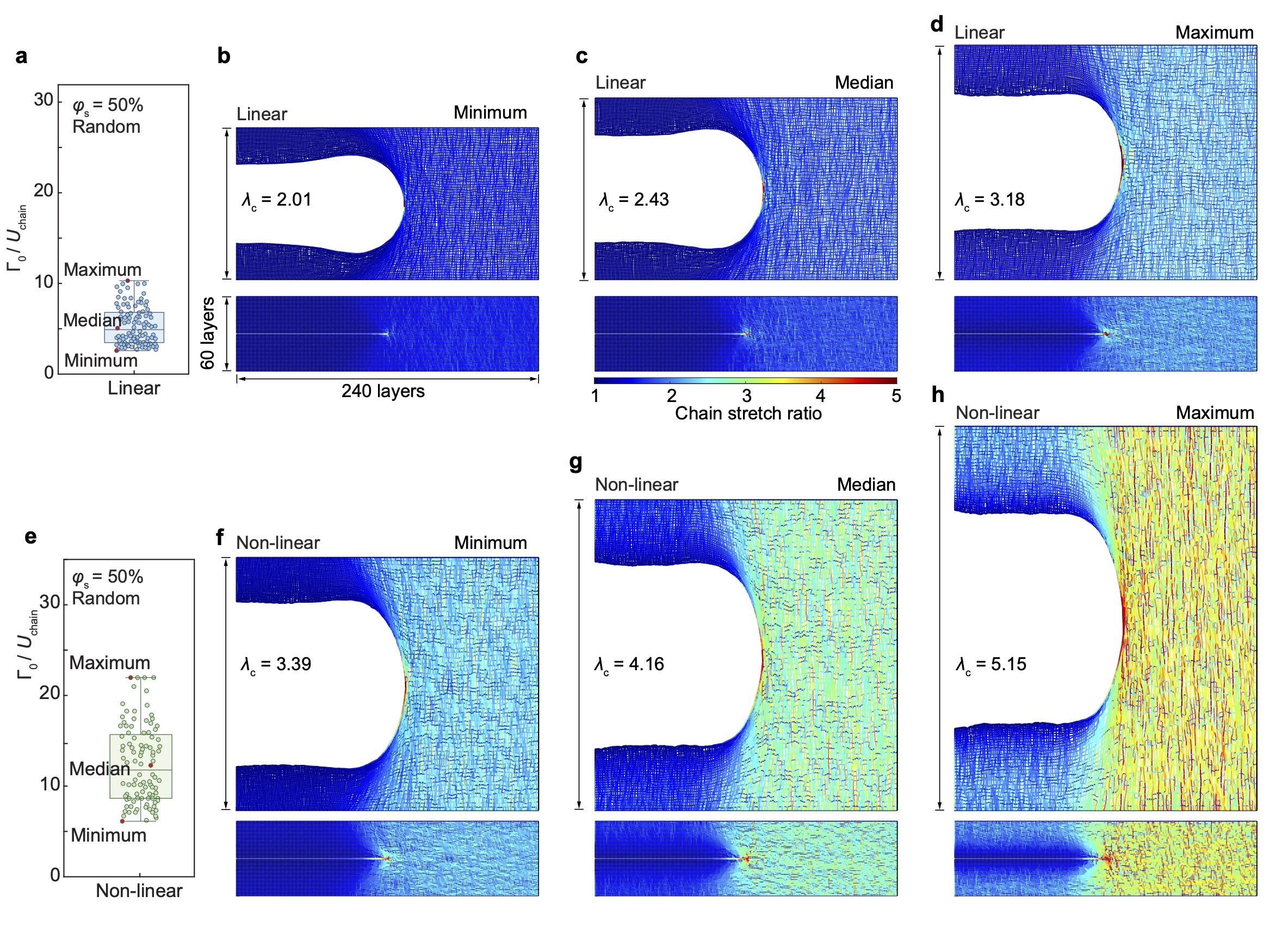}
  \caption{\textbf{Stretch ratio distribution of notched random entangled networks ($\varphi_s = 50\%$) at crack initialization, shown in deformed and undeformed states.} 
\textbf{(a)} Box plot of $\Gamma_0 / U_{\text{chain}}$ for notched networks formed by linear chains (100 samples). 
\textbf{(b)} Stretch ratio distribution of the sample with the highest intrinsic fracture energy. 
\textbf{(c)} Sample with median intrinsic fracture energy. 
\textbf{(d)} Sample with the lowest intrinsic fracture energy. 
\textbf{(e)} Box plot of $\Gamma_0 / U_{\text{chain}}$ for notched networks formed by non-linear chains (100 samples). 
\textbf{(f)} Stretch ratio distribution of the sample with the highest intrinsic fracture energy. 
\textbf{(g)} Sample with median intrinsic fracture energy. 
\textbf{(h)} Sample with the lowest intrinsic fracture energy. The networks consist of 240 horizontal layers and 60 vertical layers.}
  \label{fig-supp:SI_Fig49_Random_50_60_deformed}
\end{figure}

\vspace*{0pt}
\begin{figure}[H]
  \centering
  \includegraphics[trim={0cm 0cm 0cm 0cm},clip, width=1.0\textwidth]{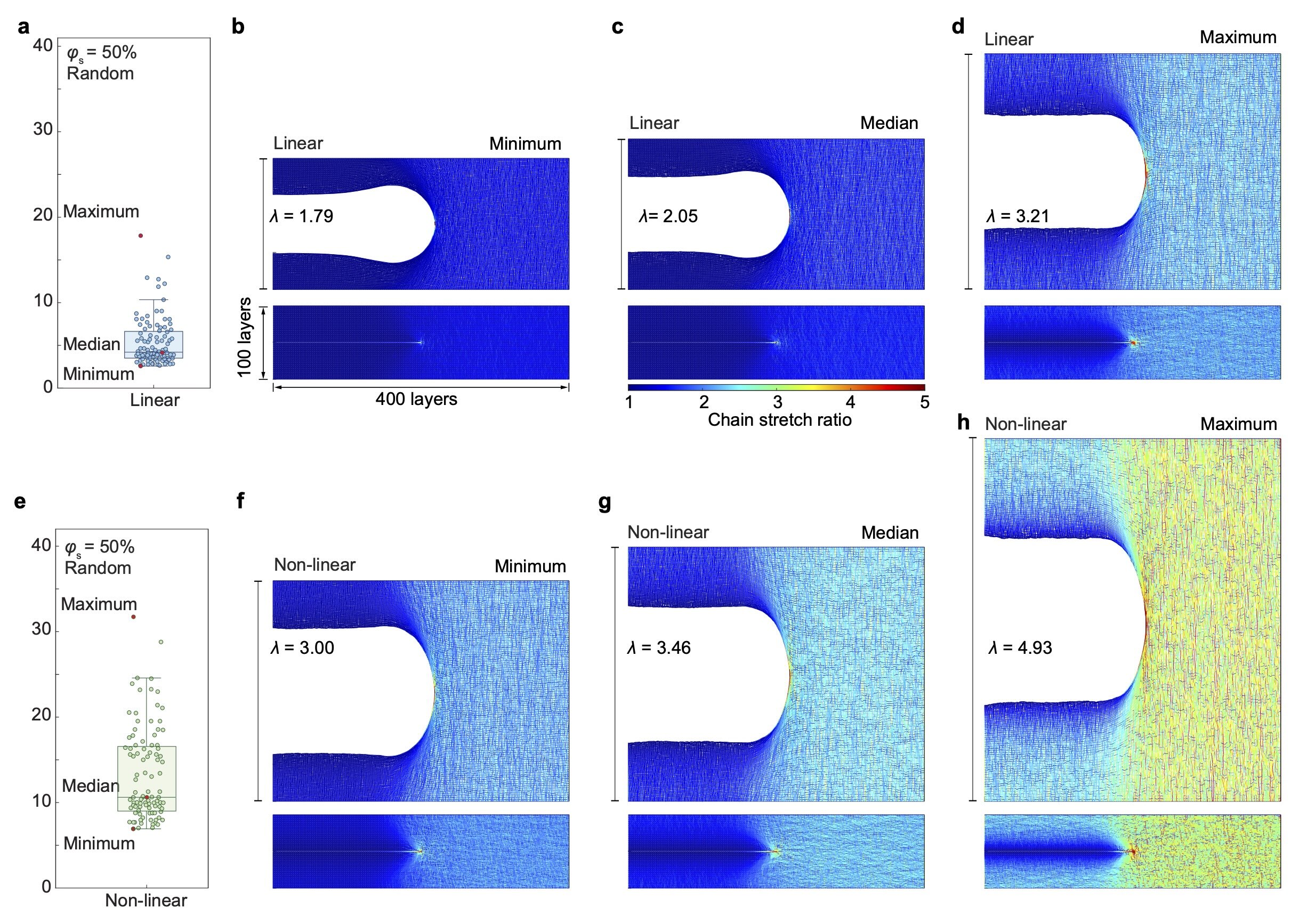}
  \caption{\textbf{Stretch ratio distribution of notched random entangled networks ($\varphi_s = 50\%$) at crack initialization, shown in deformed and undeformed states.} 
\textbf{(a)} Box plot of $\Gamma_0 / U_{\text{chain}}$ for notched networks formed by linear chains (100 samples). 
\textbf{(b)} Stretch ratio distribution of the sample with the highest intrinsic fracture energy. 
\textbf{(c)} Sample with median intrinsic fracture energy. 
\textbf{(d)} Sample with the lowest intrinsic fracture energy. 
\textbf{(e)} Box plot of $\Gamma_0 / U_{\text{chain}}$ for notched networks formed by non-linear chains (100 samples). 
\textbf{(f)} Stretch ratio distribution of the sample with the highest intrinsic fracture energy. 
\textbf{(g)} Sample with median intrinsic fracture energy. 
\textbf{(h)} Sample with the lowest intrinsic fracture energy. The networks consist of 400 horizontal layers and 100 vertical layers.}
  \label{fig-supp:SI_Fig50_Random_50_100_deformed}
\end{figure}

\vspace*{0pt}
\begin{figure}[H]
  \centering
  \includegraphics[trim={0cm 0cm 0cm 0cm},clip, width=1.0\textwidth]{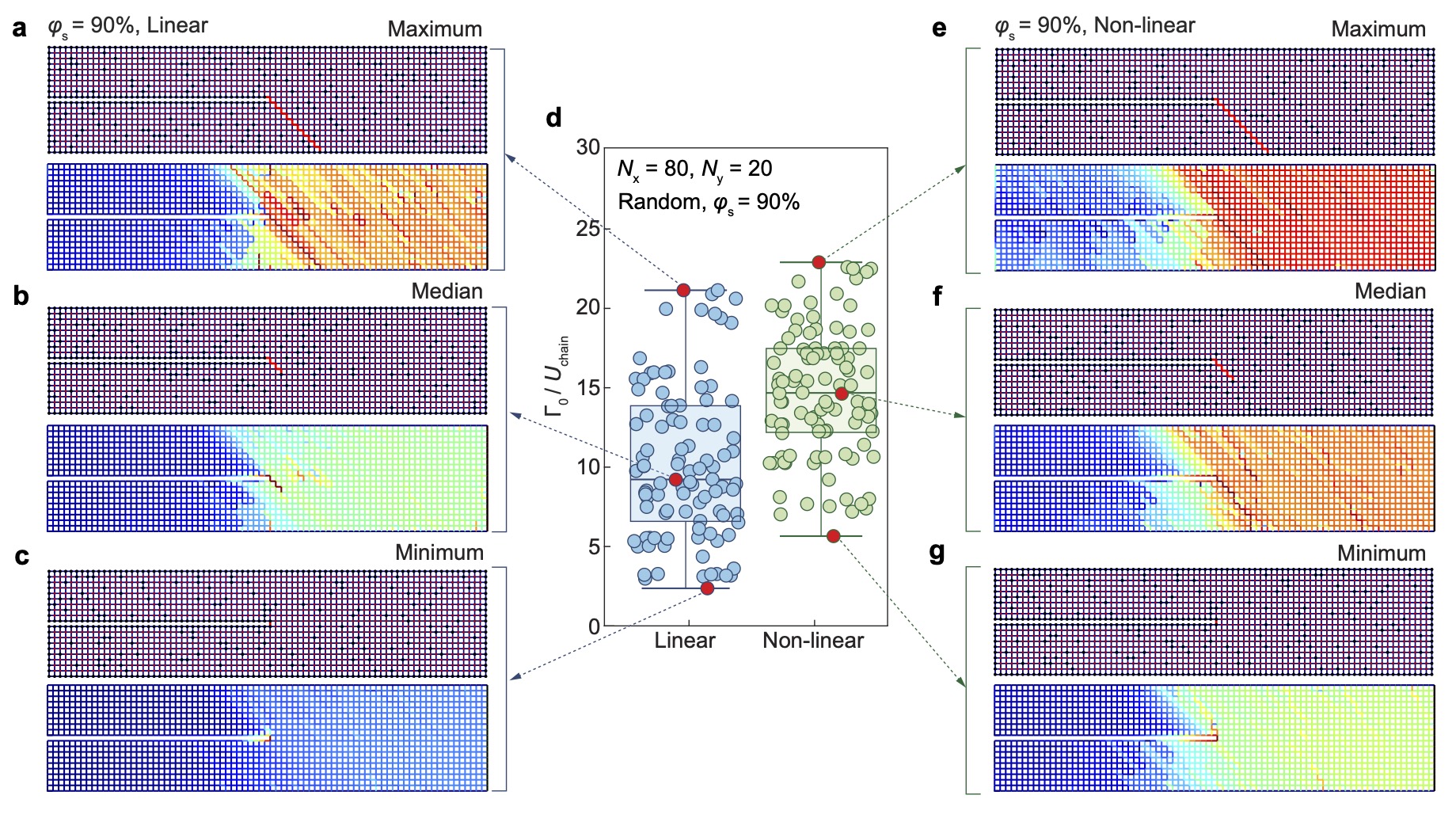}
  \caption{\textbf{Initial configuration and stretch ratio distribution of notched random entangled networks with $\varphi_s = 90\%$ at crack initialization, shown in undeformed state.}
\textbf{(a)} Sample with the highest intrinsic fracture energy in a network formed by linear chains (crack-tip chain length $l_{tip}$ = 20). 
\textbf{(b)} Sample with median intrinsic fracture energy ($l_{tip}$ = 6). 
\textbf{(c)} Sample with the lowest intrinsic fracture energy ($l_{tip}$ = 1). 
\textbf{(d)} Box plot of $\Gamma_0 / U_{\text{chain}}$ for notched networks formed by linear and non-linear chains (100 samples each). Each box shows the interquartile range, with the line inside the box indicating the median. Whiskers extend to the most extreme data points within 1.5 times the interquartile range from the lower and upper quartiles.
\textbf{(e)} Sample with the highest intrinsic fracture energy in a network formed by non-linear chains ($l_{tip}$ = 20). 
\textbf{(f)} Sample with median intrinsic fracture energy ($l_{tip}$ = 8). 
\textbf{(g)} Sample with the lowest intrinsic fracture energy ($l_{tip}$ = 1). The networks consist of 80 horizontal layers and 20 vertical layers.}
  \label{fig-supp:SI_Fig51_Random_90_20_undeformed}
\end{figure}

\vspace*{0pt}
\begin{figure}[H]
  \centering
  \includegraphics[trim={0cm 0cm 0cm 0cm},clip, width=0.9\textwidth]{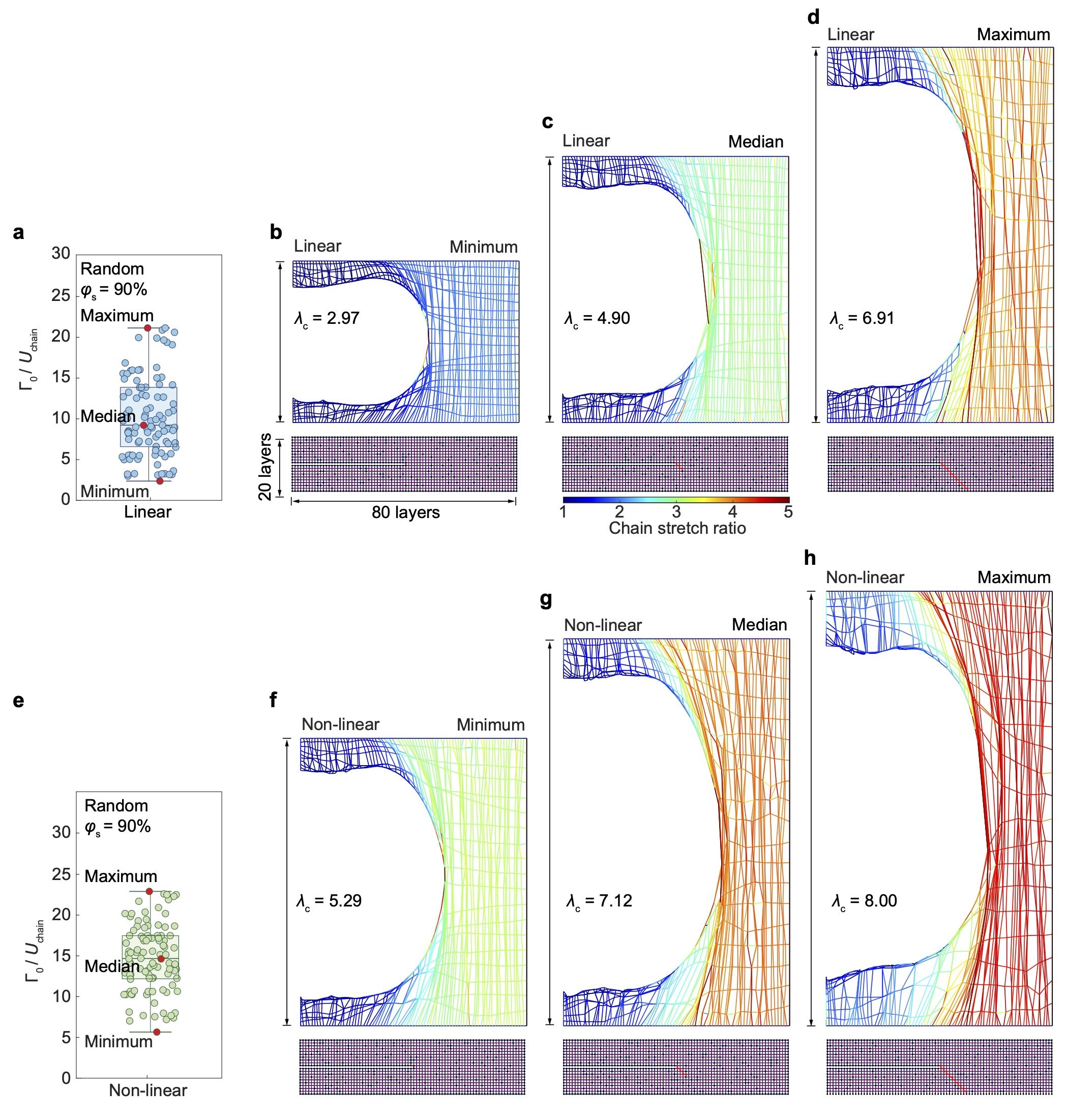}
  \caption{\textbf{Initial configuration and stretch ratio distribution of notched random entangled networks with $\varphi_s = 90\%$ at crack initialization, shown in deformed state.}
\textbf{(a)} Box plot of $\Gamma_0 / U_{\text{chain}}$ for notched networks with linear chains (100 samples). 
\textbf{(b)} Sample with the highest value. 
\textbf{(c)} Sample with median value. 
\textbf{(d)} Sample with the lowest value. 
\textbf{(e)} Box plot of $\Gamma_0 / U_{\text{chain}}$ for notched networks with non-linear chains (100 samples). 
\textbf{(f)} Sample with the highest value. 
\textbf{(g)} Sample with median value. 
\textbf{(h)} Sample with the lowest value. The networks consist of 80 horizontal layers and 20 vertical layers.}
  \label{fig-supp:SI_Fig52_Random_90_20_deformed}
\end{figure}

\vspace*{0pt}
\begin{figure}[H]
  \centering
  \includegraphics[trim={0cm 0cm 0cm 0cm},clip, width=0.96\textwidth]{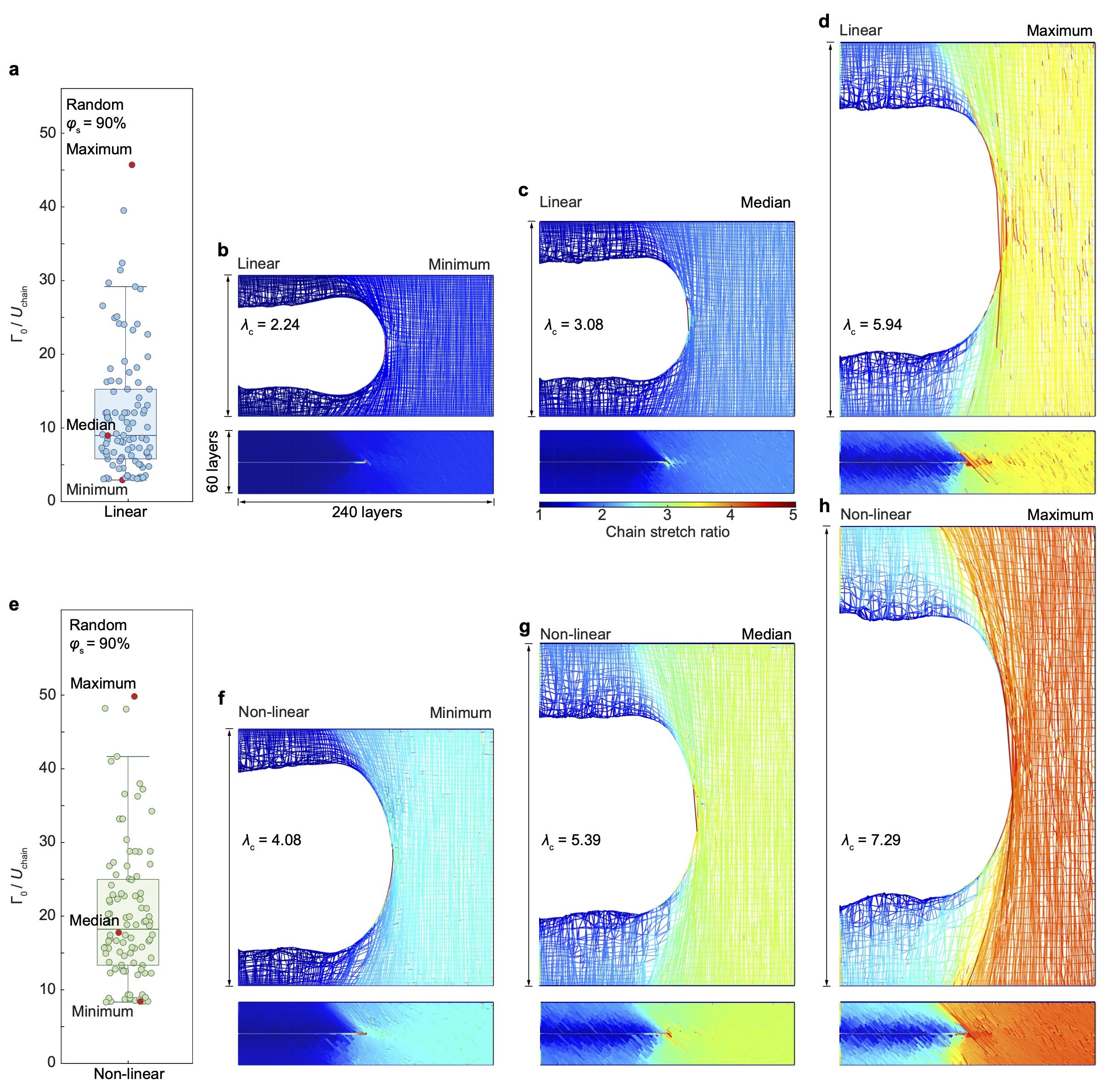}
  \caption{\textbf{Stretch ratio distribution of notched random entangled networks ($\varphi_s = 90\%$) at crack initialization, shown in deformed and undeformed states.}
\textbf{(a)} Box plot of $\Gamma_0 / U_{\text{chain}}$ for notched networks formed by linear chains (100 samples). 
\textbf{(b)} Sample with the highest intrinsic fracture energy. 
\textbf{(c)} Sample with median intrinsic fracture energy. 
\textbf{(d)} Sample with the lowest intrinsic fracture energy. 
\textbf{(e)} Box plot of $\Gamma_0 / U_{\text{chain}}$ for notched networks formed by non-linear chains (100 samples). 
\textbf{(f)} Sample with the highest intrinsic fracture energy. 
\textbf{(g)} Sample with median intrinsic fracture energy. 
\textbf{(h)} Sample with the lowest intrinsic fracture energy. The networks consist of 240 horizontal layers and 60 vertical layers.}
  \label{fig-supp:SI_Fig53_Random_90_60_deformed}
\end{figure}

\vspace*{0pt}
\begin{figure}[H]
  \centering
  \includegraphics[trim={0cm 0cm 0cm 0cm},clip, width=1.0\textwidth]{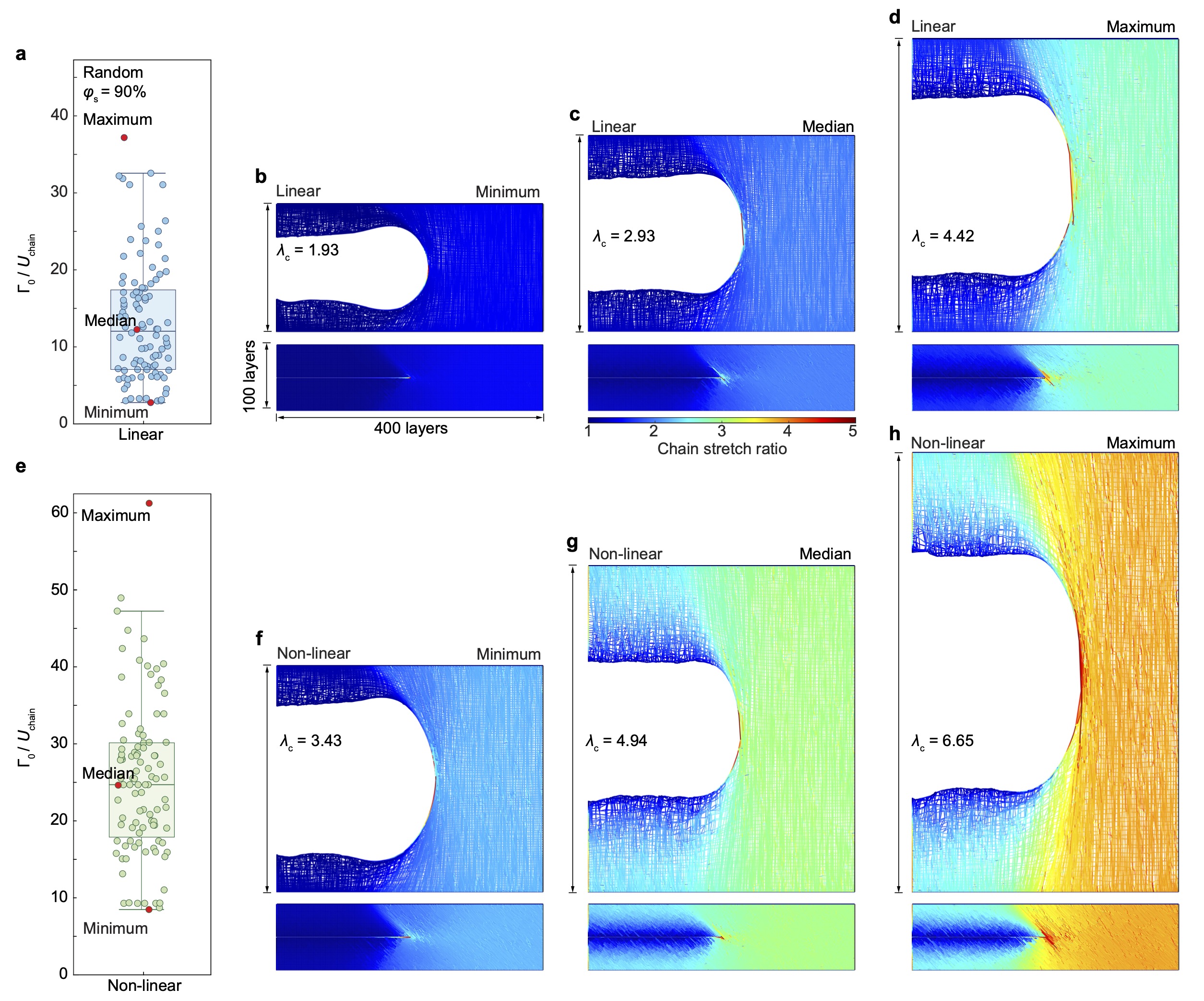}
  \caption{\textbf{Stretch ratio distribution of notched random entangled networks ($\varphi_s = 90\%$) at crack initialization, shown in deformed and undeformed states.} 
\textbf{(a)} Box plot of $\Gamma_0 / U_{\text{chain}}$ for notched networks formed by linear chains (100 samples). 
\textbf{(b)} Sample with the highest intrinsic fracture energy. 
\textbf{(c)} Sample with median intrinsic fracture energy. 
\textbf{(d)} Sample with the lowest intrinsic fracture energy. 
\textbf{(e)} Box plot of $\Gamma_0 / U_{\text{chain}}$ for notched networks formed by non-linear chains (100 samples). 
\textbf{(f)} Sample with the highest intrinsic fracture energy. 
\textbf{(g)} Sample with median intrinsic fracture energy. 
\textbf{(h)} Sample with the lowest intrinsic fracture energy. The networks consist of 400 horizontal layers and 100 vertical layers.}
  \label{fig-supp:SI_Fig54_Random_90_100_deformed}
\end{figure}

\vspace*{0pt}
\begin{figure}[H]
  \centering
  \includegraphics[trim={0cm 0cm 0cm 0cm},clip, width=1.0\textwidth]{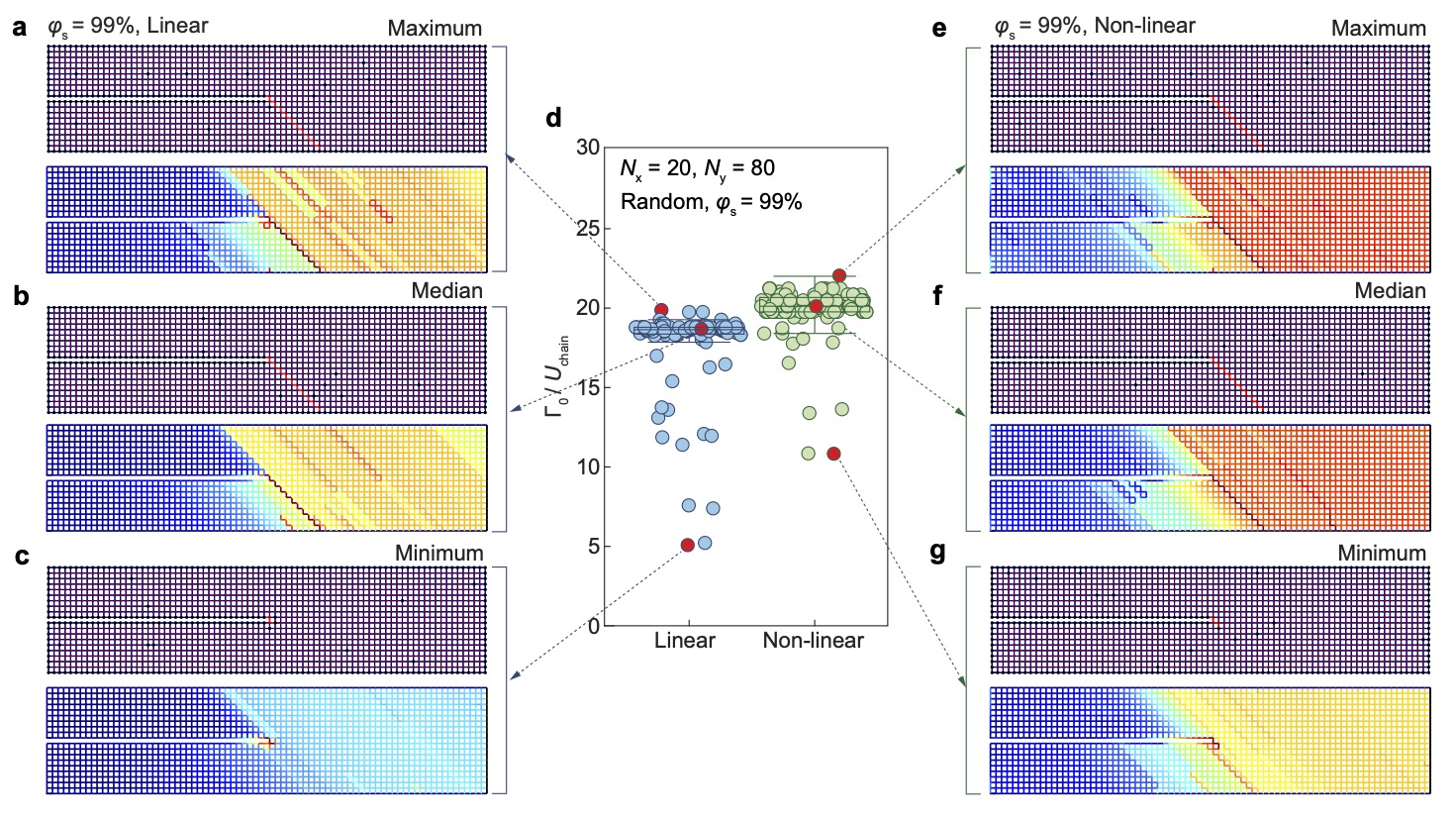}
  \caption{\textbf{Initial configuration and stretch ratio distribution of notched random entangled networks with $\varphi_s = 99\%$ at crack initialization, shown in undeformed state.}
\textbf{(a)} Sample with the highest intrinsic fracture energy in a network formed by linear chains (crack-tip chain length $l_{tip}$= 20). 
\textbf{(b)} Sample with median intrinsic fracture energy ($l_{tip}$ = 20). 
\textbf{(c)} Sample with the lowest intrinsic fracture energy ($l_{tip}$ = 3). 
\textbf{(d)} Box plot of $\Gamma_0 / U_{\text{chain}}$ for networks formed by linear and non-linear chains (100 samples each). Each box shows the interquartile range, with the line inside the box indicating the median. Whiskers extend to the most extreme data points within 1.5 times the interquartile range from the lower and upper quartiles.
\textbf{(e)} Sample with the highest intrinsic fracture energy in a network formed by non-linear chains ($l_{tip}$ = 20). 
\textbf{(f)} Sample with median intrinsic fracture energy ($l_{tip}$ = 20). 
\textbf{(g)} Sample with the lowest intrinsic fracture energy ($l_{tip}$ = 4). The network consists of 80 layers in length and 20 layers in width.}
  \label{fig-supp:SI_Fig55_Random_99_20_undeformed}
\end{figure}

\vspace*{0pt}
\begin{figure}[H]
  \centering
  \includegraphics[trim={0cm 0cm 0cm 0cm},clip, width=0.9\textwidth]{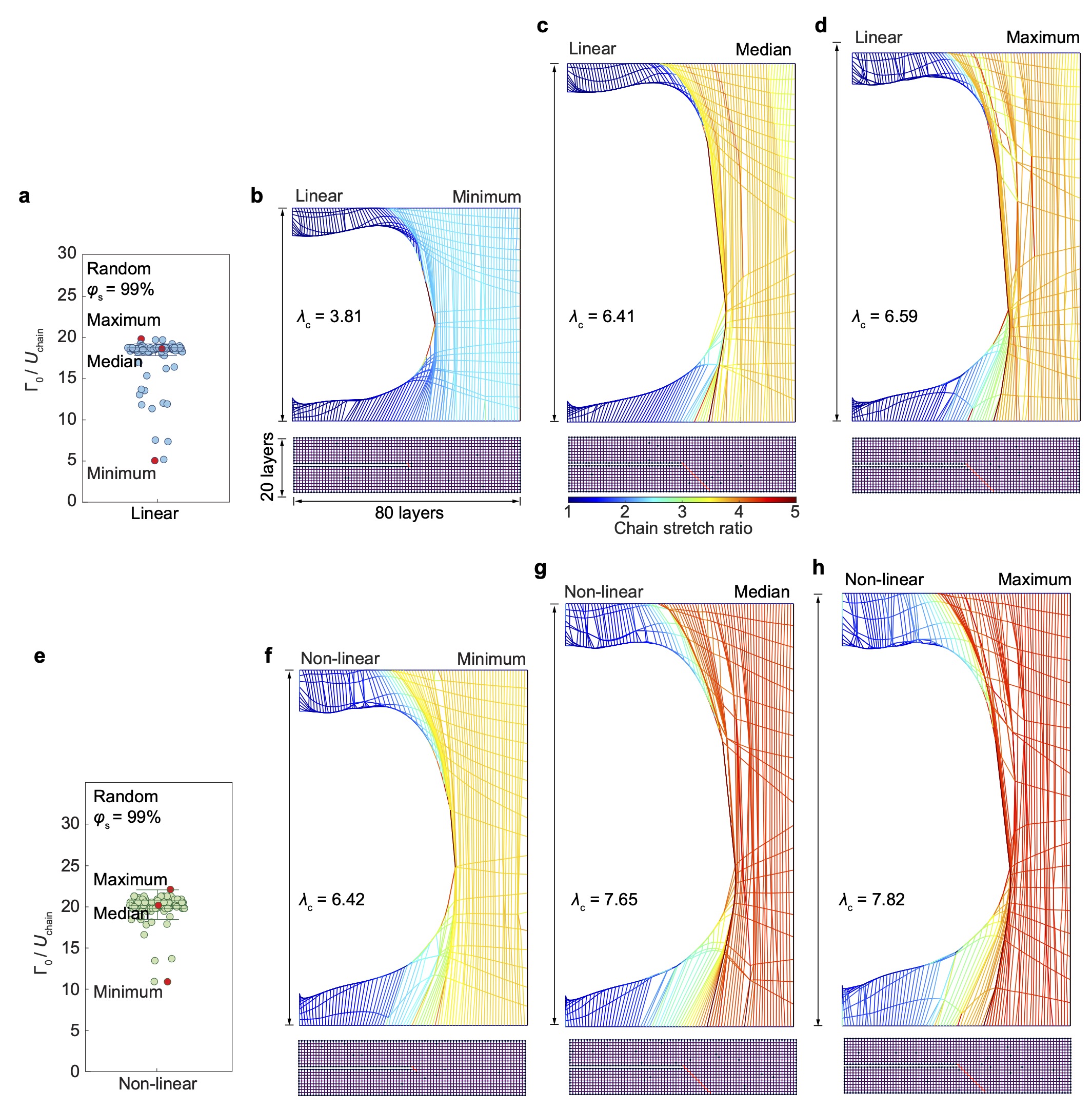}
  \caption{\textbf{Initial configuration and stretch ratio distribution of notched random entangled networks with $\varphi_s = 99\%$ at crack initialization, shown in deformed state.}
\textbf{(a)} Box plot of $\Gamma_0 / U_{\text{chain}}$ for notched networks formed by linear chains (100 samples). 
\textbf{(b)} Sample with the highest intrinsic fracture energy. 
\textbf{(c)} Sample with median intrinsic fracture energy. 
\textbf{(d)} Sample with the lowest intrinsic fracture energy. 
\textbf{(e)} Box plot of $\Gamma_0 / U_{\text{chain}}$ for notched networks formed by non-linear chains (100 samples). 
\textbf{(f)} Sample with the highest intrinsic fracture energy. 
\textbf{(g)} Sample with median intrinsic fracture energy. 
\textbf{(h)} Sample with the lowest intrinsic fracture energy. The networks consist of 80 horizontal layers and 20 vertical layers.}
  \label{fig-supp:SI_Fig56_Random_99_20_deformed}
\end{figure}

\vspace*{0pt}
\begin{figure}[H]
  \centering
  \includegraphics[trim={0cm 0cm 0cm 0cm},clip, width=0.88\textwidth]{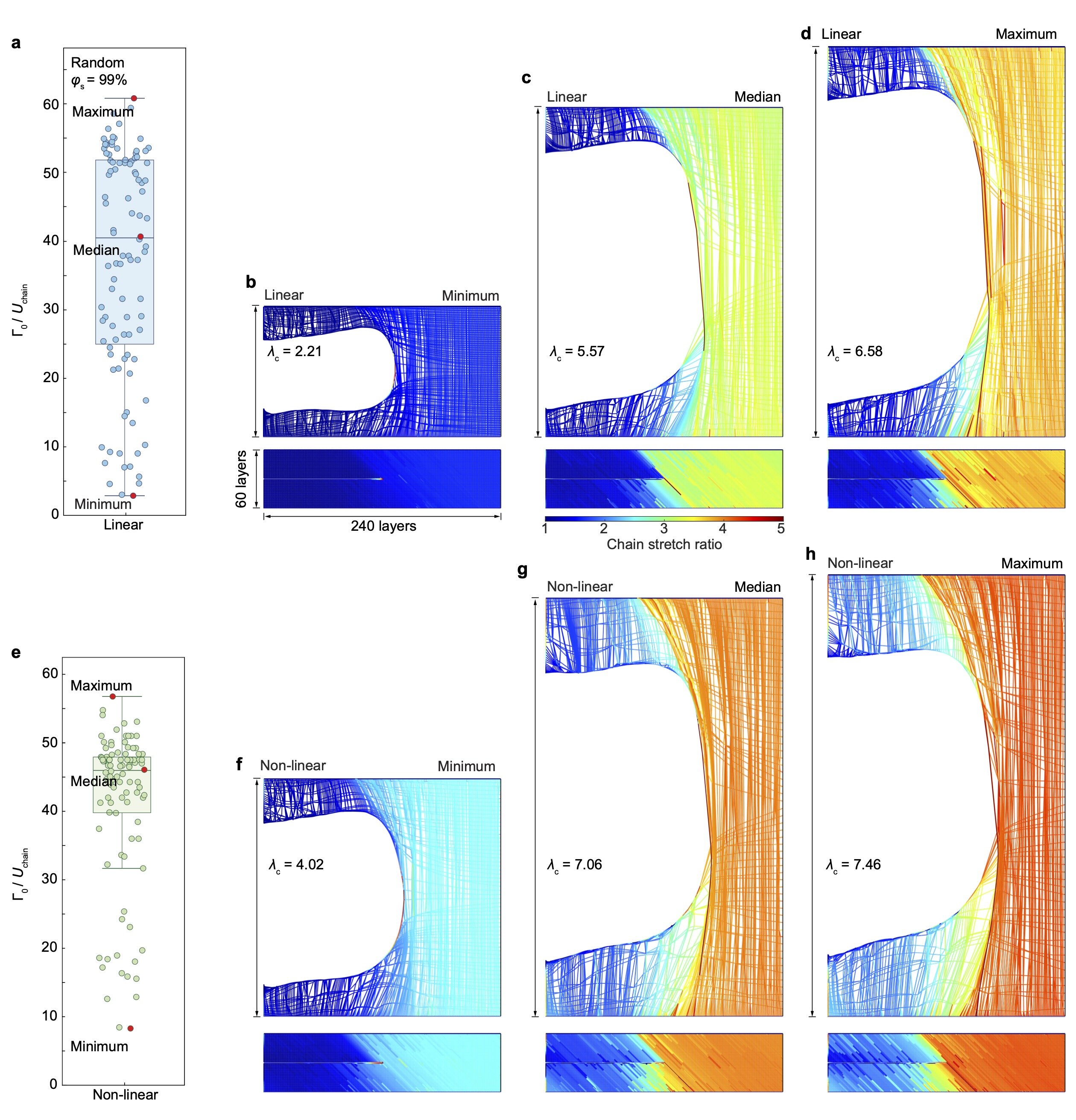}
  \caption{\textbf{Stretch ratio distribution of notched random entangled networks ($\varphi_s = 99\%$) at crack initialization, shown in deformed and undeformed states.} 
\textbf{(a)} Box plot of $\Gamma_0 / U_{\text{chain}}$ for notched networks formed by linear chains (100 samples). 
\textbf{(b)} Sample with the highest intrinsic fracture energy. 
\textbf{(c)} Sample with median intrinsic fracture energy. 
\textbf{(d)} Sample with the lowest intrinsic fracture energy. 
\textbf{(e)} Box plot of $\Gamma_0 / U_{\text{chain}}$ for notched networks formed by non-linear chains (100 samples). 
\textbf{(f)} Sample with the highest intrinsic fracture energy. 
\textbf{(g)} Sample with median intrinsic fracture energy. 
\textbf{(h)} Sample with the lowest intrinsic fracture energy. The networks consist of 240 horizontal layers and 60 vertical layers.}
  \label{fig-supp:SI_Fig57_Random_99_60_deformed}
\end{figure}

\vspace*{0pt}
\begin{figure}[H]
  \centering
  \includegraphics[trim={0cm 0cm 0cm 0cm},clip, width=0.88\textwidth]{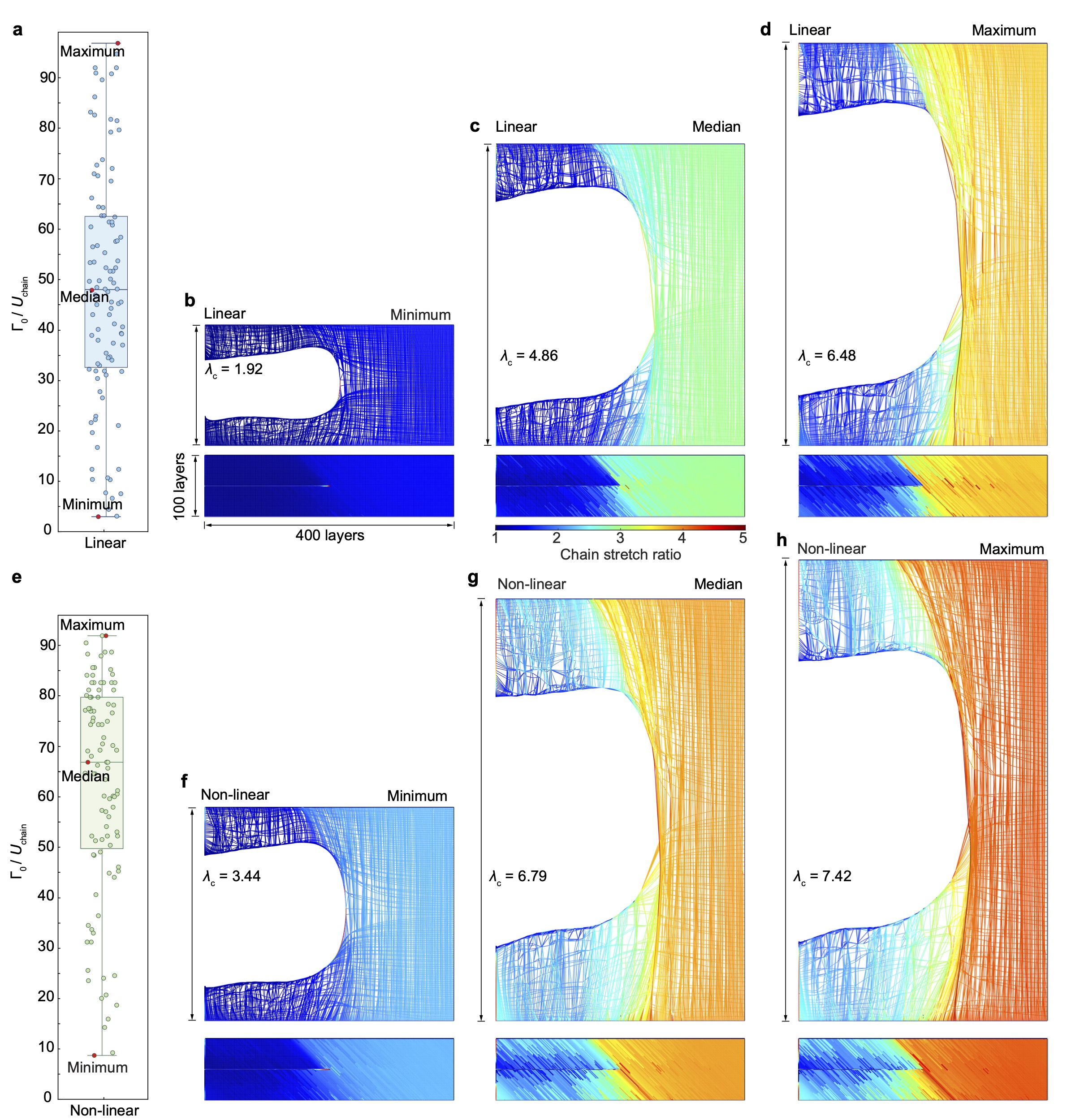}
  \caption{\textbf{Stretch ratio distribution of notched random entangled networks ($\varphi_s = 99\%$) at crack initialization, shown in deformed and undeformed states.}
\textbf{(a)} Box plot of $\Gamma_0 / U_{\text{chain}}$ for notched networks formed by linear chains (100 samples). 
\textbf{(b)} Sample with the largest intrinsic fracture energy. 
\textbf{(c)} Sample with median intrinsic fracture energy. 
\textbf{(d)} Sample with the smallest intrinsic fracture energy. 
\textbf{(e)} Box plot of $\Gamma_0 / U_{\text{chain}}$ for notched networks formed by non-linear chains (100 samples). 
\textbf{(f)} Sample with the largest intrinsic fracture energy. 
\textbf{(g)} Sample with median intrinsic fracture energy. 
\textbf{(h)} Sample with the smallest intrinsic fracture energy. The networks consist of 400 horizontal layers and 100 vertical layers.}
  \label{fig-supp:SI_Fig58_Random_99_100_deformed}
\end{figure}

\vspace*{0pt}
\begin{figure}[H]
  \centering
  \includegraphics[trim={0cm 0cm 0cm 0cm},clip, width=1.0\textwidth]{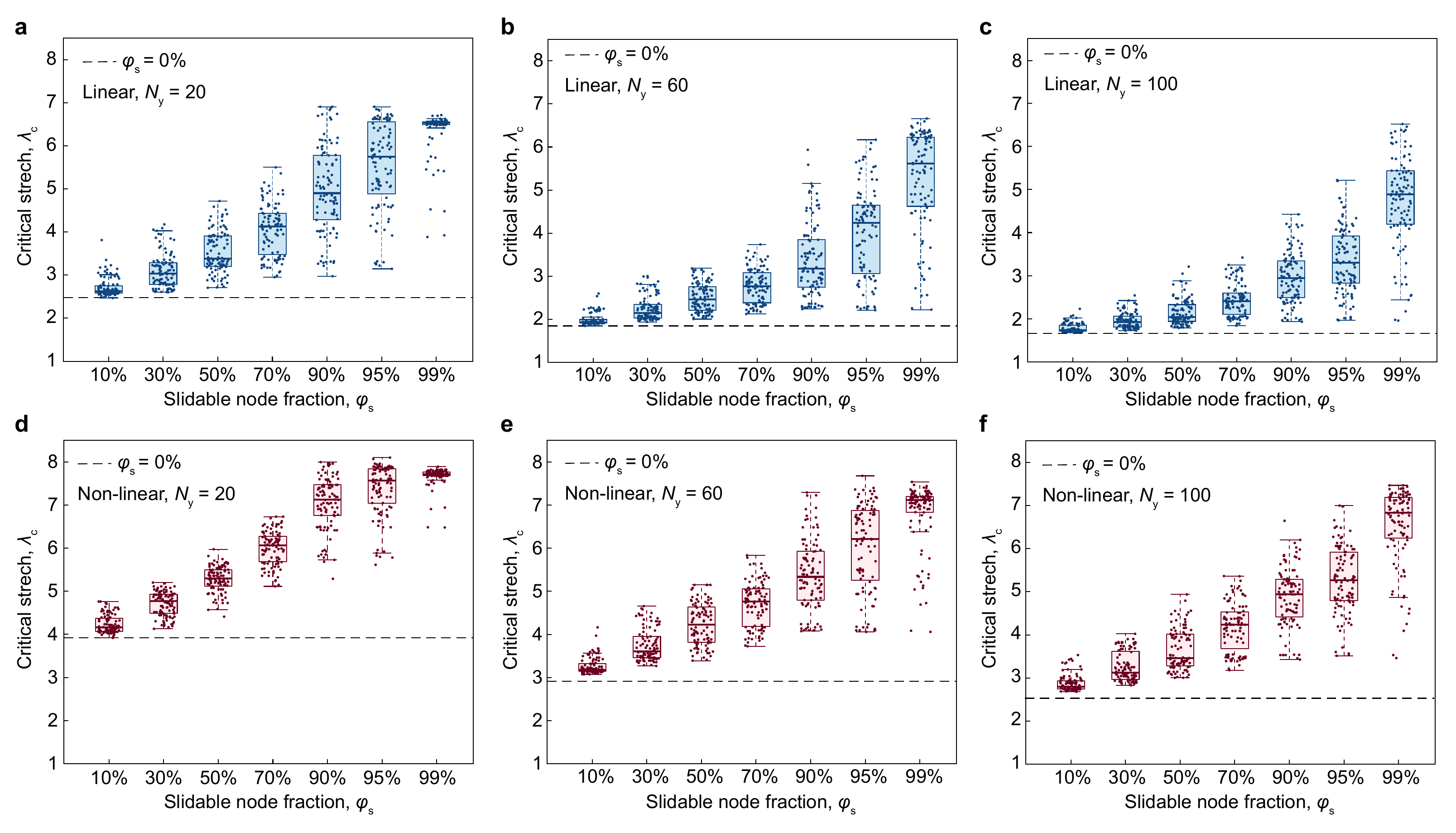}
  \caption{\textbf{Box plot of critical stretch ratio $\lambda_c$ 
  versus slidable node fraction $\varphi_s$ for notched random entangled networks.}
\textbf{(a-c)} Networks formed by linear chains with the vertical layer numbers $N_y = 20$, $60$, and $100$. 
\textbf{(d-f)} Networks formed by non-linear chains with the vertical layer numbers $N_y = 20$, $60$, and $100$. The networks have a fixed aspect ratio of $N_x/N_y = 4$ with $N_x$ being the horizontal layer number.
Each box shows the interquartile range, with the line inside the box indicating the median. Whiskers extend to the most extreme data points within 1.5 times the interquartile range from the lower and upper quartiles. Each case includes 100 samples. The dashed line in each figure corresponds to spring networks ($\varphi_s=0\%$). The critical stretch ratio $\lambda_c$ consistently increases with $\varphi_s$ across different vertical layer numbers.}
  \label{fig-supp:SI_Fig59_Critical_stretch_random_box_plot}
\end{figure}

\vspace*{0pt}
\begin{figure}[H]
  \centering
  \includegraphics[trim={0cm 0cm 0cm 0cm},clip, width=1.0\textwidth]{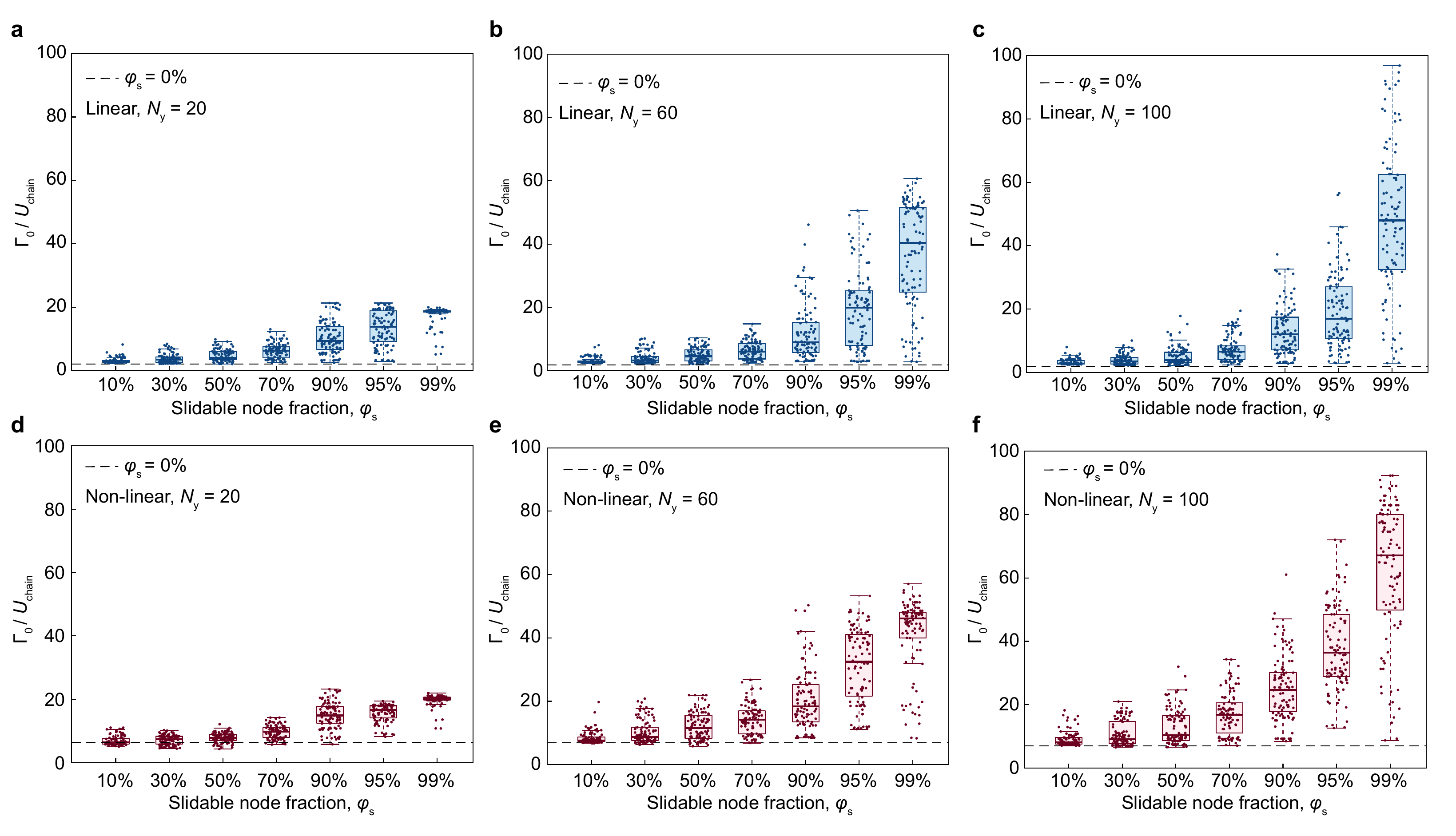}
  \caption{\textbf{Box plot of $\Gamma_0 / U_{\text{chain}}$ versus slidable node fraction $\varphi_s$ for notched random entangled networks.} 
\textbf{(a-c)} Networks formed by linear chains with the vertical layer numbers $N_y = 20$, $60$, and $100$. 
\textbf{(d-f)} Networks formed by non-linear chains with the vertical layer numbers $N_y = 20$, $60$, and $100$. The networks have a fixed aspect ratio of $N_x/N_y = 4$ with $N_x$ being the horizontal layer number.
Each box shows the interquartile range, with the line inside the box indicating the median. Whiskers extend to the most extreme data points within 1.5 times the interquartile range from the lower and upper quartiles. Each case includes 100 samples. The dashed line in each figure corresponds to spring networks ($\varphi_s=0\%$). The fracture energy consistently increases with $\varphi_s$ across different vertical layer numbers.}
  \label{fig-supp:SI_Fig60_Gamma0_random_box_plot}
\end{figure}

\vspace*{0pt}
\begin{figure}[H]
  \centering
  \includegraphics[trim={0cm 0cm 0cm 0cm},clip, width=1.0\textwidth]{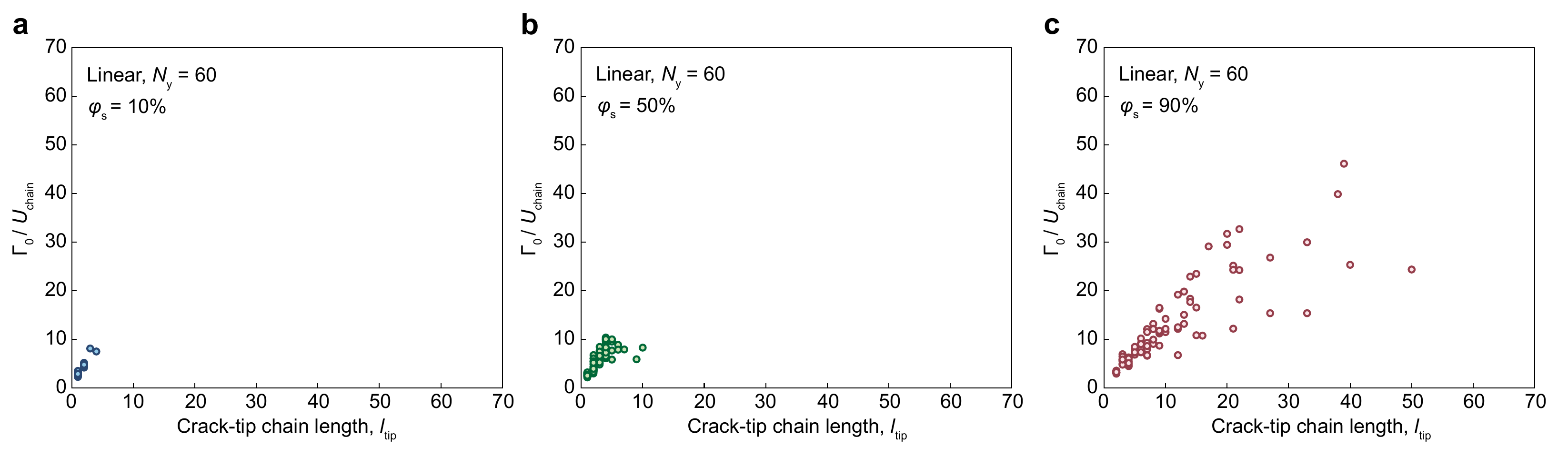}
  \caption{\textbf{$\Gamma_0 / U_{\text{chain}}$ versus crack-tip chain length $l_{\text{tip}}$ for notched random entangled networks.} 
\textbf{(a)} Network with slidable node fraction of $\varphi_s = 10\%$. 
\textbf{(b)} Network with slidable node fraction of $\varphi_s = 50\%$. 
\textbf{(c)} Network with slidable node fraction of $\varphi_s = 90\%$. The network is formed by linear chains. Each subfigure includes 100 samples. The networks consist of 240 horizontal layers and 60 vertical layers.}
  \label{fig-supp:SI_Fig61_ChainlengthVSGamma0_60_linear}
\end{figure}

\vspace*{0pt}
\begin{figure}[H]
  \centering
  \includegraphics[trim={0cm 0cm 0cm 0cm},clip, width=1.0\textwidth]{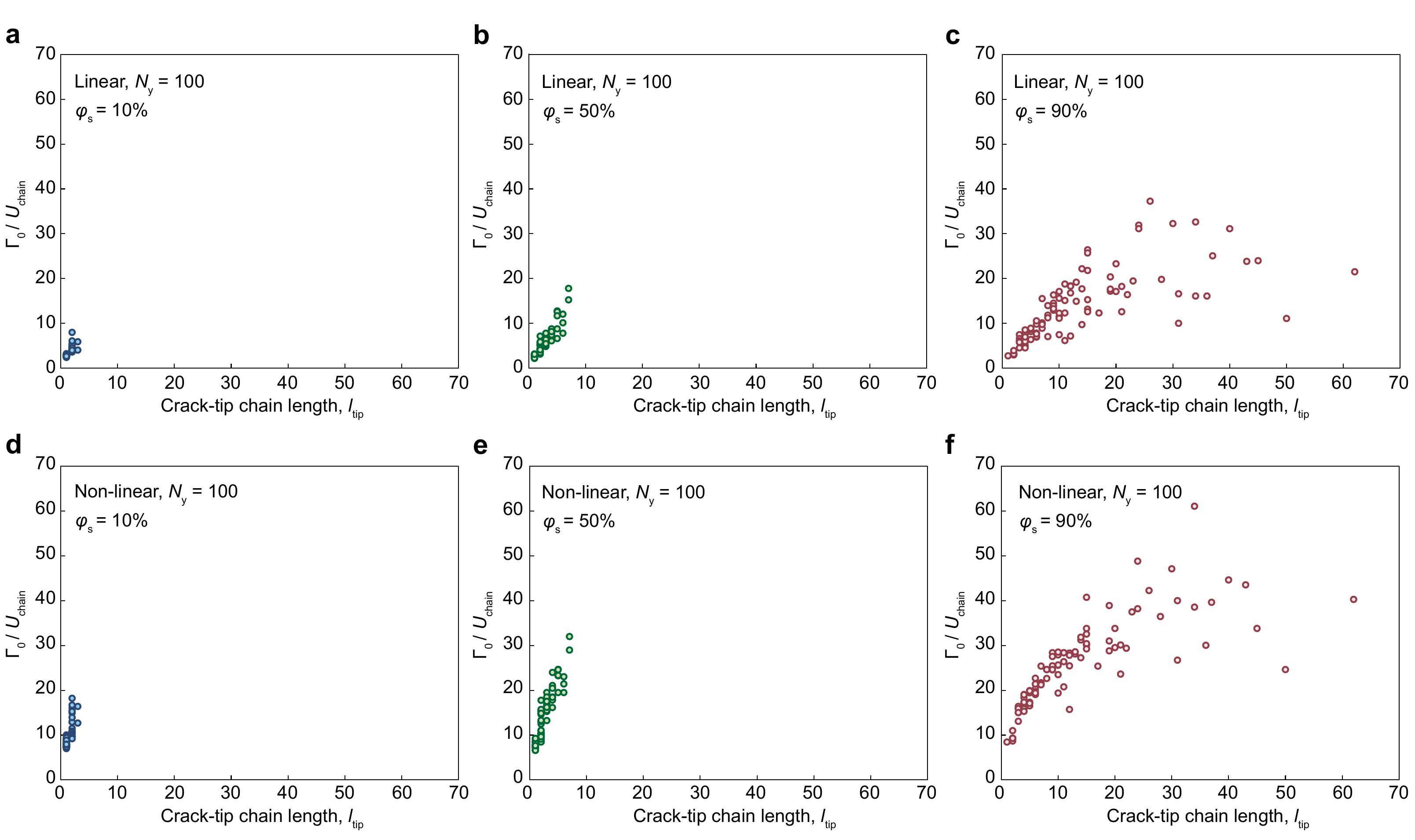}
  \caption{\textbf{$\Gamma_0 / U_{\textrm{chain}}$ versus crack-tip chain length $l_{\textrm{tip}}$ for notched random entangled networks.}
\textbf{(a-c)} Networks formed by linear chains with slidable node fraction of $\varphi_s = 10\%$, $50\%$, and $90\%$, respectively. 
\textbf{(d-f)} Networks formed by non-linear chains with slidable node fraction of $\varphi_s = 10\%$, $50\%$, and $90\%$, respectively.The network is formed by
either linear or non-linear chains. Each subfigure includes 100 samples. The networks consist of 400 horizontal layers and 100 vertical layers.}
  \label{fig-supp:SI_Fig62_ChainlengthVSGamma0_100}
\end{figure}

\vspace*{0pt}
\begin{figure}[H]
  \centering
  \includegraphics[trim={0cm 0cm 0cm 0cm},clip, width=1.0\textwidth]{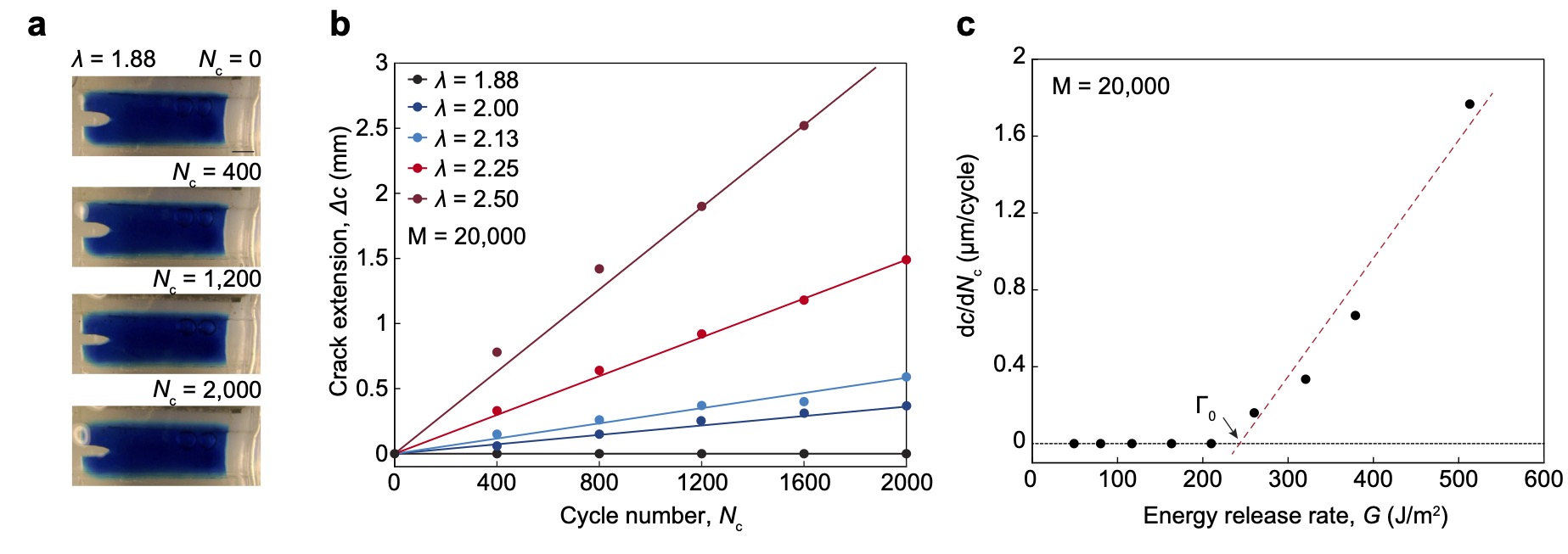}
  \caption{\textbf{Fatigue test of an entangled hydrogel with a monomer-to-crosslinker ratio of M = 20{,}000.} 
\textbf{(a)} Images of a notched hydrogel stretched to a ratio of 1.88 at various cycle numbers $N_c$, showing no visible crack growth. 
\textbf{(b)} Crack extension length versus cycle number at different stretch ratios. 
\textbf{(c)} Crack growth rate per cycle ($dc/dN_c$) versus energy release rate $G$. To prevent dehydration, the experiments were conducted in a silicone oil bath.}
  \label{fig-supp:SI_Fig63_fatigue_20K}
\end{figure}

\vspace*{0pt}
\begin{figure}[H]
  \centering
  \includegraphics[trim={0cm 0cm 0cm 0cm},clip, width=1.0\textwidth]{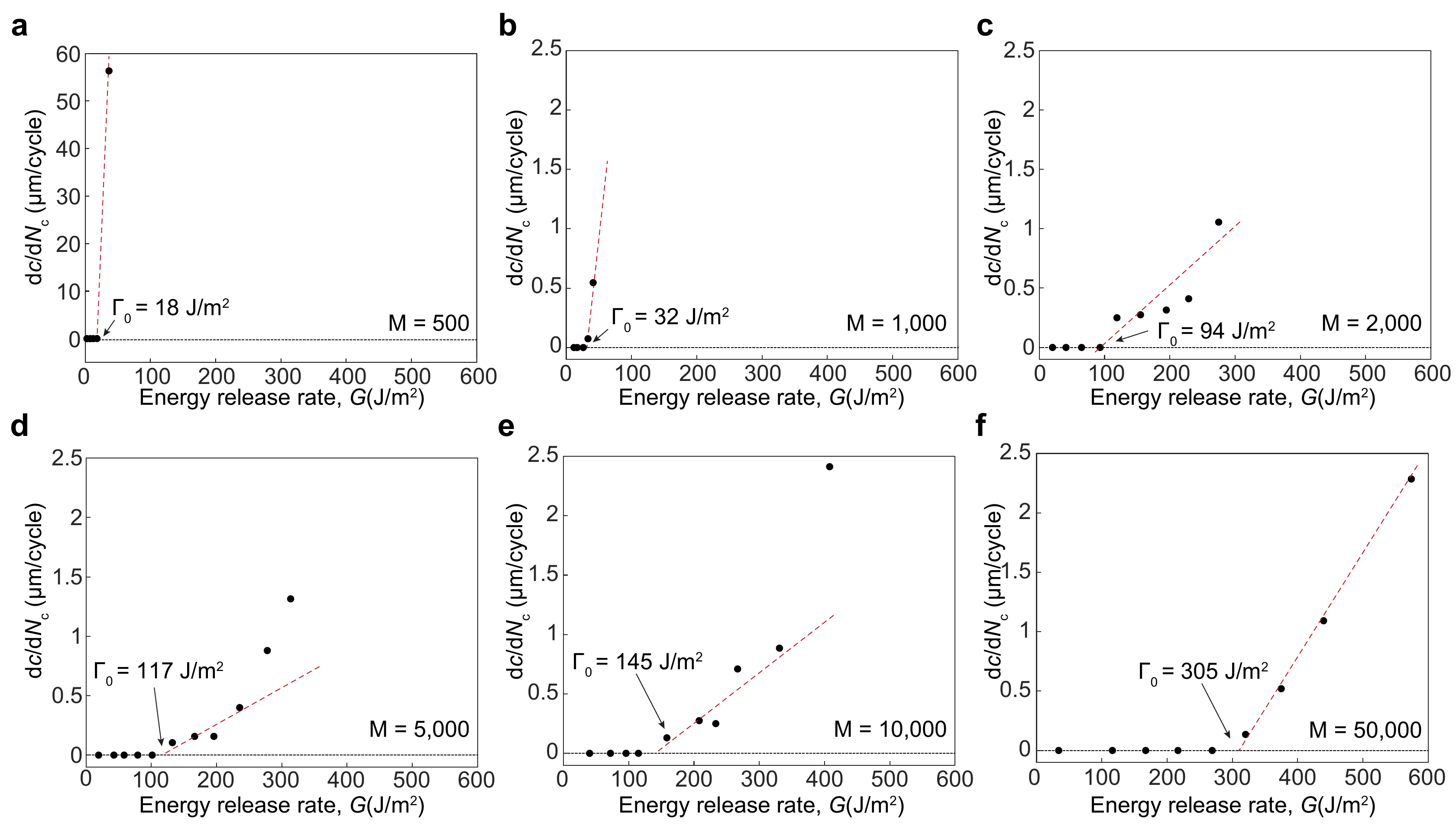}
  \caption{\textbf{Crack growth rate ($dc/dN_c$) versus energy release rate $G$ for entangled hydrogels with different monomer-to-crosslinker molar ratios M.} 
\textbf{(a)} M = 500. 
\textbf{(b)} M = 1{,}000. 
\textbf{(c)} M = 2{,}000. 
\textbf{(d)} M = 5{,}000. 
\textbf{(e)} M = 10{,}000. 
\textbf{(f)} M = 50{,}000.}
  \label{fig-supp:SI_Fig64_fatigue_all}
\end{figure}

\newpage
\section*{Movie Captions}

\noindent\textbf{Caption for Movie S1.} Photoelastic experiment and simulation results of the two-chain model with a non-slidable node (spring network). 

\bigskip

\noindent\textbf{Caption for Movie S2.} Photoelastic experiment and simulation results of the two-chain model with a slidable node (entangled network).

\bigskip

\noindent\textbf{Caption for Movie S3.} Photoelastic experiments of spring network and entangled network.

\bigskip

\noindent\textbf{Caption for Movie S4.} Photoelastic experiments of notched spring network and entangled network.

\end{document}